\documentclass[UTF8,11pt,a4paper,oneside,scheme=plain,fontset=fandol]{ctexart}
\usepackage[margin=27mm]{geometry}
\usepackage{amsmath,amssymb,amsthm,mathtools,mathrsfs,bm}
\usepackage{booktabs,array,enumitem}
\usepackage[unicode,colorlinks=true,linkcolor=blue,citecolor=blue,urlcolor=blue]{hyperref}
\allowdisplaybreaks
\ctexset{
  paragraph={runin=false,beforeskip=1.5ex plus .3ex minus .2ex,afterskip=.5ex plus .1ex},
  subparagraph={runin=false,beforeskip=1.3ex plus .3ex minus .2ex,afterskip=.5ex plus .1ex}
}
\newcommand{\papertheoremname}{Theorem}
\newtheorem{theorem}{\papertheoremname}[section]
\newcommand{\paperlemmaname}{Lemma}
\newtheorem{lemma}[theorem]{\paperlemmaname}
\newcommand{\paperpropositionname}{Proposition}
\newtheorem{proposition}[theorem]{\paperpropositionname}
\newcommand{\papercorollaryname}{Corollary}
\newtheorem{corollary}[theorem]{\papercorollaryname}
\newcommand{\paperdefinitionname}{Definition}
\theoremstyle{definition}
\newtheorem{definition}[theorem]{\paperdefinitionname}
\newcommand{\paperexamplename}{Example}
\newtheorem{example}[theorem]{\paperexamplename}
\newcommand{\paperassumptionname}{Assumption}
\newtheorem{assumption}[theorem]{\paperassumptionname}
\newcommand{\paperremarkname}{Remark}
\theoremstyle{remark}

\renewcommand{\proofname}{Proof}
\numberwithin{equation}{section}
\newcommand{\Holant}{\operatorname{Holant}}

\newcommand{\supp}{\operatorname{supp}}
\newcommand{\diag}{\operatorname{diag}}
\newcommand{\arity}{\operatorname{arity}}
\newcommand{\Tr}{\operatorname{Tr}}
\newcommand{\FP}{\mathrm{FP}}

\newcommand{\sharpP}{\#\mathrm P}
\newcommand{\Bgroup}{\mathcal B}
\newcommand{\Ggroup}{\mathcal G}
\newcommand{\Kfour}{K_4}
\newcommand{\Qeight}{Q_8}
\newcommand{\Kstd}{\mathcal K_{\mathrm{std}}}
\newcommand{\Ktw}{\mathcal K_{\mathrm{tw}}}

\newcommand{\ket}[1]{\lvert#1\rangle}

\newcommand{\paperlanguage}{en}
\renewcommand{\theHsection}{\paperlanguage.\thesection}
\let\paperoriginalappendix\appendix
\renewcommand{\appendix}{%
  \paperoriginalappendix
  \renewcommand{\theHsection}{\paperlanguage.\Alph{section}}%
}
\makeatletter
\let\paperoriginaladdtocontents\addtocontents
\renewcommand{\addtocontents}[2]{%
  \def\paperrequestedfile{#1}\def\papertocfile{toc}%
  \ifx\paperrequestedfile\papertocfile
    \paperoriginaladdtocontents{\paperlanguage toc}{#2}%
  \else
    \paperoriginaladdtocontents{#1}{#2}%
  \fi
}
\renewcommand{\tableofcontents}{%
  \section*{\contentsname}%
  \edef\papertocfilename{\paperlanguage toc}%
  \expandafter\@starttoc\expandafter{\papertocfilename}%
}
\renewcommand{\maketitle}{%
  \newpage\null\vskip 2em
  \begin{center}%
    {\LARGE\@title\par}\vskip 1.5em
    {\large\lineskip .5em\begin{tabular}[t]{c}\@author\end{tabular}\par}%
    \vskip 1em{\large\@date\par}%
  \end{center}\par\vskip 1.5em
}
\makeatother
\hypersetup{pdftitle={Unifying Dichotomy Theorems: The Klein Group Upper Case},
  pdfauthor={Mingji Xia},pdflang={en}}
\begin{document}
\pdfbookmark[0]{English version}{language-en}
\title{Unifying Dichotomy Theorems: The Klein Group Upper Case\\
Version 4: The Twisted Klein Upper Case and the Road to Completion}
\author{Mingji Xia\\Institute of Software, Chinese Academy of Sciences\\University of Chinese Academy of Sciences}
\date{September 2026}

\maketitle
\begin{abstract}
Within the nine-group classification framework, this paper continues the
study of the Klein group upper case in the third version. Transpose closure
classifies the concrete matrix representations of the Klein group into
the standard and twisted forms. For the standard form, we retain the
previous version's approach through realnumberizing and the Real-Holant
dichotomy. The main text of the fourth version supplies a direct proof for
the twisted form. It completes the dichotomy for the Klein group upper
case, in which every quaternary function is decomposable; the lower case
is not treated in the main text.

The appendix develops a route to completion based on the recent dichotomy
theorem for problems with binary disequality\cite{en:liu2026disequality}.
It determines which groups yield binary disequality, and then presents
the proofs for the cyclic group $C_1$ (see the second version), the cyclic
groups $C_k$ (see the first version), and odd dihedral groups. For the
cyclic group $C_2$ of order two, we cite the independent draft of
Yuxuan Zhou without reproducing its proof. Subject to the cited results,
these four group classes, together with the classes covered by the binary
disequality dichotomy\cite{en:liu2026disequality}, give dichotomies for all
nine group classes and thereby reach the final dichotomy theorem.
In particular, the new theorem covers the entire Klein group case,
including its lower case.
\end{abstract}
\tableofcontents
\clearpage

\section{Background and organization of this version}
Fix a finite set $\mathcal F$ of complex-valued Boolean functions.
Using its functions to form function-nets gives the counting problem
$\#\mathcal F$ of evaluating these function-nets. The original general
framework divides the finite-group cases into nine classes according to
the projective group associated with the nonzero realizable binary
functions, with the aim of establishing a complexity dichotomy for each
class. The framework also distinguishes two higher-arity cases: the
upper case, in which every nonzero ordinarily realizable quaternary
function decomposes into two binary factors, and the lower case,
in which an indecomposable quaternary function has already appeared.

The third version\cite{en:xia2026v3} focuses on the Klein group upper case.
It first uses transpose closure to identify the two concrete matrix
representations, and then analyzes common phases and realnumberizing
for the standard form. The main text of this version follows that order;
the central addition in the fourth version is the proof for the twisted
form. Under the explicit quaternary decomposition condition for the upper
case, simultaneous nonzero contractions first show that a minimal
counterexample in the ordinary gadget closure satisfies the
scattered-pearl condition. We then establish a common pairing,
support propagation, and directional contradictions at arities six and
eight, and finally give a cycle-by-cycle algorithm. This completes the
proof for the Klein group upper case in the main text; the lower case
is not treated there.

Appendix~\ref{en:sec:appendix-road} opens a separate ``Road to Completion.''
The first route uses the recent dichotomy theorem for problems with
binary disequality. It first handles groups that yield this function,
then presents the proofs for cyclic and odd dihedral groups.
The second route continues the direct proofs within the original
nine-group framework. This version records that plan without developing
its technical details.

\section{Preliminaries and reductions}\label{en:sec:foundations-zh}

\subsection{Function-nets, net-functions, and associativity}

All functions in this paper have Boolean arguments and algebraic complex values.
Fix a finite function set $\mathcal F$. The functions themselves and their arities
are not part of the input. Given a finite graph, place at each vertex a function
whose arity equals the degree of that vertex, and specify the correspondence
between its ports and the arguments of the function. The resulting object is a
\emph{function-net}, usually called a tensor network. Self-loops and parallel
edges are allowed; a self-loop occupies two distinct ports of the same vertex.

Each internal edge carries a Boolean variable. The two ports of an edge use the
same variable, so the binary kernel of an ordinary edge is equality. Multiplying
the vertex functions and summing over all internal edge variables gives the value
of the function-net. If external ports remain, sum only over internal variables;
the resulting function of the external variables is its \emph{net-function}.
The external ports are ordered. In particular, exchanging the two external ports
of a binary net-function transposes its matrix.

\begin{proposition}[Associativity of function-nets]
Evaluating a part of a function-net first and replacing that part by its
net-function does not change the net-function or value of the whole.
\end{proposition}
\begin{proof}
Partition the internal variables into those internal to the chosen part and all
others. The two expressions differ only in the order of finite summation and the
application of distributivity of multiplication over addition, so they agree.
\end{proof}

The problem of evaluating closed function-nets whose vertex functions belong to
$\mathcal F$ is denoted by $\#\mathcal F$, or $\Holant(\mathcal F)$.
With a fixed binary edge kernel $E$, write $\Holant_E(\mathcal F)$; the ordinary
model has $E=I$. This notation distinguishes an edge kernel from a vertex
function: every edge in $\Holant_N(\mathcal F)$ uses $N$, but this does not assert
that an $N$ vertex has been realized in the ordinary model.

\subsection{Gadgets, decomposition, and established dichotomies}

If $g$ is the net-function of a fixed finite gadget using only vertex functions
from $\mathcal F$, call $g$ \emph{ordinarily realizable} from $\mathcal F$.
Associativity allows this gadget to replace any number of $g$ vertices, giving
a polynomial-time reduction from $\Holant(\mathcal F\cup\{g\})$ to
$\Holant(\mathcal F)$. More generally, if adjoining $g$ gives a problem
polynomial-time Turing reducible to the original one, call $g$
\emph{Turing available}. Interpolation may provide only this latter form of
availability. Throughout, $\equiv_{\mathrm T}$ denotes polynomial-time Turing
equivalence, and $\sharpP$-hardness is with respect to these reductions.

For two disjoint sets of variables, define the tensor product by
\[
 (f\otimes g)(x,y)=f(x)g(y).
\]
A nonzero function is \emph{decomposable} if it can be written as such a product
on two nonempty variable sets, and \emph{indecomposable} otherwise. Variables
may first be permuted. Once a decomposition is obtained, we use the following
theorem to obtain its factors.

\begin{theorem}[Decomposition theorem, \cite{en:meng2025odd,en:guan2026lp}]
\label{en:thm:decomposition}\label{en:thm:decomposition-fp}
Let $\mathcal F$ be a fixed finite set of algebraic complex-valued Boolean
functions, and let $f,g$ be nonzero algebraic complex-valued Boolean functions.
At least one of the following holds:
\[
 \Holant(\mathcal F\cup\{f,g\})
 \equiv_{\mathrm T}
 \Holant(\mathcal F\cup\{f\otimes g\}),
\]
or $\Holant(\mathcal F\cup\{f\otimes g\})\in\FP$.
\end{theorem}

\begin{theorem}[Dichotomy with a nonzero odd-arity function,
\cite{en:meng2025odd,en:guan2026lp}]
\label{en:thm:odd-arity-quasi}\label{en:thm:odd-arity-fp}
If a fixed finite set $\mathcal H$ of algebraic complex-valued Boolean functions
contains a nonzero odd-arity function, then $\Holant(\mathcal H)$ is in $\FP$
or is $\sharpP$-hard.
\end{theorem}

We use the decomposition and classification results of the odd-arity paper,
together with the subsequent polynomial-time EO algorithm, to upgrade the
tractable side to $\FP$. Following the framework of the first version, whenever
a decomposable function appears we apply the decomposition theorem: the
tractable alternative is settled, and otherwise we adjoin its tensor factors
and continue. A nonzero odd-arity factor gives the odd-arity dichotomy.
The same reasoning applies to the fixed $N$-edge model: transform the whole
problem back to equality edges, perform the decomposition and reductions, and
then transform back. The group classification below includes factors obtained
in this way.

\begin{lemma}[Nonzero scaling of individual functions]\label{en:lem:per-signature-scaling}
For a fixed finite function set $\mathcal F$, choose a nonzero algebraic constant
$c_f$ for each $f\in\mathcal F$. If $\mathcal F'=\{c_ff:f\in\mathcal F\}$,
then $\Holant(\mathcal F')\equiv_{\mathrm T}\Holant(\mathcal F)$.
The same conclusion holds with any fixed edge kernel.
\end{lemma}
\begin{proof}
If an instance contains $n_f$ vertices of type $f$, the scaling multiplies its
value by $\prod_f c_f^{n_f}$. This nonzero factor can be computed exactly and
compensated. Use $c_f^{-1}$ for the reverse reduction.
\end{proof}

We also call a function available when a known nonzero scalar multiple is
available, compensating the scalar by this lemma. The function set is fixed,
so each auxiliary gadget, scaling constant, or reduction is fixed in advance
and does not vary arbitrarily with the input function-net.

\begin{theorem}[Real Boolean Holant dichotomy, \cite{en:shao2020real}]
\label{en:thm:real-holant}
Let $\mathcal H$ be a fixed finite set of algebraic real-valued Boolean
functions. If $\mathcal H$ satisfies the decidable tractability conditions of
the Shao--Cai Real Boolean Holant dichotomy theorem, then
$\Holant(\mathcal H)\in\FP$; otherwise it is $\sharpP$-hard.
We refer to the full collection of tractability conditions in that theorem as
condition \textup{(T)}.
\end{theorem}

In the standard Klein case, we first scale each function to make it algebraic
real-valued, and then test the resulting function set against condition
\textup{(T)}.

\subsection{Matrix notation and changes of basis}

Write a binary function as a $2\times2$ matrix according to the order of its
two arguments. Throughout, fix
\[
 I=\begin{pmatrix}1&0\\0&1\end{pmatrix},\qquad
 N=\begin{pmatrix}0&1\\1&0\end{pmatrix},\qquad
 Y=\begin{pmatrix}0&1\\-1&0\end{pmatrix},
\]
and
\[
 X=iN=\begin{pmatrix}0&i\\i&0\end{pmatrix},\qquad
 Z=i\diag(1,-1)=\begin{pmatrix}i&0\\0&-i\end{pmatrix}.
\]
Here $N$ is binary disequality and $Y$ is skew-symmetric. The names $X,Z$ are
interchanged relative to the third version; all formulas here use the displayed
convention. Serial connections of binary vertices give matrix multiplication,
and exchanging their two ports gives the ordinary transpose $A^{\mathsf T}$,
without conjugation. A complex orthogonal matrix satisfies $O^{\mathsf T}O=I$;
function-net summation does not use a conjugate inner product.

\begin{proposition}[Holographic transformations with an edge kernel]
\label{en:prop:holographic-reduction}
Let $P$ be an invertible algebraic matrix. List each $m$-ary vertex function's
values in a fixed argument order and define
$\widehat f=(P^{-1})^{\otimes m}f$. Then
\[
 \Holant_E(\mathcal F)
 \equiv_{\mathrm T}
 \Holant_{P^{\mathsf T}EP}(\widehat{\mathcal F}).
\]
Corresponding instances have equal values. A binary vertex function $A$ becomes
$\widehat A=P^{-1}AP^{-\mathsf T}$. In particular, if $E=I$ and $P=O$ is
complex orthogonal, the edge kernel remains $I$ and the binary matrix becomes
$O^{\mathsf T}AO$.
\end{proposition}
\begin{proof}
Insert $PP^{-1}$ at every port. First combine $P^{-1}$ with each vertex function,
and combine the two adjacent copies of $P$ with each old edge kernel $E$.
These give the stated new vertex functions and edge kernel. Associativity on
every edge preserves the value of the whole function-net. The reverse
transformation uses $P^{-1}$.
\end{proof}

For example, take
\[
 K=\frac1{\sqrt2}\begin{pmatrix}1&1\\i&-i\end{pmatrix},
 \qquad K^{\mathsf T}K=N.
\]
Applying this $K$ transformation to the ordinary model gives edge kernel $N$
and vertex functions $(K^{-1})^{\otimes m}f$. A binary vertex becomes
$K^{-1}AK^{-\mathsf T}$; combining it with an $N$ edge gives the transfer
matrix $K^{-1}AK$.

\subsection{Finite binary groups and the nine cases}

We first make precise the ``all and only'' convention of the first
version~\cite{en:xia2026v1}. Let $L$ be the number field generated by all values of
the original finite function set $\mathcal F$. Starting from $\mathcal F$ and
the ordinary two-port wire $I$, allow finitely many gadget constructions,
tensor-product decompositions of nonzero functions, and scalings by nonzero
elements of $L$. Denote the resulting set by $\mathscr C(\mathcal F)$ and call
its members the \emph{available functions} of this framework. Decomposition
factors can be chosen over $L$, as proved below. Interpolation reductions are
given separately where they are used.

For each nonzero full-rank binary matrix $A$ in this set, take both representatives
\[
 \frac{A}{s},\qquad s^2=\det A,
\]
and denote the resulting set by $\Bgroup$.

\begin{lemma}[Decomposition factors and the finite binary group]
\label{en:lem:framework-available-closure}
For every finite set $\mathcal H\subseteq\mathscr C(\mathcal F)$, either
$\Holant(\mathcal F)\in\FP$, or
\[
 \Holant(\mathcal F\cup\mathcal H)\equiv_{\mathrm T}\Holant(\mathcal F).
\]
After excluding the tractable alternative of the decomposition theorem and the
established dichotomy cases with a nonzero odd-arity function, a rank-one binary
function, or a binary function of infinite projective order, $\Bgroup$ is a
finite group. All its binary directions can be used simultaneously in a fixed
finite equivalent problem.
\end{lemma}
\begin{proof}
First, decomposition need not enlarge the number field. If a nonzero function
$F(x,y)$ decomposes across two variable sets, choose $F(a,b)\ne0$. Then
\[
 F(x,y)=F(x,b)\,\frac{F(a,y)}{F(a,b)}.
\]
Both factors on the right are over $L$. Gadget constructions use only finite
additions and multiplications and also preserve $L$. Each of finitely many
available functions has a finite derivation. Merge these derivations and adjoin
the required functions in order. Gadgets and nonzero scalings preserve
equivalence. At each decomposition, Theorem~\ref{en:thm:decomposition} either gives
an $\FP$ algorithm for the original problem, or preserves equivalence upon
adjoining both factors. Composing finitely many reductions proves the first
claim, with simultaneous and arbitrarily repeated use of these auxiliaries.

For finiteness, follow the first version. A rank-one binary function decomposes
into nonzero unary functions. A binary matrix of infinite projective order has
a Jordan block or an eigenvalue ratio allowing chain interpolation of a rank-one
matrix. These cases are handled by the odd-arity dichotomy. Every remaining
nonzero binary matrix is full rank and has finite projective order.
Write $d=[L:\mathbb Q]$. If $A$ has projective order $n$, its eigenvalue ratio is
a primitive $n$th root of unity in an extension of $L$ of degree at most two.
Thus $\varphi(n)\le2d$. Only finitely many $n$ satisfy this inequality, and each
normalized matrix has order dividing $2n$, so all these orders have a common
finite multiple. Serial connections and port exchanges give closure under
multiplication and transpose; an inverse is a positive power. Hence $\Bgroup$
is a complex matrix group of finite exponent. The finite-exponent theorem used
in the first version~\cite[Theorem~19A.9(1)]{en:legacy-a} implies that $\Bgroup$ is
finite.

Finally, choose a representative of each binary direction and a finite
derivation of that representative. There are only finitely many directions,
so the first paragraph allows all representatives to be adjoined to one fixed
finite equivalent problem. Its ordinary gadgets remain in
$\mathscr C(\mathcal F)$, and every binary direction has been adjoined, so its
complete ordinary binary group is exactly $\Bgroup$. Any further finite set of
available functions required below can likewise be adjoined simultaneously;
they remain in this closure and therefore create no binary directions outside
the group.
\end{proof}

Fix such an equivalent problem and continue to denote its function set by
$\mathcal F$. The auxiliary functions used in the proof are fixed in advance
and may then be connected in gadgets and shortloops. We classify fixed
function sets by their groups; the finite derivations in the lemma do not
require an element-by-element computation of an infinite closure.

Define the projective quotient group by
\[
 \Ggroup=\Bgroup/\{I,-I\}.
\]
Write $[A]=\{A,-A\}$, with multiplication $[A][B]=[AB]$. The classification of
finite rotation groups determines the isomorphism type of $\Ggroup$; changes
of basis for the actual functions still follow the holographic transformation
above.

Fixing a group both provides all its binary directions as auxiliaries and
forbids other binary directions. A larger group generally provides more
auxiliary power; a smaller group imposes stronger restrictions on the binary
functions, giving stronger self-discipline. These two roles lead to nine cases.

The dihedral rotation group $D_{2n}$ has order $2n$; ``odd'' and ``even'' refer
to the parity of $n$.
\begin{center}
\begin{tabular}{cl}
\toprule
Case & Projective group type\\
\midrule
1 & Trivial cyclic group $C_1$\\
2 & Cyclic group $C_2$ of order two\\
3 & Higher-order cyclic group $C_k$, $k\ge3$\\
4 & Odd dihedral group $D_{2n}$, odd $n\ge3$\\
5 & Klein four-group $D_4$\\
6 & Larger even dihedral group $D_{2n}$, even $n\ge4$\\
7 & Tetrahedral rotation group\\
8 & Octahedral rotation group\\
9 & Icosahedral rotation group\\
\bottomrule
\end{tabular}
\end{center}

Within each group type, distinguish the upper and lower cases according to
whether an indecomposable quaternary function is obtainable. In the upper case,
every nonzero quaternary function decomposes. Odd-arity factors have already
been handled by the odd-arity dichotomy, leaving only a binary-by-binary
decomposition. The decomposition lemma makes both binary factors available,
so their normalized representatives belong to this group. Thus
\emph{the upper-case quaternary functions are exactly binary pairing products
from the group}; the lower case has an indecomposable quaternary function.
In fixed $N$-edge coordinates, use the corresponding binary vertex functions
in these pairing products.

The main text begins with Case~5, follows the third version for the Klein upper
case, and completes the twisted form. Its gadget arguments also apply directly
under the ordinary-closure conditions stated there.

\section{Transpose closure and the two normal forms of the Klein group}\label{en:sec:klein}

The next three sections start from the abstract condition $\Ggroup\cong\Kfour$ supplied by the common framework. The proof has three parts. First, an orthogonal holographic transformation preserving equality edges puts the entire problem into either the standard or the twisted Klein form. Second, in the standard form we prove a common-phase property and carry out realnumberizing, then apply the Real Boolean Holant dichotomy. Third, in the twisted form we use directed cycles and the transpose action to rule out a minimal counterexample satisfying the scattered-pearl condition. The complexity analysis of both forms assumes the quaternary decomposition condition in Assumption~\ref{en:ass:klein-upper}.

\subsection{From the abstract \texorpdfstring{$\Kfour$}{K4} to two matrix normal forms}\label{en:subsec:klein-canonicalization}

Use the determinant-one Pauli notation
\[
 I=\begin{pmatrix}1&0\\0&1\end{pmatrix},\qquad
 X=\begin{pmatrix}0&i\\i&0\end{pmatrix},\qquad
 Y=\begin{pmatrix}0&1\\-1&0\end{pmatrix},\qquad
 Z=\begin{pmatrix}i&0\\0&-i\end{pmatrix},
\]
and write
\[
  \Qeight=\{\pm I,\pm X,\pm Y,\pm Z\}.
\]
These matrices satisfy
\[
 X^2=Y^2=Z^2=-I,\qquad XY=-Z,\quad YZ=-X,\quad ZX=-Y,
\]
with the sign reversed when the multiplication order is reversed.

\begin{lemma}[Lifting the projective Klein group to the quaternion group]\label{en:lem:klein-lift-q8}
Suppose $\Bgroup\leq\mathrm{SL}_2(\mathbb C)$, $\{I,-I\}\leq\Bgroup$, and
$\Bgroup/\{I,-I\}\cong\Kfour$. Then $\Bgroup$ is isomorphic to $\Qeight$.
More precisely, if $R,S\in\Bgroup$ represent two distinct nontrivial projective classes, then
\[
 R^2=S^2=-I,\qquad RS=-SR,\qquad
 \Bgroup=\{\pm I,\pm R,\pm S,\pm RS\}.
\]
\end{lemma}

\begin{proof}
The square of the nontrivial class $[R]$ is $[I]$, so $R^2=I$ or $R^2=-I$.
If $R^2=I$, then $R\ne\pm I$ forces its two eigenvalues to be $1,-1$, contradicting $\det R=1$.
Thus $R^2=-I$, and similarly $S^2=-I$.

Since $[R],[S]$ commute in the quotient, $RSR^{-1}S^{-1}=\pm I$.
If the right-hand side were $I$, then $R,S$ would commute. After diagonalizing $R$, whose two distinct eigenvalues are $i,-i$, the matrix $S$ would also be diagonal. Together with $S^2=-I$ and $\det S=1$, this would force $S=\pm R$, contrary to the distinctness of the two classes.
The commutator is therefore $-I$, or equivalently $RS=-SR$. The quotient has four elements, so the eight displayed matrices exhaust $\Bgroup$.
\end{proof}

Transpose is not a homomorphism of matrix multiplication, but it induces a homomorphism on the abelian group $\Ggroup$:
\[
 \vartheta([A])=[A^{\mathsf T}].
\]
Indeed, $[(AB)^{\mathsf T}]=[B^{\mathsf T}A^{\mathsf T}]
=[A^{\mathsf T}B^{\mathsf T}]$, and $\vartheta^2$ is the identity map.
It fixes $[I]$, so its permutation of the other three elements has only two possibilities: it fixes all three, or fixes one and exchanges the other two.

\begin{theorem}[Two orthogonal normal forms of the Klein group]\label{en:thm:two-klein-normal-forms}
Suppose $\Bgroup$ satisfies the hypotheses of Lemma~\ref{en:lem:klein-lift-q8}, has algebraic matrix entries, and is closed under transpose. Then there is an invertible algebraic matrix $O$ with $O^{\mathsf T}O=I$ such that
\[
 O^{\mathsf T}\Bgroup O=\{\pm M:M\in\Kstd\}
 \quad\text{or}\quad
 O^{\mathsf T}\Bgroup O=\{\pm M:M\in\Ktw\},
\]
where
\[
 \Kstd=\{I,X,Y,Z\}
\]
is the standard Klein form, and
\[
\begin{gathered}
 \Ktw=\{I,Z,X',Y'\},\\
 X'=\frac1{\sqrt2}\begin{pmatrix}0&-1+i\\1+i&0\end{pmatrix},\qquad
 Y'=\frac1{\sqrt2}\begin{pmatrix}0&1+i\\-1+i&0\end{pmatrix}
\end{gathered}
\]
is the twisted Klein form. Their projective transpose actions are, respectively,
\[
 [M^{\mathsf T}]=[M]\quad(M\in\Kstd),
\]
and
\[
 [I^{\mathsf T}]=[I],\quad[Z^{\mathsf T}]=[Z],\quad
 [(X')^{\mathsf T}]=[Y'],\quad[(Y')^{\mathsf T}]=[X'].
\]
\end{theorem}

\begin{proof}
Choose a representative $R$ of a nontrivial class fixed by $\vartheta$.
Then $R^{\mathsf T}=R$ or $R^{\mathsf T}=-R$.

First suppose $R^{\mathsf T}=-R$. A skew-symmetric $2\times2$ matrix has the form
$c\begin{pmatrix}0&1\\-1&0\end{pmatrix}$; the condition $\det R=1$ gives $c=\pm1$, so $R=\pm Y$.
Choose a representative $S$ of another nontrivial class. Comparing the four entries in $SY=-YS$ gives
\[
 S=\begin{pmatrix}a&b\\b&-a\end{pmatrix},\qquad a^2+b^2=-1.
\]
The matrix $S$ is symmetric, with eigenvalues $i,-i$. Let $u,v$ be corresponding eigenvectors.
Symmetry gives $u^{\mathsf T}v=0$. If $u^{\mathsf T}u=0$, then $u$ is orthogonal to both vectors in the basis $u,v$, contradicting the nondegeneracy of the bilinear form. The same argument applies to $v$.
Thus neither eigenline is isotropic. Normalize the two eigenvectors and use them as the columns of a complex orthogonal matrix $O$. Then $O^{\mathsf T}SO=Z$.
Also $O^{\mathsf T}YO=\pm Y$, so the entire group takes the standard form.

Now suppose $R^{\mathsf T}=R$. The same normalization of eigenvectors gives a complex orthogonal matrix $O$ with $O^{\mathsf T}RO=Z$.
Since $(O^{\mathsf T}AO)^{\mathsf T}=O^{\mathsf T}A^{\mathsf T}O$, this change of basis preserves the projective transpose action. In these coordinates, continue to denote the other generator by $S$.
From $SZ=-ZS$ and $S^2=-I$, we obtain
\[
 S=\begin{pmatrix}0&b\\c&0\end{pmatrix},\qquad bc=-1.
\]
If $\vartheta$ fixes all three nontrivial classes, then $S^{\mathsf T}=\pm S$.
The plus sign gives $b=c$ and $b^2=-1$; the minus sign gives $c=-b$ and $b^2=1$.
Both cases yield the standard form $\{\pm I,\pm X,\pm Y,\pm Z\}$.

The only remaining possibility is that $\vartheta$ fixes $[Z]$ and exchanges $[S]$ with $[ZS]$.
Hence
\[
 S^{\mathsf T}=\delta ZS,\qquad \delta\in\{1,-1\}.
\]
Entrywise comparison gives $c=\delta ib$ and $b^2=i/\delta$.
After changing the sign of $S$ if necessary, we may take $S=Y'$ when $\delta=1$, and $S=X'$ when $\delta=-1$; the remaining one of these two matrices represents the third nontrivial class.
This is precisely the twisted form. All entries of $O$ are obtained from the original algebraic entries by finitely many field operations and square roots, so $O$ can still be chosen algebraic.
\end{proof}

The two projective transpose actions are the identity permutation and a transposition. Since an orthogonal change of basis preserves this permutation type, the two forms are not orthogonally equivalent.

Because $O^{\mathsf T}O=I$, Proposition~\ref{en:prop:holographic-reduction} shows that the transformation in Theorem~\ref{en:thm:two-klein-normal-forms} puts the entire function-net into the standard or twisted form.
A binary matrix $A$ becomes $O^{\mathsf T}AO$, and ordinary equality edges are preserved.

\subsection{Ports, pairings, and contraction conventions}\label{en:subsec:klein-contraction}

For two ports $p,q$ of a function $F$, with $p\ne q$, insert a binary function $M$ along the chosen cycle direction and sum. Write
\begin{equation}\label{en:eq:cycle-contraction}
  \Gamma_M^{pq}(F)
  =\sum_{a,b\in\{0,1\}}M(a,b)
   F(\ldots,x_p=a,\ldots,x_q=b,\ldots).
\end{equation}
If the binary factor on the other side is written as $N(b,a)$ along the cycle, the local contraction is
\begin{equation}\label{en:eq:trace-convention}
  \sum_{a,b}M(a,b)N(b,a)=\Tr(MN).
\end{equation}
Equations~\eqref{en:eq:cycle-contraction}--\eqref{en:eq:trace-convention} define a bilinear contraction.
We also use the dual basis of this pairing to read off coordinates after a change of basis.

For $\mathcal K\in\{\Kstd,\Ktw\}$, a nonzero function of arity $2d$ is called a $\mathcal K$-basis function if, for some pairing, it has the form
\[
 c\prod_{j=1}^d M_j(x_{q_j},x_{p_j}),\qquad
 c\in\mathbb C^\times,\quad M_j\in\mathcal K.
\]
Each pair is recorded in the order $(p_j,q_j)$; writing its factor as $M_j(q_j,p_j)$ makes Equation~\eqref{en:eq:cycle-contraction} read as $\Tr(MM_j)$ along the cycle.
A nonzero nullary constant is regarded as a basis function with no binary factors.

The two matrix bases have the same cycle-pairing table. Ordering each basis by $I$ followed by its three nontrivial matrices gives
\begin{equation}\label{en:eq:klein-cycle-gram}
 \bigl(\Tr(MN)\bigr)_{M,N\in\mathcal K}
 =2\diag(1,-1,-1,-1).
\end{equation}
Indeed, $I^2=I$, each nontrivial matrix squares to $-I$, and the product of two distinct nontrivial matrices has trace zero.
Accordingly, with
\[
 c_I=\frac12,\qquad c_M=-\frac12\quad(M\ne I),
\]
there is a unique expansion on the variables $p,q$:
\begin{equation}\label{en:eq:klein-pair-expansion}
 F=\sum_{M\in\mathcal K}c_M M(x_q,x_p)\otimes
 \Gamma_M^{pq}(F).
\end{equation}
More explicitly, the bilinear dual of $M(x_q,x_p)$ is $c_M M(x_p,x_q)$, because
\[
 \sum_{a,b\in[2]}c_M M(a,b)N(b,a)
 =c_M\Tr(MN)=\mathbf1_{M=N}.
\]
Taking these dual functions pair by pair gives the dual of the corresponding tensor-product basis.

\subsection{The quaternary decomposition condition (upper case)}\label{en:subsec:klein-upper-condition}

Fix a finite algebraic function set $\mathcal F$, already transformed orthogonally as in Theorem~\ref{en:thm:two-klein-normal-forms}.
Let $\mathscr R=\mathscr R(\mathcal F)$ be its ordinary gadget closure: only the original functions, equality edges, port permutations, disjoint unions, and internal contractions are used.
Throughout the following argument, ``available'' means ordinarily realizable in this fixed closure.
By Lemma~\ref{en:lem:per-signature-scaling}, we may ignore known nonzero algebraic constants whose effects can be compensated.

\begin{assumption}[Quaternary decomposition condition (upper case)]\label{en:ass:klein-upper}
Let $\mathcal K=\Kstd$ or $\mathcal K=\Ktw$, and assume:
\begin{enumerate}[label=\textup{(\arabic*)}]
 \item All four binary directions of $\mathcal K$ are realizable in $\mathscr R$, and every nonzero binary function in $\mathscr R$ is a nonzero scalar multiple of a matrix in $\mathcal K$.
 \item Every nonzero quaternary function in $\mathscr R$ decomposes, along some pairing, as the tensor product of two binary functions.
\end{enumerate}
\end{assumption}

The first item is the completeness condition for the Klein binary group.
The second is what the main text calls the \emph{quaternary decomposition condition (upper case)}: it requires every ordinarily realizable quaternary function to admit a binary-by-binary decomposition, rather than requiring this only of the input functions.

\begin{lemma}[Immediate consequences of the upper-case condition]\label{en:lem:klein-upper-basic}
Under Assumption~\ref{en:ass:klein-upper}, $\mathscr R$ contains no nonzero odd-arity function, and every nonzero quaternary function is a $\mathcal K$-basis function.
\end{lemma}

\begin{proof}
The four matrices in $\mathcal K$ span the whole space of binary functions.
Thus, for any nonzero function and any variable pair, at least one $\mathcal K$ contraction is nonzero; otherwise every binary slice on that pair would vanish.
If a nonzero odd-arity function existed, successive nonzero contractions would yield a nonzero unary function $u$.
The disjoint union $u\otimes u$ would then be a nonzero rank-one binary function, contradicting the first item of the assumption.

Let $H=g\otimes h$ be a nonzero quaternary function with the binary-by-binary decomposition supplied by the second item.
Since $\mathcal K$ spans the binary space, some matrix in $\mathcal K$ contracts $h$ to a nonzero constant, thereby ordinarily realizing a nonzero multiple of $g$.
Similarly, a nonzero multiple of $h$ is realizable.
The first item therefore makes each factor a nonzero multiple of a $\mathcal K$ matrix, so $H$ is a basis function.
\end{proof}

Here ``decomposable'' is also equivalent to ``binary-by-binary decomposable'' for quaternary functions.
Indeed, a unary-by-ternary decomposition $u\otimes h$, followed by a nonzero $\mathcal K$ contraction on $h$, would give a nonzero rank-one binary function $u\otimes v$, again contradicting the binary condition.

The following simultaneous nonzero contraction is an arity-reduction tool shared by the two forms.

\begin{lemma}[Simultaneous nonzero contraction]\label{en:lem:klein-simultaneous-nonzero}
Let $\mathcal K=\Kstd$ or $\mathcal K=\Ktw$.
If $C,D$ are two nonzero complex-valued functions on the same set of $n\ge3$ variables, then there are a variable pair $p,q$ and $M\in\mathcal K$ such that
\[
 \Gamma_M^{pq}(C)\ne0,\qquad
 \Gamma_M^{pq}(D)\ne0.
\]
\end{lemma}

\begin{proof}
First let $n=3$. Write
\[
 \Sigma_{pq}(C)=\{M\in\mathcal K:\Gamma_M^{pq}(C)\ne0\},\qquad
 s_{pq}(C)=|\Sigma_{pq}(C)|.
\]
The four matrices span the entire binary-function space, so $s_{pq}(C)\ge1$ for every pair.
If, for example, $s_{12}(C)=1$, Equation~\eqref{en:eq:klein-pair-expansion} gives
\[
 C(x_1,x_2,x_3)=N(x_2,x_1)v(x_3)
\]
up to a nonzero constant, where $N$ is invertible and $v\ne0$.
Summing directly over row and column indices gives
\[
 \Gamma_M^{13}(C)=NMv,\qquad
 \Gamma_M^{23}(C)=N^{\mathsf T}Mv.
\]
The matrices involved are invertible, so each of the other two pairs has four nonzero branches.
If no pair has only one nonzero branch, each of the three branch counts is at least two.
In either case,
\[
 s_{12}(C)+s_{13}(C)+s_{23}(C)\ge6.
\]

Suppose, for a contradiction, that $C,D$ have no common nonzero contraction.
On each variable pair, their two supports are disjoint, so the corresponding branch counts sum to at most four.
Summing over the three pairs gives an upper bound of $12$; applying the preceding lower bound to $C$ and $D$ separately gives a lower bound of $12$.
Thus the three branch counts of each function sum to six.
If any branch count were one, the preceding calculation would make the other two equal to four, giving a sum of at least nine.
Consequently all six branch counts are two, and on every variable pair the two supports are complementary two-element sets.

For nontrivial $R\in\mathcal K$ and $p<q$, define the operator
\[
 \Omega_R^{pq}(N)=RNR
\]
on the binary matrix coordinates of $p,q$.
In tensor-product coordinates, it acts by $R^{\mathsf T}$ on variable $p$ and by $R$ on variable $q$.
The quaternion multiplication table shows that the two eigenspaces of $\Omega_R^{pq}$ are spanned by
\[
 \{I,R\},\qquad \mathcal K\setminus\{I,R\},
\]
with eigenvalues $-1,+1$, respectively.
The positive and negative eigendirections for the three choices of $R$ give all six two-element subsets of the four-element set.
Thus, for every pair $pq$, there are a unique nontrivial projective direction $R_{pq}$ and $\lambda_{pq}\in\{1,-1\}$ such that
\[
 \Omega_{R_{pq}}^{pq}C=\lambda_{pq}C,\qquad
 \Omega_{R_{pq}}^{pq}D=-\lambda_{pq}D.
\]

Compare $12$ and $13$.
If the two matrices on the shared variable $1$ represented different projective directions, they would anticommute, and so would the two full operators.
They could not then have the common nonzero eigenvector $C$.
Hence $R_{12}=R_{13}$.
Comparing the shared variable $3$ of $13$ and $23$ similarly gives $R_{13}=R_{23}$.
Denote their common direction by $R$.

In the standard form, $R^{\mathsf T}=\eta_RR$ for $\eta_R\in\{1,-1\}$.
Direct multiplication in the three tensor-product positions gives
\[
 \Omega_R^{12}\Omega_R^{13}=-\eta_R\Omega_R^{23}.
\]
The product of the three eigenvalues for $C$ is therefore $-\eta_R$.
For $D$, all three eigenvalues are reversed, so their product is $\eta_R$; the same operator identity requires it to be $-\eta_R$, a contradiction.

In the twisted form, if the common direction were $X'$ or $Y'$, compare $12$ and $23$.
On the shared variable $2$, the first ordered pair reads $R$ and the second reads $R^{\mathsf T}$.
Since $X',Y'$ are transposes of one another and anticommute, the two full operators would again anticommute.
Thus $R=Z$ is the only possibility.
Now $Z^{\mathsf T}=Z$, and
$\Omega_Z^{12}\Omega_Z^{13}=-\Omega_Z^{23}$.
The same comparison of the three opposite eigenvalues for $C,D$ again gives a contradiction.
This proves the ternary case.

For $n>3$, choose any three common variables.
For each of $C,D$, fix the remaining variables using one of its nonzero function values, obtaining two nonzero ternary slices; the two assignments need not agree.
Apply the ternary result to these slices.
The resulting nonzero contraction values are values of the respective original contraction functions, so the same pair and the same $M$ give nonzero contractions of both $C,D$.
These assignments are used only to select nonzero slices; no pinning functions are assumed available.
\end{proof}

\section{The standard Klein form: common phase and realnumberizing}\label{en:subsec:standard-klein}

In this section, take $\mathcal K=\Kstd$ in Assumption~\ref{en:ass:klein-upper} and follow the common-phase and realnumberizing approach of the third version.
Lemma~\ref{en:lem:klein-upper-basic} already ensures that every nonzero available function has even arity and every quaternary function is a basis function.

Write $V=\operatorname{span}_{\mathbb R}\Kstd$.
We shall repeatedly use one immediate consequence of the quaternary condition: for every nonzero available quaternary function $H$, there is $c\in\mathbb C^\times$ such that every scalar obtained by closing its four ports along any pairing with matrices in $V$ belongs to $c\mathbb R$.
Indeed, superimposing the factor pairing of $H$ and the closing pairing gives either one cycle of length four or two cycles of length two.
The value on each cycle is the trace of a product of $\Qeight$ matrices, hence belongs to $\{0,2,-2\}$.
Real linear expansion then gives the claim for matrices in $V$.

For an ordered pairing $P$ and a choice $\mathbf M$ of tools on its pairs, let
$\zeta_{P,\mathbf M}(F)$ denote the scalar obtained by closing all variables of $F$ along $P$.

\begin{lemma}[Common phase in the standard form]\label{en:lem:standard-klein-common-phase}
In the standard-form case of Assumption~\ref{en:ass:klein-upper}, let $F$ be a nonzero available function of even arity.
Choose any two ordered pairings $P,Q$ and two collections of $\Kstd$ tools.
If both complete closures are nonzero, then
\[
 \frac{\zeta_{Q,\mathbf N}(F)}{\zeta_{P,\mathbf M}(F)}\in\mathbb R^\times.
\]
Consequently, in the $\Kstd$ tensor-product basis determined by any pairing, all nonzero coordinates of $F$ have a common complex phase.
\end{lemma}

\begin{proof}
First change the pairing.
The symmetric difference of two perfect matchings consists of alternating even cycles.
On an alternating cycle of length at least four, replacing two edges of the current matching by two other edges can shorten the cycle.
Thus a finite sequence of two-edge switches transforms $P$ into $Q$.
For one switch, remove two tools from the current closure to obtain a nonzero available quaternary function $H$.
Since $\Kstd^{\otimes2}$ is a basis on the two pairs of the new pairing, two tools can be chosen to make the new closure nonzero.
The trace calculation at the beginning of this section shows that the ratio of the new and old nonzero values is real.
Multiplying these ratios along the sequence connects the given $P$-closure to some nonzero $Q$-closure with a real ratio.

Now fix a pairing and compare two nonzero closures.
If their tools differ on only one pair, closing all other pairs gives a nonzero available binary function with nonzero contractions in two different $\Kstd$ directions.
This contradicts the binary condition that only one basis direction occurs.
Hence this distance-one case cannot occur.

If at least two pairs differ, choose two such pairs and close the other variables according to each of the two collections, obtaining nonzero quaternary functions $H,H'$ on the same four variables.
By Lemma~\ref{en:lem:klein-simultaneous-nonzero}, there are a pair among these variables and a common $A\in\Kstd$ for which both contractions are nonzero.
Two nonzero binary functions remain.
For each, the matrices in the real vector space $V$ that make the complete closure zero form a proper subspace.
The union of two proper subspaces cannot cover $V$, so a common $B\in V$ makes both complete closures nonzero.
We have now obtained two new closing paths: they use the same pairing and the same tools $A,B$ on the selected four variables, while elsewhere they follow their respective original paths.
On each side, the old and new values are complete closures of the corresponding nonzero quaternary function, so their ratio is real.

Here $B$ is used only in an algebraic comparison of closure values; it need not be ordinarily realizable.
The two paths being compared always contain at most this one common non-$\Kstd$ tool.
At the next step it is removed before the available quaternary functions are formed.

If the original paths differed on exactly two pairs, the new paths now agree.
If other differences remain, designate the common pair carrying $B\in V$ as a special pair, and choose one pair on which the paths still differ.
After removing these two pairs, every remaining tool belongs to $\Kstd$, so two available quaternary functions are again obtained.
Repeat the preceding construction with new common tools $A\in\Kstd,B\in V$.
Each step reduces the number of differing pairs by one, so the two paths eventually coincide.
Every nonzero ratio along either path is real; hence the ratio of the original two closure values is real.
Together with the first paragraph, this proves the claim for arbitrary pairings.

Finally, Equation~\eqref{en:eq:klein-cycle-gram} shows that each basis coordinate corresponds to a unique collection of $\Kstd$ closing tools, and its complete closure is that coordinate multiplied by a nonzero real number $\pm2^d$.
Thus all nonzero coordinates have a common complex phase as well.
\end{proof}

We still need to pass from a common phase of the basis coordinates to a scaling that makes the function real-valued.
This is the purpose of the arity reduction preserving the mixed real and imaginary parts in v3.
To distinguish having two nonzero parts from having no common phase, we first specify the two complementary subspaces actually preserved by the reduction.

Set $\chi([I])=\chi([Y])=0$ and $\chi([X])=\chi([Z])=1$.
This is a group homomorphism $\Kfour\to\mathbb F_2$.
For a $\Kstd$ tensor-product basis vector on a pairing, define its parity label as the sum of the $\chi$ values of its binary factors.
Transpose preserves projective classes in the standard form.
When comparing two pairings, a nonzero matrix trace on an alternating cycle requires the projective products on the two sides to agree.
The trace-pairing table shows that all change-of-basis coefficients are real and do not mix parity labels.
Therefore the following two real subspaces are independent of the pairing: let $W_0,W_1$ be the real spans of the tensor-product basis vectors with labels $0,1$, respectively.
Basis vectors in $W_0$ are real-valued, while those in $W_1$ are purely imaginary-valued.
Inserting $M$ on any variable pair shifts the two subspaces simultaneously to
$W_{\epsilon+\chi([M])}$, so they remain disjoint.

\begin{theorem}[Per-function realnumberizing in the standard form]\label{en:thm:standard-klein-realification}
In the standard-form case of Assumption~\ref{en:ass:klein-upper}, every nonzero available function $F$ of even arity admits $\lambda_F\in\mathbb C^\times$ such that $\lambda_FF$ is real-valued.
\end{theorem}

\begin{proof}
By Lemma~\ref{en:lem:standard-klein-common-phase}, first multiply $F$ by a nonzero constant so that all its coordinates in a $\Kstd$ tensor-product basis are real.
If its support lies entirely in $W_0$, it is already real-valued; if entirely in $W_1$, multiply once more by $-i$.

Suppose instead that both subspaces have nonzero coordinates, and denote the corresponding nonzero components by $F_0,F_1$.
They lie in $W_0,W_1$ respectively and have real coordinates in the $\Kstd$ basis; equivalently, $F_0$ is real-valued and $F_1$ is purely imaginary-valued.
If $F$ has arity at least six, apply Lemma~\ref{en:lem:klein-simultaneous-nonzero} to $F_0,F_1$ to obtain a common variable pair and a common $M\in\Kstd$ whose contraction leaves both components nonzero.
The components $F_0,F_1$ are used only to select the tool and need not be separately realizable.
The actual operation is a single contraction of the current function $F_0+F_1$, so it remains in the ordinary gadget closure.
When a tensor-product basis vector is contracted, Equation~\eqref{en:eq:klein-cycle-gram} and the quaternion multiplication table show that the resulting coefficient is a real multiple of the original one.
Thus the new components still have real coordinates in the $\Kstd$ basis.
Their parity labels both shift by $\chi([M])$, so the two components remain in complementary subspaces.
Use these two nonzero components directly as the inputs to the simultaneous nonzero contraction lemma at the next step.
Repeating this process eventually produces a nonzero available quaternary function using both complementary parity subspaces.

The quaternary condition says that this function has only one nonzero basis coordinate for some pairing, and hence belongs to only one $W_\epsilon$, a contradiction.
The binary condition likewise excludes a mixed function of arity below four.
Therefore the two subspaces cannot both occur, proving the theorem.
\end{proof}

\begin{corollary}[Complexity dichotomy for the standard form]\label{en:cor:standard-klein-real-holant}
In the standard-form case of Assumption~\ref{en:ass:klein-upper}, per-function nonzero scaling gives an algebraic real-valued function set $\mathcal F_{\mathbb R}$ such that
\[
 \Holant(\mathcal F)\equiv_{\mathrm T}\Holant(\mathcal F_{\mathbb R}).
\]
Consequently, by Theorem~\ref{en:thm:real-holant}, the original problem belongs to $\FP$ if $\mathcal F_{\mathbb R}$ satisfies condition \textup{(T)}, and is $\sharpP$-hard otherwise.
\end{corollary}

\begin{proof}
The function set is fixed and finite.
For each nonzero $f$, choose the scalar $\lambda_f$ supplied by Theorem~\ref{en:thm:standard-klein-realification}.
The reciprocal of a nonzero algebraic basis coordinate removes the common phase; if necessary, multiply additionally by $-i$.
Thus the scaling constants are algebraic.
Leave zero functions unchanged.
Lemma~\ref{en:lem:per-signature-scaling} gives the stated Turing equivalence, and the scaled function set satisfies the algebraic real-valued input hypothesis of the Real Boolean Holant theorem.
\end{proof}

\section{The twisted Klein form: a complete proof}\label{en:subsec:twisted-klein}

In this section, take $\mathcal K=\Ktw$ in Assumption~\ref{en:ass:klein-upper}.
If the ordinary gadget closure contains a non-basis function, choose one of minimum arity.
Every binary contraction of this function is nonzero and is a basis function.
We call this the scattered-pearl condition.
A common pairing and support spreading reduce its arity to six or eight.
Cycle directions then force the nontrivial transpose action of the twisted form to be the identity, giving a contradiction.

First assign two-bit labels to the four projective classes:
\[
 \operatorname{lab}([I])=\binom00,\qquad
 \operatorname{lab}([Z])=\binom10,\qquad
 \operatorname{lab}([Y'])=\binom01,\qquad
 \operatorname{lab}([X'])=\binom11.
\]
Write $M_s\in\Ktw$ for the representative matrix with label $s\in\mathbb F_2^2$.
Ignoring signs in matrix multiplication, addition of labels is multiplication of projective classes.
Transpose acts on the labels by
\begin{equation}\label{en:eq:twisted-label-transpose}
 \operatorname{lab}([M^{\mathsf T}])=\tau\operatorname{lab}([M]),\qquad
 \tau=\begin{pmatrix}1&1\\0&1\end{pmatrix}.
\end{equation}
This matrix multiplication is over $\mathbb F_2$.
Thus $\tau\binom{1}{1}=\binom{0}{1}$ and $\tau\binom{0}{1}=\binom{1}{1}$, while $\binom00,\binom10$ are fixed.
In particular, the projective class $[I]$ of the matrix $I$ has label $\binom00$.
We denote the identity map on the label space by
\[
  \mathrm{Id}_{\mathrm{lab}}=\begin{pmatrix}1&0\\0&1\end{pmatrix}.
\]

\subsection{How a minimal counterexample gives the scattered-pearl condition}

All functions are taken from the fixed ordinary gadget closure $\mathscr R$ in Assumption~\ref{en:ass:klein-upper}.
Let
\[
 \mathscr R_{\mathrm{nb}}
 =\{F\in\mathscr R:F\ne0,\ \arity(F)\text{ is even},\ F\text{ is not a }\Ktw\text{-basis function}\}.
\]
If this set is nonempty, choose a member $F$ of minimum arity.

\begin{definition}[Scattered-pearl condition]\label{en:def:scattered-pearl}
A nonzero even-arity function $F$ satisfies the \emph{scattered-pearl condition} if, for every variable pair $p,q$ and every $M\in\Ktw$,
\[
 \Gamma_M^{pq}(F)\ne0,
 \qquad
 \Gamma_M^{pq}(F)\text{ is a }\Ktw\text{-basis function}.
\]
The factor pairings of the basis functions obtained by different contractions may differ.
\end{definition}

\begin{lemma}[A minimal counterexample satisfies the scattered-pearl condition]\label{en:lem:twisted-minimal-pearl}
In the twisted-form case of Assumption~\ref{en:ass:klein-upper}, if $\mathscr R_{\mathrm{nb}}\ne\varnothing$, then a minimum-arity member $F$ has arity at least six and satisfies the scattered-pearl condition.
\end{lemma}

\begin{proof}
By convention, a nonzero nullary function is a basis function.
The binary condition excludes binary counterexamples, and Lemma~\ref{en:lem:klein-upper-basic} excludes quaternary counterexamples.
Thus $F$ has arity at least six.
Fix a variable pair $p,q$, and let $b_{pq}(F)$ be the number of nonzero functions among its four contractions.
Since the four matrices span the whole binary-function space and $F\ne0$, we have $b_{pq}(F)\ge1$.
First consider a $\Ktw$-basis function.
If $p,q$ form one binary-factor edge, Equation~\eqref{en:eq:klein-cycle-gram} gives exactly one nonzero branch.
If they belong to different factors, write those factors as $A(x_r,x_p)$ and $B(x_s,x_q)$.
The contracted binary matrix is precisely $AMB^{\mathsf T}$.
All three matrices are invertible, so all four branches are nonzero.
Thus the branch count of a basis function on any variable pair is either $1$ or $4$.

If $b_{pq}(F)=1$, let $M_0$ be the unique tool giving a nonzero contraction.
Equation~\eqref{en:eq:klein-pair-expansion} gives
\[
 F=c_{M_0}M_0(x_q,x_p)\otimes\Gamma_{M_0}^{pq}(F).
\]
The second factor is nonzero, ordinarily realizable, and has smaller arity than $F$.
By minimality it is a basis function, making $F$ a basis function as well, a contradiction.

Next suppose $b_{pq}(F)=2$ or $3$.
Choose two nonzero branches
\[
 C=\Gamma_{M_1}^{pq}(F),\qquad
 D=\Gamma_{M_2}^{pq}(F).
\]
Since $F$ has arity at least six, $C,D$ still share at least four variables.
By Lemma~\ref{en:lem:klein-simultaneous-nonzero}, there are another variable pair $r,s$ and a common $N\in\Ktw$ whose contraction is nonzero for both functions.
Let
\[
 F'=\Gamma_N^{rs}(F).
\]
The pairs are disjoint, and finite summations commute, so
\[
 \Gamma_{M_a}^{pq}(F')
 =\Gamma_N^{rs}\bigl(\Gamma_{M_a}^{pq}(F)\bigr)\ne0
 \qquad(a=1,2).
\]
Branches that were zero remain zero after further contraction.
Hence $F'$ still has exactly two or three nonzero branches on the same pair.
On the other hand, $F'$ is nonzero, ordinarily realizable, and has smaller arity.
Minimality makes it a basis function, contradicting the $1/4$ branch-count criterion for basis functions.

Therefore every variable pair of a minimal counterexample has four nonzero branches.
For every $p,q,M$, the function $G=\Gamma_M^{pq}(F)$ is nonzero, has even arity, and has smaller arity than $F$.
If $G$ were not a basis function, it would itself belong to $\mathscr R_{\mathrm{nb}}$, contradicting the minimality of $F$.
Thus every such $G$ is a basis function, establishing the scattered-pearl condition.
\end{proof}

\subsection{Common pairings and an arity bound}

\begin{lemma}[Partitioning all variable edges into pairings]\label{en:lem:twisted-one-factorization}
Let $F$ be a nonzero function of arity $2d\ge6$ satisfying the scattered-pearl condition.
Then all variable pairs partition into $2d-1$ pairwise edge-disjoint pairings.
For each such pairing $P$, the function $F$ has exactly four nonzero coordinates in the $\Ktw$ tensor-product basis determined by $P$.
\end{lemma}

\begin{proof}
Fix a variable pair $p,q$ and write
\[
 G_M=\Gamma_M^{pq}(F)\qquad(M\in\Ktw).
\]
All four $G_M$ are nonzero basis functions.
The factor pairing of a basis function is uniquely recovered from the $1/4$ branch-count criterion: a variable pair has exactly one nonzero contraction branch if and only if its two variables belong to the same binary factor.
Denote the pairing of $G_M$ by $\mu_M$.

For another, disjoint variable pair $r,s$, consider the zero-one table
\[
 D_{M,N}=\mathbf1_{\Gamma_N^{rs}
                    \bigl(\Gamma_M^{pq}(F)\bigr)\ne0}
 \qquad(M,N\in\Ktw).
\]
Every row and every column has weight either $1$ or $4$; the column property follows because the two disjoint contractions commute.
Such a $4\times4$ table has only three possible forms: a permutation matrix, with each row of weight one; the all-one matrix; or a cross, with exactly one row of weight four and the unique $1$ in each of the other three rows in the same column.

For an edge $e$ on the remaining variables, let $m(e)$ be the number of pairings among the four $\mu_M$ that contain $e$.
The three table forms give $m(e)=4,0,3$, respectively.
Fix a remaining variable $r$.
It has one partner in each of the four pairings, so
\[
 \sum_{e\ni r}m(e)=4.
\]
If a summand were $3$, any other positive summand would be at least $3$, making a total of $4$ impossible.
Thus exactly one summand is $4$.
Each variable has the same partner in all four pairings, and hence
\[
 \mu_I=\mu_Z=\mu_{X'}=\mu_{Y'}=: \mu_{pq}.
\]

Set $\Pi_{pq}=\{pq\}\cup\mu_{pq}$.
If $rs\in\Pi_{pq}$, choose any nonzero entry in the permutation table above.
Contracting $pq$ and then $rs$, or doing so in the opposite order, gives the same basis function.
By uniqueness of its pairing, deleting these two edges gives $\Pi_{rs}=\Pi_{pq}$.
Consequently, two distinct pairings $\Pi$ share no edge.
Every variable edge $pq$ belongs to $\Pi_{pq}$, so these pairings partition all edges of the complete graph.
Each pairing has $d$ edges, and their number is therefore
\[
 \binom{2d}{2}/d=2d-1.
\]

Finally, apply Equation~\eqref{en:eq:klein-pair-expansion} on $pq$:
\[
 F=\sum_{M\in\Ktw}c_M M(x_q,x_p)\otimes G_M.
\]
The four functions $G_M$ have the common pairing $\mu_{pq}$.
Thus the right-hand side consists of four distinct vectors in the same tensor-product basis for $P=\Pi_{pq}$, each with a nonzero coefficient.
The same conclusion holds for every such $P$.
\end{proof}

Take two distinct pairings $P,Q$.
Since they share no edge, $P\cup Q$ decomposes into cycles whose edge colors alternate.
Fix a recorded order within each pair, and choose a traversal direction on each cycle.
If cycle $k$ has $r_k$ edges of $P$ and $r_k$ edges of $Q$, and there are $c$ cycles, then
\[
 r_1+\cdots+r_c=d.
\]
For a $P$-edge $e$, let $\epsilon_e^P=0$ when the traversal follows its recorded order and $\epsilon_e^P=1$ when it reverses that order.
Define
\[
 \sigma_k^P(\mathbf s)=
 \sum_{e\in P\cap C_k}\tau^{\epsilon_e^P}s_e.
\]
Define $\epsilon_e^Q$ similarly for $Q$-edges.
To read off change-of-basis coordinates, replace $M(x_q,x_p)$ by its bilinear dual $c_M M(x_p,x_q)$, reversing the order of the pair.
This reversal applies $\tau$ once more to the projective label, so define
\[
 \sigma_k^Q(\mathbf t)=
 \sum_{e\in Q\cap C_k}\tau^{1+\epsilon_e^Q}t_e.
\]

Denote the two sets of basis vectors by $B_P(\mathbf s),B_Q(\mathbf t)$ and their bilinear dual bases by $D_P(\mathbf s),D_Q(\mathbf t)$.
For example, if the ordered edges of $Q$ are $(p_j,q_j)$, then
\[
 B_Q(\mathbf t)=\prod_{j=1}^d M_{t_j}(x_{q_j},x_{p_j}),\qquad
 D_Q(\mathbf t)=\prod_{j=1}^d c_{M_{t_j}}M_{t_j}(x_{p_j},x_{q_j}).
\]
Define the change-of-basis matrix by
\begin{equation}\label{en:eq:twisted-change-basis}
 A^{Q\leftarrow P}_{\mathbf t,\mathbf s}
 =\sum_{\alpha\in[2]^{2d}}
   D_Q(\mathbf t)(\alpha)B_P(\mathbf s)(\alpha).
\end{equation}
This is simply the pairwise application of Equation~\eqref{en:eq:klein-pair-expansion}.
Summation along each alternating cycle in its traversal direction expresses an entry as a product of traces of matrix words around the cycles, multiplied by the constants from the dual factors.
More explicitly, a $P$-edge carries $M(x_q,x_p)$, so traversal in its recorded order reads $M^{\mathsf T}$.
A $Q$-edge carries the dual factor $c_MM(x_p,x_q)$, so traversal in its recorded order reads $c_MM$.
The actual label sum around cycle $k$ is therefore
\[
 \sum_{e\in P\cap C_k}\tau^{1+\epsilon_e^P}s_e
 +\sum_{e\in Q\cap C_k}\tau^{\epsilon_e^Q}t_e
 =\tau\bigl(\sigma_k^P(\mathbf s)+\sigma_k^Q(\mathbf t)\bigr).
\]
Since $\tau$ is invertible, this sum vanishes exactly when the two label sums agree.
Hence
\begin{equation}\label{en:eq:twisted-block-criterion}
 A^{Q\leftarrow P}_{\mathbf t,\mathbf s}\ne0
 \quad\Longleftrightarrow\quad
 \sigma_k^P(\mathbf s)=\sigma_k^Q(\mathbf t)
 \quad(1\le k\le c).
\end{equation}
For a nonzero entry, each cycle contributes absolute value $2$, and each dual factor contributes absolute value $1/2$.
The absolute value of the entry is therefore $2^{c-d}$.
Exchanging $P,Q$ gives the inverse change-of-basis matrix $A^{P\leftarrow Q}$, whose nonzero entries have the same absolute value.
In the reverse summation, both sides of the cycle condition acquire one additional application of $\tau$; applying $\tau$ to both sides restores Equation~\eqref{en:eq:twisted-block-criterion}.
The forward and inverse transformations thus use the same blocks, with row and column indices interchanged within each block.

\begin{lemma}[Support spreading in alternating-cycle blocks]\label{en:lem:twisted-support-spreading}
Partition the rows and columns by the $c$ cycle-label sums in Equation~\eqref{en:eq:twisted-block-criterion}.
Every nonzero block has size $N\times N$, where
\[
 N=4^{d-c}.
\]
If $h$ terms of the $P$-basis support of $F$ lie in one block, they produce at least $N/h$ nonzero coordinates in the $Q$ basis.
Consequently, every function satisfying the scattered-pearl condition has $d\le4$.
\end{lemma}

\begin{proof}
Fix a cycle with $r_k$ edges of $P$ and fix its label sum.
The labels on the first $r_k-1$ edges can be chosen arbitrarily, and the sum uniquely determines the last.
Thus this cycle has $4^{r_k-1}$ possible column indices, and similarly as many row indices.
Multiplying over the cycles gives
\[
 N=\prod_{k=1}^c4^{r_k-1}=4^{d-c}.
\]
The dual-basis identities show that the two changes of basis are inverse operations.
There are no nonzero entries between blocks with different label sums, so each block is itself invertible.
Every nonzero entry of the block and its inverse has absolute value
$a=2^{c-d}=N^{-1/2}$.

Let $v\ne0$ be the input coefficient vector in the block, with exactly $h$ nonzero coordinates.
Let $w=A^{Q\leftarrow P}v$ have exactly $s$ nonzero coordinates, and write
\[
 U=\max_j|v_j|>0,\qquad V=\max_i|w_i|>0.
\]
Each output coordinate is a sum of at most $h$ terms of absolute value at most $aU$, giving $V\le haU$.
Recovering $v$ from $w$ with the inverse matrix similarly gives $U\le saV$.
Therefore
\[
 U\le saV\le sha^2U.
\]
Canceling $U>0$ gives $hs\ge a^{-2}=N$, or $s\ge N/h$.

By Lemma~\ref{en:lem:twisted-one-factorization}, $F$ has exactly four nonzero coordinates in each of the $P$ and $Q$ bases.
Different blocks cannot cancel one another, and $h\le4$ in every block.
Hence $4\ge N/4=4^{d-c-1}$, or $d-c\le2$.
The two pairings share no edge, so every alternating cycle contains at least two $P$-edges, giving $c\le\lfloor d/2\rfloor$.
Together these inequalities imply $d\le4$.
\end{proof}

Since we also assume $2d\ge6$, only arities six and eight remain.
In both cases the resulting equalities give $N=16$, while each basis support has size four.
If the four $P$-basis vectors were distributed among several blocks with respective support sizes $h_1,\ldots,h_t$, then
\[
 \sum_{j=1}^t\frac{16}{h_j}>4
 \qquad(t\ge2,\ h_1+\cdots+h_t=4).
\]
Indeed, each $h_j\le4$, so the left-hand side is at least eight when $t\ge2$.
This contradicts the four-element support in the $Q$ basis.
Thus the four vectors in the $P$-basis support have the same directed label sum on every alternating cycle.

\subsection{Direction contradictions at arities six and eight}

Fix a pairing $P$ and index its four nonzero basis coordinates by the four labels on one of its edges.
For each $P$-edge $j$, the four branches in the scattered-pearl condition are all nonzero.
Consequently the four terms use each label exactly once at that position, giving a permutation $s_j$ of the four-point set.
Every such permutation has a unique expression
\[
 s_j(x)=L_jx+b_j,
\]
where $L_j$ is an invertible $2\times2$ matrix over $\mathbb F_2$.
To see this, take $b_j=s_j(0)$ and use $s_j\binom10+b_j$ and $s_j\binom01+b_j$ as its columns.
These are distinct nonzero vectors and are therefore linearly independent.
The image of the fourth point is then automatically correct, because the sum of all four points of $\mathbb F_2^2$ is zero and remains zero after any permutation.
We may arrange that the first edge has $L_0=\mathrm{Id}_{\mathrm{lab}},b_0=0$.

If an alternating cycle contains the $P$-edge positions in $J$, and $\epsilon_j$ records whether its traversal reverses the direction of $P$-edge $j$, equality of the four directed label sums means that their linear part vanishes:
\begin{equation}\label{en:eq:twisted-directed-label}
 \sum_{j\in J}\tau^{\epsilon_j}L_j=0.
\end{equation}

\begin{theorem}[No scattered-pearl function in the twisted form]\label{en:thm:twisted-no-pearl}
There is no nonzero function of arity at least six satisfying the scattered-pearl condition.
\end{theorem}

\begin{proof}
Lemma~\ref{en:lem:twisted-support-spreading} has reduced the arity to six or eight.

First consider arity six.
After renaming variables, take
\[
 P=(12)(34)(56),\qquad Q=(23)(45)(61).
\]
Among the pairings supplied by Lemma~\ref{en:lem:twisted-one-factorization}, the one containing edge $13$ cannot use any edge of $P,Q$.
Its only possible form is therefore
\[
 Q^*=(13)(25)(46).
\]
Direct the $P$-edges as $1\to2,3\to4,5\to6$.
Traverse the six-cycle in $P\cup Q$ as
\[
 1\to2\to3\to4\to5\to6\to1.
\]
All three $P$-edges are traversed forwards, so Equation~\eqref{en:eq:twisted-directed-label} gives
\[
 \mathrm{Id}_{\mathrm{lab}}+L_1+L_2=0.
\]
Traverse the six-cycle in $P\cup Q^*$ as
\[
 1\to2\to5\to6\to4\to3\to1.
\]
Now $12,56$ are traversed forwards and $34$ backwards, giving
\[
 \mathrm{Id}_{\mathrm{lab}}+\tau L_1+L_2=0.
\]
Adding these two equations over $\mathbb F_2$ gives $(\tau+\mathrm{Id}_{\mathrm{lab}})L_1=0$.
Since $L_1$ is invertible, this forces $\tau=\mathrm{Id}_{\mathrm{lab}}$, contradicting Equation~\eqref{en:eq:twisted-label-transpose}.

Now consider arity eight. Fix
\[
 P=(12)(34)(56)(78).
\]
Besides $P$, Lemma~\ref{en:lem:twisted-one-factorization} supplies six pairings.
Each forms two four-cycles with $P$, and therefore partitions the four $P$-edges into two pairs.
There are three possible groupings:
\[
 \{12,34\}\mid\{56,78\},\quad
 \{12,56\}\mid\{34,78\},\quad
 \{12,78\}\mid\{34,56\}.
\]
Fix the first grouping.
The only two ways to connect $12,34$ are $(13)(24)$ and $(14)(23)$.
Likewise, the only two ways to connect $56,78$ are $(57)(68)$ and $(58)(67)$.
If two pairings make the same choice in either place, they share an edge.
Thus this grouping contains at most two pairings.
The six pairings fall into three groupings, so each contains exactly two, whose choices are complementary in both places.

Consequently, one pairing in the first grouping uses $(14)(23)$ on $1,2,3,4$, and the other uses $(13)(24)$.
Direct $12$ as $1\to2$ and $34$ as $3\to4$.
The former pairing gives the four-cycle $1\to2\to3\to4\to1$, traversing both $P$-edges forwards, so
\[
 \mathrm{Id}_{\mathrm{lab}}+L_1=0.
\]
The latter gives the four-cycle $1\to2\to4\to3\to1$, traversing $12$ forwards and $34$ backwards, so
\[
 \mathrm{Id}_{\mathrm{lab}}+\tau L_1=0.
\]
Adding and multiplying on the right by $L_1^{-1}$ again gives
$\tau=\mathrm{Id}_{\mathrm{lab}}$, contradicting the transpose action of the twisted form.
\end{proof}

\begin{corollary}[Function structure in the twisted upper case]\label{en:cor:twisted-internal-exit}
In the twisted-form case of Assumption~\ref{en:ass:klein-upper}, every nonzero function in $\mathscr R$ is a $\Ktw$-basis function.
\end{corollary}

\begin{proof}
If $\mathscr R_{\mathrm{nb}}$ were nonempty, Lemma~\ref{en:lem:twisted-minimal-pearl} would give a minimum-arity function satisfying the scattered-pearl condition, whereas Theorem~\ref{en:thm:twisted-no-pearl} says that no such function exists.
Thus every nonzero even-arity function is a basis function.
Lemma~\ref{en:lem:klein-upper-basic} has already excluded nonzero odd-arity functions.
\end{proof}

\begin{lemma}[A cycle-by-cycle algorithm for basis functions]\label{en:lem:klein-cycle-algorithm}
Let $\mathcal H$ be a fixed finite function set, where every nonzero function is given as a nonzero constant multiple of a $\Kstd$-basis function or a $\Ktw$-basis function.
Then $\Holant(\mathcal H)\in\FP$.
\end{lemma}

\begin{proof}
If an input function-net has a vertex carrying a zero function, its value is zero.
Assume henceforth that every vertex function is nonzero.
The function set is fixed, so a witnessing pairing, the matrix on each pair, and the overall constant of each function can be preprocessed.
Split each positive-arity vertex of the input function-net into binary-factor vertices according to its pairing.
Nullary functions contribute directly to the overall constant.
Each binary-factor vertex has two ports, and the internal edges of the original function-net still connect these ports.
Every connected component of the resulting graph is therefore a cycle.
Self-loops give cycles of length one, and parallel edges may give cycles of length two; both are included.

Choose a direction on each cycle.
When traversing a binary factor, read $M$ if the direction agrees with the recorded order of its ports and $M^{\mathsf T}$ if it reverses that order; read ordinary internal edges as $I$.
Successively summing over the Boolean variables on the cycle gives exactly the trace of the product of these $2\times2$ matrices.
Multiply the traces of all cycles and the overall constants of all vertices to obtain the value of the original function-net.
The number of matrices is proportional to the input size, so this is a polynomial-time algorithm.
\end{proof}

\subsection{Conclusion for the upper case}\label{en:subsec:klein-exits}

\begin{theorem}[Complexity dichotomy for the Klein upper case]\label{en:thm:klein-upper-dichotomy}
Let $\mathcal F$ be a fixed finite algebraic function set whose nonzero ordinarily realizable binary functions have the Klein group as their projective group.
Suppose that after the orthogonal transformation of Theorem~\ref{en:thm:two-klein-normal-forms}, Assumption~\ref{en:ass:klein-upper} holds.
Then:
\begin{enumerate}[label=\textup{(\arabic*)}]
 \item In the standard form, each function can be multiplied by a nonzero algebraic constant to become real-valued.
       If the resulting real-valued function set satisfies condition \textup{(T)} of Theorem~\ref{en:thm:real-holant}, the original problem belongs to $\FP$; otherwise it is $\sharpP$-hard.
 \item In the twisted form, every nonzero ordinarily realizable function is a $\Ktw$-basis function, and the original problem belongs to $\FP$.
\end{enumerate}
\end{theorem}

\begin{proof}
Theorem~\ref{en:thm:two-klein-normal-forms} exhausts the possibilities with the two normal forms.
The standard-form conclusion follows from Corollary~\ref{en:cor:standard-klein-real-holant}.
The twisted-form conclusion follows from Corollary~\ref{en:cor:twisted-internal-exit} and Lemma~\ref{en:lem:klein-cycle-algorithm}.
\end{proof}

This completes the Klein case under the quaternary decomposition condition in the main text.
The lower case, in which a quaternary function has no binary-by-binary decomposition, is left untreated here.

\paragraph*{Statement on the proof work}

As a rough division of the work, the first three versions may be regarded as the author's own work.
The proof of the twisted Klein case in the fourth version was created by AI, and the author has checked the complete proof.

More precisely, AI developed the twisted Klein proof with guidance from the author through conversation.
The author observed that, in the standard form, the real/imaginary division is also the division of the Klein group into $\{I,Y\}$ and $\{X,Z\}$.
He conjectured that the other two divisions---$\{I,X\}$ versus $\{Y,Z\}$, and $\{I,Z\}$ versus $\{X,Y\}$---should have the same effect, and wondered whether the other two divisions in the twisted form would have the same effect as well.
AI proved the claim for the standard form and directly reported a stronger result, the simultaneous nonzero contraction lemma, which also holds in the twisted form.
This was the author's main guiding contribution in the conversations, and its role already seemed quite small.
The lengthy exchanges before and after it, supplying materials and instructions, appeared to play an even smaller role.

To describe the work up to this point more precisely: apart from the twisted Klein case, only the tetrahedral arity-reduction lemma for the standard form was independently proved in full by a paid AI service after the author supplied the proposition, while he was busy grading final examinations; see the acknowledgments in the third version.
The author subsequently checked and rewrote that proof.
Looking back, he considers formulating the proposition an important contribution, and the proof was within his own capabilities.
In all remaining parts, free AI tools were used only for work requiring relatively little original creation, such as literature retrieval, adjusting proof wording, simple matrix calculations, and Chinese-to-English translation, with only modest results.
For example, they could not identify which Burnside lemma to use; a search engine ultimately resolved the matter.

\clearpage
\appendix
\section{The Road to Completion}\label{en:sec:appendix-road}
Recently, the dichotomy theorem with binary disequality\cite{en:liu2026disequality} was submitted to arXiv on August 31, 2026.
I saw it on September 8 and also heard in conversation that a complete dichotomy theorem had been submitted and was awaiting public display online. After thinking for a few hours, I decided to quickly prepare a sprint version,
filling in the remaining cases from the dichotomy with binary disequality to the final dichotomy, except for the C2 case, for which a complete independent manuscript already exists.
Of the cases needing completion, all except the odd dihedral groups were close to being proved in the first three versions. I subscribed to a stronger paid AI service to help complete some proofs and organize the text.

The proof in the paper establishing the dichotomy with binary disequality\cite{en:liu2026disequality} is long and covers many cases, including the entire Klein group case discussed here, as well as the remaining even dihedral groups and the three regular polyhedral rotation groups. Incidentally, I once tried to unify the proofs for all larger groups containing a Klein group: I wanted first to establish conclusions common to the four Klein matrices and then show that the additional elements in each group provide capabilities with similar effects, leading to the respective complexity dichotomies. What I did not manage to do was accomplished by this dichotomy theorem (one only needs an auxiliary analysis showing that binary disequality can be obtained in these larger groups).
With this major theorem as a foundation, the final stage of the original nine-class group framework has a shorter route: carry out a brief group-theoretic analysis to identify the groups that can obtain binary disequality after a complex orthogonal holographic transformation, and then invoke the theorem for those cases.
Thus, from the standpoint of completing the complexity classification, the main text can retain the direct proof for the Klein upper case, while the lower case is handled by this route in the appendix; this version does not develop the direct proof for the lower case.

After this analysis, the cases in which binary disequality cannot yet be obtained lie only among cyclic groups and odd dihedral groups, corresponding to
$C_1,C_2,C_k\ (k\ge3)$ and $D_{2n}\ (n\ge3\text{ odd})$ in the nine classes.
All cyclic groups of odd order remain; for cyclic groups of even order and odd dihedral groups, only matrix representations in which binary disequality cannot yet be obtained remain.
Coincidentally, at the beginning of the semester after this year's Spring Festival, two of my junior graduate students chose $C_2$ and odd dihedral groups, respectively, as directions for exploration and research.
The dichotomy for $C_2$ has been completed by Yuxuan Zhou in independent work
\cite{en:zhou-draft}; I am not an author of that paper, and this appendix only cites its result.
The higher-order cyclic groups $C_k$ were already proved in the first draft, v1\cite{en:xia2026v1}; the main proof for $C_1$ was given in
v2\cite{en:xia2026v2}, but some finite cases remained unverified at the time, and some lower cases had not been fully treated.
Much of what needs to be added in this round happens to have a foundation in earlier work by my students and me that was almost entirely manual (including help from free AI with capabilities slightly stronger than a search engine). However, after completion and rewriting with the current paid AI, some arguments from the main text through the appendix have been refined and no longer retain quite their original flavor.

Since a complete dichotomy theorem has also been submitted and is awaiting public display online,
let me borrow a metaphor of my own: my colleagues' toes have already touched the finish line, and once their paper is displayed publicly, their heels will have crossed it too.
I want to try to use this short interval to organize an AI-assisted version, let my own toes touch the finish line as well, and experience the satisfaction of following them across it, bringing a stage of my pursuit of more than twenty years to a close.
This sprint version rests on the major theorem of~\cite{en:liu2026disequality}; the task of this appendix is to identify the parts of the nine-class group framework that can directly use it, and then complete the remaining proofs.

This appendix has three parts, in order:
\begin{enumerate}[label=\textup{(\arabic*)}]
 \item a brief analysis determining which cases among the nine group classes can obtain binary disequality and invoke the new theorem;
 \item proofs for the cyclic groups, focusing on completing the remaining $C_1$ cases with AI assistance, organizing the existing $C_k$ proof,
       and citing the independent $C_2$ result;
 \item the dichotomy theorem for odd dihedral groups. The proof presented in this appendix was completed independently by AI after it was given the framework and research questions.
\end{enumerate}

\paragraph*{Why retain the original route?}
After completing this sprint with the new theorem, I still hope to follow the original framework and complete direct proofs, class by class, without invoking that theorem. This is a second route to the finish, which this appendix does not develop.

I conjecture that the organization of the proof in~\cite{en:liu2026disequality} may resemble the earlier $\#\mathrm{EO}$ dichotomy in that new auxiliary functions are continually obtained and added during the investigation.
Studying which auxiliary functions can be obtained, and which properties the original functions must satisfy when they cannot be obtained, reflects what I call the ``auxiliary role'' and ``self-discipline,'' respectively.
The continual addition of auxiliary functions is a part of the $\#\mathrm{EO}$ framework that I particularly like; my favorite strategic method within it is realnumberizing. The basic idea of this framework is not complicated, but proving all cases requires handling many stages. I prefer to separate the stages, taking each new group situation as a stage and dividing the cases by group, so that the conditions for each case are clearer.

As a metaphor, such a proof is like explaining how an army starts with a commander alone, gradually gains special forces, marines, and a navy, and ultimately forms three armed services. Each stage has its own personnel and capabilities, along with corresponding discipline and constraints; the proof must explain the auxiliary role and self-discipline at each stage.
I personally prefer to use the nine-class group classification directly: from the outset, the complete binary group identifies the auxiliary capabilities already possessed by the class and the restrictions imposed by its inability to produce other binary functions, after which the analysis proceeds within that class.

The marines have their own personnel structure and discipline, as does the navy. When studying them separately, one can state their respective auxiliary capabilities and self-discipline directly, without taking ``how the marines develop into a navy'' as the main thread of the proof and then distinguishing cases in which such development is possible from those in which it is not: which personnel are added in the former cases, and which conclusions follow from the original discipline in the latter.
I favor the nine-class group framework precisely because it makes these characteristics of each class easier for me to see.

As discussed in the first version, the number of available binary auxiliary functions and their distribution in the matrix space or subspace under consideration affect the strength of arity reduction. When the number of auxiliary directions exceeds the dimension of the relevant space, the auxiliary power is usually already strong enough that one can expect substantial arity-reduction results; the proof methods for these groups also often share common abstract features.
Explaining these common features to AI can help advance direct proofs for the polyhedral groups and higher-order dihedral groups.
Thus, after this sprint using~\cite{en:liu2026disequality}, I remain willing to work through the original framework once more: groups with stronger auxiliary power may admit similar methods, although completing the proofs one by one requires substantial work; groups with fewer auxiliary directions than, or exactly as many as, the dimension of the relevant space come in fewer types, and their proofs often have more individual character and more ``bite.''

Below we develop only the first route. To facilitate the passage from arity-reduction analysis to complexity dichotomies, we first restate the required model and notation.
The product-type class $\mathcal P$ is used in the EO classification for the order-one group and in the final tractability criterion for odd dihedral groups;
the binary pairing product class $\mathcal P_{\mathcal B}$ describes functions that have already decomposed into binary factors during arity reduction.

\paragraph*{The problem, notation, and reductions in this appendix}
An $m$-ary Boolean function is a map $f:\{0,1\}^m\to\overline{\mathbb Q}$.
The function set is fixed and finite; function values are given by exact representations of algebraic numbers.
We call what is usually known as a tensor network a function-net.
A function-net consists of an underlying graph, which may have self-loops and parallel edges, a function at each vertex, and an ordering of the edge ends at each vertex.
External half-edges are also ordered. Summing over the internal edge variables while retaining the external variables gives the net-function
\[
 f_G(\mathbf x)=\sum_{\mathbf z\in\{0,1\}^{E(G)}}
       \prod_{v\in V(G)}f_v(\mathbf z,\mathbf x).
\]
Each internal variable appears at two edge ends; a self-loop appears twice at the same vertex.
When there are no external half-edges, the net-function is a constant, and the problem of computing it is denoted by $\Holant(\mathcal F)$.
The empty product is $1$. Zero functions and nullary functions are treated as known scalars.
The associativity of finite summation allows any sub-function-net to be combined first, so a fixed gadget can replace its net-function.

The net-function of a gadget is called an ordinarily realizable function; tensor factors obtained by decomposition are also treated as available functions, according to Lemma~\ref{en:lem:framework-available-closure}.
Function slices are used to read the values on specified inputs, and interpolation reductions are given where they are used.
Known nonzero scalars can be compensated for: if a gadget realizes $cf$ and is substituted at $r$ locations in an instance, the answer is multiplied by $c^r$, and this nonzero factor is divided out at the end.

A binary function is written as a matrix according to the order of its two variables. Serial composition and swapping the two external ports give matrix multiplication and transpose, respectively.
Throughout the paper, we fix
\[
 I=\begin{pmatrix}1&0\\0&1\end{pmatrix},\qquad
 N=\begin{pmatrix}0&1\\1&0\end{pmatrix},\qquad
 Y=\begin{pmatrix}0&1\\-1&0\end{pmatrix}.
\]
$N$ is binary disequality, and $Y$ is antisymmetric.
In the conventional notation for the quaternion group, the other two axes are
$X=iN$ and $Z=i\diag(1,-1)$.

\paragraph*{Vertex functions and edges after a change of basis}
For an invertible matrix $P$, let $\widehat f=(P^{-1})^{\otimes m}f$.
Combining the two adjacent matrices on each old edge gives the new edge kernel $Q=P^{\mathsf T}P$.
Thus the entire problem becomes exactly a function-net with edge kernel $Q$ and vertex functions $\widehat{\mathcal F}$.
For a binary function $A$, the actual transformation is
\begin{equation}\label{en:eq:binary-congruence}
 \widehat A=P^{-1}AP^{-\mathsf T},\qquad
 \widehat A Q=P^{-1}AP.
\end{equation}
Combining the new binary vertex with a new edge gives the port transfer action represented by the latter matrix.
The old edge kernel is preserved only when $P^{\mathsf T}P=I$.
Here and below, $\mathsf T$ denotes transpose without conjugation.

In particular, take
\begin{equation}\label{en:eq:K-basis}
 K=\frac1{\sqrt2}\begin{pmatrix}1&1\\ i&-i\end{pmatrix},
 \qquad K^{\mathsf T}K=N.
\end{equation}
The new problem has fixed $N$ edges, and its binary vertex functions are $K^{-1}AK^{-\mathsf T}$.
In this model, when two ports are shorted with a binary vertex $A$, the effective contraction matrix is $NAN$, whereas attaching $A$ to one port gives the transfer matrix $AN$.
We also denote this native $N$-edge model by $\Holant(N\mid\Gamma)$ below, synonymously with $\Holant_N(\Gamma)$ in the main text; regarding each $N$ edge as a binary vertex gives the corresponding bipartite function-net.

\paragraph*{Two distinct product classes}
The function class $\mathcal P$ is the product-type class for counting constraint satisfaction problems: a function can be written as a product of unary functions, equality constraints, and disequality constraints, with variables allowed to recur in these factors.
It includes generalized equality functions of every arity.
In contrast, $\mathcal P_{\mathcal B}$ consists of the zero functions and the tensor products of full-rank binary functions from the current group over a perfect matching, allowing a nonzero overall scalar.
These binary directions have been added to the current function set according to the general framework.
In the latter class, each variable appears in exactly one binary factor.
When all binary matrices are full-rank diagonal or antidiagonal matrices, a nonzero $2m$-ary
$\mathcal P_{\mathcal B}$ function has exactly $2^m$ support inputs, and
$\mathcal P_{\mathcal B}\subseteq\mathcal P$.
The quaternary equality function $=_4=\ket{0000}+\ket{1111}$ belongs to $\mathcal P$,
but not to $\mathcal P_{\mathcal B}$.

We also use three known tractable classes below. Functions in the affine class $\mathcal A$ have the form
$c\,\mathbf1_{L}(x)i^{Q(x)}$, where $L$ is an affine subspace over $\mathbb F_2$ and
$Q$ is a quadratic polynomial modulo four whose cross-term coefficients are even.
Let $\alpha=e^{\pi i/4}$ and $T=\diag(1,\alpha)$, and define
$\mathcal A_\alpha=\{T^{\otimes m}f:f\in\mathcal A,\ \arity(f)=m\}$.
The local affine class $\mathcal L$ requires that, for each $\sigma\in\supp(f)$,
the function $\alpha^{\sum_j\sigma_jx_j}f(x)$ belongs to $\mathcal A$.
These classes all include the zero functions. A $2d$-ary function is an EO function if its support contains only inputs of weight $d$.
An EO function has the ARS property if
$f(x)=\overline{f(\bar x)}$, where $\bar x$ denotes bitwise complementation.
This equation describes only the real structure of the function values; function-net summation still involves no conjugation.

\subsection{Which Groups Have Binary Disequality?}\label{en:subsec:road-N}
\label{en:sec:binary-group}
Following the group framework in the foundations, $\mathcal B$ contains all determinant-one representatives of the available binary functions and is closed under multiplication and transpose; its finiteness was proved in Lemma~\ref{en:lem:framework-available-closure}.
The projective quotient group is
\[
 \mathcal G=\mathcal B/\{I,-I\}.
\]
This section only needs to determine which groups contain a symmetric, traceless binary matrix.
Orthogonally diagonalizing such a matrix and then applying a simple orthogonal change of basis gives binary disequality $N$.

\begin{lemma}\label{en:lem:N-symmetric}
There exists a complex orthogonal holographic transformation after which the current problem ordinarily realizes a nonzero multiple of $N$ if and only if $\mathcal B$ contains a nonzero, symmetric, traceless matrix.
\end{lemma}
\begin{proof}
Orthogonal congruence is also similarity, so it preserves symmetry and trace.
Conversely, suppose $A\in\mathcal B$ is symmetric and traceless.
Its determinant is one, so its eigenvalues are $i,-i$. Corresponding eigenvectors $u,v$ satisfy
$u^{\mathsf T}v=0$. Both eigendirections are non-isotropic: if, for example, $u^{\mathsf T}u=0$,
then for the nondegenerate bilinear form in dimension two, $u^\perp=\mathbb C u$, forcing $v$ to be collinear with $u$, a contradiction.
Normalizing them separately gives an orthogonal eigenbasis that transforms $A$ into $i\diag(1,-1)$.
Then applying $2^{-1/2}\left(\begin{smallmatrix}1&1\\1&-1\end{smallmatrix}\right)$
gives $iN$. All scalars can be compensated for.
\end{proof}

\begin{lemma}\label{en:lem:involution-route}
Suppose the finite group $\mathcal G$ has even order. If $Y\notin\mathcal B$, then $N$ can be obtained orthogonally.
If $Y\in\mathcal B$, then $N$ can be obtained orthogonally if and only if $[Y]$ belongs to a Klein four subgroup of $\mathcal G$.
\end{lemma}
\begin{proof}
Pair the elements that are neither the identity nor of order two as $g,g^{-1}$.
Since the group order is even, the number of nonidentity elements of order two is odd.
Transpose acts as an involution on this odd set, so it fixes some coset $[A]$.
A determinant-one lift of a nonidentity coset of order two satisfies $A^2=-I$ and $\operatorname{tr}A=0$.
Since its coset is fixed, $A^{\mathsf T}=\pm A$.
The only antisymmetric two-dimensional matrices of determinant one are $\pm Y$.
Thus, when $Y\notin\mathcal B$, there must be a symmetric traceless matrix.

Now suppose $Y\in\mathcal B$. Every symmetric traceless matrix
$A=\left(\begin{smallmatrix}a&b\\b&-a\end{smallmatrix}\right)$ satisfies $AY=-YA$.
Hence $[A],[Y]$ are two distinct commuting elements of order two and generate a Klein four group.
Conversely, if $[A]\ne[I],[Y]$ is an element of a Klein subgroup containing $[Y]$,
then $AY=\pm YA$. If the two matrices commute, then $A=uI+vY$; zero trace forces $u=0$,
and hence $[A]=[Y]$, a contradiction. Therefore $AY=-YA$, and direct comparison of the four entries gives
$A^{\mathsf T}=A$ and $\operatorname{tr}A=0$. Apply Lemma~\ref{en:lem:N-symmetric}.
\end{proof}

In the classification of finite three-dimensional rotation groups, this criterion has the following direct interpretation.
\begin{center}
\begin{tabular}{@{}>{\raggedright\arraybackslash}p{.56\linewidth}>{\raggedright\arraybackslash}p{\dimexpr.44\linewidth-2\tabcolsep\relax}@{}}
\toprule
Projective group & Can $N$ be obtained orthogonally?\\
\midrule
Cyclic groups of odd order (including the trivial group) & No\\
Cyclic groups of even order & If and only if $Y\notin\mathcal B$\\
Odd dihedral groups & If and only if $Y\notin\mathcal B$\\
Even dihedral groups and the Klein four group & Yes\\
Tetrahedral, octahedral, and icosahedral rotation groups & Yes\\
\bottomrule
\end{tabular}
\end{center}
Here the dihedral group $D_{2n}$ has order $2n$, and odd or even refers to $n$.
This table does not require choosing an arbitrary Klein subgroup and assuming that it is preserved by transpose.
Odd cyclic groups have no elements of order two. Cyclic groups have no Klein subgroups; in an odd dihedral group, the product of two distinct elements of order two has nontrivial odd order, so it too has no Klein subgroup.
In an even dihedral group, the central half-turn and any other half-turn generate a Klein subgroup, so every element of order two belongs to such a subgroup.
Likewise, each half-turn in the three regular polyhedral rotation groups belongs to a triple of half-turns about pairwise perpendicular axes.
This can be checked directly in coordinates: take the tetrahedron with vertices $(1,1,1),(1,-1,-1),(-1,1,-1),(-1,-1,1)$;
the octahedron with the three coordinate unit vectors and their negatives as its six vertices; and the icosahedron with vertices
$(0,\pm1,\pm\varphi),(\pm1,\pm\varphi,0),(\pm\varphi,0,\pm1)$,
where $\varphi=(1+\sqrt5)/2$.
The three half-turns about the coordinate axes preserve each corresponding vertex set.
These are precisely the three rotations of order two for the tetrahedron; the half-turn axes of the icosahedron are conjugate by edge transitivity.
The octahedron also has half-turn axes through the midpoints of opposite edges; the representative axis $(1,1,0)$, together with
$(1,-1,0)$ and $(0,0,1)$, gives the required three perpendicular axes.
Thus every element of order two lies in a Klein subgroup, and Lemma~\ref{en:lem:involution-route} applies.

\begin{theorem}[The existing dichotomy with binary disequality\cite{en:liu2026disequality}]
For every finite set $\mathcal H$ of algebraic complex-valued Boolean functions,
$\Holant(\mathcal H\cup\{N\})$ has an $\FP/\sharpP$-hard dichotomy,
and its tractability criterion is decidable from exact algebraic input.
\end{theorem}
Thus an actually realizable symmetric traceless binary matrix suffices to resolve the corresponding branch.
The $N$ in this theorem is an available vertex function, rather than a kernel assigned to edges after a general change of basis.

\subsection{Complete Proofs for the Cyclic Groups}\label{en:subsec:road-cyclic}
This section has two subsections, for $C_1$ and $C_k\ (k\ge3)$, each treating the upper case and then the lower case;
the independent result for $C_2$ is stated here first.
Under the preceding decomposition convention, nonzero quaternary functions in the upper case are pairing products of binary functions from the current group, while the lower case contains an indecomposable quaternary function.
The arity-reduction proofs below use common structural lemmas for pairing products from Section~\ref{en:subsec:road-odd};
these lemmas are proved in full at the cited location.

\paragraph*{Citation for the cyclic group of order two}\phantomsection\label{en:sec:c2-citation}
The complete complexity dichotomy for the cyclic group of order two was established by Yuxuan Zhou in independent work\cite{en:zhou-draft}.
I am not an author of that work, and only cite its result here.
During discussions at the Yanqihu campus of the University of Chinese Academy of Sciences, I heard presentations of most of that manuscript and, following the explanations of its author, Yuxuan Zhou, checked the portions I heard at the time.
For representations that can obtain binary disequality, Section~\ref{en:subsec:road-N} already provides a resolution;
the remaining antisymmetric representation, in ordinary equality-edge coordinates, is
\[
 \mathcal B=\{\pm I,\pm Y\},\qquad
 Y=\begin{pmatrix}0&1\\-1&0\end{pmatrix},\qquad
 \mathcal G\cong C_2.
\]
This appendix retains the citation to that independent result without reproducing its proof.

\addtocontents{toc}{\protect\setcounter{tocdepth}{4}}
\subsubsection{\texorpdfstring{The $C_1$ Case}{The C1 Case}}\label{en:subsec:cyclic-c1}

\paragraph{Upper case}\label{en:subsec:c1-upper}\label{en:sec:c1-strict}
In native $N$-edge coordinates, the only binary functions for $C_1$ are nonzero multiples of $N$.
The decomposition lemma gives two binary factors for each upper-case quaternary function, so these functions are all $\lambda N\otimes N$, up to port permutations.
We first reduce the problem to a single function, then organize the common phase of its EO part, and finally connect two copies of the function to force overall realnumberizing.

\subparagraph*{Quaternary pairing products and the proof objective}

Use the coordinates of
\eqref{en:eq:K-basis}, with fixed edge kernel $N$, and write $\Gamma$ for the current finite
function set. Every nonzero function is assumed to have even arity, and we assume that
\begin{enumerate}[label=\textup{(C\arabic*)}]
\item\label{en:item:c1-binary} every nonzero binary function realizable by an ordinary gadget
      is a nonzero scalar multiple of $N$;
\item\label{en:item:c1-quaternary} every nonzero quaternary function realizable by an ordinary
      gadget has the form $\lambda N\otimes N$ under some port pairing, where $\lambda\ne0$.
\end{enumerate}
By the decomposition lemma, the second condition is exactly the quaternary condition for the $C_1$ upper case. The following proof establishes the dichotomy under these two conditions.

For even $m$, write
\[
 \mathrm{EO}^{m}=\{\alpha\in\{0,1\}^{m}:|\alpha|_1=m/2\},\qquad
 r(\alpha)=|\alpha|_1-|\alpha|_0.
\]
The EO restriction of a function retains its values on this layer and is zero elsewhere.
We write $[m]=\{1,\ldots,m\}$.

\begin{theorem}[Dichotomy for the $C_1$ upper case]
\label{en:thm:c1-strict-dichotomy}
Let $\Gamma$ be a finite set of algebraic complex-valued functions of even arity satisfying
\ref{en:item:c1-binary} and \ref{en:item:c1-quaternary}. Then $\Holant(N\mid\Gamma)$ is in $\FP$
or is $\sharpP$-hard. Apart from the direct zero-value and common-one-sided EO cases,
the proof either transforms the original problem by a Turing equivalence into an ordinary
real-valued Holant problem, or obtains a hardness witness from a finite set of actual EO gadgets.
\end{theorem}

\subparagraph*{Shared zero behavior for fixed vertex labels and reduction to one function}

\begin{lemma}\label{en:lem:c1-active-single}
Fix a list of vertex functions and their port orders, and vary the perfect matching that
closes all ports. Either all resulting closed function-nets have value zero, or all have nonzero
value. In the latter case, reconnecting only two edges changes the value by a factor in
$\{1/2,1,2\}$. Consequently, either every nonempty instance has value zero, or the original
problem is Turing equivalent to
\[
 \Holant(N\mid\{H\}),
\]
where $H$ is an even-arity function realizable by an ordinary gadget and every complete
closure of a single $H$ vertex is nonzero.
\end{lemma}
\begin{proof}
Cut the two edges to be reconnected. The remaining function-net has a quaternary net-function $q$.
If $q=0$, both closures have value zero. Otherwise, \ref{en:item:c1-quaternary} gives
$q=\lambda N\otimes N$, allowing a port permutation. Every perfect-matching closure has
value $2\lambda$ or $4\lambda$. Any two perfect matchings are connected by finitely many
two-edge reconnections, so zero behavior and the stated nonzero ratios propagate along this path.

Call $f\in\Gamma$ active if one complete closure of a single vertex is nonzero;
equivalently, every such closure is nonzero. If an instance contains an inactive function,
reconnect its fixed list of vertices so that each vertex closes separately. This gives value
zero, and hence the original instance also has value zero. Delete the inactive functions.
If none remain, every nonempty instance has value zero, while the empty function-net has value $1$.
Otherwise, let $H$ be the tensor product of all remaining functions.
Expanding $H$ gives one reduction direction. In the other direction, simulate an original
vertex with one copy of $H$, retain the required factor, and close every other factor by a
previously chosen nonzero closure. Divide out the resulting known nonzero scalar.
There is at least one nonzero closure of $H$. Applying the shared zero behavior to its fixed
list of factor labels shows that every closure is nonzero.
\end{proof}

\begin{corollary}\label{en:cor:c1-active-powers}
For every fixed $k\ge1$, every complete closure of a single $H^{\otimes k}$ vertex is nonzero.
The same holds for any function $F$ obtained from $H^{\otimes k}$ by adding $N$ self-loops.
In particular, $F|_{\mathrm{EO}}\ne0$, and every two-port $N$ shortloop of $F$ is nonzero.
\end{corollary}
\begin{proof}
Expand $H^{\otimes k}$ into $k$ copies of $H$. Closing the copies separately gives a nonzero
value, so all closures are nonzero. Every complete closure of $F$ is also a complete closure
of this list of $k$ copies of $H$. If its EO restriction or one two-port shortloop were zero,
closing the remaining ports would give zero, a contradiction.
\end{proof}

\subparagraph*{EO complement symmetry and a common phase}

\begin{lemma}\label{en:lem:c1-eo-symmetry}
Under \ref{en:item:c1-binary}, the EO restriction of every even-arity function $f$ realizable
by an ordinary gadget satisfies $f(\alpha)=f(\bar\alpha)$.
\end{lemma}
\begin{proof}
Induct on arity. The binary case follows from \ref{en:item:c1-binary}.
An $N$ shortloop deletes one $0$ and one $1$, so it commutes with taking the EO restriction.
By induction, every $N$ shortloop of $f|_{\mathrm{EO}}$ is complement symmetric.
Let $R(\alpha)=f|_{\mathrm{EO}}(\alpha)-f|_{\mathrm{EO}}(\bar\alpha)$.
Then $R$ is an EO function and all its $N$ shortloops vanish.
If an EO function $R$ of arity at least four is nonzero, choose three positions of a nonzero
input forming $100$. Let $a,b,c$ be its values on the three permutations of this local word,
with $a\ne0$. The three shortloops at the corresponding remaining inputs give
$a+b,a+c,b+c$, which cannot all vanish. This contradicts the vanishing of every shortloop
of $R$. Thus $R=0$, completing the induction.
\end{proof}

\begin{lemma}\label{en:lem:c1-common-phase}
For the function $H$ from Lemma~\ref{en:lem:c1-active-single}, there is a nonzero algebraic number
$\theta$ such that $H|_{\mathrm{EO}}/\theta$ is a nonzero rational-valued, complement-symmetric
function. After renaming $H/\theta$ as $H$, the function $H^{\otimes k}|_{\mathrm{EO}}$ is
rational-valued and complement symmetric for every $k\ge1$.
\end{lemma}
\begin{proof}
For a perfect matching $M$ of the ports of $H$, add $N$ self-loops along $M$ and denote the
closure value by $z_M$. Fix $M_0$ and set $\theta=z_{M_0}\ne0$.
Lemma~\ref{en:lem:c1-active-single} gives $z_M=2^{t_M}\theta$, where $t_M\in\mathbb Z$.

Use the common value on each pair of complementary EO inputs as an unknown.
Each perfect-matching closure supplies a linear equation with rational coefficients and
right-hand side $z_M$. This system has full rank on the complement-symmetric EO subspace.
Indeed, suppose a nonzero complement-symmetric EO function $R$ lies in its homogeneous kernel.
Choose three positions of a support input with local word $100$, and write $a,b,c$ for the
values on its three permutations, with $a\ne0$. The shortloop values $a+b,a+c,b+c$ cannot
all vanish. We may therefore reduce arity while retaining a nonzero function, EO support,
and complement symmetry. Eventually we obtain a nonzero binary function $\lambda N$.
Closing it gives $2\lambda\ne0$, contradicting membership in the homogeneous kernel of all
perfect-matching equations. Dividing the system by $\theta$ therefore gives a unique rational solution.

Apply the same argument to $H^{\otimes k}$, using Corollary~\ref{en:cor:c1-active-powers} for
activity. Fix a nonzero EO input of $H$. Its $k$-fold concatenation has a nonzero rational value,
which forces the new common phase to be rational. Thus the entire EO layer of every tensor power is rational-valued.
\end{proof}

\subparagraph*{Absorbing imbalance and choosing the least generating imbalance}

If all support imbalances of $H$ are nonnegative, or all are nonpositive, every $N$ edge
contributes one $0$ and one $1$, and hence every nonzero global term has total imbalance zero.
No strictly imbalanced term can contribute. The original problem is therefore equal,
instance by instance, to $\Holant(N\mid H|_{\mathrm{EO}})$, and is classified by the general
EO dichotomy and its FP upgrade~\cite{en:meng2025eo,en:guan2026lp}.
We henceforth consider the case where positive and negative support imbalances both occur.

\begin{lemma}\label{en:lem:c1-absorb-charge}
There is an algebraic diagonal matrix $D=\diag(c^{1/2},c^{-1/2})$, with $c\ne0$, such that
$\widetilde H=D^{\otimes m}H$ is globally complement symmetric.
Since $DND=N$, this is an edge-preserving holographic transformation. It leaves the EO
layer of every tensor power unchanged.
\end{lemma}
\begin{proof}
First prove that the support is closed under complementation. If $r(\alpha)=a>0$, choose
a support input $\beta$ with $r(\beta)=-b<0$. Concatenating $b$ copies of $\alpha$ and $a$
copies of $\beta$ gives an EO input. Complement symmetry of the EO layer of the tensor power gives
\[
 H(\alpha)^bH(\beta)^a=H(\bar\alpha)^bH(\bar\beta)^a.
\]
The left side is nonzero, so both complementary values are nonzero. Negative imbalance is
handled in the same way, and zero imbalance follows directly from EO symmetry.

Set $u(\alpha)=H(\alpha)/H(\bar\alpha)$. For two inputs $\alpha,\beta$ of the same imbalance,
the concatenation $\alpha\bar\beta$ shows that $u(\alpha)=u(\beta)$; denote this value by $u_r$.
For any integer relation $\sum_jn_jr_j=0$, replace negative multiplicities by complementary
inputs. EO complement symmetry yields $\prod_ju_{r_j}^{n_j}=1$.
Thus $r\mapsto u_r$ defines a multiplicative homomorphism on the generated subgroup
$g\mathbb Z\le\mathbb Z$. Choose an algebraic root $c^g=u_g$. Then $u_r=c^r$.
On imbalance $r$, the transformation $D^{\otimes m}$ multiplies the value by $c^{-r/2}$,
so the transformed complementary values agree. At imbalance zero the multiplier is one.
The edge-kernel identity follows by direct multiplication.
\end{proof}

Rename $\widetilde H$ as $H$. Conditions \ref{en:item:c1-binary} and \ref{en:item:c1-quaternary}
remain valid because $DND=N$. The function $H$ remains active, is globally complement
symmetric, and has a rational EO layer in every tensor power.

\begin{lemma}[Phase character]\label{en:lem:c1-phase-character}
Let $R=\{r(\alpha):\alpha\in\supp(H)\}$ and
$g=\gcd\{|r|:r\in R,r\ne0\}$. The phase class of a nonzero function value in
$\mathbb C^*/\mathbb R^*$ depends only on its imbalance and defines a homomorphism
\[
 \chi:g\mathbb Z\longrightarrow\{\mathbb R^*,i\mathbb R^*\}.
\]
Some fixed tensor power $H^{\otimes k}$ has a nonzero input of imbalance exactly $g$.
\end{lemma}
\begin{proof}
The input $\alpha\bar\alpha$ is EO, so global complement symmetry gives
$H(\alpha)^2\in\mathbb R^*$. Every phase class therefore has order at most two.
If $r(\alpha)=r(\beta)$, then $\alpha\bar\beta$ is EO and
$H(\alpha)H(\beta)\in\mathbb R^*$, so the two phase classes agree.
For any integer zero-sum relation, replace negative multiplicities by complementary inputs.
The resulting EO concatenation has a nonzero real value, so the corresponding product of
phase classes is the identity. This proves that the homomorphism is well defined.
Finally, express $g$ by a B\'ezout integer relation. Complement-closed support allows negative
multiplicities to be replaced by positive multiplicities of complementary inputs. The
resulting concatenation is nonzero and gives the required fixed tensor power.
\end{proof}

\begin{lemma}[Shaping while retaining the least generating imbalance]\label{en:lem:c1-shape-generator}
From the tensor power in the preceding lemma, an ordinary gadget realizes a $g$-ary function
\[
 F=Q+a(\ket{0^g}+\ket{1^g}),\qquad a\ne0,
\]
where $Q$ is a nonzero rational-valued, complement-symmetric EO function, $a$ has phase class
$\chi(g)$, and $g\ge6$. Every complete closure of a single $F$ vertex is nonzero.
\end{lemma}
\begin{proof}
Retain a nonzero input of imbalance $g$. If it still contains both $0,1$, then $g>0$
allows us to choose a local word $011$. Denote the values on its three permutations by
$a_0,b_0,c_0$, with $a_0\ne0$. The three shortloop values
$a_0+b_0,a_0+c_0,b_0+c_0$ cannot all vanish. Choose a shortloop that remains nonzero.
It deletes one $0$ and one $1$, preserves imbalance, and lowers arity. After finitely many
steps, the retained input is all ones and the remaining arity is $g$.
Shortloop preserves imbalance, so every nonzero final imbalance is still an integer
multiple of $g$. On $g$ ports the only nonzero possibilities are $\pm g$.
Global complement symmetry gives two equal nonzero endpoint values.
The preceding lemmas preserve rationality on the EO layer and activity.
Every summand of imbalance $g$ lies on the same real line $\chi(g)$, so the nonzero sum $a$
still has this phase class. If $g=2$, a nonzero binary diagonal entry occurs; if $g=4$, a
nonzero quaternary endpoint occurs. These contradict \ref{en:item:c1-binary} and
\ref{en:item:c1-quaternary}, respectively. Hence $g\ge6$.
\end{proof}

\subparagraph*{Turning EO shortloops into matching products}

Write $\partial_{ij}f=\sum_{a\ne b}f(x_i=a,x_j=b,\ldots)$, retaining the order of the
remaining ports. For the shaped function $F$, every
$J_{ij}=\partial_{ij}F=\partial_{ij}Q$ is a nonzero rational-valued, complement-symmetric
EO function, and every complete closure of a single vertex is nonzero.
The ARS--EO theorem~\cite[Theorem 3.1]{en:cai2020arseo}, applied to the finite set
$\mathcal E=\{J_{ij}:i<j\}$, gives tractability if $\mathcal E\subseteq\mathcal A$ or
$\mathcal E\subseteq\mathcal P$, and hardness otherwise.
Here $\mathcal A$ is the class of fourth-root quadratic phases on affine supports:
its functions have the form $\lambda\mathbf1_{Ax=b}i^{L(x)+2Q_2(x)}$, where $L$ is an
integer linear polynomial and $Q_2$ is an integer quadratic polynomial.
The class $\mathcal P$ is the product-type class defined earlier.
Our functions are real-valued and complement symmetric, and hence satisfy the ARS condition
of that theorem. In the hard case, expansion of the fixed gadgets transfers hardness to
the original problem. We continue in a branch where the entire set lies in one of the two
tractable classes.

\begin{lemma}[Cleanup of active EO functions in the tractable classes]\label{en:lem:c1-cleanup}
Let $J$ be a nonzero rational-valued, complement-symmetric EO function realizable by an
ordinary gadget. Suppose every complete closure of a single vertex is nonzero.
If $J\in\mathcal P\cup\mathcal A$, then under \ref{en:item:c1-quaternary}, $J$ is a nonzero
scalar multiple of a perfect-matching product of $N$ functions.
\end{lemma}
\begin{proof}
First suppose $J\in\mathcal P$. Decompose it into the connected blocks of its equality and
disequality constraints. Complement-closed support excludes blocks of fixed variables.
Each block supports two complementary assignments with weights $a,b\ne0$.
Assignments to different blocks are independent, whereas the total support is EO, so each
block is itself balanced. Comparing global complementation while fixing the other blocks
gives $b/a=\pm1$ for each block. If a block has ratio $-1$, close it by pairing opposite
positions of a support assignment and close every other block internally. The total value
is zero, a contradiction. Thus the two weights of every block agree.
Closing all other blocks with nonzero value extracts any selected block by an ordinary gadget.
If a block has arity at least four, repeatedly use $N$ to short one port of its $0$ class
with one port of its $1$ class. This yields a nonzero scalar multiple of a port permutation
of $D_4=\ket{0011}+\ket{1100}$, contradicting \ref{en:item:c1-quaternary}.
Consequently every block is a scalar multiple of $N$.

Now suppose $J\in\mathcal A$. Parameterize its affine support as $x_i=\ell_i(z)+a_i$,
where each $\ell_i$ is a linear form over $\mathbb F_2$.
The EO condition says that $\sum_i(-1)^{\ell_i(z)+a_i}=0$ for every $z$.
Walsh characters associated with distinct linear forms are linearly independent.
For each fixed $\ell$, the numbers of coordinates equal to $\ell(z)$ and to $\ell(z)+1$
are therefore equal. Pair these coordinates to obtain
\[
 J=\prod_{t=1}^d[x_t\ne y_t]A(x_1,\ldots,x_d),\qquad A\in\mathcal A.
\]
Choose a nonzero value $\lambda$ of $A$. Rationality implies that every nonzero value is
$\pm\lambda$. Take two copies of $J$ and, for each coordinate, join $y_t$ of the first
copy to $x_t$ of the second by an $N$ edge. Retain $x_t$ of the first copy and $y_t$ of the
second. Each support assignment extends uniquely, and the resulting function is
$\lambda^2\prod_t[x_t\ne y_t]\chi(x)$, where $\chi$ indicates the original affine support.

If $\chi$ is a proper affine subspace, choose free variables among the original coordinates
and a relation for one dependent variable, $x_j+\sum_{t\in T}x_t=c$.
Complement-closed support implies that the number of variables in this row is positive and
even. Add $N$ self-loops to the pairs $x_t,y_t$ of all other variables, eliminating the
other dependent variables and the free variables absent from the row. Only this equation
remains, accompanied by a nonzero power of $2$.
If the equation has at least four variables, choose two free variables $p,q$ in it and
short $x_p,y_q$ and $x_q,y_p$ across the pairs. This forces $x_p=x_q$, eliminates them, and multiplies
the function by $2$. The equation loses two variables. Continuing until two variables remain,
and retaining their opposite pairs, yields a permutation of $D_4$, a contradiction.

Thus $\chi$ is the full cube and $A(x)=\lambda(-1)^{q(x)}$, where $q$ is a polynomial of
degree at most two over $\mathbb F_2$. If $q$ contains a nonzero mixed quadratic term,
short the pair $x_j,y_j$ of a variable participating in such a term. The result is
\[
 \lambda(-1)^{q(0,z)}
 \bigl(1+(-1)^{q(1,z)-q(0,z)}\bigr).
\]
The exponent difference is a nonconstant affine linear form, so the function is nonzero
and has proper affine support. Repeating the preceding squaring and elimination gadgets
again produces $D_4$, a contradiction. Hence $q$ has only linear terms.
If a linear coefficient is nonzero, closing $J$ along all its original pairs introduces
a factor $1-1=0$, contrary to activity. Thus $q$ is constant, proving the claim.
\end{proof}

\subparagraph*{A structural theorem when every shortloop is a matching product}

We now reconstruct the overall structure of an EO function from the matching structures of its $N$ shortloops. We also write $\partial^{ij}$ for $\partial_{ij}$.

Throughout this argument shortloops use binary disequality. For distinct $i,j$, abbreviate
the shortloop notation as
\[
 \partial^{ij}Q:=\sum_{a\ne b}Q(x_i=a,x_j=b,\ldots).
\]
Only the fixed edge-kernel subscript is omitted; the remaining port order is unchanged.
For $e=\{i,j\}$, write $\partial^e$. If $M$ is a perfect matching on an even number of ports, set
\[
 P_M=\prod_{\{u,v\}\in M}[x_u\ne x_v].
\]

First define the arity-eight exception. Let $F'_8$ have value $1$ on the following support
and value $0$ elsewhere:
\begin{equation}\label{en:eq:c1-f8}
\begin{split}
\supp(F'_8)=\{x\in\{0,1\}^8:\ &\textstyle\sum_{i=1}^8x_i=4,\\
&x_1+x_2+x_3+x_4\equiv x_5+x_6+x_7+x_8\equiv0\pmod2,\\
&x_1+x_2+x_5+x_6\equiv x_1+x_3+x_5+x_7\equiv0\pmod2\}.
\end{split}
\end{equation}
Its support has fourteen words, and it is a 0--1-valued, complement-symmetric EO function.
References to this exception allow a variable permutation and an overall nonzero scalar.

\begin{theorem}[EO arity reduction for the trivial group]\label{en:thm:c1-eo-matching-reduction}
Let $Q$ be a rational-valued, complement-symmetric EO function of arity at least six.
Suppose that for every pair of distinct indices $i,j$, there are a nonzero scalar
$\theta_{ij}$ and a perfect matching $M_{ij}$ on the remaining ports such that
\begin{equation}\label{en:eq:c1-all-contractions}
 \partial^{ij}Q=\theta_{ij}P_{M_{ij}}.
\end{equation}
Then exactly one of the following holds:
\begin{enumerate}[label=\textup{(\arabic*)}]
\item $Q$ is a nonzero scalar multiple of a perfect-matching product of binary disequality functions;
\item $Q$ has arity eight and, after permuting variables, $Q=\lambda F'_8$ for some $\lambda\ne0$.
\end{enumerate}
\end{theorem}

We first give a uniform argument for large arities, and then compute both finite base cases
directly. The following local shortloop rule is used in both parts.

\begin{lemma}[One shortloop of a matching product]\label{en:lem:c1-matching-merge}
Let $M$ be a perfect matching. If $\{a,b\}\in M$, then
\[
 \partial^{ab}P_M=2P_{M\setminus\{\{a,b\}\}}.
\]
If $a,b$ are paired with $a',b'$ in $M$, respectively, and $a'\ne b$, then
\[
 \partial^{ab}P_M
 =P_{(M\setminus\{\{a,a'\},\{b,b'\}\})\cup\{\{a',b'\}\}}.
\]
\end{lemma}
\begin{proof}
In the first case, $(x_a,x_b)=(0,1),(1,0)$ each contribute once.
In the second case, the three constraints $x_a\ne x_{a'}$, $x_a\ne x_b$, and $x_b\ne x_{b'}$
are jointly equivalent to $x_{a'}\ne x_{b'}$; every allowed pair $(x_{a'},x_{b'})$ has exactly
one corresponding $(x_a,x_b)$. All other matching edges remain unchanged.
Moreover, the support of $P_M$ uniquely determines $M$: endpoints of a matching edge are
always opposite, whereas variables on distinct edges can be assigned independently.
Consequently equality of nonzero matching products permits comparison of both their matchings
and their scalar coefficients.
\end{proof}

\begin{lemma}[Vanishing of all disequality shortloops]\label{en:lem:c1-zero-merges}
If a function of arity $m\ge3$, denoted by $R$, satisfies $\partial^{ij}R\equiv0$ for every $i\ne j$,
then $R$ can be nonzero only at $0^m,1^m$.
\end{lemma}
\begin{proof}
Any input containing both $0,1$ has three positions that can be ordered as $001$ or $110$.
For $001\delta$, the three shortloop equations give
\[
 R(001\delta)+R(010\delta)=0,\quad
 R(001\delta)+R(100\delta)=0,\quad
 R(010\delta)+R(100\delta)=0.
\]
Adding the first two equations and subtracting the third gives $2R(001\delta)=0$.
The case $110\delta$ is identical.
\end{proof}

\begin{lemma}[Propagation of a common edge from two shortloops]\label{en:lem:c1-two-merge-propagation}
Suppose $Q$ satisfies \eqref{en:eq:c1-all-contractions}. Let distinct index pairs $p,q$ both be
disjoint from $\{u,v\}$, and suppose $[x_u\ne x_v]$ is a factor of both functions
$\partial^pQ,\partial^qQ$. Then for every index pair $e$ disjoint from
$\{u,v\}\cup p\cup q$, the function $[x_u\ne x_v]$ is also a factor of $\partial^eQ$.
\end{lemma}
\begin{proof}
Suppose instead that $u,v$ are not paired in $M_e$, but have partners $u',v'$, respectively.
Compare $\partial^p\partial^eQ=\partial^e\partial^pQ$.
The right side retains its original factor $[x_u\ne x_v]$, since $e$ avoids its ports.
The left side starts from $M_e$. By Lemma~\ref{en:lem:c1-matching-merge}, if $\{u,v\}$ was
absent, the only way the shortloop at $p$ can create it is $p=\{u',v'\}$.
The analogous comparison with $q$ gives $q=\{u',v'\}$, contradicting $p\ne q$.
\end{proof}

\begin{lemma}[A fixed triangle of shortloops]\label{en:lem:c1-obtain-triangle}
If $Q$ has arity at least ten and satisfies \eqref{en:eq:c1-all-contractions}, there are
pairwise distinct $u,v,r,s,t$ such that $[x_u\ne x_v]$ divides all of
\[
 \partial^{rs}Q,\qquad \partial^{st}Q,\qquad \partial^{rt}Q. \tag{C1.10}
\]
\end{lemma}
\begin{proof}
Fix $\partial^{12}Q$ and reorder the remaining ports according to its matching so that
\[
 \partial^{12}Q=\theta_{12}[x_3\ne x_4][x_5\ne x_6]
 [x_7\ne x_8][x_9\ne x_{10}]\cdots . \tag{C1.11}
\]
Compare $\partial^{34}\partial^{12}Q$ with $\partial^{12}\partial^{34}Q$.
The former has the $d-2$ edges left after deleting $\{3,4\}$ from (C1.11).
In the latter, shorting $1,2$ in $M_{34}$ replaces at most two edges by one new edge.
Thus at most one edge of the resulting matching was absent from $M_{34}$.
Since $d\ge5$ and $d-2\ge3$, we can fix an edge $\{u,v\}$ common to $M_{12},M_{34}$.

Choose three distinct indices $r,s,t$ in $[2d]\setminus\{1,2,3,4,u,v\}$.
Apply Lemma~\ref{en:lem:c1-two-merge-propagation} with the fixed pairs $p=\{1,2\}$ and
$q=\{3,4\}$ to obtain (C1.10). The edge $u,v$ is selected before the auxiliary triangle
and $u,v$ remain fixed throughout the following argument.
\end{proof}

\begin{lemma}[Propagation from a triangle to every shortloop]\label{en:lem:c1-triangle-to-all}
Under Lemma~\ref{en:lem:c1-obtain-triangle}, fix its $u,v,r,s,t$.
For every pair $i,j$ disjoint from $\{u,v\}$,
\[
 [x_u\ne x_v]\mid \partial^{ij}Q. \tag{C1.12}
\]
\end{lemma}
\begin{proof}
If $\{i,j\}$ is disjoint from $\{r,s,t\}$, use any two distinct edges of the triangle in
Lemma~\ref{en:lem:c1-two-merge-propagation}. Otherwise, the set $\{u,v,r,s,t,i,j\}$ has at
most seven indices. Since the total arity is at least ten, choose three distinct indices
$r',s',t'$ outside it. The first case shows that their three shortloops still contain
the fixed factor $[x_u\ne x_v]$, giving a new triangle disjoint from $\{i,j\}$.
Apply the first case again to obtain (C1.12). Only the auxiliary triangle changes; $u,v$
remain fixed.
\end{proof}

\begin{lemma}[Extracting a fixed factor at arity at least ten]\label{en:lem:c1-large-factor}
If $Q$ has arity at least ten and satisfies \eqref{en:eq:c1-all-contractions}, then $Q$ is a
nonzero scalar multiple of a perfect-matching product of binary disequality functions.
\end{lemma}
\begin{proof}
The preceding two lemmas fix $u,v$ such that (C1.12) holds for all pairs avoiding $u,v$.
Split $Q$ into the four slices $Q^{00}_{uv},Q^{01}_{uv},Q^{10}_{uv},Q^{11}_{uv}$ according
to $(x_u,x_v)=00,01,10,11$. For every $i,j\notin\{u,v\}$, the common factor gives
\[
 \partial^{ij}Q^{00}_{uv}=\partial^{ij}Q^{11}_{uv}=0,
 \qquad \partial^{ij}(Q^{01}_{uv}-Q^{10}_{uv})=0. \tag{C1.13}
\]
By Lemma~\ref{en:lem:c1-zero-merges}, these three functions can be nonzero only at their
all-zero and all-one inputs. But $Q$ is an EO function of arity $2d$, with $d\ge5$:
the total weights of these slice endpoints belong to $\{0,2d-2\}$, $\{1,2d-1\}$, or
$\{2,2d\}$, none equal to $d$. Hence
\[
 Q^{00}_{uv}=Q^{11}_{uv}=0,\qquad
 Q^{01}_{uv}=Q^{10}_{uv}=:Q',
\]
and $Q=[x_u\ne x_v]Q'$. Finally use the previously unused shortloop $\partial^{uv}$.
On the one hand, $\partial^{uv}Q=2Q'$; on the other,
\eqref{en:eq:c1-all-contractions} gives $\partial^{uv}Q=\theta_{uv}P_{M_{uv}}$.
Thus $Q'=(\theta_{uv}/2)P_{M_{uv}}$, as required.
\end{proof}

The quantifiers in this argument first produce one fixed witness
\[
 \exists\{u,v\}\ \forall\{i,j\}\cap\{u,v\}=\emptyset,
 \qquad [x_u\ne x_v]\mid \partial^{ij}Q,
\]
and only then extract this factor from $Q$ itself. In particular, $u,v$ do not depend on
$i,j$ or on an input selected in a contradiction argument.

We now complete the arity-six and arity-eight base cases by explicit calculations.
Continue to use $Q,\theta_{ij},M_{ij}$ from \eqref{en:eq:c1-all-contractions}, the functions
$P_M$, and $F'_8$ from \eqref{en:eq:c1-f8}; abbreviate $\{i,j\}$ by $ij$.
This is a structural argument about functions, with no new realizability or factorization
assumptions. At arity six we enumerate linear branches. At arity eight we first determine
all matchings and coefficients and then recover the function.
Uniqueness of recovery follows from Lemma~\ref{en:lem:c1-zero-merges}: if two EO functions
have identical shortloops, their difference can be supported only at the two endpoints,
which the EO condition excludes. Thus the two functions are equal.

\subparagraph*{Arity six: explicit calculation of five two-parameter families}

By a nonzero overall scaling and a variable permutation, first arrange
\[
 \partial^{12}Q=P_{\{34,56\}}.
\]
For a three-element subset $A\subseteq[6]$, let $1_A$ be its indicator word and write
$q_A=Q(1_A)=q_{[6]\setminus A}$. All solutions of the initial condition can be written as
\[
\begin{array}{lll}
q_{134}=a,&q_{156}=-a,&q_{135}=b,\\
q_{146}=1-b,&q_{136}=c,&q_{145}=1-c,
\end{array}
\qquad y_r=q_{12r}\quad(3\le r\le6).
\]
Only the seven numbers $a,b,c,y_3,y_4,y_5,y_6$ remain.
On the three pairs of complementary EO inputs, a quaternary matching product has values
$(0,t,t),(t,0,t)$, or $(t,t,0)$, where $t\ne0$.
Computing the following five shortloops entry by entry gives
\begin{equation}\label{en:eq:c1-finite-six-rows}
\begin{array}{c|ccc}
 e&\multicolumn{3}{c}{\partial^e Q\text{ on the three complementary EO pairs}}\\ \hline
 13&y_4-a&y_5+1-b&y_6+1-c\\
 14&y_3-a&y_5+c&y_6+b\\
 25&y_3+b&y_4+1-c&y_6-a\\
 23&y_4+a&y_5+b&y_6+c\\
 16&y_3+1-c&y_4+b&y_5+a
\end{array}
\end{equation}
For example, the third entry in row $25$ is $q_{126}+q_{156}=y_6-a$, and the first entry
in row $16$ is $q_{123}+q_{236}=y_3+1-c$.
Every row must have exactly one zero entry and two equal nonzero entries.

Let $h_e\in\{1,2,3\}$ denote the unique zero coordinate of row $e$.
Start with rows $13,14$. If $(h_{13},h_{14})$ is $(1,2)$ or $(1,3)$, the equality in the
first row forces all three entries of the second row to vanish. If it is $(2,1)$ or $(3,1)$,
the reverse argument makes the entire first row zero. These cases are excluded.
Only $(1,1),(2,2),(2,3),(3,2),(3,3)$ remain.
Next write $A,B,C$ for the entries of row $25$.
In branch $(2,3)$, the first two rows imply $2A+B+3C=0$; in branch $(3,2)$, they imply
$A+2B+3C=0$. One of $A,B,C$ is zero and the other two equal a nonzero $t$.
Substituting into either equation and dividing by the resulting nonzero positive integer gives
$t=0$, a contradiction. In branches $(2,2)$ and $(3,3)$, the first two entries of row $25$ agree,
so only its third entry can be zero. In branch $(1,1)$, row $25$ still allows all three choices.
This exhausts the possibilities into the following five families, with $w=y_5,z=y_6$:
\begingroup\small\setlength{\arraycolsep}{4pt}
\[
\begin{array}{c|c|ccc|c}
 &(h_{13},h_{14},h_{25})&a&b&c&y_3=y_4\\ \hline
\mathrm I&(1,1,1)&\frac{-w+2z-1}{3}&\frac{w-2z+1}{3}&\frac{-2w+z+1}{3}&a\\
\mathrm {II}&(1,1,2)&\frac{-w+2z-1}{3}&\frac{2w-z+2}{3}&\frac{-w+2z+2}{3}&a\\
\mathrm {III}&(1,1,3)&z&\frac{w-z+1}{2}&\frac{-w+z+1}{2}&z\\
\mathrm {IV}&(2,2,3)&z&w+1&-w&w+2z+1\\
\mathrm V&(3,3,3)&z&-z&z+1&w+2z+1
\end{array}
\]
\endgroup
Each row gives all solutions of its listed zero-entry and equality equations; the parameters
must still satisfy the nonzero conditions. For example, in the first family we initially have
$y_3=y_4=a$ and $b=c+w-z$. Row $25$ then gives $a+b=0$ and $a+1-c=z-a$, which yield the
first row of the table by direct solution.
Now use row $16$ and then row $23$. Set $s=w+z$. Substitution gives
\begingroup\small\setlength{\arraycolsep}{3pt}
\[
\begin{array}{c|c|c|c}
 &16&23\text{ after imposing row 16}&z\\ \hline
\mathrm I&((s+1)/3,0,(2s-1)/3)&(2z-2,3-2z,2z-1)&1\\
\mathrm {II}&(0,(s+1)/3,(2s-1)/3)&(2z-2,4-2z,2z)&1\\
\mathrm {III}&((s+1)/2,(s+1)/2,s)&(2z,\frac12-2z,\frac12+2z)&0\\
\mathrm {IV}&(2s+2,2s+2,s)&(2z+1,1-2z,2z)&0\\
\mathrm V&(s+1,s+1,s)&(2z+1,-2z,2z+1)&0
\end{array}
\]
\endgroup
Since row $16$ must have exactly one zero entry, the first two families require $s=2$,
and the last three require $s=0$.
After substitution in row $23$, set each of its three entries to zero in turn and check
equality of the other two. This gives exactly the last column of the table.
For instance, the three choices in the first family give $z=1,3/2,1/2$; in the latter two
cases the other entries are $(1,2)$ and $(-1,2)$, respectively, so both are excluded.
Each of the other four rows similarly has only its listed value.
Thus exactly five points remain:
\begin{equation}\label{en:eq:c1-finite-six-points}
\begin{array}{c|ccc|cccc|c}
&a&b&c&y_3&y_4&y_5&y_6&Q\\ \hline
1&0&0&0&0&0&1&1&P_{\{13,24,56\}}\\
2&0&1&1&0&0&1&1&P_{\{14,23,56\}}\\
3&0&\frac12&\frac12&0&0&0&0&\frac12P_{\{12,34,56\}}\\
4&0&1&0&1&1&0&0&P_{\{16,25,34\}}\\
5&0&0&1&1&1&0&0&P_{\{15,26,34\}}
\end{array}
\end{equation}
All five functions are matching products. Lemma~\ref{en:lem:c1-matching-merge} verifies all
their shortloops omitted from the table as well, so the arity-six case is exhaustive.

\subparagraph*{Arity eight: fixing four variable pairs with a maximum coefficient}

At arity eight we need not list all $35$ independent function values as in the arity-six
calculation. Instead, first determine every matching and coefficient, and then recover the
function by the uniqueness rule above.

Shortloops on disjoint index pairs $e,f$ commute. Comparing
$\partial^f\partial^e Q=\partial^e\partial^f Q$ using Lemma~\ref{en:lem:c1-matching-merge}
gives $\theta_e/\theta_f\in\{1/2,1,2\}$. If $e,f$ intersect, choose a pair $g$ outside their
union and compare each with $g$. Hence the ratio of any two $\theta_e$ is positive real.
Choose one of these $28$ nonzero coefficients with maximum absolute value. Scale globally
and relabel ports so that
\[
 \theta_{12}=1,\qquad M_{12}=\{34,56,78\},\qquad
 0<\theta_e\le1\quad\text{for every }e.
\]
Compare the two orders of shortloop at $12$ and $34$. One side has coefficient $2$;
the other has coefficient $\theta_{34}$ or $2\theta_{34}$.
Maximality forces $\theta_{34}=1$ and $12\in M_{34}$, whence
$M_{34}=\{12,56,78\}$. The same argument applies to $56,78$.
Writing $M_0=\{12,34,56,78\}$ gives
\begin{equation}\label{en:eq:c1-finite-four-edges}
 \partial^r Q=P_{M_0\setminus\{r\}}
 \qquad(r\in M_0).
\end{equation}
The maximum is always taken over the $28$ shortloops of this same function $Q$;
we neither change the function nor invoke a minimum-counterexample hypothesis.

Call the four pairs blocks $A=(1,2),B=(3,4),C=(5,6),D=(7,8)$.
Consider the pair $13$ crossing two blocks. Compare it separately with the pairs $56,78$.
The double shortloops on the side supplied by \eqref{en:eq:c1-finite-four-edges} both have
coefficient $1$. Thus $56,78$ must either both belong to $M_{13}$ or both be absent:
if exactly one belonged, the comparisons would require both $2\theta_{13}=1$ and
$\theta_{13}=1$. If both belong, the remaining edge must be $24$.
If both are absent, we still have $24\in M_{13}$. Otherwise, creating $24$ by shortloop
at $56$ would require the partners of $2,4$ to be $5,6$, whereas creating the same edge by
shortloop at $78$ would require their partners to be $7,8$, a contradiction.
There are consequently only three possibilities:
\begin{equation}\label{en:eq:c1-finite-three-choices}
\begin{array}{c|c}
 M_{13}&\theta_{13}\\ \hline
 \{24,56,78\}&1/2\\
 \{24,57,68\}&1\\
 \{24,58,67\}&1
\end{array}
\end{equation}
The same reasoning applies to every pair crossing two blocks.
Call the first possibility the internal type: the other two blocks are each paired internally.
Call the latter two possibilities the crossing type: the other two blocks are joined by
one of their two cross-block perfect matchings.

\subparagraph*{Arity eight: the three block partitions must have the same type}

The four blocks have three partitions into two pairs:
$AB|CD,AC|BD,AD|BC$. Fix $AB|CD$, take $e$ crossing $A,B$, and take $f$ crossing $C,D$.
If $e$ has internal type, shortloop at $f$ does not remove an edge of $M_e$, so
$\partial^f\partial^e Q$ has coefficient $1/2$. If $f$ has crossing type, the coefficient
on the other side is $1$ or $2$, an impossibility.
Thus one internal-type $e$ forces every opposite-side $f$ to have internal type.
Comparing back forces every crossing pair on the first side to have internal type as well.
For each block partition, all eight cross-block pairs therefore have internal type, or all
have crossing type.

Different partitions cannot have different types. Compare $e=13$ and $f=25$.
In $M_{13}$, the partner of $2$ is always $4$, so $25$ is not a matching edge.
In $M_{25}$, the partner of $1$ is always $6$, so $13$ is not a matching edge either.
Commuting these shortloops therefore gives $\theta_{13}=\theta_{25}$ directly, and
$AB|CD$ and $AC|BD$ have the same type.
Using $e=13,f=27$ likewise shows that $AB|CD$ and $AD|BC$ have the same type.
Consequently all three partitions have internal type, or all have crossing type.

If all have internal type, \eqref{en:eq:c1-finite-four-edges} and
\eqref{en:eq:c1-finite-three-choices} specify all $28$ shortloops. They are exactly the
shortloops of $\frac12P_{M_0}$. Uniqueness for EO functions gives
\[
 Q=\tfrac12P_{\{12,34,56,78\}}.
\]

\subparagraph*{Arity eight: only two orientation patterns remain in the crossing type}

Assume every cross-block pair has crossing type, so all its coefficients are $1$.
Label the two ports in each block by $0,1$ in their established order.
The direction of an edge joining two blocks is the XOR $p\in\{0,1\}$ of its port labels.
Each direction consists of exactly two edges, which together form a perfect matching of
the two blocks.

Fix $AB|CD$. For every $e$ crossing $AB$ and every $f$ crossing $CD$, commute the
shortloops and compare their coefficients, which are $1$ or $2$. This gives
\begin{equation}\label{en:eq:c1-finite-incidence}
 f\in M_e\quad\Longleftrightarrow\quad e\in M_f.
\end{equation}
Two $AB$ edges $e,e'$ of the same direction occur together, or are both absent, in every
$M_f$. By \eqref{en:eq:c1-finite-incidence}, the matchings $M_e,M_{e'}$ therefore choose the
same direction on $CD$. This defines a direction map $\phi:\{0,1\}\to\{0,1\}$.
It cannot be constant. Otherwise, choose an edge $f$ of its selected direction; the same
equivalence would force $M_f$ to contain all four $AB$ edges, a contradiction.
Thus $\phi$ is a bijection and has the form $p\mapsto p\oplus s_1$.
The other two block partitions similarly give $s_2,s_3\in\{0,1\}$.
Equation~\eqref{en:eq:c1-finite-incidence} also determines each reverse map, leaving no further
choices.

Finally compare two specific pairs of double shortloops to determine the three bits.
For $13,25$, their matchings and the resulting double-shortloop matchings are
\[
\begin{array}{c|c|c}
 &\text{original matching}&\text{after shorting the other pair}\\ \hline
 s_1=0&M_{13}=\{24,57,68\}&\{47,68\}\\
 s_1=1&M_{13}=\{24,58,67\}&\{48,67\}\\ \hline
 s_2=0&M_{25}=\{16,38,47\}&\{47,68\}\\
 s_2=1&M_{25}=\{16,37,48\}&\{48,67\}
\end{array}
\]
Hence $s_1=s_2$. Comparing $13,27$ gives
\[
\begin{array}{c|c|c}
 &\text{original matching}&\text{after shorting the other pair}\\ \hline
 s_1=0&M_{13}=\{24,57,68\}&\{45,68\}\\
 s_1=1&M_{13}=\{24,58,67\}&\{46,58\}\\ \hline
 s_3=0&M_{27}=\{18,36,45\}&\{45,68\}\\
 s_3=1&M_{27}=\{18,35,46\}&\{46,58\}
\end{array}
\]
Thus $s_1=s_3$. Only two data sets for all shortloop matchings and coefficients remain:
$s_1=s_2=s_3=0$ or $s_1=s_2=s_3=1$.
By uniqueness, each data set corresponds to at most one EO function.

\subparagraph*{Direct verification of the arity-eight exception and completion}

We must still show that these two data sets arise from functions, rather than merely from
formally compatible matchings. Label $x_1,\ldots,x_8$ by the eight points
\[
 000,001,010,011,100,101,110,111
\]
of $\mathbb F_2^3$, in this order. The four check vectors in \eqref{en:eq:c1-f8} are
\[
11110000,\quad00001111,\quad11001100,\quad10101010.
\]
They span the truth-table space of the constant function and the three coordinate functions.
They are linearly independent and mutually orthogonal, so their common nullspace is this
same four-dimensional space. Every nonconstant affine function
$z\mapsto\alpha\cdot z+\beta$, with $\alpha\ne0$, has weight $4$; the two constant functions
have weights $0,8$. Thus the fourteen support words of $F'_8$ are precisely the truth tables
of the fourteen nonconstant affine functions.

Fix distinct points $u,v$ and let $\delta=u+v$.
Taking different values at $u,v$ is equivalent to $\alpha\cdot\delta=1$, giving four choices
of $\alpha$ and two choices of $\beta$. The remaining six points form three pairs under
$z\leftrightarrow z+\delta$, with opposite function values on each pair.
The restrictions of these eight affine functions to the six points are distinct.
Indeed, if two restrictions were equal, their difference would be an affine function vanishing
on six points; a nonzero affine function has weight $4$ or $8$ and cannot be supported only
on the remaining two points. The three disequality constraints have exactly eight assignments,
so each assignment corresponds to exactly one internal assignment. Therefore
\begin{equation}\label{en:eq:c1-finite-f8-merges}
 \partial^{uv}F'_8
 =\prod_{\substack{\{z,z+\delta\}\subseteq
       \mathbb F_2^3\setminus\{u,v\}\\\text{each pair counted once}}}
 [x_z\ne x_{z+\delta}].
\end{equation}
This verifies all $28$ shortloops, not merely representative index pairs.
In particular, $M_{12}=\{34,56,78\}$ and $M_{13}=\{24,57,68\}$, so $F'_8$ corresponds
exactly to the data with all three direction bits equal to $0$.
After exchanging $x_7,x_8$, the data for the four within-block edges are unchanged, while
$M_{13}$ becomes $\{24,58,67\}$. This is the data with all three bits equal to $1$.
Thus, after normalization by a maximum shortloop coefficient and relabeling according
to its matching, the arity-eight possibilities are precisely
\begin{equation}\label{en:eq:c1-finite-eight-points}
 \tfrac12P_{\{12,34,56,78\}},\qquad F'_8,\qquad (78)F'_8.
\end{equation}
The last two differ by a variable permutation and form the single exceptional class in
the theorem.

\begin{proof}[Proof of Theorem~\ref{en:thm:c1-eo-matching-reduction}]
At arity six, the exhaustive branching calculation from \eqref{en:eq:c1-finite-six-rows}
gives \eqref{en:eq:c1-finite-six-points}.
At arity eight, a maximum shortloop coefficient fixes four blocks. The three cross-block
choices, uniformity of their types, and compatibility of their direction bits then give
\eqref{en:eq:c1-finite-eight-points}. All candidates are directly verified by the matching
shortloop rule and \eqref{en:eq:c1-finite-f8-merges}.
At arity at least ten, apply Lemma~\ref{en:lem:c1-large-factor}.
Undo the nonzero scalings and variable permutations in the finite base cases.
These cases exhaust every even arity at least six.
\end{proof}

\subparagraph*{Two copies of the shaped function force global realnumberizing}

\begin{lemma}\label{en:lem:c1-double-F}
Let $F=Q+a(\ket{0^m}+\ket{1^m})$ be the function from
Lemma~\ref{en:lem:c1-shape-generator}. If $Q$ satisfies all hypotheses of
Theorem~\ref{en:thm:c1-eo-matching-reduction}, then $m=8$, after a port permutation
$Q=\lambda F'_8$, and $a=\pm\lambda\in\mathbb R^*$.
\end{lemma}
\begin{proof}
First suppose $m=2d$ and $Q=\lambda P_M$.
Take two copies of $F$, retain the two ports of the same edge of $M$ in each copy, and join
all other corresponding ports by $N$ edges. The $Q\times Q$ contribution has $d-1$ closed
cycles and equals $2^{d-1}\lambda^2N_{12}N_{34}$.
Among products of endpoint terms, only complementary assignments pass through the joining
edges, contributing $a^2D_4$. Mixed terms vanish because an internal matching requires
opposite bits whereas an endpoint term is constant on all ports.
The three complementary EO values of the resulting quaternary function are
\[
 (a^2,2^{d-1}\lambda^2,2^{d-1}\lambda^2).
\]
They are all nonzero, so the function is not a nonzero $N$ matching product, contradicting
\ref{en:item:c1-quaternary}.

Only $Q=\lambda F'_8$ remains. Use the eight-point labeling by $\mathbb F_2^3$.
In two copies of $F$, retain the same two points $u,v$ and join the six other equally labeled
points by $N$ edges. Two affine functions complementary on six points are complementary
on all eight. For prescribed values at $u,v$, there are four affine functions.
If the two values agree, exactly one is constant; if they differ, all four are nonconstant.
Thus $Q\times Q$ contributes $4\lambda^2N_{13}N_{24}-\lambda^2D_4$.
Mixed terms vanish because a nonconstant affine function cannot take the same constant
value at six points. Including products of endpoint terms gives the quaternary function
\[
 4\lambda^2N_{13}N_{24}+(a^2-\lambda^2)D_4.
\]
Its three complementary EO values are $(a^2+3\lambda^2,0,4\lambda^2)$.
Condition~\ref{en:item:c1-quaternary} forces the first and third entries to agree, so
$a^2=\lambda^2$. Since $Q$ is nonzero and rational-valued, $\lambda$ is a nonzero rational
number. Hence $a=\pm\lambda$ is real.
\end{proof}

\begin{proof}[Proof of Theorem~\ref{en:thm:c1-strict-dichotomy}]
Apply Lemma~\ref{en:lem:c1-active-single}: either handle the identically zero branch, or reduce
by Turing equivalence to a single function. Send the common-one-sided branch to the general
EO dichotomy. In every remaining branch, common-phase normalization and imbalance absorption
give an active, complement-symmetric function $H$ whose tensor powers have rational EO layers.
Lemma~\ref{en:lem:c1-phase-character} selects the positive generator $g$ of the subgroup generated
by all support imbalances, and Lemma~\ref{en:lem:c1-shape-generator} produces $F$.
Apply ARS--EO simultaneously to all $\partial_{ij}F$. If the result is hard, actual gadget
expansion transfers hardness to the original problem. Otherwise, Lemma~\ref{en:lem:c1-cleanup}
makes every nonzero shortloop a matching product.
The structural theorem and Lemma~\ref{en:lem:c1-double-F} then give a real endpoint value $a$.
Its phase class is $\chi(g)$, so $\chi(g)=\mathbb R^*$.
The phase homomorphism is therefore trivial everywhere, and the entire function $H$ is real-valued.

Return to ordinary equality-edge coordinates using \eqref{en:eq:K-basis}.
For an output $y$,
\[
 (K^{\otimes m}H)(y)
 =2^{-m/2}i^{|y|}\sum_x(-1)^{y\cdot x}H(x).
\]
Complement symmetry makes the real sum on the right zero when $|y|$ is odd.
When the weight is even, $i^{|y|}$ is real. Thus the entire transformed single function is real-valued.
The Real-Holant dichotomy~\cite{en:shao2020real} classifies this problem, and the preceding
Turing equivalences transfer its conclusion to the original $\Gamma$.
\end{proof}

\paragraph{Lower case}\label{en:subsec:c1-lower}\label{en:sec:c1-quaternary}
Take an indecomposable quaternary function. First determine its form using binary closure, then distinguish cases according to whether its two endpoints are nonzero.
When both endpoints are nonzero, use the existing quaternary dichotomy; when only one is nonzero, obtain pins by interpolation; when neither is nonzero, proceed according to the number of centers in the EO layer.

\subparagraph*{A dichotomy for indecomposable quaternary functions}

We discuss the lower case of $C_1$ in the native $N$-edge model.
Let $\Gamma$ be a finite set of algebraic complex-valued functions of even arity. Under the framework convention above,
the available nonzero binary functions are exactly $\lambda N$, where $\lambda\ne0$.
Write $\delta_x$ for the function taking value $1$ only on input $x$.

\begin{proposition}[The dichotomy for the lower case of $C_1$]\label{en:prop:c1-quaternary-interface}
Under the binary closure condition above, if an indecomposable quaternary function is available, then
$\Holant(N\mid\Gamma)$ is in $\FP$ or is $\sharpP$-hard.
\end{proposition}

\subparagraph*{The quaternary form forced by binary closure}

\begin{lemma}\label{en:lem:c1-general-q4-shape}
Every ordinary quaternary gadget has the form
\begin{equation}\label{en:eq:c1-general-q4-shape}
 q=s\delta_{0000}+t\delta_{1111}
   +a(\delta_{0011}+\delta_{1100})
   +b(\delta_{0101}+\delta_{1010})
   +c(\delta_{0110}+\delta_{1001}).
\end{equation}
\end{lemma}
\begin{proof}
Short any two ports $i,j$ with an $N$ edge. Both diagonal entries of the resulting binary matrix must vanish.
Writing $u_i=q(e_i)$ for the four coefficients of weight one, we obtain $u_i+u_j=0$ for every $i\ne j$.
Any three indices imply that their coefficients vanish, and then the fourth vanishes as well.
The same argument handles the four coefficients of weight three.

Let the differences within the three complementary pairs of weight two be
\[
 d_a=q_{0011}-q_{1100},\quad
 d_b=q_{0101}-q_{1010},\quad
 d_c=q_{0110}-q_{1001}.
\]
After shorting $12$, $13$, and $23$, the two antidiagonal entries of the binary matrix are equal, giving
$d_b=d_c$, $d_a=d_c$, and $d_a=-d_b$, respectively.
Thus all three differences vanish, proving the stated form. The values $s,t$ do not enter these shortloops and impose no further constraints.
\end{proof}

\subparagraph*{The case of two nonzero endpoints}

\begin{lemma}\label{en:lem:c1-general-two-endpoints}
If $st\ne0$ in~\eqref{en:eq:c1-general-q4-shape}, then the whole function set has a decidable dichotomy.
\end{lemma}
\begin{proof}
If $q$ has no nontrivial tensor-product decomposition, apply the
Quaternary $K$-Holant dichotomy of~\cite{en:liu2026disequality} directly.
Its hypotheses are precisely even support, two nonzero endpoints, and tensor-product indecomposability;
its conclusion gives the dichotomy for $\Holant(N\mid\Gamma)$.

If $q$ decomposes, the two nonzero endpoints force every factor to be nonzero on both its all-zero and all-one inputs.
A $1+3$ decomposition would also produce a nonzero input of odd weight, so it is impossible.
In a $2+2$ decomposition, both factors must be full-rank diagonal binary functions:
a nonzero antidiagonal entry of either factor, multiplied by the all-zero entry of the other, would violate even support.
By the decomposition lemma, both factors are available, contradicting the condition that all nonzero binary functions are multiples of $N$.
Thus the decomposable case does not occur in this branch.
\end{proof}

\subparagraph*{Uniform interpolation for a quaternary function with one endpoint}

\begin{lemma}[Jordan interpolation at a single endpoint]\label{en:lem:c1-general-jordan}
If exactly one endpoint in~\eqref{en:eq:c1-general-q4-shape} is nonzero, then instances with arbitrarily many occurrences
of the quaternary pin at that endpoint Turing-reduce to $\Holant(N\mid\Gamma)$.
Consequently, the whole function set has a decidable dichotomy.
\end{lemma}
\begin{proof}
We first construct a sequence of quaternary chains whose endpoint term has an extra factor equal to the chain length,
and then isolate that term by interpolation.
Suppose first that $s\ne0,t=0$; the target pin is $E=\delta_{0000}$.
If $a=b=c=0$, it is already a nonzero multiple of an actual gadget.
Otherwise, permute the ports so that $a\ne0$. Flatten in the natural order $12\mid34$ and write
\[
 Q=\begin{pmatrix}
 s&0&0&a\\0&b&c&0\\0&c&b&0\\a&0&0&0
 \end{pmatrix},\qquad J=N\otimes N.
\]
Take $m\ge1$ copies of $q$, and connect the two column ports of each copy to the two row ports of the next copy using $N$ edges.
The resulting actual quaternary gadget $q_m$ has matrix
\begin{equation}\label{en:eq:c1-general-chain}
 Q_m=(QJ)^{m-1}Q=(QJ)^mJ.
\end{equation}
On the $00,11$ subspace, the matrix of $QJ$ is
$\left(\begin{smallmatrix}a&s\\0&a\end{smallmatrix}\right)$;
on the $01,10$ subspace, it is
$\left(\begin{smallmatrix}c&b\\b&c\end{smallmatrix}\right)$, which is always semisimple.
Write
\[
 A=\delta_{0101}+\delta_{1010},\quad
 B=\delta_{0110}+\delta_{1001},\quad
 C=\delta_{0011}+\delta_{1100}.
\]
Taking powers gives the exact formula
\begin{equation}\label{en:eq:c1-general-jordan-expansion}
 q_m=m\frac{s}{a}a^mE+a^mC
      +(c+b)^m\frac{A+B}{2}
      +(c-b)^m\frac{B-A}{2}.
\end{equation}
A term with zero base vanishes for every $m\ge1$ and can be deleted.

Consider an arbitrary target function-net containing $M\ge1$ vertices labeled $E$.
By associativity, separate these vertices from the remainder: the former form $E^{\otimes M}$,
and the latter form a fixed context.
Replace the former simultaneously by $q_m^{\otimes M}$, and denote the resulting value by $Z_m$.
The part that varies with $m$ supplies the coefficients of the equations, whereas the fixed context supplies the unknowns to be recovered.
Let $\Lambda$ be the set of distinct nonzero elements among $a,c+b,c-b$; its size is at most three.
After collecting equal bases, multilinear expansion gives
\begin{equation}\label{en:eq:c1-general-exp-poly}
 Z_m=\sum_{\mu\in\mathcal L_M}P_\mu(m)\mu^m,
 \qquad \deg P_\mu\le M,
\end{equation}
where $\mathcal L_M$ consists of the distinct nonzero products of $M$ elements chosen from $\Lambda$.
Its size is at most $\binom{M+2}{2}$.
Only the endpoint term in~\eqref{en:eq:c1-general-jordan-expansion} contains the factor $m$, so
\begin{equation}\label{en:eq:c1-general-leading-coeff}
 [m^M]P_{a^M}(m)=\left(\frac{s}{a}\right)^M Z_{\mathrm{target}}.
\end{equation}
Root-of-unity relations or coincident products among the spectral bases do not affect this equality:
degree $M$ requires choosing the endpoint term at all $M$ positions.

We verify explicitly that finitely many interpolation samples suffice.
Let $L=(M+1)|\mathcal L_M|=O(M^3)$. Using the samples $m=1,\ldots,L$,
solve for the unknown coefficients $[m^j]P_\mu$ with the matrix
\[
 V_{m,(\mu,j)}=m^j\mu^m,\qquad 0\le j\le M.
\]
If $(v_1,\ldots,v_L)$ were a row dependence, set $p(z)=\sum_{m=1}^L v_mz^m$.
The dependence says that $(z\,d/dz)^jp(\mu)=0$ for every $\mu\in\mathcal L_M$ and $0\le j\le M$.
Since $\mu\ne0$, these conditions are equivalent to $p^{(j)}(\mu)=0$:
the conversion between the two families of derivatives is triangular with nonzero diagonal.
Hence $p$ is divisible by the degree-$L$ polynomial
$\prod_{\mu\in\mathcal L_M}(z-\mu)^{M+1}$.
But $p(0)=0$, whereas this product has nonzero constant term, so $p=0$.
Therefore $V$ is invertible. Recovering the coefficient in~\eqref{en:eq:c1-general-leading-coeff} and dividing by
the known nonzero number $(s/a)^M$ gives the target value.
Each sample uses $mM=O(M^4)$ copies of a fixed gadget, and the total number of queries is $O(M^3)$.
All numbers lie in the fixed number field generated by the fixed algebraic coefficients;
exact comparison, exponentiation, and linear-system solving have bit complexity polynomial in the input size.
If $q$ has first been rescaled, compensate by the corresponding known power.

When $s=0,t\ne0$, the endpoint block in the same calculation is a lower triangular Jordan block and yields
$\delta_{1111}$ directly; no additional bit-flip function is required.

After obtaining the quaternary pin, apply the decomposition theorem~\ref{en:thm:decomposition}:
either the problem is already in $\FP$, or a nonzero unary factor $\delta_0$ or $\delta_1$ becomes available.
The odd-arity dichotomy~\ref{en:thm:odd-arity-fp} then completes the proof.
\end{proof}

\subparagraph*{Pure EO quaternary functions, by the number of centers}

\begin{lemma}\label{en:lem:c1-general-three-centers}
If $s=t=0$ and $abc\ne0$, then $\Holant(N\mid\Gamma)$ is $\sharpP$-hard.
\end{lemma}
\begin{proof}
The six-vertex dichotomy of Cai--Fu--Xia~\cite[Theorem~3.1]{en:cai2018sixvertex}
applies to $\Holant(N\mid\{q\})$ on general graphs.
Its exceptions are $q\in\mathcal A$, $q\in\mathcal P$, or the condition that each of the three complementary pairs contains a zero entry.
Here the support of $q$ has size exactly six. A nonzero support in $\mathcal A$ is an affine space,
and a nonzero support in $\mathcal P$ is an independent product of blocks of two complementary points;
both have size a power of $2$.
Thus the first two exceptions are impossible, and the third also fails. The single-$q$ problem is therefore hard.
Substituting the actual gadget transfers this hardness to the whole set $\Gamma$.
\end{proof}

\begin{lemma}[A chain reduction for two centers]\label{en:lem:c1-general-two-centers}
Suppose that $s=t=0$ and exactly two centers are nonzero. After a port permutation, write
\[
 q=aA+bB,\qquad
 A=\delta_{0101}+\delta_{1010},\quad
 B=\delta_{0110}+\delta_{1001},\qquad ab\ne0.
\]
If $a=b$, then $q=aN_{12}N_{34}$; if $a=-b$, then $q=aY_{12}Y_{34}$.
In all other cases, instances with arbitrarily many occurrences of $A$ Turing-reduce to the original problem,
so the whole function set has a decidable dichotomy.
\end{lemma}
\begin{proof}
Flattening along $13\mid24$ gives
\[
 Q=\diag(a,b,b,a)J,\qquad J=N\otimes N.
\]
A two-wire $N$ chain of length $m$ gives the actual gadget
\begin{equation}\label{en:eq:c1-general-two-center-powers}
 q_m=a^mA+b^mB.
\end{equation}
Suppose $r=a/b$ is not a root of unity. Combine the $M$ vertices labeled $A$ in the target function-net into $A^{\otimes M}$,
and regard the remainder as a fixed context. Replacing them all by $q_m$ and dividing the output by $b^{mM}$
gives the value at $r^m$ of a polynomial of degree at most $M$.
For $m=1,\ldots,M+1$, these points are distinct;
ordinary Vandermonde interpolation recovers the target value as its leading coefficient.
This uses $O(M)$ queries, each with $O(M^2)$ copies of a fixed gadget,
and permits arbitrarily many occurrences of $A$ in Turing access to $\Gamma\cup\{A\}$.
Lemma~\ref{en:lem:cyc-complement-carrier} classifies the whole function set containing the complementary two-point quaternary function $A$.
Composing reductions transfers its conclusion to the original problem.

Now suppose that $r$ is a root of unity and $r\ne\pm1$.
Flatten instead along $12\mid34$; the nonzero middle block is
$\left(\begin{smallmatrix}a&b\\b&a\end{smallmatrix}\right)$.
Connecting two copies of $q$ by two $N$ edges gives the middle block
\begin{equation}\label{en:eq:c1-general-root-conversion}
 \begin{pmatrix}a&b\\b&a\end{pmatrix}
 \begin{pmatrix}0&1\\1&0\end{pmatrix}
 \begin{pmatrix}a&b\\b&a\end{pmatrix}
 =\begin{pmatrix}2ab&a^2+b^2\\a^2+b^2&2ab\end{pmatrix}.
\end{equation}
If $r=\pm i$, this is already a nonzero multiple of a complementary two-point quaternary function.
Otherwise both new coefficients are nonzero, and one of their ratios is
\[
 \frac{a^2+b^2}{2ab}=\frac{r+r^{-1}}2.
\]
Writing $r=e^{i\theta}$, this is the nonzero real number $\cos\theta$, strictly between $-1$ and $1$,
and hence is not a root of unity. Apply the preceding paragraph to this actual gadget.
Finally, the two tensor-product identities for $a=\pm b$ follow by entrywise verification.
\end{proof}

\begin{proof}[Proof of Proposition~\ref{en:prop:c1-quaternary-interface}]
Take an indecomposable quaternary function and write it in the form of Lemma~\ref{en:lem:c1-general-q4-shape}.
Two nonzero endpoints are handled by Lemma~\ref{en:lem:c1-general-two-endpoints},
and exactly one nonzero endpoint by Lemma~\ref{en:lem:c1-general-jordan}.
The remaining functions are pure EO.
Three nonzero centers give hardness by Lemma~\ref{en:lem:c1-general-three-centers}.
Exactly one nonzero center is a complementary two-point quaternary function, handled by Lemma~\ref{en:lem:cyc-complement-carrier}.
Two nonzero centers are handled by Lemma~\ref{en:lem:c1-general-two-centers}:
when $a=\pm b$, both forms admit a $2+2$ decomposition, contrary to indecomposability; all other cases give a dichotomy.
With no nonzero center, the function is zero, which has been excluded.
\end{proof}

\subsubsection{\texorpdfstring{The $C_k$ Case ($k\ge3$)}{The Ck Case (k>=3)}}\label{en:sec:cyclic}
We use the framework's convention on available functions: gadgets and their nonzero
tensor factors are all counted as available. The binary closure condition says that
every nonzero available binary function has full rank and that the projective group
of their determinant-normalized representatives is $\Ggroup\cong C_k$, with $k\ge3$.
The finitely many binary directions have already been adjoined to the current function
set in one step. Only finitely many other functions are used in each proof; these too
are adjoined simultaneously under the framework before the gadget calculations.
Write $\mu_k=\{z\in\mathbb C:z^k=1\}$ for the set of $k$th roots of unity.
When there is a non-isotropic eigenvector, orthogonal triangularization and transpose
closure put the entire cyclic group into diagonal form. If $k$ is even, its order-two
coset has a symmetric traceless representative, so the existing disequality dichotomy
already applies. Below we work separately in the equality-edge diagonal model with
odd $k$ and the native $N$-edge model arising from an isotropic eigenbasis. We first
reduce arity in the upper case, then handle the quaternary functions in the lower
case, and finally combine the two conclusions.

A nonzero function is called \emph{indecomposable} if, after any port permutation
and nonzero rescaling, it cannot be written as the tensor product of two functions
of positive arity. By the decomposition theorem, every nonzero factor is available
in the unresolved cases. Thus, in the upper case, a quaternary function that admits
a $2+2$ decomposition must be a binary pairing product from the current group.

\begin{theorem}[Quaternary dichotomy in the native disequality-edge model\cite{en:liu2026disequality}]
\label{en:thm:cyc-external-q4}
Let $\Gamma$ be a finite algebraic function set with edge kernel $N$, ordinarily
realizing a tensor-indecomposable quaternary function $q$ such that
\[
 q_x=0\quad(|x|\text{ is odd}),\qquad q_{0000}q_{1111}\ne0.
\]
Then the entire $\Gamma$ problem has a decidable $\FP/\sharpP$-hard dichotomy.
Its tractability condition is that the entire function set satisfies the seven-part
predicate $\operatorname{Tract}(K\Gamma)$ in the cited paper.
\end{theorem}
We use the Quaternary $K$-Holant dichotomy of~\cite{en:liu2026disequality}, whose proof
does not assume that $I$ is an available vertex function. The matrix $K$ is exactly
the one in~\eqref{en:eq:K-basis}. If $q$ is already a gadget, it and the finitely many
auxiliary functions used may be adjoined, preserving ordinary gadget equivalence.

\subparagraph*{Isotropic cyclic groups retain native disequality edges}

If both eigendirections of a cyclic generator $R$ are isotropic, choose an eigenbasis
$K_0$ with $K_0^{\mathsf T}K_0=N$, and let $D=K_0^{-1}RK_0$.
The similarity representation is $D^j$, but the actual binary vertex function is
\begin{equation}\label{en:eq:cyc-native-binary}
 A_j=K_0^{-1}R^jK_0^{-\mathsf T}=D^jN,
 \qquad NA_jN=ND^j,\qquad A_jN=D^j.
\end{equation}
These three matrices are, respectively, the binary function, the effective matrix
for a two-port shortloop, and the one-port transfer matrix. Port actions therefore
give diagonal weights, whereas shortloops compare the $01$ and $10$ slices.
All isotropic cases are treated in this native $N$-edge model, with $k\ge3$ either
odd or even.

For a quaternary function, two antidiagonal phase probes with distinct phases,
together with the antidiagonal binary outputs, force each coefficient on an equal
remaining input to vanish. Every odd-weight four-bit string contains two unequal
bits whose removal leaves two equal bits. Hence all quaternary gadgets have even
support. Write one natural flattening as
\begin{equation}\label{en:eq:cyc-native-q4}
 Q=\begin{pmatrix}
 s&0&0&a\\0&c&e&0\\0&f&d&0\\b&0&0&t
 \end{pmatrix}.
\end{equation}
Applying Lemma~\ref{en:lem:oddq-root-ratio} to the central block with three effective
antidiagonal probes gives the following: each central complementary pair is either
zero in both entries or nonzero in both; in the latter case the ratio belongs to
$\mu_k$. If two central pairs are active, all four corresponding entries are nonzero.
The last case of the root-ratio lemma says that their products agree and that all
ratios among these four values belong to $\mu_k$. Consequently, if all three central
pairs are active, then
\begin{equation}\label{en:eq:cyc-native-center-products}
 ab=cd=ef=P\ne0.
\end{equation}
These ratio relations concern only the central values; $s,t$ are treated in the
quaternary classification.

\paragraph{Upper case}\label{en:subsec:ck-upper}
Quaternary functions in the upper case are pairing products of binary functions from the current group.
Following the approach of the first version, we reduce the arity of functions of higher arity that are not pairing products: in the equality-edge model we compare the $00,11$ slices, and in the native disequality-edge model we compare the $01,10$ slices, proving the required pairing structure in each setting.

\subparagraph*{Arity reduction in the equality-edge model}

We first treat the diagonal cyclic group with equality edges. The proof has three
steps: obtain a nonzero value on the all-zero input, use three diagonal auxiliary
functions to recover the pairing structure of the slices, and combine three views
to reconstruct the entire function as a pairing product.

\begin{lemma}[Arity reduction for a diagonal cyclic group]\label{en:lem:cyc-diagonal-arity}
Assume the binary closure condition of this section for the equality-edge diagonal
cyclic group, with $k\ge3$. If an even-arity realizable function lies outside the
group's binary pairing product class, then a quaternary realizable function lies
outside that class.
\end{lemma}
\begin{proof}
Suppose there is no such quaternary function, and choose a non-pairing-product
function $F$ of minimum arity $2m\ge6$. This minimum is taken among ordinary
gadgets over the current fixed function set. Every realizable function of smaller
even arity is zero or a pairing product of full-rank diagonal binary functions
from the current group.

First, three nonzero diagonal pairing products satisfying a linear relation with
all coefficients nonzero must use the same matching. Otherwise, the support count
in Lemma~\ref{en:lem:odda-three-products} forces the third support to be exactly the
symmetric difference of the first two. But the all-zero input belongs to every
purely diagonal pairing support and to none of these symmetric differences, a
contradiction.

Choose three positions in a nonzero input of $F$. At least two positions $p,q$ have
the same value. Thus at least one of $F_{pq}^{00},F_{pq}^{11}$ is nonzero. Exactly
one cannot be nonzero: otherwise every diagonal shortloop gives a multiple of that
nonzero slice, so the slice is a lower-arity diagonal pairing product. Along its
matching, choose probes that close each pair to a nonzero value. Applying these
shortloops to the original $F$ leaves a binary function with exactly one nonzero
diagonal entry. The complete binary closure also forces its off-diagonal entries
to vanish, contradicting full rank.

Write the two nonzero slices as $A,B$. Three probes give $A+t_jB$. If $A,B$ are
dependent, both are proportional to a nonzero lower-arity diagonal pairing product.
Hence $A(0)\ne0$, giving $F(0^{2m})\ne0$. If they are independent, all three
shortloops are nonzero and pairwise nonproportional, and the preceding argument
gives a common matching. By Lemma~\ref{en:lem:odda-one-changing-factor}, they have
common factors on all but at most one pair. Solving for the slices yields
$A=E_0Q$ and $B=E_1Q$, where $Q$ is an actual full-rank diagonal pairing product
and $E_0,E_1$ are both nonzero. Closing $Q$ to a nonzero value in the original $F$
gives an actual quaternary function. Its $00$ slice is nonzero, so it is not the zero
function; the lower-arity condition forces it to be a full-rank diagonal pairing
product. Its all-zero value is therefore nonzero. Thus $E_0(00)\ne0$, again giving
$F(0^{2m})\ne0$.

Now repeat this three-probe and common-factor argument for every pair of ports.
The proofs of Lemmas~\ref{en:lem:odda-common-slice-factor}
and~\ref{en:lem:odda-two-slice-forms} use only diagonal probes, a common matching,
and the actual quaternary closure, so both apply here. Each pair of $00,11$ slices
therefore consists either of proportional complete pairing products or of pairing
products with complementary pins on the same pair. Next use
Lemmas~\ref{en:lem:odda-three-views} and~\ref{en:lem:odda-two-rhinos} to choose three
views on $x,y,z,w$, including at least two pairs of complementary pinned pairing
products. These proofs also use only diagonal closures, with no antidiagonal port
action. Since the variables are Boolean, at least one of $x=z$, $y=z$, and $x=y$
holds; the corresponding slice conditions then force both $x=y$ and $z=w$.
Factor out the common external product and close it to a nonzero value, obtaining
an actual quaternary residual factor. The lower-arity condition makes this factor
a diagonal pairing product as well, so $F$ itself is a pairing product, a
contradiction. This last step is the same as in
Theorem~\ref{en:thm:odd-arity-reduction}; once the first step above gives a nonzero
all-zero value, all remaining steps require only diagonal auxiliary functions.
\end{proof}

\subparagraph*{Arity reduction in the native disequality-edge model}\phantomsection\label{en:sec:native-cyclic-arity}

Shortloops in the native disequality-edge model compare the $01$ and $10$ slices.
The proof again has three steps. First separate the EO part from the two endpoints.
Next use the common pairings of the shortloops to show that a pointwise power of the
EO part is an unweighted matching. Finally exclude the arity-eight exception and
recover the original binary weights. The resulting minimum counterexample consists
only of a matching part and endpoints, which are treated in the next part.

We work in the native $N$-edge model. Up to nonzero scalars, the binary closure
condition gives all nonzero available binary functions as
\[
 A_j=D^jN,\qquad
 D=\operatorname{diag}(\lambda,\lambda^{-1}),\qquad
 \lambda=e^{\pi i/k},\qquad k\ge3,
\]
where $j$ ranges over the integers. No parity restriction is imposed on $k$.
Let $\mathcal P_{\mathcal B}$ consist of the zero function and the pairing tensor
products of actually realizable full-rank antidiagonal binary functions, allowing
an overall nonzero scalar. The ratio of the two antidiagonal entries of each
nonzero binary factor belongs to $\mu_k$. When the two ports of an $A_j$ are used
for a shortloop, the effective kernel is $NA_jN$. As $j$ varies, these kernels,
up to nonzero scalars, are exactly
\begin{equation}\label{en:eq:nca-probes}
 K_t=\begin{pmatrix}0&1\\t&0\end{pmatrix},\qquad t\in\mu_k.
\end{equation}
Thus the actual shortloops in this section have the form
$F_{pq}^{01}+tF_{pq}^{10}$.

Consider a nonzero ordinarily realizable function $F$ of arity $m=2d\ge6$, and assume
that every ordinarily realizable function of smaller even arity belongs to
$\mathcal P_{\mathcal B}$. We will eventually satisfy this assumption by choosing a
non-pairing-product function of minimum arity. Pointwise powers, slices, and EO
restrictions below are used to analyze function values; gadgets are still built
by connecting the specified binary auxiliary functions.

\subparagraph*{The middle layer and consistency of pairings}

\begin{lemma}[Only the EO layer and the two endpoints remain]\label{en:lem:nca-eo-endpoints}
There are an EO function $Q$ and constants $a,b$ such that
\[
 F=Q+a\delta_{0^m}+b\delta_{1^m}.
\]
Here $Q=0$ is allowed.
\end{lemma}

\begin{proof}
Take any input containing both $0$ and $1$ whose weight is not $d$. Choose ports
$p,q$ with opposite values. After deleting these two bits, the remaining input
still lies outside the corresponding EO layer. Every actual shortloop
$F_{pq}^{01}+tF_{pq}^{10}$ is a lower-arity antidiagonal pairing product or zero,
and hence is an EO function. At this remaining input, two distinct parameters $t$
give two zero-value equations, forcing both slice values to vanish. Thus every
value outside the EO layer and the all-zero and all-one inputs is zero.
\end{proof}

\begin{lemma}[Common matching for cyclic probes]\label{en:lem:nca-common-matching}
Fix any two ports $p,q$. All nonzero actual shortloops
$F_{pq}^{01}+tF_{pq}^{10}$, $t\in\mu_k$, have the same support and matching.
\end{lemma}

\begin{proof}
Write $A=F_{pq}^{01}$ and $B=F_{pq}^{10}$. If they are dependent, all nonzero
outputs are proportional. Suppose instead that they are independent. Then the
$k$ shortloops are all nonzero and have distinct projective directions.

If $k\ge4$, Lemma~\ref{en:lem:odda-three-products} says that once three pairing
products on the same line have different supports, that line contains exactly
three projective points that are pairing products; every other nonzero point has
six-point support on a local set of four variables. Thus at least four distinct
projective pairing-product points cannot coexist on this line, proving the claim.

It remains to consider $k=3$. Suppose that the three outputs $U,V,W$ exhibit the
four-variable exception, and write their nondegenerate relation as
$\alpha U+\beta V+\gamma W=0$. The two entries of every actual antidiagonal binary
factor have equal cubes. Hence the pointwise cubes of $U,V,W$ on their respective
supports are nonzero constants $u,v,w$. In the exception, each pairwise support
intersection is nonempty, whereas the triple intersection is empty. Evaluate the
relation at inputs belonging only to $U,V$, only to $U,W$, and only to $V,W$,
respectively, and cube the resulting equations. This gives
\[
 \alpha^3u=-\beta^3v,\qquad
 \alpha^3u=-\gamma^3w,\qquad
 \beta^3v=-\gamma^3w.
\]
These equations contradict the nonzero values of all three quantities. Thus the
exception is impossible for cube roots as well.
\end{proof}

\begin{lemma}[Local forms of the two opposite slices]\label{en:lem:nca-slice-forms}
For any ports $p,q$, either $F_{pq}^{01}=F_{pq}^{10}=0$, or a pair $z,w$ of remaining
variables can be chosen so that
\begin{equation}\label{en:eq:nca-slice-factors}
 F_{pq}^{01}=E_0(z,w)Q_{\mathrm{out}},\qquad
 F_{pq}^{10}=E_1(z,w)Q_{\mathrm{out}},
\end{equation}
where $Q_{\mathrm{out}}$ is a pairing product of actual full-rank antidiagonal binary
factors. There are exactly two nonzero forms:
\begin{enumerate}
 \item $E_0=\eta_0 E$ and $E_1=\eta_1 E$, where $E$ is a full-rank antidiagonal
 function, the ratio of its two nonzero values belongs to $\mu_k$, and
 $\eta_1/\eta_0\in\mu_k$;
 \item $E_0,E_1$ are nonzero only at $01,10$, respectively, or only at $10,01$,
 respectively, and the ratio of their two nonzero values belongs to $\mu_k$.
\end{enumerate}
All coefficients specified as nonzero in these two forms are nonzero.
\end{lemma}

\begin{proof}
If the two slices are dependent and not both zero, take the actual pairing
decomposition of a nonzero shortloop and factor out all but one pair. If the
slices are independent, Lemma~\ref{en:lem:nca-common-matching} gives a common matching
for three shortloops, and Lemma~\ref{en:lem:odda-one-changing-factor} shows that the
factor changes on at most one pair. Both cases give~\eqref{en:eq:nca-slice-factors},
where for now $E_0,E_1$ are possibly degenerate antidiagonal functions, not both zero.

For each factor of $Q_{\mathrm{out}}$, choose an effective kernel $K_t$ whose closure
value is nonzero. The closure value is a nonzero linear polynomial in $t$, and so
vanishes for at most one $t$. Connect these actual auxiliary functions to the
original $F$, obtaining a nonzero quaternary function $H(p,q,z,w)$. Its $01,10$
slices are a common nonzero constant times $E_0,E_1$. The lower-arity assumption
makes $H$ an actual antidiagonal pairing product.

If the matching of $H$ is $pq|zw$, the two slices are multiples of the same full-rank
$zw$ factor. Their ratio is the ratio of the two antidiagonal values of the $pq$
factor, and belongs to $\mu_k$. If the matching crosses, opposite values on $p,q$
pin $z,w$ to two complementary opposite inputs. The ratio of the corresponding
function values is a product of ratios of two binary factors, and also belongs
to $\mu_k$. This proves the two forms and simultaneously excludes the possibility
that exactly one slice is zero.
\end{proof}

\subparagraph*{All disequality shortloops of the pointwise power}

For any function $f$, write $f^{[k]}(x)=f(x)^k$ for its pointwise power, as distinct
from a matrix power. First suppose that $Q\ne0$ in
Lemma~\ref{en:lem:nca-eo-endpoints}, and let
\[
 S=Q^{[k]}.
\]
This is a nonzero EO function. Write $\partial^{pq}$ for the unweighted binary
disequality shortloop, as in Section~\ref{en:sec:c1-strict}.

\begin{lemma}[Shortloops of the pointwise power are matchings]\label{en:lem:nca-power-caps}
For every $p\ne q$, there are a nonzero constant $\theta_{pq}$ and a perfect matching
$M_{pq}$ on the remaining variables such that
\begin{equation}\label{en:eq:nca-all-power-caps}
 \partial^{pq}S=\theta_{pq}P_{M_{pq}},\qquad
 P_M=\prod_{\{u,v\}\in M}[x_u\ne x_v].
\end{equation}
\end{lemma}

\begin{proof}
The endpoint terms vanish on the $01,10$ slices, so these are also the corresponding
slices of $Q$. Therefore
\[
 \partial^{pq}S=(F_{pq}^{01})^{[k]}+(F_{pq}^{10})^{[k]}.
\]
This expression is the sum of the pointwise powers of two slices; it does not
interchange a shortloop with taking the power of the whole function. If both slices
are zero, the shortloop is zero. Otherwise apply Lemma~\ref{en:lem:nca-slice-forms}.
The function $Q_{\mathrm{out}}^{[k]}$ is a nonzero constant times an unweighted
disequality matching. In the first form, $E^{[k]}$ is a nonzero constant times $N$,
and $\eta_0^k=\eta_1^k\ne0$; the sum has nonzero coefficient $2\eta_0^k$. In the
second form, the $k$th powers of the two nonzero pinned values are equal and nonzero,
so their sum is again a nonzero constant times $N$. Thus every shortloop is either
zero or a nonzero constant times an unweighted disequality matching.

We now exclude zero shortloops. If $\partial^eS=0$ and a disjoint pair of ports $f$
satisfies $\partial^fS\ne0$, commuting the two shortloops gives
\[
 0=\partial^f\partial^eS=\partial^e\partial^fS.
\]
But the right side is a disequality shortloop of a nonzero disequality matching
product, and remains nonzero: closing along a matching edge gives a factor $2$,
whereas shorting across two matching edges merges them and gives a factor $1$.
Thus one zero shortloop forces every shortloop disjoint from it to be zero.

The graph on port pairs in which disjoint pairs are adjacent is connected. Indeed,
disjoint pairs are already adjacent; two distinct intersecting pairs have a union
of only three ports, and since $m\ge6$, a pair outside that union gives a path of
length two. Hence one zero shortloop would force all shortloops to be zero.
A nonzero EO function cannot have all its shortloops zero: choose three positions
with values $100$ in a nonzero input, and denote the function values obtained by
the three placements of the $1$ by $a,b,c$, with $a\ne0$. The corresponding
shortloop values are $a+b,a+c,b+c$, which cannot all vanish. Every shortloop is
therefore nonzero, proving~\eqref{en:eq:nca-all-power-caps}.
\end{proof}

\begin{lemma}[Rational values and bit-flip symmetry]\label{en:lem:nca-rational-power}
There is a nonzero constant $\theta$ such that $R=S/\theta$ is a nonzero rational-valued,
bit-flip-symmetric EO function, and each of its disequality shortloops is still
a nonzero constant times an unweighted disequality matching.
\end{lemma}

\begin{proof}
Commute two shortloops on disjoint pairs $e,f$. A shortloop of a matching product
introduces only a coefficient $1$ or $2$, so~\eqref{en:eq:nca-all-power-caps} gives
\[
 \theta_e c_e P_{M'}=\theta_f c_f P_{M''},\qquad c_e,c_f\in\{1,2\}.
\]
Both sides are nonzero and have equal supports, so their matchings and coefficients
agree. Hence $\theta_e/\theta_f\in\{1/2,1,2\}$. Following paths in the disjointness
graph above, fix $e_0$ and set $\theta=\theta_{e_0}$ to obtain
$\theta_e/\theta\in2^{\mathbb Z}$ for every $e$. Thus all shortloops of $R=S/\theta$
are rational-valued.

Treat the values on all EO inputs as unknowns. All disequality shortloops together
form a linear system with rational coefficients. This joint linear map is injective
on the EO subspace: a nonzero EO function in its kernel would contradict the
three-local-sums argument above. The system has a unique solution and a rational
right-hand side, so Gaussian elimination shows that $R$ itself is rational-valued.

Every unweighted disequality matching product is bit-flip symmetric, and disequality
shortloops commute with the simultaneous flip of all bits. Hence all shortloops of
$R(x)-R(\bar x)$ vanish. This difference is still an EO function, so injectivity
gives $R(x)=R(\bar x)$.
\end{proof}

\subparagraph*{Excluding the arity-eight exception and recovering the weights}

By Lemma~\ref{en:lem:nca-rational-power}, Theorem~\ref{en:thm:c1-eo-matching-reduction}
applies to $R$. It follows that $S$ is a nonzero constant times an unweighted
disequality matching or, after a variable permutation, a nonzero multiple of the
arity-eight exception $F'_8$.

\begin{lemma}[The arity-eight exception cannot occur]\label{en:lem:nca-exclude-f8}
Here $S$ cannot be an arity-eight exception. Consequently, there are a nonzero
constant $c$ and a perfect matching $M$ such that $S=cP_M$.
\end{lemma}

\begin{proof}
The fourteen support inputs in~\eqref{en:eq:c1-f8} can be viewed as the truth tables
of all nonconstant affine functions $a+b\cdot v$ on the eight points of
$\mathbb F_2^3$. To verify this description, note that the four linear parity-check
equations are independent, so their solution space has dimension four. All sixteen
affine truth tables satisfy these equations and hence give exactly that space.
Removing the two constant functions of weights zero and eight leaves fourteen
inputs of weight four.

Fix any two distinct coordinates $v,w$ to $0,1$. This imposes two independent
equations on the four affine coefficients and leaves exactly four truth tables.
The evaluation linear form at any third coordinate $u$ is not in the span of the
first two evaluation forms: its constant coefficient is one, excluding zero and
the sum of the first two forms, while $u\ne v,w$ excludes the first two forms
themselves. Thus every remaining coordinate takes both zero and one among these
four tables.

If $S$ were a nonzero multiple of $F'_8$, then $Q$ would have the same support,
since pointwise powers preserve support. Every $01$ slice of $Q$ would have four
support inputs, with none of the six remaining ports fixed. On the other hand,
Lemma~\ref{en:lem:nca-slice-forms} writes this slice as $E_0Q_{\mathrm{out}}$. Here
$Q_{\mathrm{out}}$ has two full-rank binary factors and therefore support size four.
Thus $E_0$ must have support size one, fixing two of the remaining ports. This
contradicts the preceding property.
\end{proof}

\begin{lemma}[Recovering an actual weighted matching]\label{en:lem:nca-recover-weighted}
If $Q\ne0$, then $Q\in\mathcal P_{\mathcal B}$, and its binary factors are indeed
ordinarily realizable.
\end{lemma}

\begin{proof}
Since $S=cP_M$ and pointwise powers preserve support, the support of $Q$ is exactly
the full support of $P_M$. For any pair $p,q$ in the original matching $M$, both
slices $F_{pq}^{01}=Q_{pq}^{01}$ and $F_{pq}^{10}=Q_{pq}^{10}$ have the full support
of the remaining matching, of size $2^{d-1}$. In
Lemma~\ref{en:lem:nca-slice-forms}, the external product has support size $2^{d-2}$,
so the local factor must have full rank and cannot be in the complementary-pin
case. The two slices are therefore proportional.

View the two orientations of each matching edge as a logical bit. This
proportionality says that changing any logical bit changes the function value by
a ratio independent of the other logical bits. All values on the full matching
support are nonzero. Starting at a fixed input and changing the logical bits one
at a time therefore gives
\[
 Q=\gamma\prod_{e\in M}B_e,\qquad \gamma\ne0,
\]
where every $B_e$ is a full-rank antidiagonal function. The local forms already
show that its ratio belongs to $\mu_k$.

The actual realizability of the factors can also be checked directly. In the
original function $F$, close all other matching edges of $M$ using effective
kernels $K_t$ with nonzero closure values. Each antidiagonal kernel used for a
closure kills the all-zero and all-one endpoint terms. The result is therefore
a nonzero multiple of $B_e$, with scalar equal to the product of the other
factors' nonzero closure values. This is an ordinary binary gadget, and the
complete binary closure places it in the actual binary set. Consequently,
$Q\in\mathcal P_{\mathcal B}$.
\end{proof}

\begin{theorem}[Shape of a minimum counterexample in the native cyclic model]\label{en:thm:native-cyclic-shape}
Assume the binary closure condition of this section, and suppose that every nonzero
ordinary quaternary gadget belongs to $\mathcal P_{\mathcal B}$. If there is an
ordinarily realizable even-arity function that is not a pairing product, choose
one of minimum arity and call it $F$. Its arity is even, $m\ge6$, and
\[
 F=Q+a\delta_{0^m}+b\delta_{1^m},
\]
where $Q=0$ or $Q$ is an actual weighted antidiagonal binary pairing product, and
at least one of $a,b$ is nonzero. This holds for every $k\ge3$.
\end{theorem}

\begin{proof}
Minimality puts every ordinarily realizable function of smaller even arity in
$\mathcal P_{\mathcal B}$. Arities zero, two, and four have been excluded, so
$m\ge6$. Lemma~\ref{en:lem:nca-eo-endpoints} gives the displayed decomposition into
layers. If $Q=0$, neither pointwise powers nor the EO structure theorem are needed:
$F$ is already a pure endpoint function. If $Q\ne0$, apply in order the lemmas on
matching consistency, shortloops of pointwise powers, rational values, exclusion
of the arity-eight exception, and recovery of weights. These give
$Q\in\mathcal P_{\mathcal B}$. If also $a=b=0$, then
$F=Q\in\mathcal P_{\mathcal B}$, contrary to its choice.
\end{proof}

\subparagraph*{Two-copy gadgets and multiline interpolation for the endpoints}\phantomsection\label{en:subsec:native-cyclic-tail}

Continue in the native $N$-edge model, where the binary closure condition gives
the functions $D^jN$. Connecting the actual binary function $D^jN$ to a port has
transfer matrix $D^j$. Thus for every $r\in\mu_k$, the one-port transfer
$\diag(1,r)$ is available, up to a known nonzero scalar. The next two lemmas
treat the endpoint part obtained in Theorem~\ref{en:thm:native-cyclic-shape}.
Both hold for every $k\ge3$.

\begin{lemma}[A quaternary support contradiction for two endpoints]\label{en:lem:native-cyclic-two-endpoints}
Suppose every nonzero ordinary quaternary gadget belongs to the antidiagonal
binary pairing product class $\mathcal P_{\Bgroup}$. If an ordinary gadget of
arity $2d\ge6$ satisfies
\begin{equation}\label{en:eq:native-cyclic-tail-shape}
 F=Q+a\delta_{0^{2d}}+b\delta_{1^{2d}},
\end{equation}
where either $Q=0$ or, on some matching,
\[
 Q=\prod_{i=1}^{d}E_i(x_i,y_i),\qquad
 E_i=\begin{pmatrix}0&u_i\\v_i&0\end{pmatrix},\qquad u_iv_i\ne0,
\]
then $ab=0$. Excluding two nonzero endpoints here does not require a root-of-unity
ratio condition on the factors $E_i$.
\end{lemma}
\begin{proof}
First suppose $Q\ne0$. Take two whole copies of the gadget $F$, leaving the two
ports of the first matching pair external in each copy. For $i\ge2$, connect
the two copies of $x_i$ by an $N$ edge, and likewise connect the two copies of $y_i$.
Order the four external ports as $x_1,y_1,x'_1,y'_1$. When the two $Q$ terms are
paired, the $i$th internal matching pair contributes
\[
 \sum_{x,y}E_i(x,y)E_i(1-x,1-y)=2u_iv_i.
\]
Writing $C=\prod_{i=2}^{d}2u_iv_i\ne0$, this contribution is
$CE_1\otimes E_1$. Pairing opposite endpoints contributes
$ab(\delta_{0011}+\delta_{1100})$; equal endpoints cannot be connected across
$N$ edges. Every mixed term pairing an endpoint with a $Q$ is also zero: at least
one complete internal matching pair is forced to be $00$ or $11$, where its
antidiagonal factor vanishes. The actual quaternary output is therefore exactly
\begin{equation}\label{en:eq:native-cyclic-tail-six-support}
 H=CE_1\otimes E_1+ab(\delta_{0011}+\delta_{1100}).
\end{equation}
The first term has exactly four nonzero support inputs, $0101,0110,1001,1010$,
disjoint from the second term's $0011,1100$. If $ab\ne0$, then $H$ has six support
inputs, whereas a pairing product of two full-rank antidiagonal factors has
exactly four, a contradiction.

If $Q=0$, choose any pair of external ports in each copy and make the same
connections. The output is $ab(\delta_{0011}+\delta_{1100})$. For $ab\ne0$ this has
two support inputs and is again not an actual binary pairing product. Thus
$ab=0$ in both cases.
\end{proof}

\begin{lemma}[Interpolation from a single endpoint at any even arity]\label{en:lem:native-cyclic-endpoint-interpolation}
Let $\Gamma$ be a finite algebraic even-arity function set with edge kernel $N$,
whose complete ordinary binary closure has exactly the nonzero directions $D^jN$.
Suppose $\Gamma$ ordinarily realizes a function of the
form~\eqref{en:eq:native-cyclic-tail-shape} with exactly one nonzero endpoint, where
$Q=0$ or $Q$ is a full-rank antidiagonal matching product whose factors have
nonzero-entry ratios in $\mu_k$. Then the entire $\Gamma$ problem has a decidable
$\FP/\sharpP$-hard dichotomy. This lemma does not assume that every ordinary
quaternary gadget is a pairing product.
\end{lemma}
\begin{proof}
First suppose $Q\ne0$, $a\ne0$, and $b=0$. Apply the actually available transfer
$T_i=\diag(1,u_i/v_i)$ to the first port of each pair, so that $T_iE_i=u_iN$.
After applying these transfers to the whole $F$ and dividing by the known nonzero
scalar $c=\prod_i u_i$, we obtain
\[
 F_* =\prod_{i=1}^{d}N(x_i,y_i)+s\delta_{0^{2d}},\qquad s=a/c\ne0,
\]
where an ordinary realization may carry an additional known nonzero scalar.
If initially one only knows that all nonzero values of $Q$ have the same $k$th
power, fixing all other matching pairs and reversing the $i$th pair gives
$(u_i/v_i)^k=1$, so the same transfers are allowed.

Flatten with one end of each matching pair on the row side and the other on the
column side. Let $J=N^{\otimes d}$, $z=0^d$, and $o=1^d$. In matrix form,
\[
 F_*=J+sE_{z,z},\qquad F_*J=I+sE_{z,o},\qquad E_{z,o}^{\,2}=0.
\]
Put $\ell$ whole copies of the gadget in series, connecting the column ports of
each copy to the corresponding row ports of the next by $d$ native $N$ edges.
The actual output is
\begin{equation}\label{en:eq:native-cyclic-multiline-chain}
 (F_*J)^{\ell-1}F_*=(I+\ell sE_{z,o})J=J+\ell sE_{z,z}.
\end{equation}
Every vertex block in this chain is a whole copy of $F_*$, and adjacent copies
are connected by native $N$ edges. If the actual standard gadget is $\kappa F_*$,
a chain of length $\ell$ carries the additional known nonzero factor $\kappa^\ell$.

Given an instance with $M$ occurrences of the target $\delta_{0^{2d}}$, first
separate these target vertices from the remaining fixed function-net. Replace
each occurrence by a chain of length $\ell$ and expand multilinearly over these
$M$ positions. After division by $\kappa^{\ell M}$, the oracle answer is a
polynomial in $\ell$ of degree at most $M$. Its coefficient of $\ell^M$ is exactly
$s^M$ times the answer to the target instance. Query at $\ell=1,\ldots,M+1$,
interpolate the leading coefficient using $M+1$ queries, and divide by $s^M$.
Each query uses $O(M^2)$ copies of the fixed gadget for the replacements. Hence
\begin{equation}\label{en:eq:native-cyclic-endpoint-access}
 \Holant(N\mid\Gamma,\delta_{0^{2d}})
 \equiv_T\Holant(N\mid\Gamma).
\end{equation}
The reverse reduction is immediate because the original function set is retained.

If the nonzero endpoint is $1^{2d}$, the same normalization gives $J+sE_{o,o}$,
with $F_*J=I+sE_{o,z}$ and $E_{o,z}^{\,2}=0$. The same chain and interpolation
give $\delta_{1^{2d}}$. If $Q=0$, the required single-endpoint function is already
ordinarily available and interpolation is unnecessary.

We can now adjoin a tensor power of a nonzero unary function while preserving
Turing equivalence. Use the fixed $K$ transformation of~\eqref{en:eq:K-basis} to put
the entire problem back into the equality-edge model. The added function remains
a tensor power of a nonzero unary function. By the decomposition theorem,
Theorem~\ref{en:thm:decomposition}, either the problem belongs to $\FP$ or this unary
factor can be adjoined. The odd-arity dichotomy,
Theorem~\ref{en:thm:odd-arity-fp}, completes the latter case.
\end{proof}

\paragraph{Lower case}\label{en:subsec:ck-lower}
The lower case already has an indecomposable quaternary function. We give the quaternary forms in the two edge models, adjust their coefficients with available binary functions,
and then construct equality functions, pins, or functions satisfying the hypotheses of an existing quaternary dichotomy.

\subparagraph*{Actual four-port normalization for diagonal odd-order groups with equality edges}

First let $k\ge3$ be odd. Up to known nonzero scalars, the available diagonal
auxiliaries are $D_r=\diag(1,r)$ for $r\in\mu_k$. Every binary shortloop
of a quaternary function $q$ is diagonal. Two independent diagonal probes
therefore force every odd-weight entry to vanish. Write
\[
 s=q_{0000},\quad t=q_{1111},\qquad
 a_I=q_{\mathbf1_I},\quad b_I=q_{\mathbf1_{I^c}},\qquad I=12,13,14.
\]
For each complementary pair, apply all $D_r$ in both orientations and use
Lemma~\ref{en:lem:oddq-root-ratio}. The values $s,t$ cannot have exactly one nonzero entry.
Each central pair is either zero in both entries or nonzero in both, and in
the latter case $b_I/a_I\in\mu_k$. If $st\ne0$, then also
$t/s\in\mu_k$, every nonzero central pair satisfies $a_Ib_I=st$, and the
ratio of each of its nonzero entries to $s$ belongs to $\mu_k$.

\begin{lemma}[Four-port balancing in odd order]\label{en:lem:cyc-balance}
Using only the actual port actions of the available $D_r$, one can make the
two values in every complementary pair of $q$ equal. If the endpoints are
nonzero, one can also make the endpoint values equal, with the common value
of every nonzero central pair equal to the endpoint value.
\end{lemma}
\begin{proof}
Apply $D_{r_i}$ at the four respective ports, and put $R=\prod_i r_i$ and
$u_I=\prod_{i\in I}r_i$. If the endpoints are nonzero, set $R=s/t$;
otherwise set $R=1$. For every active central pair choose
\[
 u_I^2=(b_I/a_I)R.
\]
Squaring is bijective on the odd-order group $\mu_k$, so these choices are
possible; choose $u_I$ arbitrarily for inactive pairs. Every prescribed
$(R,u_{12},u_{13},u_{14})\in\mu_k^4$ has an actual four-port solution:
\[
 r_1^2=u_{12}u_{13}u_{14}/R,\qquad
 r_2=u_{12}/r_1,\quad r_3=u_{13}/r_1,\quad r_4=u_{14}/r_1.
\]
The new central values are $a_Iu_I=b_IR/u_I$, and the endpoint values are
$s,tR$. If $st\ne0$, each common central value $c$ satisfies $c^2=s^2$
and $c/s\in\mu_k$. Since $-1\notin\mu_k$, we must have $c=s$.
Every action used here is an actual binary attachment.
\end{proof}

Let $h$ be the balanced function, with common endpoint value $s$ and common
central values $a,b,c$, writing zero for an inactive pair. Apply $K^{-1}$
simultaneously to the entire function set, changing the edge kernel from
$I$ to $N$. The actual new function $h'=(K^{-1})^{\otimes4}h$ satisfies
\begin{equation}\label{en:eq:cyc-K-fourier}
 h'_y=\frac14\sum_x h_x(-i)^{|x|}(-1)^{x\cdot y},\qquad
 h'_{0000}=h'_{1111}=\frac{s-a-b-c}{2}.
\end{equation}
Because $h_x=h_{\bar x}$ and $h$ has even support, complementary terms
cancel at odd-weight inputs $y$. Thus $h'$ also has even support. Invertible
portwise transformations preserve indecomposability.

\begin{proposition}[Quaternary cases for odd-order cyclic groups with equality edges]\label{en:prop:cyc-equality-q4}
Under the binary closure conditions of this section, suppose the diagonal group with equality edges is $C_k$ of odd order, with $k\ge3$. If some quaternary function is not a pairing product of binary functions from this group, then the entire problem has a complexity dichotomy.
\end{proposition}
\begin{proof}
First suppose $st\ne0$, and let $r$ be the number of active central pairs.
After Lemma~\ref{en:lem:cyc-balance}, the new endpoint value
in~\eqref{en:eq:cyc-K-fourier} is $s(1-r)/2$. It is nonzero for $r=0,2,3$.
In any of these cases, a positive-arity factorization would, by the two
endpoints and even support, have to be a product of two full-rank diagonal
binary factors, whose support has only four points. These cases are
therefore indecomposable, and Theorem~\ref{en:thm:cyc-external-q4} applies.
More explicitly, even support and the occurrence of both values at every
variable exclude a unary factor. In a $2+2$ factorization, each factor
must have fixed parity, since independent choices of their inputs would
otherwise produce an odd-support word. The two endpoints then force both
factors to be full-rank diagonal functions. For $r=1$, the identity
$a_Ib_I=st$ already gives a product of two available diagonal binary
factors, contrary to the case under consideration.

Now suppose $s=t=0$. One active central pair gives complementary two-point
support, and after balancing the new endpoint is $-a/2\ne0$. The same
theorem applies. With two active central pairs, write their common values
as $a,b$. If $a^2\ne b^2$, then $a+b\ne0$, and the determinant for the
only possible $2+2$ factorization is nonzero. This again gives an
indecomposable quaternary entry. If $a^2=b^2$, an appropriate port order gives
a nonzero multiple of $N\otimes N$ or $Y\otimes Y$. After decomposition, $N$ or $Y$ is also an available binary function, whereas every binary function in the present group is diagonal. This contradicts the binary closure condition.

Finally, exclude the case in which all three central pairs are active.
After balancing, $a,b,c$ are all nonzero. Fix a $2+2$ flattening and the
corresponding port orientations on two copies, and write
\[
 Q_a=\begin{pmatrix}
 0&0&0&a\\0&b&c&0\\0&c&b&0\\a&0&0&0
 \end{pmatrix}.
\]
Join the first two ports of the two copies in corresponding order by
ordinary equality edges, retaining the last two ports of each copy. The
actual output is
\[
 Q_a^{\mathsf T}Q_a=
 \begin{pmatrix}
 a^2&0&0&0\\0&b^2+c^2&2bc&0\\
 0&2bc&b^2+c^2&0\\0&0&0&a^2
 \end{pmatrix}.
\]
Both endpoints and the central pair $(2bc,2bc)$ are nonzero. The
root-of-unity ratio lemma forces $a^4=4b^2c^2$. The other two pairings
likewise give $b^4=4a^2c^2$ and $c^4=4a^2b^2$. Multiplying and dividing by
the nonzero quantity $(abc)^4$ yields $1=64$, a contradiction. Thus the case in which all three central pairs are nonzero also contradicts the binary closure condition.
\end{proof}

\subparagraph*{A carrier exit for arbitrary complementary two-point quaternaries}

\begin{lemma}[Complementary two-point carrier]\label{en:lem:cyc-complement-carrier}
Suppose a finite algebraic function set $\Gamma$ with native $N$ edges
ordinarily realizes
\[
 g=c_0\delta_x+c_1\delta_{\bar x},\qquad
 x\in\{0,1\}^4,\quad c_0c_1\ne0.
\]
Then the problem over the entire function set has a decidable
$\FP/\sharpP$-hard dichotomy. The word $x$ may have weight two; nonzero
endpoints of $g$ are not required.
\end{lemma}
\begin{proof}
Use the Universal one-pair reduction and Single-occurrence replication
lemmas of\cite{en:liu2026disequality}, with the ordered function
$\Delta=\delta_0\otimes\delta_1$ serving only as a Turing carrier. The
first lemma states that, for any finite function set, instances containing
at most one $\Delta$ reduce to the original problem. Start the construction with this one copy of $\Delta$.

Attach the two ports of $\Delta$, through $N$ edges, to the same numbered
port in two copies of $g$. The two copies of $g$ are forced to take opposite
values at that port and therefore select complementary terms. The six
remaining external ports have exactly three $0$'s and three $1$'s, with
coefficient $c_0c_1$. A fixed reordering gives
\[
 c_0c_1\,\Delta(P)\otimes\Delta^{\otimes2}(H).
\]
Retain $P$ as an ordered two-port carrier. Feed only this retained block
into the next copy of the same context, without identifying any other
external ports. After $\ell$ nested copies, the unique initial $\Delta$
therefore realizes
\[
 (c_0c_1)^\ell\Delta(P)\otimes
 \bigotimes_{j=1}^{\ell}\Delta^{\otimes2}(H_j),
\]
with size $O(\ell)$. Closing the two ports of $\Delta$ by one $N$ edge
has value $1$, satisfying the nonzero closed-instance condition of the
replication lemma. Arbitrarily many occurrences of $\Delta^{\otimes2}$
thus Turing-reduce to $\Gamma$. Explicitly, first take the disjoint union
of the target instance with the closed instance $J$ consisting of one
$\Delta$ with an $N$ self-loop; its value is $1$. Starting from this unique
carrier, nest the replication contexts to fill every target position.
Apply the single-$\Delta$ reduction, and divide by
$(c_0c_1)^\ell Z(J)$ to recover the original answer.

We may therefore adjoin $\Delta^{\otimes2}$ while preserving Turing equivalence. Use the fixed transformation $K$ to return the entire problem to the equality-edge model; the new function remains a tensor product of nonzero unary functions. By the decomposition theorem, Theorem~\ref{en:thm:decomposition}, either the problem is in $\FP$, or these unary factors may be adjoined. In the latter case, the odd-arity dichotomy, Theorem~\ref{en:thm:odd-arity-fp}, completes the proof.
\end{proof}

\subparagraph*{Two-wire chain interpolation for one-endpoint quaternaries}

\begin{lemma}[One-endpoint interpolation exit]\label{en:lem:cyc-native-one-endpoint}
Suppose the complete ordinary binary closure of a function set with native $N$ edges has exactly the nonzero directions listed in~\eqref{en:eq:cyc-native-binary}. If a quaternary gadget has exactly one nonzero endpoint and at most two active
central complementary pairs, then the entire function set has a decidable
dichotomy.
\end{lemma}
\begin{proof}
Treat the nonzero endpoint $0000$; the case $1111$ is analogous. With no
central entries, the function is already a nonzero multiple of
$\delta_{0000}$. With one active central pair, its two values have ratio
in $\mu_k$. Available diagonal port transfers and a port permutation
therefore normalize the function to
\[
 q=s\delta_{0000}+\delta_{0011}+\delta_{1100},\qquad s\ne0.
\]
In the $12\mid34$ flattening, restrict to the $\{00,11\}$ block:
\[
 Q=\begin{pmatrix}s&1\\1&0\end{pmatrix},\qquad
 J=(N\otimes N)|_{\{00,11\}}=\begin{pmatrix}0&1\\1&0\end{pmatrix}.
\]
Form a chain of $m$ copies, each time joining the column ports of one copy
to the row ports of the next by two $N$ edges. Its output is
$(QJ)^{m-1}Q$, and direct multiplication gives
\[
 (QJ)^{m-1}Q=\begin{pmatrix}ms&1\\1&0\end{pmatrix}.
\]

With two active central pairs, the rank-one relation among their four
values and the $\mu_k$ ratios give a product of two available
anti-diagonal binary factors. Actual diagonal port transfers normalize
the values to the expression
\[
 q=s\delta_{0000}+N_{12}N_{34}.
\]
We now form a chain directly from copies of the entire function $q$. Change to the
$13\mid24$ flattening and put $J=N\otimes N$. Then
\[
 Q=J+sE_{00,00},\qquad QJ=I+sE_{00,11},\qquad E_{00,11}^2=0.
\]
The same two-wire chain actually realizes
$(QJ)^{m-1}Q=J+msE_{00,00}$. Thus in either case $O(m)$ vertices realize
$q_m=ms\delta_{0000}+q_0$, where $q_0$ is fixed. If the normalized gadget
actually realizes $cq$, a chain of length $m$ realizes $c^mq_m$. A query
with $M$ replacements below must first be divided by the known nonzero
scalar $c^{mM}$.

Consider an instance with $M$ target occurrences of $\delta_{0000}$, and first separate these target vertices from the remaining fixed function-net. Replace all targets by $q_m$ and expand multilinearly over these $M$ positions. The answer is a polynomial in $m$ of degree at most $M$, whose leading
coefficient is exactly $s^M$ times the target answer. Query
$m=1,\ldots,M+1$, using $M+1$ oracle calls and replacement size $O(M^2)$
per query. Recover the leading coefficient by interpolation and divide by
$s^M$. This gives Turing access to arbitrarily many occurrences of
$\delta_{0000}$. After returning the entire problem to the equality-edge model, the new function remains a tensor power of a nonzero unary function. By the decomposition theorem, Theorem~\ref{en:thm:decomposition}, either the problem belongs to $\FP$, or this unary factor may be adjoined. The odd-arity dichotomy, Theorem~\ref{en:thm:odd-arity-fp}, then completes the proof.
\end{proof}

\begin{lemma}[The three-block contradiction when the endpoint product vanishes]\label{en:lem:cyc-native-three-centers}
Under the binary closure conditions of this section for the native $N$ edge model, a quaternary function with $st=0$ cannot have three active central complementary pairs.
\end{lemma}
\begin{proof}
Use the port orientation in~\eqref{en:eq:cyc-native-q4}. Join the row ports
of two copies in corresponding order by two $N$ edges, retaining the
column ports. The actual output is $Q^{\mathsf T}(N\otimes N)Q$.
By~\eqref{en:eq:cyc-native-center-products}, the products of its three
central complementary pairs are
\[
 (st+P)^2,\qquad (2cf)(2de)=4P^2,\qquad (cd+ef)^2=4P^2.
\]
If $st=0$, all three pairs are active. Reapplying central-ratio rigidity
would give $P^2=4P^2$, a contradiction. Here the connection matrix between the two copies is $N\otimes N$.
\end{proof}

\begin{proposition}[Quaternary exit for cyclic groups with native disequality edges]\label{en:prop:cyc-native-q4}
Under the binary closure conditions of this section for the native $N$ edge model, let $k\ge3$. If some quaternary function is not a pairing product of binary functions from this group, then the entire problem has a decidable complexity dichotomy.
\end{proposition}
\begin{proof}
The quaternary function has even support. If both endpoints are nonzero
and the function is indecomposable, apply
Theorem~\ref{en:thm:cyc-external-q4} directly. If it factors, the two nonzero
endpoints exclude unary factors, and even support forces a product of two
full-rank diagonal binary factors. After decomposition, both factors are available, contradicting the fact that every binary function in this group is anti-diagonal.

If exactly one endpoint is nonzero, three active central pairs have been
excluded by Lemma~\ref{en:lem:cyc-native-three-centers}; the remaining cases
are resolved by Lemma~\ref{en:lem:cyc-native-one-endpoint}. If both endpoints
are zero, having no active central pairs gives the zero function. One active pair
is resolved by Lemma~\ref{en:lem:cyc-complement-carrier}. Two active pairs
already form a product of two actually available anti-diagonal factors,
by the central rank-one relation, contradicting the hypothesis. Three
active pairs are again excluded by the three-block contradiction.
\end{proof}

\subparagraph*{An example with only EO quaternary functions}
\begin{example}\label{en:ex:cyclic-no-endpoint}
In the native $N$-edge model, consider
$\Gamma=\{=_6\}\cup\{D^jN:0\le j<k\}$, where
$D=\diag(e^{\pi i/k},e^{-\pi i/k})$.
After suppressing the binary vertices, the equality vertices are joined by
disequality edges with nonzero weights. A connected component contributes zero
if it is not bipartite; otherwise there are only two global assignments, whose
contributions can be computed directly. Thus the problem is in $\FP$.
Its complete projective binary group is $C_k$.

On the other hand, suppose a nonzero connected four-port gadget has $a,b$
equality vertices in its two bipartition classes and $u,v$ external ports on
the respective sides. Summing degrees gives $6(a-b)=u-v$, while $u+v=4$.
Hence $a=b$ and $u=v=2$, and its quaternary support lies entirely in the EO layer.
Components containing only binary vertices do not change this conclusion.
Thus, even when an original function has two nonzero endpoints, it may be
impossible to construct a quaternary function with two nonzero endpoints.
Connecting two copies of $=_6$ by four $N$ edges does give a complementary-pair
EO quaternary function. Such functions are handled by
Lemma~\ref{en:lem:cyc-complement-carrier}.
\end{example}

\subparagraph*{Combining the upper and lower cases}

Combining arity reduction in the upper case with the quaternary classification
in the lower case gives a complete dichotomy in each edge model.

\begin{corollary}\label{en:cor:cyc-equality-dichotomy}
Under the binary closure condition of this section for the equality-edge diagonal
group $C_k$ of odd order $k\ge3$, every finite algebraic complex-valued even-arity
function set has an $\FP/\sharpP$-hard dichotomy.
\end{corollary}
\begin{proof}
If all input functions are actual binary pairing products, expand their fixed
decompositions. Every connected component becomes a cycle of binary matrices.
Multiply the matrices in the appropriate orientations and take the trace to
evaluate it in polynomial time. Self-loops and parallel edges are covered by
the same algorithm. Otherwise, Lemma~\ref{en:lem:cyc-diagonal-arity} gives a
quaternary gadget that is not a pairing product, and
Proposition~\ref{en:prop:cyc-equality-q4} completes the proof.
\end{proof}

\begin{corollary}[Complete dichotomy for higher-order cyclic groups with native disequality edges]\label{en:cor:native-cyclic-dichotomy}
Let $k\ge3$, and let $\Gamma$ be a finite algebraic complex-valued even-arity
function set satisfying the binary closure condition of this section in the
native $N$-edge model. Then the entire $\Gamma$ problem has a decidable
$\FP/\sharpP$-hard dichotomy.
\end{corollary}
\begin{proof}
If every input function belongs to the actual binary pairing product class,
expand these fixed decompositions. Each connected component is a cycle of
binary functions, and its value is the trace of the product obtained by
alternating the actual binary matrices with the edge kernel $N$. The entire
problem is therefore computable in polynomial time. Zero functions, nullary
scalars, self-loops, and parallel edges are all handled by the same rule.

Otherwise, some original input function is already a non-pairing-product
function. If an ordinarily realizable quaternary non-pairing-product function
exists, apply Proposition~\ref{en:prop:cyc-native-q4} directly. In the remaining
case, every ordinary quaternary gadget is zero or a binary pairing product from
the current group. Among the ordinary non-pairing-product gadgets over the
current fixed function set, choose one of minimum arity and call it $F$.
It exists because the input function just identified belongs to this set.
The binary closure condition and the quaternary assumption ensure that its
even arity is at least six. Theorem~\ref{en:thm:native-cyclic-shape} gives
$F=Q+a\delta_{0^{2d}}+b\delta_{1^{2d}}$, where $Q=0$ or $Q$ is an actual
antidiagonal matching product whose factor ratios belong to $\mu_k$.
If $a=b=0$, then $F$ is zero or belongs to the actual binary pairing product
class, contrary to its choice. Lemma~\ref{en:lem:native-cyclic-two-endpoints}
excludes $ab\ne0$, so exactly one endpoint is nonzero. Whether or not $Q=0$,
Lemma~\ref{en:lem:native-cyclic-endpoint-interpolation} gives the dichotomy for
the original problem.

Finally, we explain decidability. Membership in the binary pairing product
class can be decided by finite enumeration of matchings and algebraic-number
arithmetic. For the fixed function set to which the group's binary directions
have already been adjoined, if not all input functions are pairing products,
enumerate ordinary gadgets until finding either a quaternary non-pairing-product
gadget or a gadget of the form~\eqref{en:eq:native-cyclic-tail-shape} with at least
one nonzero endpoint and $Q$ satisfying the conditions above. The existence
proof guarantees that this enumeration terminates. A certificate of the first
type is handled by the quaternary dichotomy. For a certificate of the second
type with two nonzero endpoints, the two-copy construction in
Lemma~\ref{en:lem:native-cyclic-two-endpoints} directly produces a quaternary
non-pairing-product gadget; with one nonzero endpoint, the interpolation lemma
applies directly. Every certificate condition is decidable by finite checks of
support, matching, and root-of-unity ratios. The predicates on the complete
function set in the external dichotomies are also decidable. This procedure
therefore terminates and determines the complexity of the original problem.
\end{proof}

\addtocontents{toc}{\protect\setcounter{tocdepth}{2}}

\subsection{Odd Dihedral Groups}\label{en:subsec:road-odd}
For the dichotomy theorem for odd dihedral groups, see also the draft by Jiayi Zheng\cite{en:zheng-draft}. This appendix retains an AI-produced proof.

For odd dihedral groups that cannot yet obtain binary disequality, we first follow the original steps of
``orthogonal triangularization, followed by diagonalization forced by transpose closure'' to obtain a normal form,
then use diagonal and antidiagonal auxiliary functions to reduce to arity four, and finally give the complexity dichotomy for the entire function set.

\subsubsection{Normal Form in the Absence of Binary Disequality}
\begin{lemma}\label{en:lem:odd-dihedral-normal}
If $\mathcal G$ is a dihedral group of order $2n$, where $n\ge3$ is odd, and the condition of Lemma~\ref{en:lem:N-symmetric} does not hold, then there is a complex orthogonal basis in which
\begin{equation}\label{en:eq:odd-normal}
 \mathcal B=\langle D,Y\rangle,\qquad
 D=\diag(\lambda,\lambda^{-1}),\qquad \lambda=e^{\pi i/n}.
\end{equation}
\end{lemma}
\begin{proof}
Write $\mathcal G=\langle r,s:r^n=s^2=1,\ srs=r^{-1}\rangle$.
Its nonidentity elements of order two are precisely $s,sr,\ldots,sr^{n-1}$, of which there are $n$.
Transpose acts as an involution on these cosets and therefore fixes at least one.
A lift $S$ of a fixed coset has eigenvalues $i,-i$ and satisfies $S^{\mathsf T}=S$ or $S^{\mathsf T}=-S$.
The former case would give $N$, so necessarily $S=\pm Y$, and hence $Y\in\mathcal B$.

The cyclic subgroup $C=\langle r\rangle$ consists exactly of all elements of odd order, so transpose preserves $C$.
Choose a lift $R$ with $R^n=-I$; the full preimage of $C$ is $\langle R\rangle$.
If both eigenlines of $R$ were isotropic, they would be
$\mathbb C(1,i)^{\mathsf T}$ and $\mathbb C(1,-i)^{\mathsf T}$.
These are also the eigenlines of $Y$, so $RY=YR$.
But in the odd dihedral group, $[Y]$ is an element of order two outside the cyclic subgroup and conjugates $r$ to $r^{-1}\ne r$, a contradiction.
Thus $R$ has a non-isotropic eigenvector $u$, which can be extended, after normalization, to a complex orthogonal basis.
In this basis $R$ is upper triangular, as is every element of $\langle R\rangle$.
This cyclic subgroup remains closed under transpose, so $R^{\mathsf T}$ is also upper triangular; hence $R$ is diagonal.
Choose a cyclic generator so that it is the matrix $D$ in Equation~\eqref{en:eq:odd-normal}.
A two-dimensional orthogonal matrix $M$ satisfies $M^{\mathsf T}YM=(\det M)Y$, so $Y$ remains in the group.
The classes $[D]$ and $[Y]$ already generate the entire quotient; these two matrices, together with $-I=D^n$, therefore generate $\mathcal B$.
\end{proof}

The argument first triangularizes the full preimage of the cyclic subgroup in an orthogonal basis, then uses transpose closure of that subgroup to obtain diagonalization.

\subsubsection{Direct arity reduction for odd dihedral groups}
\label{en:sec:odd-arity}

This section uses the equality-edge model.
Suppose every nonzero ordinarily realizable binary function is a full-rank diagonal or antidiagonal matrix.
Let $\mathcal B$ be the set of their determinant-normalized matrices, retaining both signs, and let $\mathcal G=\mathcal B/\{I,-I\}$ be the corresponding quotient group.
Assume that three pairwise nonproportional diagonal auxiliary functions are realizable.
Ignoring compensable nonzero overall scalars, write them as
\begin{equation}\label{en:eq:odda-three-probes}
 L_j=\operatorname{diag}(1,t_j),\qquad
 t_j\ne0,\qquad t_i\ne t_j\quad(i\ne j),\qquad j=1,2,3.
\end{equation}
Also assume that an invertible antidiagonal port action and its inverse are realizable.
These operations preserve the set of actually realizable binary functions.
The odd dihedral normal form
\[
 \mathcal B=\langle D,Y\rangle,\qquad
 D=\operatorname{diag}(\lambda,\lambda^{-1}),\qquad
 \lambda=e^{\pi i/n},\quad n\ge3\text{ is odd},\qquad
 Y=\begin{pmatrix}0&1\\-1&0\end{pmatrix}
\]
satisfies these conditions: choose the three diagonal directions from $I,D,D^2$, and use $Y$ as the antidiagonal action, with inverse $-Y$.

Let $\mathcal P_{\mathcal B}$ consist of the zero functions and the following binary pairing products: take a tensor product of actually realizable full-rank binary functions on a perfect matching, and multiply by a nonzero overall scalar.
Thus a nonzero member of arity $2d$ has the form
\[
 C(x_1,\ldots,x_{2d})=
 c\prod_{\{u,v\}\in M}B_{uv}(x_u,x_v),\qquad c\ne0.
\]
Every binary factor is supported either on the two equal inputs or on the two unequal inputs, so
\begin{equation}\label{en:eq:odda-product-support-size}
 |\operatorname{supp}(C)|=2^d.
\end{equation}
Nullary constants are treated as empty pairing products.

The proof follows the same steps as for higher-order cyclic groups.
First, a linear relation among three shortloops gives a common pairing.
Next, recover the structures of the $00$ and $11$ slices.
Finally, combine three views to obtain two rhinoceros slices and an antelope slice, and reconstruct the whole function as a binary pairing product.
The additional antidiagonal auxiliaries are used to flip bits and to exclude the four-variable exception.
Slices and linear combinations analyze function values; actual gadgets connect the specified binary auxiliaries.

\paragraph{Three-term relations among binary pairing products}

\begin{lemma}[A common pairing or a four-variable exception]\label{en:lem:odda-three-products}
Let $U,V,W$ be nonzero binary pairing products of arity $2d$, formed from full-rank diagonal or antidiagonal binary functions, satisfying
\[
 aU+bV+cW=0,\qquad abc\ne0,
\]
and pairwise nonproportional.
Then one of the following two cases holds:
\begin{enumerate}
 \item The three functions have the same support, and hence the same matching and the same equality or disequality constraint on each edge.
 \item Their matchings differ only on four variables, where they use the three distinct matchings.
 A common nonzero binary pairing product $Q$ can be factored out on the external variables.
 After suitable bit flips on the four variables, the union of the three local supports is exactly the set of all six inputs of Hamming weight two.
 In the two-dimensional space spanned by the three local functions, every nonzero function other than these three projective points is nonzero at all six positions.
\end{enumerate}
\end{lemma}

\begin{proof}
The support of a binary pairing product of full-rank monomial matrices determines its matching and edge constraints.
For any variable, only its paired variable is always equal to it or always its complement; the free bits of all other pairs vary independently.
If $U,V$ have the same support, the linear relation places $\operatorname{supp}(W)$ inside that support.
Since all three supports have the same size, they must be equal.

If $U,V$ have different supports, the nondegenerate relation gives
\[
 \operatorname{supp}(U)\mathbin{\triangle}\operatorname{supp}(V)
 \subseteq\operatorname{supp}(W).
\]
By Equation~\eqref{en:eq:odda-product-support-size},
\begin{equation}\label{en:eq:odda-half-intersection}
 |\operatorname{supp}(U)\cap\operatorname{supp}(V)|\ge2^{d-1}.
\end{equation}
The union of the two matchings decomposes into alternating cycles, with a common edge treated as a cycle of two parallel edges.
If the equality and disequality constraints around a cycle are inconsistent, the intersection is empty.
If they are consistent, each cycle leaves only one free bit.
A cycle on $2r$ vertices leaves $r$ free bits in one matching but only one in the intersection.
Thus different matchings have an intersection of size at most $2^{d-1}$.
Equality holds only when there is exactly one four-variable alternating cycle and all other edges and constraints agree.
Different edge constraints on the same matching give disjoint supports, which also cannot satisfy Equation~\eqref{en:eq:odda-half-intersection}.

Consequently the intersection has size exactly $2^{d-1}$, and the support of $W$ is precisely the symmetric difference above.
Its matching agrees outside the four variables and must use the third matching inside them.
The local intersection contains two inputs.
At either such input, $W$ vanishes while $U,V$ do not; substituting into the relation shows that the external product factors of $U,V$ are proportional.
Substitution at any local input in the symmetric difference shows that the external factor of $W$ is proportional to them as well.
Hence a common factor $Q$ can be extracted.

To describe the standard form of the local supports, write the first two local matchings as $12|34$ and $13|24$.
Choose an input in their common support and flip its bits to $0000$.
All constraints in the first two matchings become equalities, so both constraints in the third matching $14|23$ are disequalities.
Now flip variables 2 and 3.
The three local supports become
\[
 \operatorname{supp}(Y_{12}Y_{34}),\quad
 \operatorname{supp}(Y_{13}Y_{24}),\quad
 \operatorname{supp}(Y_{14}Y_{23}),\qquad
 Y=\begin{pmatrix}0&1\\-1&0\end{pmatrix}.
\]
These functions describe only the supports; the local nonzero values need not equal the corresponding products of $Y$ values.
The union is the set of all six four-bit inputs of weight two.

Finally, consider the local two-dimensional linear space.
The two exclusive parts of the supports cannot disappear in a nondegenerate linear combination.
The two positions in the intersection must cancel simultaneously at the same parameter; otherwise a third binary pairing product could not occur.
Except for the first two functions themselves and this unique cancellation direction, every nonzero combination is therefore nonzero at all six positions.
\end{proof}

The four-variable exception does occur; for example,
\[
 Y_{12}Y_{34}-Y_{13}Y_{24}+Y_{14}Y_{23}=0.
\]
The next step excludes it by applying shortloops on different ports.

\paragraph{Changing ports to exclude the four-variable exception}

The next several lemmas concern a nonzero ordinarily realizable function $F$ of arity $2m$, with $m\ge3$.
Assume that every ordinarily realizable function of smaller even arity belongs to $\mathcal P_{\mathcal B}$.
This assumption will hold in the final minimal-counterexample argument.

\begin{lemma}[Consistent pairings for the three shortloops]\label{en:lem:odda-cap-consistency}
Fix any two ports $p,q$ of $F$.
Among the results of shorting them with the three diagonal auxiliaries in Equation~\eqref{en:eq:odda-three-probes}, all nonzero results have the same support, and hence the same matching and edge constraints.
\end{lemma}

\begin{proof}
Denote the four mathematical slices by
\[
 A=F_{pq}^{00},\quad B=F_{pq}^{11},\quad
 C=F_{pq}^{01},\quad E=F_{pq}^{10}.
\]
The three actual shortloop results are $H_j=A+t_jB$.
If $A,B$ are linearly dependent, all nonzero results are proportional and the claim follows.
If they are independent, the three results are nonzero, pairwise nonproportional, and satisfy a linear relation with all three coefficients nonzero.
The lower-arity assumption makes every $H_j$ a binary pairing product, so Lemma~\ref{en:lem:odda-three-products} applies.

Suppose its four-variable exception occurs.
Denote the local variables by $x_1,\ldots,x_4$ and the remaining variables by $z$.
Factor out the common external function $Q(z)$, and set
\[
 \mathcal S=\operatorname{supp}(Q),\qquad s=|\mathcal S|=2^{m-3}.
\]
Use the available invertible antidiagonal actions to perform the local bit flips in the lemma.
These actions only flip input bits and multiply by nonzero weights, preserving support sizes, membership in $\mathcal P_{\mathcal B}$, and the lower-arity assumption.
Keep the same notation after these actions.
The union of the three special supports is now
\[
 E_2\times\mathcal S,\qquad
 E_2=\{x\in\{0,1\}^4:|x|=2\}.
\]
In this two-dimensional space, $A,B$ correspond to parameters $0,\infty$, whereas the three binary pairing product parameters $t_j$ are nonzero and finite.
The final assertion of Lemma~\ref{en:lem:odda-three-products} therefore says that $A,B$ are nonzero at exactly all $6s$ positions and zero elsewhere.

Fix $i\in\{1,2,3,4\}$ and instead use the same three diagonal auxiliaries on ports $p,x_i$, obtaining functions $K_j^{(i)}$.
Their remaining ports are $q,x_{\ne i},z$, and
\begin{align*}
 K_j^{(i)}(q=0)&=A(x_i=0)+t_jE(x_i=1),\\
 K_j^{(i)}(q=1)&=C(x_i=0)+t_jB(x_i=1).
\end{align*}
Consider the set of positions
\begin{equation}\label{en:eq:odda-anchor-set}
 q+|x_{\ne i}|=2,\qquad z\in\mathcal S.
\end{equation}
It contains $6s$ positions.
At each one, the preceding linear expression in $t_j$ has at least one known nonzero coefficient: from $A$ when $q=0$, and from $B$ when $q=1$.
Such an expression vanishes at at most one of the three distinct parameters, so
\begin{equation}\label{en:eq:odda-support-lower-bound}
 \sum_{j=1}^{3}|\operatorname{supp}(K_j^{(i)})|\ge12s.
\end{equation}
On the other hand, every result is ordinarily realizable with strictly smaller arity, and is therefore zero or a binary pairing product.
Its arity is $2m-2$, so a nonzero result has support size $2^{m-1}=4s$.
Equality must therefore hold in Equation~\eqref{en:eq:odda-support-lower-bound}.
Thus all three results vanish at every position outside Equation~\eqref{en:eq:odda-anchor-set}.
A linear expression vanishing at three distinct parameters has both coefficients zero.

As $i$ ranges over the four local variables, we obtain the following restrictions: if a nonzero input of $C$ has $x_i=0$, then $|x|=1$; if a nonzero input of $E$ has $x_i=1$, then $|x|=3$.
Their possible supports on the local four variables therefore satisfy
\[
 \operatorname{supp}_x(C)\subseteq E_1\cup\{1111\},\qquad
 \operatorname{supp}_x(E)\subseteq E_3\cup\{0000\},
\]
where $E_r=\{x:|x|=r\}$ and $\operatorname{supp}_x$ denotes projection of the nonzero support onto these four variables.
No extra assumption on the external variables was used, so these inclusions also cover inputs that might lie outside $\mathcal S$.

Now exchange $p,q$ and apply the same three shortloops on ports $q,x_i$.
The identical count gives the opposite restrictions
\[
 \operatorname{supp}_x(C)\subseteq E_3\cup\{0000\},\qquad
 \operatorname{supp}_x(E)\subseteq E_1\cup\{1111\}.
\]
The two allowed sets for each function are disjoint, so $C=E=0$.
But then every $K_j^{(i)}$ is nonzero at all $6s$ positions of Equation~\eqref{en:eq:odda-anchor-set}, exceeding the $4s$ positions allowed for a binary pairing product.
This contradiction excludes the four-variable exception.
\end{proof}

\paragraph{A common factor of two slices}

Inserting invertible antidiagonal auxiliaries at suitable ports flips a nonzero input of $F$ to the all-zero input.
Both these operations and their inverses are realizable, so they preserve membership in $\mathcal P_{\mathcal B}$ and the lower-arity assumption.
We may therefore arrange that
\begin{equation}\label{en:eq:odda-all-zero-nonzero}
 F(0^{2m})\ne0.
\end{equation}

\begin{lemma}[At most one changing factor]\label{en:lem:odda-one-changing-factor}
Suppose three nonzero functions are tensor products over the same partition of the variables, are pairwise nonproportional, and satisfy a nondegenerate three-term linear relation.
For any two of them, their factors can be nonproportional on at most one variable group.
\end{lemma}

\begin{proof}
If the factors were nonproportional on at least two variable groups, flatten the functions between one such group and all remaining variables.
The first two functions become $uv^{\mathsf T}$ and $xy^{\mathsf T}$, with $u,x$ independent and $v,y$ independent.
A sum with both coefficients nonzero therefore has rank two and cannot be the rank-one flattening of the third function.
\end{proof}

\begin{lemma}[A common external factor of the slices]\label{en:lem:odda-common-slice-factor}
For any ports $p,q$, a pair $z,w$ can be chosen among the remaining variables so that
\begin{equation}\label{en:eq:odda-slice-factorization}
 F_{pq}^{00}=E_0(z,w)Q,\qquad
 F_{pq}^{11}=E_1(z,w)Q.
\end{equation}
Here $Q$ is a binary pairing product of actually realizable full-rank diagonal binary factors.
The matrices $E_0,E_1$ are diagonal, but may at this stage be degenerate, and $E_0(00)\ne0$.
\end{lemma}

\begin{proof}
The values of the three shortloops at the all-zero assignment of the remaining variables are
\[
 F_{pq}^{00}(0)+t_jF_{pq}^{11}(0).
\]
The first term is nonzero, so at least two results have the all-zero input in their support.
If the support of a full-rank monomial binary pairing product contains the all-zero input, all its binary factors must be diagonal.

If the two slices are linearly dependent, every nonzero shortloop is proportional to this purely diagonal binary pairing product, and the conclusion follows immediately.
If they are independent, all three shortloops are nonzero and, by Lemma~\ref{en:lem:odda-cap-consistency}, have the same support and matching.
Thus all three are purely diagonal binary pairing products.
By Lemma~\ref{en:lem:odda-one-changing-factor}, their factors are nonproportional on at most one pair.
Extract the other common factors as $Q$, and recover the slices from the linear equations for two shortloops, obtaining Equation~\eqref{en:eq:odda-slice-factorization}.
The factors of $Q$ can be chosen as factors in the decomposition of an actual shortloop result, so they do belong to the set of actually realizable binary functions.
Finally, Equation~\eqref{en:eq:odda-all-zero-nonzero} gives $E_0(00)\ne0$.
\end{proof}

\begin{lemma}[Complete pairings and complementarily pinned pairs]\label{en:lem:odda-two-slice-forms}
In Equation~\eqref{en:eq:odda-slice-factorization}, exactly one of the following two cases holds:
\begin{enumerate}
 \item Both $E_0,E_1$ have full rank and are proportional; the two slices have the same complete diagonal-pairing support.
 \item The function $E_0$ is nonzero only at $00$, and $E_1$ only at $11$.
 In the two slices, the same pair $z,w$ is pinned to $00$ and $11$, respectively, while all remaining pairing factors have full rank.
\end{enumerate}
In particular, both slices are nonzero.
\end{lemma}

\begin{proof}
For each full-rank diagonal factor of $Q$, choose an $L_j$ whose bilinear closure with it is nonzero.
The closure value is a nonzero linear polynomial in $t_j$, so at most one probe makes it zero; a choice is always possible.
Connect these auxiliaries on the corresponding external pairs of the original function $F$ to obtain an actual quaternary function $H(p,q,z,w)$.
Its two slices are a common nonzero scalar times $E_0,E_1$, and $H(0000)\ne0$.

Arity four is strictly below $2m$, so $H\in\mathcal P_{\mathcal B}$.
Since its support contains the all-zero input, it is a purely diagonal binary pairing product.
If its matching is $pq|zw$, the two slices are nonzero multiples of the same full-rank $zw$ factor.
If its matching is $pz|qw$ or $pw|qz$, pinning $p=q=0$ leaves only $z=w=0$, and pinning $p=q=1$ leaves only $z=w=1$, with both corresponding values nonzero.
These are precisely the two stated cases.
\end{proof}

To retain the visual names for these slice structures, call a complete-pairing slice in the first case an \emph{antelope}.
Call a slice in the second case, with a pair of variables fixed to $00$ or $11$, a \emph{0-rhinoceros} or \emph{1-rhinoceros}, respectively, and call that pair its horn pair.
The terms antelope and rhinoceros below refer to these two slice structures.
In the first case, both slices are proportional to a nonzero actual shortloop result, so all their pairing factors can be chosen from the available full-rank diagonal functions.
In the second case, Lemma~\ref{en:lem:odda-common-slice-factor} ensures that the common factors outside the horn pair are available binary functions.

\paragraph{Assembling pairings from three views}

\begin{lemma}[Internal pairing on four variables]\label{en:lem:odda-three-views}
There are four distinct variables $x,y,z,w$ such that the three slices
\[
 F_{xy}^{00},\qquad F_{xz}^{00},\qquad F_{yz}^{00}
\]
pair $z,y,x$, respectively, with $w$.
If any of these slices is a rhinoceros, its horn pair is this internal pair.
The three slices have the same matching on all remaining variables, and their external product factors are proportional.
\end{lemma}

\begin{proof}
Choose any $F_{xy}^{00}$.
If it is an antelope, choose any remaining variable $z$ and let $w$ be its partner.
If it is a rhinoceros, let $z,w$ be its horn pair.
In either case, among the remaining variables of the common slice $F_{xyz}^{000}$, exactly $w$ is forced to zero, while all other variables remain freely paired by full-rank diagonal factors.

Now view this slice from $F_{xz}^{00}$, by pinning $y$ to zero.
If this view is an antelope, the unique additionally fixed variable is the partner of $y$, so that partner is $w$.
If this view is a rhinoceros and $y$ does not lie in its horn pair, pinning $y$ forces three remaining variables to zero: the two in the horn pair and the partner of $y$.
This is a contradiction.
Thus $y$ belongs to the horn pair and the other member is $w$.
The same argument for $F_{yz}^{00}$ shows that $x$ is paired with $w$, and any horn pair must be $x,w$.

In each view, now pin the remaining two variables among $x,y,z,w$ to zero.
All three give the same nonzero function $F_{xyzw}^{0000}$.
On the external variables this is a full-rank diagonal binary pairing product.
Its support determines its matching, so the three external matchings agree.
Equality of the function values also makes the three external product factors proportional.
\end{proof}

\begin{lemma}[At least two rhinoceros views]\label{en:lem:odda-two-rhinos}
At least two of the three slices in Lemma~\ref{en:lem:odda-three-views} are 0-rhinoceroses.
\end{lemma}

\begin{proof}
If at least two were antelopes, rename $x,y,z$ so that they are $F_{xy}^{00}$ and $F_{xz}^{00}$.
Their external factors are proportional, so a common collection of diagonal auxiliaries can close this common factor to a nonzero value.
Connect these auxiliaries to the original function $F$, obtaining an actual quaternary function $H(x,y,z,w)$ with
\[
 H(0000)\ne0,\qquad H(0011)\ne0,\qquad H(0101)\ne0.
\]
The lower-arity assumption and $H(0000)\ne0$ force $H$ to be a purely diagonal binary pairing product.
But of the three matchings on four variables, $xy|zw$ excludes $0101$, $xz|yw$ excludes $0011$, and $xw|yz$ excludes both inputs.
This is a contradiction.
\end{proof}

\begin{theorem}[Reduction to a quaternary function outside the binary pairing products]\label{en:thm:odd-arity-reduction}
Under the binary rigidity and auxiliary-function conditions of this section, if an ordinarily realizable even-arity function lies outside $\mathcal P_{\mathcal B}$, then an ordinarily realizable quaternary function also lies outside $\mathcal P_{\mathcal B}$.
\end{theorem}

\begin{proof}
Suppose, for a contradiction, that no such quaternary function exists.
Choose a function $F$ of minimum arity among the ordinarily realizable functions outside $\mathcal P_{\mathcal B}$.
Arities zero, two, and four are excluded, so $F$ has arity $2m\ge6$, and every ordinarily realizable function of smaller even arity belongs to $\mathcal P_{\mathcal B}$.
Apply invertible antidiagonal port actions to ensure $F(0^{2m})\ne0$.
Undoing these actions preserves membership in $\mathcal P_{\mathcal B}$.

Apply the preceding lemmas to obtain four variables $x,y,z,w$.
Rename $x,y,z$ so that $F_{xz}^{00}$ and $F_{yz}^{00}$ are 0-rhinoceroses, with horn pairs $y,w$ and $x,w$, respectively.
The complementary $11$ slices in Lemma~\ref{en:lem:odda-two-slice-forms} give the following conditions on every nonzero input of $F$:
\begin{itemize}
 \item If $x=z$, then $x=y=z=w$.
 \item If $y=z$, then $x=y=z=w$.
 \item If $x=y$, the $00$ or $11$ slice in the third view gives $z=w$.
\end{itemize}
Among three bits $x,y,z$, at least two are equal.
Thus the entire support satisfies
\begin{equation}\label{en:eq:odda-final-support}
 x=y,\qquad z=w.
\end{equation}

Now consider $F_{xy}^{00}$ and $F_{xy}^{11}$.
If they are antelopes, they are proportional and pair $z,w$.
If they are complementary rhinoceroses, their horn pair is $z,w$.
In either case, extract a common external factor formed from actually realizable binary functions to write
\[
 F_{xy}^{00}=R_0(z,w)Q,\qquad
 F_{xy}^{11}=R_1(z,w)Q.
\]
Equation~\eqref{en:eq:odda-final-support} excludes both $F_{xy}^{01}$ and $F_{xy}^{10}$.
The common factor of the two slices is therefore a factor of the entire function:
\[
 F=J(x,y,z,w)Q.
\]
On each pair of $Q$, choose an actually realizable diagonal auxiliary with a nonzero closure value.
This ordinarily realizes a nonzero scalar multiple of $J$.
The function $J$ is quaternary, of strictly smaller arity, so the lower-arity assumption gives $J\in\mathcal P_{\mathcal B}$.
Since $Q$ is already an actual binary pairing product, $F\in\mathcal P_{\mathcal B}$ as well, contradicting its choice.
\end{proof}

\subsubsection{The Quaternary Base Case and Complexity Resolution for Odd Dihedral Groups}
\label{en:sec:oddq}

Throughout this section we use the equality-edge model and ordinary transpose. Suppose that all realizable nonzero binary functions, after determinant normalization and retention of both signs, form
\[
 \mathcal B=\langle D,Y\rangle,\qquad
 D=\operatorname{diag}(\lambda,\lambda^{-1}),\qquad
 \lambda=e^{\pi i/n},\qquad n\ge3\text{ odd},\qquad
 Y=\begin{pmatrix}0&1\\-1&0\end{pmatrix}.
\]
Its projective quotient group is $\mathcal G=\mathcal B/\{I,-I\}$. This section treats precisely the standard odd dihedral form without
$N=\left(\begin{smallmatrix}0&1\\1&0\end{smallmatrix}\right)$.
All actual gadgets may carry known nonzero overall scalars, which are compensated for according to the number of gadget occurrences.

Let $\mu_n$ denote the set of all $n$th roots of unity. Up to nonzero scalars, the two types of actually available binary probes can be written as
\begin{equation}\label{en:eq:oddq-probes}
 L_\rho=\begin{pmatrix}1&0\\0&\rho\end{pmatrix},\qquad
 K_\rho=\begin{pmatrix}0&1\\-\rho&0\end{pmatrix},
 \qquad \rho\in\mu_n.
\end{equation}
Thus the ratio of the two nonzero values of a nonzero diagonal binary function belongs to $\mu_n$; the ratio of the lower nonzero value to the upper one in a nonzero antidiagonal binary function belongs to $-\mu_n$.
In the latter case, the $n$th powers of the two values are negatives of one another.

Let $\mathcal P_{\mathcal B}$ consist of the zero functions and the tensor products of actually realizable full-rank binary functions over any perfect matching, multiplied by a nonzero scalar. Separately, let $\mathcal P$ denote the usual CSP product-type class,
whose definition allows unary, binary equality, and binary disequality factors to share variables. These classes are distinct; for example, quaternary equality
$\mathrm{EQ}_4$ belongs to $\mathcal P$ but not to $\mathcal P_{\mathcal B}$.

\paragraph{Shortloop probes and the total parity of quaternary functions}

\begin{lemma}[Fixed total parity]\label{en:lem:oddq-parity}
Under the binary rigidity conditions above, the support of every nonzero realizable quaternary function has fixed total parity.
Attaching $Y$ to any one port of a function with odd support changes it to a function with even support; this invertible actual port operation preserves membership in $\mathcal P_{\mathcal B}$.
\end{lemma}

\begin{proof}
Let $V_0,V_1$ be the two-dimensional spaces of all diagonal and all antidiagonal binary matrices, respectively, so
$V_0\cap V_1=\{0\}$. Fix a $2+2$ pairing of the quaternary function. Shorting one pair of ports gives a linear map
$T:\mathbb C^{2\times2}\to\mathbb C^{2\times2}$.

First we prove that $T(V_0)$ is contained in some $V_i$. This is clear at rank zero. At rank one, the
$L_\rho$ in~\eqref{en:eq:oddq-probes} span $V_0$, so at least one probe has a nonzero image;
this allowed binary image belongs to some $V_i$ and spans the entire image line. At rank two, $T|_{V_0}$ is injective,
so at least three distinct projective probe directions give three distinct nonzero image directions. If the two-dimensional image space is contained in neither
$V_0$ nor $V_1$, its intersection with each has dimension at most one, so it can contain at most two allowed projective directions, a contradiction.
The rank-one case already includes situations in which some shortloop results vanish; at rank two, no nonzero probe lies in the kernel.

Applying the same argument to $K_\rho$ shows that $T(V_1)$ is also contained in some $V_i$. Shorting the other pair of ports gives the same conclusion in the opposite direction.
Thus, after the rows and columns of the quaternary flattening are each divided into two parity groups, every row group and every column group has at most one nonzero parity block.
The nonzero blocks of this $2\times2$ block table form a partial matching: two nonzero blocks lie either both on the diagonal or both off the diagonal; a single nonzero block also has fixed total parity.

Attaching $Y$ to a port merely flips that input and multiplies by a nonzero weight; its inverse is realized by $Y^{-1}=-Y$.
It changes total parity and turns each binary factor into a product of existing matrices, with ordinary transpose when needed.
Hence it preserves membership in $\mathcal P_{\mathcal B}$ in both directions.
\end{proof}

\paragraph{Root-of-unity ratios and the six parity blocks}

\begin{lemma}[Two-port root-of-unity ratios]\label{en:lem:oddq-root-ratio}
Suppose $n\ge3$, and let
\[
 Q=\begin{pmatrix}u&v\\w&z\end{pmatrix}.
\]
If, for every $\rho\in\mu_n$,
\begin{equation}\label{en:eq:oddq-root-two-directions}
 (u+\rho w)^n=(v+\rho z)^n,\qquad
 (u+\rho v)^n=(w+\rho z)^n,
\end{equation}
then exactly one of the following four cases holds:
\begin{enumerate}
 \item all four entries are zero;
 \item $v=w=0$, $u,z\ne0$, and $u^n=z^n$;
 \item $u=z=0$, $v,w\ne0$, and $v^n=w^n$;
 \item all four entries are nonzero, $uz=vw$, and their $n$th powers are all equal.
\end{enumerate}
Conversely, all four cases satisfy~\eqref{en:eq:oddq-root-two-directions}.
\end{lemma}

\begin{proof}
Expand on all $n$th roots of unity and compare the discrete Fourier coefficients. The two constant-term equations are
\[
 u^n+w^n=v^n+z^n,\qquad
 u^n+v^n=w^n+z^n,
\]
so $u^n=z^n$ and $v^n=w^n$. The remaining coefficients, after canceling the nonzero binomial coefficients, give
\[
 u^{n-r}w^r=v^{n-r}z^r,\qquad
 u^{n-r}v^r=w^{n-r}z^r\quad(1\le r<n).
\]
The zero patterns immediately leave only the first three cases or the case of four nonzero entries. In the latter case, dividing the two equations in the first family for $r=1,2$ gives $uz=vw$, and substituting back into $r=1$ gives $u^n=v^n$.
Sufficiency follows by substitution in each case.
\end{proof}

By Lemma~\ref{en:lem:oddq-parity}, we may assume below that the quaternary function $F$ has been transformed to have even support.
Use the natural flattening, with row index $x_1x_2$ and column index $x_3x_4$, both ordered as $00,01,10,11$.
Write
\begin{equation}\label{en:eq:oddq-eight-parameters}
 N_{12|34}(F)=
 \begin{pmatrix}
 s&0&0&e\\
 0&a&c&0\\
 0&d&b&0\\
 f&0&0&t
 \end{pmatrix}.
\end{equation}
That is,
\[
 s=F_{0000},\quad t=F_{1111},\quad
 a=F_{0101},\ b=F_{1010},\quad
 c=F_{0110},\ d=F_{1001},\quad
 e=F_{0011},\ f=F_{1100}.
\]
The even and odd blocks for the three pairings are, respectively,
\begin{equation}\label{en:eq:oddq-six-blocks}
\begin{array}{c|c|c}
 \text{Pairing}&\text{Even block}&\text{Odd block}\\ \hline
 12|34&\begin{pmatrix}s&e\\f&t\end{pmatrix}
       &\begin{pmatrix}a&c\\d&b\end{pmatrix}\\[4pt]
 13|24&\begin{pmatrix}s&a\\b&t\end{pmatrix}
       &\begin{pmatrix}e&c\\d&f\end{pmatrix}\\[4pt]
 14|23&\begin{pmatrix}s&c\\d&t\end{pmatrix}
       &\begin{pmatrix}e&a\\b&f\end{pmatrix}
\end{array}
\end{equation}
The row and column indices of the even blocks are $00,11$, and those of the odd blocks are $01,10$.

\begin{lemma}[Ratio constraints on the parity blocks]\label{en:lem:oddq-block-rigidity}
Every even block in~\eqref{en:eq:oddq-six-blocks} satisfies Lemma~\ref{en:lem:oddq-root-ratio}.
For every odd block $Q=\left(\begin{smallmatrix}u&v\\w&z\end{smallmatrix}\right)$, the modified matrix
\[
 Q^\sharp=\begin{pmatrix}u&-v\\-w&z\end{pmatrix}
\]
also satisfies that lemma. In particular, if all four entries of an odd block are nonzero, then
\begin{equation}\label{en:eq:oddq-odd-block-sign}
 uz=vw,\qquad u^n=z^n=-v^n=-w^n.
\end{equation}
\end{lemma}

\begin{proof}
For an even block, shorting with $L_\rho$ from each side gives a diagonal or zero binary function, directly yielding
\eqref{en:eq:oddq-root-two-directions}.

For an odd block, shorting with $K_\rho$ from each side gives an antidiagonal or zero binary function, so
\[
 (u-\rho w)^n=-(v-\rho z)^n,\qquad
 (u-\rho v)^n=-(w-\rho z)^n.
\]
Since $n$ is odd, these are precisely the equations asserting that $Q^\sharp$ satisfies
\eqref{en:eq:oddq-root-two-directions}. Zero shortloop results also satisfy these equations.
\end{proof}

\paragraph{Exhaustive classification of quaternary functions}

\begin{theorem}[Quaternary equality or an actual binary product]\label{en:thm:oddq-quaternary}
Under the binary rigidity conditions of this section, every nonzero realizable quaternary function $F$ belongs to exactly one of the following two mutually exclusive cases:
\begin{enumerate}
 \item $F\in\mathcal P_{\mathcal B}$;
 \item the support of $F$ consists of two complementary inputs, and attaching existing $Y$ and $D^j$ functions to its ports ordinarily realizes a nonzero scalar multiple of quaternary equality $\mathrm{EQ}_4$.
\end{enumerate}
\end{theorem}

\begin{proof}
First apply invertible port operations as in Lemma~\ref{en:lem:oddq-parity} to obtain the form
\eqref{en:eq:oddq-eight-parameters}. The ratio constraints on the three even blocks show that $s,t$ are either both zero or both nonzero, as are the two entries in each of the three pairs $(a,b),(c,d),(e,f)$.

If $s,t$ are both nonzero, each nonzero pair $(p,q)$ satisfies
\[
 pq=st,\qquad p^n=q^n=s^n=t^n.
\]
If two pairs are nonzero, however, the corresponding odd block requires, by~\eqref{en:eq:oddq-odd-block-sign}, that the nonzero
$n$th powers in the two pairs be negatives of one another, a contradiction. Thus at most one pair is nonzero.

If all three pairs vanish, then $F=s\lvert0000\rangle+t\lvert1111\rangle$ with $s^n=t^n$.
If exactly one pair is nonzero, permute ports so that it is $(e,f)$. Then $st=ef$, and
\[
 F(x_1,x_2,x_3,x_4)=A(x_1,x_2)B(x_3,x_4),\qquad
 A=\operatorname{diag}(1,f/s),\quad B=\operatorname{diag}(s,e).
\]
Both diagonal ratios lie in $\mu_n$, so the two factors really are nonzero scalar multiples of existing $D^j$ functions.
This proves actual membership in $\mathcal P_{\mathcal B}$, not merely algebraic decomposability.

Now suppose $s=t=0$. Within each nonzero pair, the two $n$th powers are equal. All three pairs vanishing would give $F=0$, which was excluded.
If exactly one pair is nonzero, permute ports so that it is $(e,f)$. Attaching $Y$ to each of ports 3 and 4 changes the support
$\{0011,1100\}$ to $\{0000,1111\}$, and the two nonzero values still have equal $n$th powers.

If exactly two pairs are nonzero, permute ports so that they are $(a,b),(c,d)$. The odd-block constraints give
\[
 ab=cd,\qquad a^n=b^n=-c^n=-d^n.
\]
Therefore
\[
 F(x_1,x_2,x_3,x_4)=A(x_1,x_2)B(x_3,x_4),\qquad
 A=\begin{pmatrix}0&1\\d/a&0\end{pmatrix},\quad
 B=\begin{pmatrix}0&a\\c&0\end{pmatrix}.
\]
Both antidiagonal ratios $d/a,c/a$ have $n$th power $-1$; since $n$ is odd, both ratios belong to
$-\mu_n$. Thus the two factors are nonzero scalar multiples of existing $D^jY$ functions.

If all three pairs are nonzero, denote the common $n$th powers within the three pairs by $A_0,C_0,E_0$, all nonzero.
The three odd blocks require
\[
 A_0=-C_0,\qquad E_0=-C_0,\qquad E_0=-A_0.
\]
The first two equations give $A_0=E_0$, contradicting the third. This exhausts all zero patterns.

The remaining nonproduct cases have been transformed by actual port operations into
\[
 H=s\lvert0000\rangle+t\lvert1111\rangle,
 \qquad s,t\ne0,\qquad(s/t)^n=1.
\]
Choose $j\in\{0,\ldots,n-1\}$ so that $\lambda^{-2j}=s/t$, and attach $D^j$ to any port.
The two nonzero values become $s\lambda^j$ and $t\lambda^{-j}=s\lambda^j$, respectively, ordinarily realizing
$s\lambda^j\mathrm{EQ}_4$. Reversing the initial invertible port operations accounts for the original function.
The two cases are mutually exclusive because their support sizes are four and two, respectively.
\end{proof}

\begin{corollary}[The quaternary resolution from direct arity reduction]\label{en:cor:oddq-eq4-exit}
Together with Theorem~\ref{en:thm:odd-arity-reduction}, if a realizable even-arity function outside
$\mathcal P_{\mathcal B}$ exists, then a nonzero scalar multiple of $\mathrm{EQ}_4$ is ordinarily realizable.
\end{corollary}

\begin{proof}
Theorem~\ref{en:thm:odd-arity-reduction} gives a quaternary function $F\notin\mathcal P_{\mathcal B}$; then
Theorem~\ref{en:thm:oddq-quaternary} realizes a nonzero multiple of $\mathrm{EQ}_4$.
\end{proof}

\paragraph{The complexity dichotomy for the entire function set}

We use the dichotomy theorem for sets of complex-valued Boolean functions in $\#\mathrm{CSP}^{2}$
\cite[Theorem~16]{en:lin2021cspd}: for a fixed finite set $\mathcal H$ of algebraic complex-valued functions,
the $\#\mathrm{CSP}^{2}$ problem, when the entire set is contained in one of
$\mathcal P,\mathcal A,\mathcal A^\alpha,\mathcal L$, belongs to
$\mathrm{FP}$, and is $\#\mathrm P$-hard otherwise. Here $\alpha^2=i$,
$M_\alpha=\operatorname{diag}(1,\alpha)$, and
$\mathcal A^\alpha=\{M_\alpha^{\otimes k}g:g\in\mathcal A\}$.
The local affine class $\mathcal L$ requires that, for each support point $\sigma$, applying to the $j$th port the matrix
$\operatorname{diag}(1,\alpha^{\sigma_j})$ gives a function in $\mathcal A$.

\begin{lemma}[An odd-order diagonal tool excludes the affine tractable classes]\label{en:lem:oddq-exclude-affine}
The function $D$ of this section does not belong to $\mathcal A\cup\mathcal A^\alpha\cup\mathcal L$.
\end{lemma}

\begin{proof}
The ratio $\rho=\lambda^{-2}$ of its two nonzero values has odd order $n\ge3$, so
$\rho\notin\mu_4$. The ratio of any two nonzero values of an affine function lies in $\mu_4$, and hence
$D\notin\mathcal A$.

If $D\in\mathcal A^\alpha$, then
$(M_\alpha^{-1}\otimes M_\alpha^{-1})D$ belongs to $\mathcal A$. The ratio of its two nonzero values is
$\alpha^{-2}\rho=-i\rho$, still outside $\mu_4$, a contradiction.
If $D\in\mathcal L$, choose the support point $00$: the local affine definition directly requires the unchanged $D$ to belong to $\mathcal A$, again a contradiction.
\end{proof}

\begin{theorem}[Dichotomy for the standard odd dihedral form without disequality]\label{en:thm:oddq-dichotomy}
Let $\mathcal F$ be a fixed finite set of algebraic complex-valued Boolean functions of even arity satisfying the binary group assumptions of this section.
Then, in these standard coordinates, if $\mathcal F\subseteq\mathcal P$, we have
$\operatorname{Holant}(\mathcal F)\in\mathrm{FP}$; otherwise,
$\operatorname{Holant}(\mathcal F)$ is $\#\mathrm P$-hard.
\end{theorem}

\begin{proof}
If $\mathcal F\subseteq\mathcal P$, expand the product-type constraints into unary weights and equality and disequality relations.
For each consistent connected component, it suffices to sum the weights of its two possible assignments; an inconsistent component contributes zero. This gives a polynomial-time algorithm.

Conversely, $\mathcal F\not\subseteq\mathcal P$, while
$\mathcal P_{\mathcal B}\subseteq\mathcal P$, so some input function lies outside
$\mathcal P_{\mathcal B}$. Corollary~\ref{en:cor:oddq-eq4-exit} ordinarily realizes a nonzero scalar multiple of
$\mathrm{EQ}_4$.

Let $\mathcal H=\mathcal F\cup\{D,Y\}$. The added functions already have fixed realization gadgets, so
$\operatorname{Holant}(\mathcal H)\equiv_T\operatorname{Holant}(\mathcal F)$.
Ordinary edges provide $\mathrm{EQ}_2$; connecting one port of $\mathrm{EQ}_{2k}$ to one port of $\mathrm{EQ}_4$ gives
$\mathrm{EQ}_{2k+2}$. Thus every CSP variable that occurs an even number of times can be connected by the corresponding even-arity equality gadget, giving
\[
 \#\mathrm{CSP}^{2}(\mathcal H)
 \equiv_T\operatorname{Holant}(\mathcal H).
\]
For the reverse direction, simply regard each Holant edge as a CSP variable occurring twice. All known nonzero gadget scalars are compensated for according to their numbers of occurrences.

The set $\mathcal H$ is not contained in $\mathcal P$, and its member $D$, by
Lemma~\ref{en:lem:oddq-exclude-affine}, simultaneously excludes the other three tractable classes. The dichotomy theorem for function sets in
$\#\mathrm{CSP}^{2}$ therefore gives hardness.
\end{proof}

If a uniform complex orthogonal transformation was previously applied to the original problem, the test
$\mathcal F\subseteq\mathcal P$ should be applied to the entire transformed function set in the standard form of this section.

\paragraph*{Continuing the first route}\label{en:sec:scope}
The three parts above refer the groups that can obtain binary disequality to~\cite{en:liu2026disequality},
and develop proofs for the order-one group, higher-order cyclic groups, and odd dihedral groups;
the result for the cyclic group of order two is cited independently from~\cite{en:zhou-draft}.
These resolutions all concern the entire function set; obtaining a quaternary function is only an intermediate step in the reduction.
In the future, I will continue the second route within the original nine-class group framework, organizing proofs class by class without invoking the dichotomy theorem with binary disequality.


\clearpage
\renewcommand{\paperlanguage}{zh}
\setcounter{section}{0}
\setcounter{subsection}{0}
\setcounter{subsubsection}{0}
\setcounter{paragraph}{0}
\setcounter{equation}{0}
\setcounter{theorem}{0}
\setcounter{tocdepth}{2}
\renewcommand{\thesection}{\arabic{section}}
\ctexset{section={name={},number=\arabic{section},numbering=true}}
\renewcommand{\theHsection}{\paperlanguage.\thesection}
\makeatletter
\def\Hy@chapapp{\Hy@chapterstring}
\makeatother
\renewcommand{\contentsname}{目录}
\renewcommand{\abstractname}{摘要}
\renewcommand{\refname}{参考文献}
\renewcommand{\appendixname}{附录}
\renewcommand{\proofname}{证明}
\renewcommand{\papertheoremname}{定理}
\renewcommand{\paperlemmaname}{引理}
\renewcommand{\paperpropositionname}{命题}
\renewcommand{\papercorollaryname}{推论}
\renewcommand{\paperdefinitionname}{定义}
\renewcommand{\paperexamplename}{例}
\renewcommand{\paperassumptionname}{假设}
\renewcommand{\paperremarkname}{注}
\pdfbookmark[0]{中文原文}{language-zh}
\title{统一二分定理之克莱因群上\\第四版：麻花型克莱因的上集证明与通往最后的路程}
\author{夏盟佶\\中国科学院软件研究所\\中国科学院大学}
\date{2026年9月}
\maketitle
\begin{abstract}
本文在九群分类框架下，延续第三版关于克莱因群上集的研究。
转置封闭将克莱因群的实际矩阵表示归为标准型与麻花型。
标准型沿用上一版实数化与调用 Real-Holant 的路线。
第四版正文集中补入麻花型的直接证明。
正文完成克莱因群上集的二分，上集的条件是任意四元函数均可分解，下集暂不处理。

附录另行整理基于最近含二元不等函数二分定理\cite{zh:liu2026disequality}的通往终点线路线，依次讨论哪些群可取得二元不等函数、
$C_1$ 循环群的证明（参见第二版）、$C_k$ 循环群的证明（参见第一版）与奇二面体群的证明。
二阶循环群 $C_2$ 直接引用 Yuxuan Zhou 的独立草稿，本文不复述其证明。
在所引用结果的基础上，以上四类群与含二元不等函数二分定理\cite{zh:liu2026disequality}覆盖的群类合在一起，
给出全部九类群的二分，从而达到最终的二分定理。其中，新定理覆盖整个克莱因群情形，包括其下集。
\end{abstract}
\tableofcontents
\clearpage

\section{背景与本版安排}
任意取定一个有限布尔复值函数集 $\mathcal F$，以其中的函数构成函网，
便得到求函网值的计数问题 $\#\mathcal F$。
原有总框架按照非零二元可实现函数所对应的射影群，把待处理的有限群情形分为九类，
目标是逐类建立复杂性二分。这一框架同时区分两种高元情形：
所有非零四元普通可实现函数均可二二分解的上集，以及已经出现不可二二分解四元函数的下集。

第三版\cite{zh:xia2026v3}集中讨论克莱因群的上集，
先利用转置封闭识别两种实际矩阵表示，再对标准型进行共同相位与实数化分析。
本版正文延续这一次序，第四版新增工作的中心是麻花型证明。
在明确的二二得四上集条件下，先由双活收缩证明普通闭包中的最小反例满足散珍珠条件，
再依次证明共同配对、支撑扩散以及六元和八元的方向矛盾，最后给出逐圈算法。
正文据此完成克莱因群上集的证明；下集暂不处理。

附录~\ref{zh:sec:appendix-road}另开“通往最后的路程”。
第一条路线利用最近的含二元不等函数二分定理，先处理能够取得该函数的群，
然后整理循环群和奇二面体群的证明。第二条路线继续按原来的九类群框架完成直接证明，
本轮仅记录这一计划，不展开它的技术内容。

\section{基础记号与归约}\label{zh:sec:foundations-zh}

\subsection{函网、网函与结合律}

本文的函数均以布尔变量为自变量，以代数复数为函数值。取定一个有限函数集
$\mathcal F$，函数本身及其元数都不属于输入的一部分。给定一个有限图，在每个点上
放置一个元数等于该点度数的函数，并指定该点各端口与函数自变量的对应次序，就得到
一个\emph{函网}。允许自环和平行边；自环占用同一点的两个不同端口。

每条内部边带一个布尔变量。同一边的两个端口使用同一个变量，所以普通边的二元核是
相等函数。把各点函数相乘，再对所有内部边变量求和，得到函网的值。若留有外端口，
只对内部变量求和，所得关于外端口变量的函数称为该函网的\emph{网函}。
这个定义也固定了外端口次序：交换一个二元网函的两个外端口，会把它的矩阵转置。

\begin{proposition}[函网的结合律]
把函网中的一部分先求成网函，再用这个网函替换该部分，不改变整个函网的网函或值。
\end{proposition}
\begin{proof}
把内部变量分成该部分的内部变量和其余变量。所求的两种表达式只交换有限求和的次序，
并使用乘法对加法的分配律，因而相等。
\end{proof}

只允许使用 $\mathcal F$ 中的点函数所构成的闭合函网作为输入，求其值的问题记为
$\#\mathcal F$，也记为 $\Holant(\mathcal F)$。若允许固定二元边核 $E$，则记为
$\Holant_E(\mathcal F)$；普通模型就是 $E=I$。这种记号区分边核与点函数：
$\Holant_N(\mathcal F)$ 中每条边使用 $N$，并不表示普通模型已经实现了一个 $N$ 点。

\subsection{构件、分解与已知的复杂性出口}

若 $g$ 是一个只使用 $\mathcal F$ 中点函数的固定有限函网的网函，就称 $g$ 从
$\mathcal F$ \emph{普通可实现}。结合律允许用这个构件替换任意多个 $g$ 点，故
$\Holant(\mathcal F\cup\{g\})$ 可多项式时间归约到 $\Holant(\mathcal F)$。
更一般地，若这个加入 $g$ 的问题可多项式时间 Turing 归约到原问题，就称 $g$ 在原问题中
\emph{Turing 可用}。插值得到的函数可能只有后一种可用性。以下 $\equiv_{\mathrm T}$
表示多项式时间 Turing 等价，$\sharpP$ 困难性也在这一归约意义下使用。

对互不相交的两组变量，定义张量积
\[
 (f\otimes g)(x,y)=f(x)g(y).
\]
非零函数若能在两个非空变量组之间写成这样的乘积，就称它可分解；否则称它不可分解。
允许先置换变量。得到这样的分解后，使用下面的分解定理取得因子。

\begin{theorem}[分解定理，\cite{zh:meng2025odd,zh:guan2026lp}]
\label{zh:thm:decomposition}\label{zh:thm:decomposition-fp}
设 $\mathcal F$ 是固定有限的代数复值布尔函数集，$f,g$ 是非零代数复值布尔函数。
以下两项至少有一项成立：
\[
 \Holant(\mathcal F\cup\{f,g\})
 \equiv_{\mathrm T}
 \Holant(\mathcal F\cup\{f\otimes g\}),
\]
或者 $\Holant(\mathcal F\cup\{f\otimes g\})\in\FP$。
\end{theorem}

\begin{theorem}[非零奇元函数的二分，\cite{zh:meng2025odd,zh:guan2026lp}]
\label{zh:thm:odd-arity-quasi}\label{zh:thm:odd-arity-fp}
若固定有限的代数复值布尔函数集 $\mathcal H$ 含有一个非零奇数元函数，则
$\Holant(\mathcal H)$ 属于 $\FP$，或者为 $\sharpP$ 困难的。
\end{theorem}

这里采用奇元论文的分解与分类结果，并结合后来的 EO 多项式时间算法，把易侧提升为
$\FP$。按照第一版的框架，遇到可分解函数就使用分解定理：易算情形已经完成；
否则将它的张量因子加入函数集，继续证明。若取得非零奇元因子，再用奇元二分即可。
这也适用于固定 $N$ 边模型：先把整个问题变回相等边模型，分解及归约后再变回来。
后面的群分类把这样取得的因子一起计算在内。

\begin{lemma}[逐函数非零缩放]\label{zh:lem:per-signature-scaling}
对固定有限函数集 $\mathcal F$，给每个 $f\in\mathcal F$ 选定一个非零代数常数 $c_f$。
记 $\mathcal F'=\{c_ff:f\in\mathcal F\}$，则
$\Holant(\mathcal F')\equiv_{\mathrm T}\Holant(\mathcal F)$。
同样的结论适用于任意固定边核。
\end{lemma}
\begin{proof}
若输入中 $f$ 型点的数目为 $n_f$，逐点缩放只使答案乘上
$\prod_f c_f^{n_f}$。这一非零因子可以精确计算并补偿；反向使用 $c_f^{-1}$ 即可。
\end{proof}

以下把已知非零标量倍的辅助函数也称为可用，均由这个引理补偿标量。函数集固定，
所以每个辅助函数的构件、缩放常数或归约都先行固定，不随输入函网任意变化。

\begin{theorem}[Real Boolean Holant 二分，\cite{zh:shao2020real}]
\label{zh:thm:real-holant}
设 $\mathcal H$ 是固定有限的代数实值布尔函数集。若 $\mathcal H$ 满足
Shao--Cai 的 Real Boolean Holant 二分定理中的可判定易解条件，则
$\Holant(\mathcal H)\in\FP$；否则 $\Holant(\mathcal H)$ 为 $\sharpP$ 困难的。
本文把该文的这一整组易解条件称为条件 \textup{(T)}。
\end{theorem}

正文标准型先将各函数缩放成代数实值函数，再判断所得函数集是否满足条件 \textup{(T)}。

\subsection{矩阵记号与换基}

把二元函数按两个自变量的先后次序写成 $2\times2$ 矩阵，全文固定
\[
 I=\begin{pmatrix}1&0\\0&1\end{pmatrix},\qquad
 N=\begin{pmatrix}0&1\\1&0\end{pmatrix},\qquad
 Y=\begin{pmatrix}0&1\\-1&0\end{pmatrix},
\]
以及
\[
 X=iN=\begin{pmatrix}0&i\\i&0\end{pmatrix},\qquad
 Z=i\diag(1,-1)=\begin{pmatrix}i&0\\0&-i\end{pmatrix}.
\]
$N$ 是二元不等函数，$Y$ 是反对称矩阵。这里 $X,Z$ 的命名与第三版互换，本文所有公式
都按上式理解。串接二元点给出矩阵乘法，交换两端口给出不带共轭的转置
$A^{\mathsf T}$。复正交矩阵指 $O^{\mathsf T}O=I$；函网求和中不采用共轭内积。

\begin{proposition}[带边核的全息变换]\label{zh:prop:holographic-reduction}
设 $P$ 为可逆代数矩阵，按固定自变量次序排列每个 $m$ 元点函数的值表，并定义
$\widehat f=(P^{-1})^{\otimes m}f$。则
\[
 \Holant_E(\mathcal F)
 \equiv_{\mathrm T}
 \Holant_{P^{\mathsf T}EP}(\widehat{\mathcal F}).
\]
对应输入的值实际相等。二元点函数 $A$ 变为
$\widehat A=P^{-1}AP^{-\mathsf T}$。特别地，当 $E=I$ 且
$P=O$ 为复正交矩阵时，边核仍是 $I$，二元矩阵变为 $O^{\mathsf T}AO$。
\end{proposition}
\begin{proof}
在每个端口插入 $PP^{-1}$。把 $P^{-1}$ 与点函数先结合，把相邻的两个 $P$ 与旧边核
$E$ 先结合，分别得到所述新点函数和新边核。逐边使用结合律，整个函网的值保持不变。
反向使用 $P^{-1}$。
\end{proof}

例如，取
\[
 K=\frac1{\sqrt2}\begin{pmatrix}1&1\\i&-i\end{pmatrix},
 \qquad K^{\mathsf T}K=N.
\]
从普通模型使用 $K$ 变换后，边核为 $N$，点函数为 $(K^{-1})^{\otimes m}f$。
二元点变为 $K^{-1}AK^{-\mathsf T}$；再与一条 $N$ 边结合，得到传递矩阵 $K^{-1}AK$。

\subsection{有限二元群与九类分派}

先落实第一版中“能实现且仅能实现”的约定\cite{zh:xia2026v1}。
设 $L$ 是由原有限函数集 $\mathcal F$ 的全部函数值生成的数域。
从 $\mathcal F$ 和普通双端导线 $I$ 出发，允许有限次构件构造、非零函数的张量积分解，
以及 $L$ 中非零常数的缩放。所得函数组成集合 $\mathscr C(\mathcal F)$，称为本框架的
\emph{可用函数}。分解因子可选为 $L$ 值函数，理由见下面的证明。
插值则在使用时另行给出归约。

对其中每个非零满秩二元矩阵 $A$，同时取两个代表
\[
 \frac{A}{s},\qquad s^2=\det A,
\]
所得集合记为 $\Bgroup$。

\begin{lemma}[分解因子与有限二元群]\label{zh:lem:framework-available-closure}
对任意有限集 $\mathcal H\subseteq\mathscr C(\mathcal F)$，或者 $\Holant(\mathcal F)\in\FP$，
或者
\[
 \Holant(\mathcal F\cup\mathcal H)\equiv_{\mathrm T}\Holant(\mathcal F).
\]
排除分解定理的易算情形，以及非零奇元、秩一二元和无限射影阶二元的已有二分情形后，
$\Bgroup$ 是有限群。可以在一个固定有限的等价问题中，同时使用这个群的全部二元方向。
\end{lemma}
\begin{proof}
首先，分解不必扩大数域。若非零函数 $F(x,y)$ 在两组变量间可分解，选取 $F(a,b)\ne0$，则
\[
 F(x,y)=F(x,b)\,\frac{F(a,y)}{F(a,b)}.
\]
右边两个因子都在 $L$ 上。构件构造只用有限次加法和乘法，也保持函数值在 $L$ 上。
任意有限多个可用函数各有有限的推导，把这些推导合在一起，按次序加入所需函数。
构件和非零缩放保持问题等价；每次分解由定理~\ref{zh:thm:decomposition} 保证，
或者已得到原问题的 $\FP$ 算法，或者加入两个因子后仍与原问题等价。
有限次归约复合就证明了第一项，且允许这些辅助函数同时、任意多次使用。

剩下证明有限性，方法与第一版相同。秩一二元函数分解后给出非零一元函数；
无限射影阶二元矩阵的 Jordan 块或两个特征值之比，用链式插值得到秩一矩阵。
这些情形均由奇元二分处理。余下每个非零二元矩阵满秩且有有限射影阶。
记 $d=[L:\mathbb Q]$。若 $A$ 的射影阶为 $n$，其特征值之比是本原 $n$ 次单位根，
又属于 $L$ 的至多二次扩张，因此 $\varphi(n)\le2d$。
满足此式的 $n$ 只有有限多个，每个归一化矩阵的阶又整除 $2n$，故这些阶有共同的有限倍数。
二元串接和端口交换给出乘法与转置封闭，逆矩阵是某个正整数次幂，
所以 $\Bgroup$ 是有限指数的复矩阵群。由第一版使用的有限指数群定理
\cite[Theorem~19A.9(1)]{zh:legacy-a}，$\Bgroup$ 有限。

最后，对群的每个二元方向选取一个代表及其有限推导。方向只有有限多个，
由第一段可将全部代表同时加入一个固定有限的等价问题。
它的普通构件仍属于 $\mathscr C(\mathcal F)$，所有二元方向又都已经加入，
故普通构件的完整二元群恰为 $\Bgroup$。后面的证明若再需要有限多个可用函数，
可同样一并加入；它们仍在这个闭包中，因而不会产生群外的二元方向。
\end{proof}

以下固定这样一个等价问题，仍记它的函数集为 $\mathcal F$。
证明中用到的辅助函数都先行固定，随后可以直接作构件和短路。
这里按群分类讨论固定函数集；引理中的有限推导不要求对无限闭包逐项进行计算。

定义射影商群
\[
 \Ggroup=\Bgroup/\{I,-I\}.
\]
记 $[A]=\{A,-A\}$，乘法为 $[A][B]=[AB]$。
有限旋转群分类确定 $\Ggroup$ 的同构类；具体函数的换基仍按前述全息变换进行。

固定一个群，既给出了可以使用的全部二元方向，也限制了其他二元方向的出现。
群较大时，辅助函数较多，辅佐性较强；群较小时，二元函数的形式限制较强，
自律性较强。根据这两方面作用，将证明分成下面九类。

九类分派如下。二面体旋转群 $D_{2n}$ 的阶为 $2n$；奇、偶均指 $n$ 的奇偶性。
\begin{center}
\begin{tabular}{cl}
\toprule
序号 & 射影群的类型\\
\midrule
1 & 一阶循环群 $C_1$\\
2 & 二阶循环群 $C_2$\\
3 & 高阶循环群 $C_k$，$k\ge3$\\
4 & 奇二面体群 $D_{2n}$，奇数 $n\ge3$\\
5 & 克莱因四元群 $D_4$\\
6 & 大偶二面体群 $D_{2n}$，偶数 $n\ge4$\\
7 & 正四面体旋转群\\
8 & 正八面体旋转群\\
9 & 正二十面体旋转群\\
\bottomrule
\end{tabular}
\end{center}

在每一类群中，再按是否能得到不可分解的四元函数分成上下两集。
上集的每个非零四元函数都可分解；奇元因子已经由奇元二分处理，故这里只可能二二分解。
由分解引理，两个二元因子均可用，归一化后就在本群中。
因此，\emph{上集的四元函数恰为本群二元函数的配对积}；下集则有真四元函数。
若在固定 $N$ 边坐标中讨论，就把本群对应的二元点函数作同样的配对积。

本次正文从第 5 类进入，沿第三版证明克莱因群的二二得四上集，并补齐麻花型。
正文的构件论证还直接适用于那里写明的普通闭包条件。

\section{转置封闭与克莱因群的两种规范形}\label{zh:sec:klein}

下面三节从共同框架交付的抽象条件 $\Ggroup\cong\Kfour$ 出发。证明分成三部分。第一部分
证明整个问题经过保持相等边的正交全息变换后，只可能得到标准型或麻花型；第二部分在
标准型中证明共同相位和实数化，再调用 Real Boolean Holant 二分定理；第三部分在麻花型中
用有向圈和转置作用排除散珍珠最小反例。两型的复杂性分析都以
假设~\ref{zh:ass:klein-upper} 中的二二得四上集条件为前提。

\subsection{从抽象 \texorpdfstring{$\Kfour$}{K4} 到两种矩阵规范形}\label{zh:subsec:klein-canonicalization}

采用行列式为 $1$ 的 Pauli 记号
\[
 I=\begin{pmatrix}1&0\\0&1\end{pmatrix},\qquad
 X=\begin{pmatrix}0&i\\i&0\end{pmatrix},\qquad
 Y=\begin{pmatrix}0&1\\-1&0\end{pmatrix},\qquad
 Z=\begin{pmatrix}i&0\\0&-i\end{pmatrix},
\]
并记
\[
  \Qeight=\{\pm I,\pm X,\pm Y,\pm Z\}.
\]
它们满足
\[
 X^2=Y^2=Z^2=-I,\qquad XY=-Z,\quad YZ=-X,\quad ZX=-Y,
\]
反向相乘时改变符号。

\begin{lemma}[从射影 Klein 群提升到四元数群]\label{zh:lem:klein-lift-q8}
设 $\Bgroup\leq\mathrm{SL}_2(\mathbb C)$，$\{I,-I\}\leq\Bgroup$，并且
$\Bgroup/\{I,-I\}\cong\Kfour$。则 $\Bgroup$ 同构于 $\Qeight$。更具体地，若
$R,S\in\Bgroup$ 代表两个不同的非平凡射影类，则
\[
 R^2=S^2=-I,\qquad RS=-SR,\qquad
 \Bgroup=\{\pm I,\pm R,\pm S,\pm RS\}.
\]
\end{lemma}

\begin{proof}
非平凡类 $[R]$ 的平方是 $[I]$，故 $R^2=I$ 或 $R^2=-I$。若 $R^2=I$，则
$R\ne\pm I$，$R$ 的两个特征值只能是 $1,-1$，与 $\det R=1$ 矛盾。因此
$R^2=-I$；同理 $S^2=-I$。

$[R],[S]$ 在商群中交换，所以 $RSR^{-1}S^{-1}=\pm I$。若等号右边是 $I$，则
$R,S$ 交换。把具有两个不同特征值 $i,-i$ 的 $R$ 对角化后，$S$ 也为对角矩阵；结合
$S^2=-I$ 和 $\det S=1$，只能有 $S=\pm R$，与两个类不同矛盾。
故换位子为 $-I$，也就是 $RS=-SR$。商群有四个元素，所以列出的八个矩阵已经是
$\Bgroup$ 的全部元素。
\end{proof}

转置不是矩阵乘法的同态，但在交换群 $\Ggroup$ 上可诱导一个同态：
\[
 \vartheta([A])=[A^{\mathsf T}].
\]
事实上 $[(AB)^{\mathsf T}]=[B^{\mathsf T}A^{\mathsf T}]
=[A^{\mathsf T}B^{\mathsf T}]$，并且 $\vartheta^2$ 是恒等映射。它固定 $[I]$，在另外三个
元素上的置换因而只有两种可能：三个都固定，或者固定一个并交换另外两个。

\begin{theorem}[Klein 群的两种正交规范形]\label{zh:thm:two-klein-normal-forms}
设 $\Bgroup$ 满足引理~\ref{zh:lem:klein-lift-q8} 的条件，矩阵元均为代数数，并且关于转置封闭。
则存在代数可逆
矩阵 $O$，满足 $O^{\mathsf T}O=I$，使
\[
 O^{\mathsf T}\Bgroup O=\{\pm M:M\in\Kstd\}
 \quad\text{或}\quad
 O^{\mathsf T}\Bgroup O=\{\pm M:M\in\Ktw\},
\]
其中
\[
 \Kstd=\{I,X,Y,Z\}
\]
称为标准型，而
\[
\begin{gathered}
 \Ktw=\{I,Z,X',Y'\},\\
 X'=\frac1{\sqrt2}\begin{pmatrix}0&-1+i\\1+i&0\end{pmatrix},\qquad
 Y'=\frac1{\sqrt2}\begin{pmatrix}0&1+i\\-1+i&0\end{pmatrix}
\end{gathered}
\]
称为麻花型。两型的射影转置作用分别为
\[
 [M^{\mathsf T}]=[M]\quad(M\in\Kstd),
\]
以及
\[
 [I^{\mathsf T}]=[I],\quad[Z^{\mathsf T}]=[Z],\quad
 [(X')^{\mathsf T}]=[Y'],\quad[(Y')^{\mathsf T}]=[X'].
\]
\end{theorem}

\begin{proof}
从 $\vartheta$ 固定的一个非平凡类中取代表 $R$。于是
$R^{\mathsf T}=R$ 或 $R^{\mathsf T}=-R$。

先设 $R^{\mathsf T}=-R$。二阶反对称矩阵只有
$c\begin{pmatrix}0&1\\-1&0\end{pmatrix}$ 这种形状；由 $\det R=1$ 得 $c=\pm1$，故
$R=\pm Y$。取另一个非平凡类的代表 $S$。由 $SY=-YS$，直接比较四个矩阵元得到
\[
 S=\begin{pmatrix}a&b\\b&-a\end{pmatrix},\qquad a^2+b^2=-1.
\]
$S$ 是对称矩阵，特征值为 $i,-i$。设对应特征向量为 $u,v$，由对称性得到
$u^{\mathsf T}v=0$。若 $u^{\mathsf T}u=0$，则 $u$ 与基 $u,v$ 中的两个向量都正交，
违反该双线性型的非退化性；$v$ 同理。因此两条特征直线都不是迷向的。
把两个特征向量分别规范化后作为列，得到复正交矩阵 $O$，使
$O^{\mathsf T}SO=Z$。同时 $O^{\mathsf T}YO=\pm Y$，所以整个群变成标准型。

再设 $R^{\mathsf T}=R$。同样用规范化的特征向量得到复正交矩阵 $O$，使
$O^{\mathsf T}RO=Z$。由于
$(O^{\mathsf T}AO)^{\mathsf T}=O^{\mathsf T}A^{\mathsf T}O$，这一换基保持射影转置作用。
换到这个坐标后，把另一个生成元仍记为 $S$。由 $SZ=-ZS$ 及
$S^2=-I$，有
\[
 S=\begin{pmatrix}0&b\\c&0\end{pmatrix},\qquad bc=-1.
\]
若 $\vartheta$ 固定三个非平凡类，则 $S^{\mathsf T}=\pm S$。当取正号时
$b=c$、$b^2=-1$，当取负号时 $c=-b$、$b^2=1$；两种情况都给出标准型
$\{\pm I,\pm X,\pm Y,\pm Z\}$。

余下只可能是 $\vartheta$ 固定 $[Z]$，并交换 $[S]$ 与 $[ZS]$。因此
\[
 S^{\mathsf T}=\delta ZS,\qquad \delta\in\{1,-1\}.
\]
逐项比较给出 $c=\delta ib$ 和 $b^2=i/\delta$。改变 $S$ 的正负号后，$\delta=1$
时可取 $S=Y'$，$\delta=-1$ 时可取 $S=X'$；第三个非平凡类是二者中剩下的一个。
这正是麻花型。所有矩阵元都由原来的代数矩阵元经有限次四则运算和开平方得到，所以
$O$ 仍可取为代数矩阵。
\end{proof}

两型的射影转置作用分别为恒等置换和一个换位，而正交换基保持这一置换型，故两型不正交等价。

由于 $O^{\mathsf T}O=I$，由命题~\ref{zh:prop:holographic-reduction}，
定理~\ref{zh:thm:two-klein-normal-forms} 中的变换把整个函网转到标准型或麻花型。
二元矩阵 $A$ 成为 $O^{\mathsf T}AO$，普通相等边保持不变。

\subsection{统一的端口、pairing 与收缩约定}\label{zh:subsec:klein-contraction}

对函数 $F$ 的第 $p,q$ 个端口（$p\ne q$），沿选定圈方向接入二元函数 $M$ 并求和，记为
\begin{equation}\label{zh:eq:cycle-contraction}
  \Gamma_M^{pq}(F)
  =\sum_{a,b\in\{0,1\}}M(a,b)
   F(\ldots,x_p=a,\ldots,x_q=b,\ldots).
\end{equation}
若对侧二元因子按圈方向记作 $N(b,a)$，则局部收缩为
\begin{equation}\label{zh:eq:trace-convention}
  \sum_{a,b}M(a,b)N(b,a)=\Tr(MN).
\end{equation}
式~\eqref{zh:eq:cycle-contraction}--\eqref{zh:eq:trace-convention} 给出双线性收缩。
下面也用这一配对的对偶基读取换基坐标。

对 $\mathcal K\in\{\Kstd,\Ktw\}$，若一个非零 $2d$ 元函数在某一份 pairing 下可以写成
\[
 c\prod_{j=1}^d M_j(x_{q_j},x_{p_j}),\qquad
 c\in\mathbb C^\times,\quad M_j\in\mathcal K,
\]
就称它为一个 $\mathcal K$-基函数。每对的记录次序是 $(p_j,q_j)$；因子写成
$M_j(q_j,p_j)$，正好使式~\eqref{zh:eq:cycle-contraction} 按圈读成 $\Tr(MM_j)$。
零元非零常数视为没有二元因子的基函数。

两套矩阵基都有同一张沿圈配对表：按各自的 $I$ 和三个非平凡矩阵排列，
\begin{equation}\label{zh:eq:klein-cycle-gram}
 \bigl(\Tr(MN)\bigr)_{M,N\in\mathcal K}
 =2\diag(1,-1,-1,-1).
\end{equation}
这是因为 $I^2=I$，三个非平凡矩阵的平方均为 $-I$，两个不同非平凡矩阵的乘积迹为零。
相应地，若
\[
 c_I=\frac12,\qquad c_M=-\frac12\quad(M\ne I),
\]
则在 $p,q$ 两变量上有唯一展开
\begin{equation}\label{zh:eq:klein-pair-expansion}
 F=\sum_{M\in\mathcal K}c_M M(x_q,x_p)\otimes
 \Gamma_M^{pq}(F).
\end{equation}
具体地，函数 $M(x_q,x_p)$ 的双线性对偶函数是
$c_M M(x_p,x_q)$，因为
\[
 \sum_{a,b\in[2]}c_M M(a,b)N(b,a)
 =c_M\Tr(MN)=\mathbf1_{M=N}.
\]
对一份 pairing 逐对取上述对偶函数，就得到对应张量积基的对偶基。

\subsection{二二得四上集条件}\label{zh:subsec:klein-upper-condition}

固定有限代数函数集 $\mathcal F$，并已按定理~\ref{zh:thm:two-klein-normal-forms} 作正交变换。
记 $\mathscr R=\mathscr R(\mathcal F)$ 为其普通网函闭包：只使用原函数、相等边、端口置换、
不交并及内部收缩。以下“可用”均指这一固定闭包中的普通可实现性；由
引理~\ref{zh:lem:per-signature-scaling}，行文可忽略已知且可补偿的非零代数常数。

\begin{assumption}[二二得四上集条件]\label{zh:ass:klein-upper}
取 $\mathcal K=\Kstd$ 或 $\mathcal K=\Ktw$，并假设：
\begin{enumerate}[label=\textup{(\arabic*)}]
 \item $\mathcal K$ 的四个二元方向均在 $\mathscr R$ 中可实现，且 $\mathscr R$ 中每个非零
       二元函数都是 $\mathcal K$ 中某个矩阵的非零标量倍；
 \item $\mathscr R$ 中每个非零四元函数都能沿某一份 pairing 分解为两个二元函数的张量积。
\end{enumerate}
\end{assumption}

第一项是克莱因二元群的完整性条件；第二项是本正文所称的\emph{二二得四上集条件}，
它要求所有普通可实现的四元函数可二二分解，而不只要求输入函数具有这种性质。

\begin{lemma}[上集条件的直接后果]\label{zh:lem:klein-upper-basic}
在假设~\ref{zh:ass:klein-upper} 下，$\mathscr R$ 中没有非零奇元函数，且每个非零四元函数
都是 $\mathcal K$-基函数。
\end{lemma}

\begin{proof}
四个 $\mathcal K$ 矩阵张成全部二元函数空间。因此，对任意非零函数及任意变量对，
至少有一个 $\mathcal K$ 收缩非零，否则该变量对上的每个二元切片都为零。
若有非零奇元函数，逐次作非零收缩可得到非零一元函数 $u$；不交并
$u\otimes u$ 就是非零秩一二元函数，与假设的第一项矛盾。

设非零四元函数 $H=g\otimes h$ 是第二项给出的二二分解。由于 $\mathcal K$ 张成二元
空间，可用其中一个矩阵把 $h$ 收缩为非零常数，从而普通实现 $g$ 的非零倍数；同理可实现
$h$ 的非零倍数。由第一项，两个因子各自都是 $\mathcal K$ 矩阵的非零倍数，故 $H$ 是基函数。
\end{proof}

在这里，四元“可分解”也等价于“可二二分解”：若有一三分解 $u\otimes h$，用
$\mathcal K$ 工具在 $h$ 上作非零收缩，可得到非零秩一二元函数 $u\otimes v$，
同样违反二元条件。

下面的同时非零收缩是两型共用的降元工具。

\begin{lemma}[同时非零收缩（双活引理）]\label{zh:lem:klein-simultaneous-nonzero}
取 $\mathcal K=\Kstd$ 或 $\mathcal K=\Ktw$。若 $C,D$ 是同一组 $n\ge3$ 个变量上的
两个非零复值函数，则存在变量对 $p,q$ 和 $M\in\mathcal K$，使
\[
 \Gamma_M^{pq}(C)\ne0,\qquad
 \Gamma_M^{pq}(D)\ne0.
\]
\end{lemma}

\begin{proof}
先设 $n=3$。记
\[
 \Sigma_{pq}(C)=\{M\in\mathcal K:\Gamma_M^{pq}(C)\ne0\},\qquad
 s_{pq}(C)=|\Sigma_{pq}(C)|.
\]
四个矩阵张成全部二元函数空间，故每个 $s_{pq}(C)\ge1$。若例如
$s_{12}(C)=1$，式~\eqref{zh:eq:klein-pair-expansion} 给出
\[
 C(x_1,x_2,x_3)=N(x_2,x_1)v(x_3)
\]
乘一个非零常数的形式，其中 $N$ 可逆、$v\ne0$。直接按行列指标求和可得
\[
 \Gamma_M^{13}(C)=NMv,\qquad
 \Gamma_M^{23}(C)=N^{\mathsf T}Mv.
\]
三个矩阵都可逆，所以后两对各有四个非零分支。若没有分支数为 $1$，三个分支数各至少为
$2$。两种情况都说明
\[
 s_{12}(C)+s_{13}(C)+s_{23}(C)\ge6.
\]

反设 $C,D$ 没有共同非零收缩。每个变量对上两份支撑不交，故对应分支数之和至多为
$4$。对三个变量对求和，上界是 $12$；对 $C,D$ 分别使用刚才的下界，总和又至少为
$12$。所以两个函数各自的三个分支数之和都等于 $6$。若其中一个分支数为 $1$，前面的
计算会使另外两个分支数都等于 $4$，总和至少为 $9$；因此六个分支数只能全为 $2$。
每个变量对上的两份支撑于是都是互补的二元素集合。

对非平凡 $R\in\mathcal K$，在 $p<q$ 时定义算子
\[
 \Omega_R^{pq}(N)=RNR
\]
作用在 $p,q$ 的二元矩阵坐标上；在张量坐标中，它分别在变量 $p,q$ 上作用
$R^{\mathsf T},R$。由四元数乘法表，$\Omega_R^{pq}$ 的两个特征空间分别由
\[
 \{I,R\},\qquad \mathcal K\setminus\{I,R\}
\]
张成，特征值分别为 $-1,+1$。三种 $R$ 的正、负特征方向恰好给出四元素集合的全部六个
二元素子集。因此对每个 $pq$，存在唯一的非平凡射影方向 $R_{pq}$ 和
$\lambda_{pq}\in\{1,-1\}$，使
\[
 \Omega_{R_{pq}}^{pq}C=\lambda_{pq}C,\qquad
 \Omega_{R_{pq}}^{pq}D=-\lambda_{pq}D.
\]

比较 $12$ 与 $13$：共享变量 $1$ 上的两个矩阵若代表不同射影方向就反对易，两个整体算子
也反对易，不可能有共同非零特征向量 $C$，故 $R_{12}=R_{13}$。比较 $13$ 与 $23$ 的
共享变量 $3$，同理得到 $R_{13}=R_{23}$。记共同方向为 $R$。

标准型中 $R^{\mathsf T}=\eta_RR$，其中 $\eta_R\in\{1,-1\}$。在三个张量位置上直接相乘，
\[
 \Omega_R^{12}\Omega_R^{13}=-\eta_R\Omega_R^{23}.
\]
于是 $C$ 的三个特征值乘积是 $-\eta_R$；$D$ 在三对上都取相反特征值，乘积则是
$\eta_R$，但同一个算子恒等式又要求它等于 $-\eta_R$，矛盾。

麻花型中，若共同方向是 $X'$ 或 $Y'$，则比较 $12$ 与 $23$：共享变量 $2$ 在前一有序对
中读 $R$，在后一有序对中读 $R^{\mathsf T}$，而 $X',Y'$ 互为转置并且反对易，仍会使
两个整体算子反对易。因此只能有 $R=Z$。这时 $Z^{\mathsf T}=Z$，仍有
$\Omega_Z^{12}\Omega_Z^{13}=-\Omega_Z^{23}$，对 $C,D$ 的三个相反特征值作同样比较，
再次矛盾。三元情形得证。

若 $n>3$，任取三个共同变量。分别从 $C,D$ 的一个非零函数值中固定其余变量，得到两个
非零三元切片；两次固定值不必相同。对这两个切片使用三元结论。所得两个非零收缩值也是
原来两个收缩函数各自的一个函数值，所以同一变量对和同一 $M$ 对原来的 $C,D$ 仍同时
非零。这里的固定只用于选择非零切片，不假定钉扎函数可用。
\end{proof}

\section{克莱因群的标准型：共同相位与实数化}\label{zh:subsec:standard-klein}

本节在假设~\ref{zh:ass:klein-upper} 下取 $\mathcal K=\Kstd$，沿用第三版的共同相位与实数化路线。
引理~\ref{zh:lem:klein-upper-basic} 已保证所有非零可用函数均为偶元，且四元函数都是基函数。

记 $V=\operatorname{span}_{\mathbb R}\Kstd$。四元条件有一个反复使用的直接后果：若
$H$ 是非零四元可用函数，则存在 $c\in\mathbb C^\times$，使沿任意 pairing 用
$V$ 中的矩阵把四个端口全部闭合后，所得标量都属于 $c\mathbb R$。因为把 $H$ 的基函数
pairing 与闭合 pairing 叠起来只得到一个四圈或两个二圈；逐圈的值是若干 $\Qeight$
矩阵乘积的迹，属于 $\{0,2,-2\}$，再对 $V$ 中矩阵作实线性展开即可。

对有序 pairing $P$ 及每条边所选的工具组 $\mathbf M$，记
$\zeta_{P,\mathbf M}(F)$ 为把 $F$ 的全部变量按 $P$ 闭合所得的标量。

\begin{lemma}[标准型的共同相位]\label{zh:lem:standard-klein-common-phase}
在假设~\ref{zh:ass:klein-upper} 的标准型情形下，设 $F$ 是非零、可用的偶元函数。任取两份
有序 pairing $P,Q$ 和两组 $\Kstd$ 工具。若两个完全闭合值都非零，则
\[
 \frac{\zeta_{Q,\mathbf N}(F)}{\zeta_{P,\mathbf M}(F)}\in\mathbb R^\times.
\]
因此在任意一份 pairing 所确定的 $\Kstd$ 张量积基中，$F$ 的全部非零坐标有共同复相位。
\end{lemma}

\begin{proof}
先改变 pairing。两份 perfect matching 的对称差由交替偶圈组成；在一个长度至少为四的
交替圈上，把当前 matching 的两条边换成另外两条，可以把该圈缩短。因此可用有限次
二边换接把 $P$ 变成 $Q$。一次换接时，剪掉原闭合中的两个工具，得到一个非零四元可用
函数 $H$。在新 pairing 的两对上，$\Kstd^{\otimes2}$ 是基，所以可选两个工具使新闭合
非零；本节开头的迹计算说明新旧两个非零值的比例为实数。逐步相乘，便可从给定的
$P$-闭合走到某个非零 $Q$-闭合，比例仍为实数。

再固定同一份 pairing，比较两个非零闭合。若两组工具只在一对上不同，闭合其余各对后
得到一个非零二元可用函数，它在两个不同 $\Kstd$ 方向上的收缩都非零，与二元条件中
“只有一个基方向”矛盾。因此这种距离一的情况不会发生。

若至少有两对不同，先选两对，分别按两组闭合其余变量，得到同一组四个变量上的非零
四元函数 $H,H'$。由引理~\ref{zh:lem:klein-simultaneous-nonzero}，存在这四个变量中的一对
及同一个 $A\in\Kstd$，使两次收缩都非零。它们留下两个非零二元函数；在实向量空间
$V$ 中，各自使完全闭合为零的矩阵组成真子空间，两个真子空间的并不能覆盖 $V$，故可取
同一个 $B\in V$，使两次完全闭合都非零。这样得到两条新闭合路径：它们在选出的四个变量
上使用同一 pairing 和同一组 $A,B$，其余地方仍分别沿原来的路径。对每一边，旧值与新值
都是相应非零四元函数的完全闭合值，所以比例为实数。

这里 $B$ 只参与闭合值的代数比较，不要求它普通可实现。两条比较路径始终至多含有
这一个共同的非 $\Kstd$ 工具；下一步先剪掉它，再形成可用四元函数。

若原来恰有两对不同，两条新路径已经碰头。若还有别的不同处，就把共同的 $B\in V$ 所在
变量对作为特别对，再选一对仍不同的变量；剪掉这两对后，其余工具全部属于 $\Kstd$，所以
又得到两个可用四元函数。重复上一段，换成新的共同 $A\in\Kstd,B\in V$，每次使不同处
减少一个。有限步后两路完全相同。沿两路所得的所有非零比例都是实数，所以原来两个闭合
值的比例是实数。结合第一段便得到任意两份 pairing 的结论。

最后，式~\eqref{zh:eq:klein-cycle-gram} 表明：每个基坐标都对应唯一一组 $\Kstd$ 闭合工具，
相应完全闭合值等于该坐标乘一个非零实数 $\pm2^d$。因此所有非零坐标也有共同复相位。
\end{proof}

还要把“基坐标有共同相位”推进到“函数可以整体变成实值”。这是 v3 中保留实虚混合降元
的作用。为避免把“实部、虚部都非零”误当成“没有共同相位”，先写出降元真正保持的两个
互补子空间。

令 $\chi([I])=\chi([Y])=0$，$\chi([X])=\chi([Z])=1$。这是
$\Kfour\to\mathbb F_2$ 的群同态。对一份 pairing 下的 $\Kstd$ 张量积基向量，把各二元因子的
$\chi$ 值相加，所得值称为它的奇偶标记。因为标准型的转置不改变射影类，比较两份 pairing
时，每条交替圈上非零的矩阵迹要求两边射影乘积相同。由迹配对表，换基系数均为实数，
且不混合奇偶标记，因此下面两个实子空间的定义不依赖 pairing。
由标记为 $0$ 和 $1$ 的张量积基向量分别作实线性张成，所得两个子空间记为 $W_0,W_1$。
$W_0$ 中的基向量是实值的，$W_1$ 中的基向量是纯虚值的；给任意变量对接入 $M$ 后，
两个子空间同时平移为
$W_{\epsilon+\chi([M])}$，仍然互不相交。

\begin{theorem}[标准型的逐函数实数化]\label{zh:thm:standard-klein-realification}
在假设~\ref{zh:ass:klein-upper} 的标准型情形下，每个非零、可用的偶元函数 $F$ 都存在
$\lambda_F\in\mathbb C^\times$，使 $\lambda_FF$ 为实值函数。
\end{theorem}

\begin{proof}
由引理~\ref{zh:lem:standard-klein-common-phase}，先乘一个非零常数，使 $F$ 在某一份
$\Kstd$ 张量积基中的坐标全为实数。若其支撑只落在 $W_0$，函数已经实值；若只落在
$W_1$，再乘 $-i$ 即为实值。

反设两个子空间中的坐标都非零，把相应的两个非零分量记为 $F_0,F_1$。它们分别落在
$W_0,W_1$，并且在 $\Kstd$ 基中的坐标都是实数；等价地，$F_0$ 是实值函数，而 $F_1$
是纯虚值函数。若 $F$ 的元数至少为六，应用
引理~\ref{zh:lem:klein-simultaneous-nonzero} 于 $F_0,F_1$，得到同一变量对和同一
$M\in\Kstd$，使两个分量收缩后仍都非零。$F_0,F_1$ 只用于选择工具，不要求它们分别
可实现；实际执行的是对当前函数 $F_0+F_1$ 作同一次收缩，因此始终留在普通闭包中。
收缩一个张量积基向量时，式
~\eqref{zh:eq:klein-cycle-gram} 和四元数乘法表说明所得系数只是原系数的实数倍，故两个
新分量的 $\Kstd$ 基坐标仍为实数；奇偶标记则同时加上 $\chi([M])$，所以它们仍落在两个
互补子空间。下一步直接把这两个非零分量作为双活引理的输入。这样反复进行，最终得到一个
同时使用两个互补奇偶子空间的非零四元可用函数。

四元条件却说它在某一 pairing 下只有一个非零基坐标，因而只属于一个 $W_\epsilon$，
矛盾。二元条件同样排除了低于四元的混合函数。因此两个子空间不能同时出现，结论成立。
\end{proof}

\begin{corollary}[标准型的复杂性二分]\label{zh:cor:standard-klein-real-holant}
在假设~\ref{zh:ass:klein-upper} 的标准型情形下，存在逐函数非零缩放后的代数实值函数集
$\mathcal F_{\mathbb R}$，
使
\[
 \Holant(\mathcal F)\equiv_{\mathrm T}\Holant(\mathcal F_{\mathbb R}).
\]
因此由定理~\ref{zh:thm:real-holant}，若 $\mathcal F_{\mathbb R}$ 满足 condition \textup{(T)}，
则原问题属于 $\FP$；否则原问题为 $\sharpP$-hard。
\end{corollary}

\begin{proof}
函数集固定且有限，对每个非零 $f$ 选定定理~\ref{zh:thm:standard-klein-realification} 给出的
$\lambda_f$：可用一个非零代数基坐标的倒数消去共同相位，必要时再乘 $-i$，故缩放仍为代数数。
零函数保持不动。引理~\ref{zh:lem:per-signature-scaling} 给出所述 Turing 等价，
而缩放后的函数集满足 Real Boolean Holant 定理的代数实值输入条件。
\end{proof}

\section{克莱因群的麻花型：完整证明}\label{zh:subsec:twisted-klein}

本节在假设~\ref{zh:ass:klein-upper} 下取 $\mathcal K=\Ktw$。若普通闭包中存在非基函数，
就取其中元数最小者；它的每一种二元收缩都非零并且成为基函数。我们称这个结论为
散珍珠条件。共同 pairing 和支撑扩散把它压到六元或八元，最后圈方向却迫使麻花型的
非平凡转置作用成为恒等，得到矛盾。

先给四个射影类选二维二元标签：
\[
 \operatorname{lab}([I])=\binom00,\qquad
 \operatorname{lab}([Z])=\binom10,\qquad
 \operatorname{lab}([Y'])=\binom01,\qquad
 \operatorname{lab}([X'])=\binom11.
\]
用 $M_s\in\Ktw$ 表示标签为 $s\in\mathbb F_2^2$ 的代表矩阵。忽略矩阵乘法中的正负号后，
标签相加就是射影类相乘。转置在标签上的作用为
\begin{equation}\label{zh:eq:twisted-label-transpose}
 \operatorname{lab}([M^{\mathsf T}])=\tau\operatorname{lab}([M]),\qquad
 \tau=\begin{pmatrix}1&1\\0&1\end{pmatrix}.
\end{equation}
这里的矩阵乘法在 $\mathbb F_2$ 上进行。于是
$\tau\binom{1}{1}=\binom{0}{1}$、
$\tau\binom{0}{1}=\binom{1}{1}$，而 $\binom00,\binom10$ 不动。
其中，矩阵 $I$ 的射影类 $[I]$ 对应标签 $\binom00$。下文把标签空间上的恒等映射记为
\[
  \mathrm{Id}_{\mathrm{lab}}=\begin{pmatrix}1&0\\0&1\end{pmatrix}.
\]

\subsection{最小反例怎样得到散珍珠条件}

始终在假设~\ref{zh:ass:klein-upper} 所固定的普通闭包 $\mathscr R$ 中取函数。令
\[
 \mathscr R_{\mathrm{nb}}
 =\{F\in\mathscr R:F\ne0,\ \arity(F)\text{ 为偶数},\ F\text{ 不是 }\Ktw\text{-基函数}\}.
\]
若此集合非空，选取其中元数最小的函数 $F$。

\begin{definition}[散珍珠条件]\label{zh:def:scattered-pearl}
非零偶元函数 $F$ 满足\emph{散珍珠条件}，是指对每个变量对 $p,q$ 和每个
$M\in\Ktw$，都有
\[
 \Gamma_M^{pq}(F)\ne0,
 \qquad
 \Gamma_M^{pq}(F)\text{ 是一个 }\Ktw\text{-基函数}.
\]
各个收缩所得基函数的因子 pairing 可以不同。
\end{definition}

\begin{lemma}[最小反例满足散珍珠条件]\label{zh:lem:twisted-minimal-pearl}
在假设~\ref{zh:ass:klein-upper} 的麻花型情形下，若 $\mathscr R_{\mathrm{nb}}\ne\varnothing$，则其中
的最小元数函数 $F$ 至少为六元，并满足散珍珠条件。
\end{lemma}

\begin{proof}
非零零元函数按约定是基函数；二元条件排除二元反例，
引理~\ref{zh:lem:klein-upper-basic} 排除四元反例。因此 $F$
至少为六元。固定变量对 $p,q$，把四个收缩中非零函数的个数记为 $b_{pq}(F)$。四个矩阵
张成全部二元函数空间，而 $F\ne0$，所以 $b_{pq}(F)\ge1$。先看一个
$\Ktw$-基函数：若 $p,q$ 原来是一条二元因子边，式~\eqref{zh:eq:klein-cycle-gram} 说明恰有
一个分支非零；若 $p,q$ 分属两个因子，可把这两个因子分别记成
$A(x_r,x_p)$ 和 $B(x_s,x_q)$。收缩后的二元矩阵正是
$AMB^{\mathsf T}$；三个矩阵都可逆，所以四个分支全非零。因此基函数在任意
变量对上的分支数只能是 $1$ 或 $4$。

若 $b_{pq}(F)=1$，设唯一非零工具为 $M_0$。由式~\eqref{zh:eq:klein-pair-expansion}，
\[
 F=c_{M_0}M_0(x_q,x_p)\otimes\Gamma_{M_0}^{pq}(F).
\]
第二个因子非零、普通可实现，且元数比 $F$ 小；由最小性，它是基函数。因此 $F$ 也是
基函数，矛盾。

再设 $b_{pq}(F)=2$ 或 $3$。任选两个非零分支
\[
 C=\Gamma_{M_1}^{pq}(F),\qquad
 D=\Gamma_{M_2}^{pq}(F).
\]
由于 $F$ 至少为六元，$C,D$ 至少还有四个共同变量。由双活引理
~\ref{zh:lem:klein-simultaneous-nonzero}，可在另外一对变量 $r,s$ 上取同一个
$N\in\Ktw$，使两者收缩都非零。令
\[
 F'=\Gamma_N^{rs}(F).
\]
两对变量不交，有限求和可交换，所以
\[
 \Gamma_{M_a}^{pq}(F')
 =\Gamma_N^{rs}\bigl(\Gamma_{M_a}^{pq}(F)\bigr)\ne0
 \qquad(a=1,2).
\]
原来为零的分支在进一步收缩后仍为零，故 $F'$ 在同一变量对上仍恰有两个或三个非零分支。
另一方面，$F'$ 非零、普通可实现且元数更小，故由最小性它是基函数，
与基函数的 $1/4$ 分支判据矛盾。

所以最小反例的每个变量对都有四个非零分支。
对任意 $p,q,M$，函数
$G=\Gamma_M^{pq}(F)$ 非零、偶元且元数更小。若 $G$ 不是基函数，它本身就在
$\mathscr R_{\mathrm{nb}}$ 中，与 $F$ 的最小性矛盾。故每个 $G$ 都是
基函数，散珍珠条件成立。
\end{proof}

\subsection{共同配对与元数压缩}

\begin{lemma}[全部变量边怎样分成 pairing]\label{zh:lem:twisted-one-factorization}
设非零 $2d$ 元函数 $F$ 满足散珍珠条件，且 $2d\ge6$。则全部变量对被分成
$2d-1$ 份两两不共边的 pairing。对其中任意一份 $P$，$F$ 在 $P$ 所确定的
$\Ktw$ 张量积基中恰有四个非零坐标。
\end{lemma}

\begin{proof}
固定变量对 $p,q$，写
\[
 G_M=\Gamma_M^{pq}(F)\qquad(M\in\Ktw).
\]
四个 $G_M$ 都是非零基函数。基函数的因子 pairing 可由 $1/4$ 分支判据唯一恢复：一对
变量恰有一个收缩分支非零，当且仅当它们原来属于同一个二元因子。把 $G_M$ 的 pairing
记为 $\mu_M$。

对另一对不交变量 $r,s$，考虑零一表
\[
 D_{M,N}=\mathbf1_{\Gamma_N^{rs}
                    \bigl(\Gamma_M^{pq}(F)\bigr)\ne0}
 \qquad(M,N\in\Ktw).
\]
每行和每列的重量都只能是 $1$ 或 $4$；列也有此性质，是因为两个不交收缩可以交换。
这样的 $4\times4$ 表只有三种形状：四行各重一的置换形；全一形；或者恰有一行重四、
其余三行唯一的 $1$ 落在同一列的十字形。

对剩余变量边 $e$，令 $m(e)$ 为四份 $\mu_M$ 中含 $e$ 的份数。三种形状分别给出
$m(e)=4,0,3$。固定一个剩余变量 $r$，它在四份 pairing 中各有一个搭档，故
\[
 \sum_{e\ni r}m(e)=4.
\]
和式中若出现 $3$，其余正项至少也是 $3$，不可能凑成 $4$；所以只能出现唯一一个 $4$。
每个变量在四份 pairing 中都有相同搭档，从而
\[
 \mu_I=\mu_Z=\mu_{X'}=\mu_{Y'}=: \mu_{pq}.
\]

令 $\Pi_{pq}=\{pq\}\cup\mu_{pq}$。若 $rs\in\Pi_{pq}$，取上述置换表中的任一非零位置，
先收缩 $pq$ 再收缩 $rs$ 与交换次序所得基函数相同；用 pairing 的唯一性删去这两条边，
得到 $\Pi_{rs}=\Pi_{pq}$。故两份不同的 $\Pi$ 不共边。每条变量边 $pq$ 又都属于
$\Pi_{pq}$，所以这些 pairing 恰好分割完全图的全部边；每份有 $d$ 条边，总数为
\[
 \binom{2d}{2}/d=2d-1.
\]

最后沿 $pq$ 使用式~\eqref{zh:eq:klein-pair-expansion}：
\[
 F=\sum_{M\in\Ktw}c_M M(x_q,x_p)\otimes G_M.
\]
四个 $G_M$ 有共同 pairing $\mu_{pq}$，所以右边是在同一份
$P=\Pi_{pq}$ 张量积基中的四个不同基向量；四个系数都非零。对每份 $P$ 都同样成立。
\end{proof}

取两份不同的 pairing $P,Q$。由于它们不共边，$P\cup Q$ 分解成若干条边色交替圈。
在每条 pair 内固定记录次序，在每条圈上再任选行进方向。若第 $k$ 条圈含 $r_k$ 条
$P$ 边和 $r_k$ 条 $Q$ 边，并且共有 $c$ 条圈，则
\[
 r_1+\cdots+r_c=d.
\]
对 $P$ 边 $e$，令 $\epsilon_e^P=0$ 表示沿记录次序经过，$1$ 表示反向；定义
\[
 \sigma_k^P(\mathbf s)=
 \sum_{e\in P\cap C_k}\tau^{\epsilon_e^P}s_e.
\]
对 $Q$ 边同样定义 $\epsilon_e^Q$。读取换基坐标时，须把
$M(x_q,x_p)$ 换成双线性对偶函数 $c_M M(x_p,x_q)$，相当于反转这一对的次序。
这个反转使射影标签多受一次 $\tau$ 的作用，所以定义
\[
 \sigma_k^Q(\mathbf t)=
 \sum_{e\in Q\cap C_k}\tau^{1+\epsilon_e^Q}t_e.
\]

以 $B_P(\mathbf s),B_Q(\mathbf t)$ 分别表示两套基向量，并把它们的双线性对偶基
记为 $D_P(\mathbf s),D_Q(\mathbf t)$。例如，若 $Q$ 的有序边为 $(p_j,q_j)$，则
\[
 B_Q(\mathbf t)=\prod_{j=1}^d M_{t_j}(x_{q_j},x_{p_j}),\qquad
 D_Q(\mathbf t)=\prod_{j=1}^d c_{M_{t_j}}M_{t_j}(x_{p_j},x_{q_j}).
\]
定义换基矩阵
\begin{equation}\label{zh:eq:twisted-change-basis}
 A^{Q\leftarrow P}_{\mathbf t,\mathbf s}
 =\sum_{\alpha\in[2]^{2d}}
   D_Q(\mathbf t)(\alpha)B_P(\mathbf s)(\alpha).
\end{equation}
这只是式~\eqref{zh:eq:klein-pair-expansion} 的逐对应用。沿每条交替圈依行进方向求和，
矩阵元成为圈上矩阵词的迹，再乘各对偶因子的常数。具体地，$P$ 边上的因子写成
$M(x_q,x_p)$，所以它沿记录次序经过时读 $M^{\mathsf T}$；$Q$ 边上的对偶因子写成
$c_MM(x_p,x_q)$，所以沿记录次序经过时读 $c_MM$。第 $k$ 条圈上的实际标签和因此是
\[
 \sum_{e\in P\cap C_k}\tau^{1+\epsilon_e^P}s_e
 +\sum_{e\in Q\cap C_k}\tau^{\epsilon_e^Q}t_e
 =\tau\bigl(\sigma_k^P(\mathbf s)+\sigma_k^Q(\mathbf t)\bigr).
\]
由于 $\tau$ 可逆，这个和为零当且仅当下面两个标签和相等，故
\begin{equation}\label{zh:eq:twisted-block-criterion}
 A^{Q\leftarrow P}_{\mathbf t,\mathbf s}\ne0
 \quad\Longleftrightarrow\quad
 \sigma_k^P(\mathbf s)=\sigma_k^Q(\mathbf t)
 \quad(1\le k\le c).
\end{equation}
非零时，每条圈贡献绝对值 $2$，每个对偶因子另贡献绝对值 $1/2$，
所以矩阵元的绝对值为 $2^{c-d}$。
交换 $P,Q$ 得到逆换基矩阵 $A^{P\leftarrow Q}$；它的非零矩阵元也有同一绝对值。
逆向求和时，圈条件的两边分别多受一次 $\tau$ 的作用，同时作用 $\tau$ 就恢复
式~\eqref{zh:eq:twisted-block-criterion}。所以正向与逆向使用的是同一组块，
只把每个块的行、列指标对调。

\begin{lemma}[交替圈块的支撑扩散]\label{zh:lem:twisted-support-spreading}
按式~\eqref{zh:eq:twisted-block-criterion} 的 $c$ 个圈标签和给行列分块，每个非零块都是
$N\times N$，其中
\[
 N=4^{d-c}.
\]
若 $F$ 的 $P$ 基支撑中有 $h$ 项落在同一块，则这 $h$ 项在 $Q$ 基中产生至少 $N/h$
个非零坐标。因此，满足散珍珠条件的函数必有 $d\le4$。
\end{lemma}

\begin{proof}
固定一条含 $r_k$ 条 $P$ 边的圈及其标签和。任选前 $r_k-1$ 条边的标签，最后一条由和式
唯一确定，所以该圈有 $4^{r_k-1}$ 个列标；行标同理。逐圈相乘便得到
\[
 N=\prod_{k=1}^c4^{r_k-1}=4^{d-c}.
\]
对偶基恒等式说明两次换基互为逆运算；标签和不同的块之间没有非零矩阵元，
所以每个块也各自可逆。这个块及其逆矩阵的非零项都有绝对值
$a=2^{c-d}=N^{-1/2}$。

设该块中的输入系数向量为 $v\ne0$，恰有 $h$ 个非零坐标；输出为
$w=A^{Q\leftarrow P}v$，恰有 $s$ 个非零坐标。记
\[
 U=\max_j|v_j|>0,\qquad V=\max_i|w_i|>0.
\]
每个输出坐标至多是 $h$ 个绝对值不超过 $aU$ 的项之和，故 $V\le haU$。
再从 $w$ 使用逆换基矩阵恢复 $v$，同理有 $U\le saV$。因此
\[
 U\le saV\le sha^2U.
\]
约去 $U>0$，得到 $hs\ge a^{-2}=N$，即 $s\ge N/h$。

由引理~\ref{zh:lem:twisted-one-factorization}，$F$ 在 $P,Q$ 两套基中都恰有四个非零坐标。
各块不能互相抵消，且每块 $h\le4$，所以
$4\ge N/4=4^{d-c-1}$，即 $d-c\le2$。两份 pairing 不共边，每条交替圈至少含两条
$P$ 边，故 $c\le\lfloor d/2\rfloor$。合并两个不等式得到 $d\le4$。
\end{proof}

因为这里还假定 $2d\ge6$，候选元数只剩六元、八元。两种情况下等号都给
$N=16$，而两套基中的支撑大小均为四。若 $P$ 基支撑中的四个基向量分在大小分别为
$h_1,\ldots,h_t$ 的多个块，则
\[
 \sum_{j=1}^t\frac{16}{h_j}>4
 \qquad(t\ge2,\ h_1+\cdots+h_t=4),
\]
这里每个 $h_j\le4$，所以左边在 $t\ge2$ 时至少为 $8$。这与 $Q$ 基支撑只有四项矛盾。
因此 $P$ 基支撑中的四个基向量在每条交替圈上有同一个有向标签和。

\subsection{六元、八元的方向矛盾}

固定一份 pairing $P$，用它的一条边上的四个标签给四个非零基坐标编号。对第 $j$ 条
$P$ 边，散珍珠条件的四个分支均非零，所以四项在该位置恰好各用一个标签，给出四点集合
的一个置换 $s_j$。任何这样的置换都唯一写成
\[
 s_j(x)=L_jx+b_j,
\]
其中 $L_j$ 是 $\mathbb F_2$ 上的可逆 $2\times2$ 矩阵：取
$b_j=s_j(0)$，再把 $s_j\binom10+b_j$、$s_j\binom01+b_j$ 作为两列即可；它们是两个
不同的非零向量，因而线性无关。第四个点的像也自动正确，因为 $\mathbb F_2^2$ 四个点
的和为零，任意置换后的四个点之和仍为零。可令第一条边满足
$L_0=\mathrm{Id}_{\mathrm{lab}},b_0=0$。

若某条交替圈包含的 $P$ 边位置集合为 $J$，并令 $\epsilon_j$ 记录沿圈经过第 $j$ 条
$P$ 边时是否反向，则四项的有向标签和相同意味着其线性部分为零：
\begin{equation}\label{zh:eq:twisted-directed-label}
 \sum_{j\in J}\tau^{\epsilon_j}L_j=0.
\end{equation}

\begin{theorem}[麻花型没有散珍珠函数]\label{zh:thm:twisted-no-pearl}
不存在元数至少为六、满足散珍珠条件的非零函数。
\end{theorem}

\begin{proof}
引理~\ref{zh:lem:twisted-support-spreading} 已把元数压到六或八。

先看六元。变量重命名后取
\[
 P=(12)(34)(56),\qquad Q=(23)(45)(61).
\]
引理~\ref{zh:lem:twisted-one-factorization} 给出的 pairing 中，含边 $13$ 的那一份不能使用
$P,Q$ 的边，所以只能是
\[
 Q^*=(13)(25)(46).
\]
固定 $P$ 边方向为 $1\to2,3\to4,5\to6$。在 $P\cup Q$ 的六圈上沿
\[
 1\to2\to3\to4\to5\to6\to1
\]
行进，三条 $P$ 边都顺向，式~\eqref{zh:eq:twisted-directed-label} 给出
\[
 \mathrm{Id}_{\mathrm{lab}}+L_1+L_2=0.
\]
在 $P\cup Q^*$ 的六圈上沿
\[
 1\to2\to5\to6\to4\to3\to1
\]
行进，$12,56$ 顺向而 $34$ 反向，所以
\[
 \mathrm{Id}_{\mathrm{lab}}+\tau L_1+L_2=0.
\]
两式在 $\mathbb F_2$ 上相加，得到 $(\tau+\mathrm{Id}_{\mathrm{lab}})L_1=0$。由于
$L_1$ 可逆，只能有 $\tau=\mathrm{Id}_{\mathrm{lab}}$，与
式~\eqref{zh:eq:twisted-label-transpose} 矛盾。

再看八元。固定
\[
 P=(12)(34)(56)(78).
\]
除 $P$ 外，引理~\ref{zh:lem:twisted-one-factorization} 还给出六份 pairing。每份与 $P$
叠成两条四圈，因而把四条 $P$ 边分成两对。
共有三种分组：
\[
 \{12,34\}\mid\{56,78\},\quad
 \{12,56\}\mid\{34,78\},\quad
 \{12,78\}\mid\{34,56\}.
\]
固定第一种分组。连接 $12,34$ 只有 $(13)(24)$ 与 $(14)(23)$ 两种；连接 $56,78$ 也只有
$(57)(68)$ 与 $(58)(67)$ 两种。两份 pairing 若在任何一处作相同选择就会共边，所以这一类
至多有两份；六份 pairing 分入三类，故每类恰有两份，而且两处选择都互补。

于是第一类中有一份在 $1,2,3,4$ 上用 $(14)(23)$，另一份用 $(13)(24)$。固定
$12$ 方向为 $1\to2$、$34$ 方向为 $3\to4$。前一份给出四圈
$1\to2\to3\to4\to1$，两条 $P$ 边都顺向，故
\[
 \mathrm{Id}_{\mathrm{lab}}+L_1=0.
\]
后一份给出四圈 $1\to2\to4\to3\to1$，$12$ 顺向而 $34$ 反向，故
\[
 \mathrm{Id}_{\mathrm{lab}}+\tau L_1=0.
\]
相加并右乘 $L_1^{-1}$，又得到
$\tau=\mathrm{Id}_{\mathrm{lab}}$，仍与麻花型的转置作用矛盾。
\end{proof}

\begin{corollary}[麻花型上集的函数结构]\label{zh:cor:twisted-internal-exit}
在假设~\ref{zh:ass:klein-upper} 的麻花型情形下，$\mathscr R$ 中每个非零函数都是
$\Ktw$-基函数。
\end{corollary}

\begin{proof}
若 $\mathscr R_{\mathrm{nb}}$ 非空，引理~\ref{zh:lem:twisted-minimal-pearl} 给出满足散珍珠条件的
最小元数函数，而定理~\ref{zh:thm:twisted-no-pearl} 说明它不存在。因此所有非零偶元函数都是
基函数；引理~\ref{zh:lem:klein-upper-basic} 已排除非零奇元函数。
\end{proof}

\begin{lemma}[基函数的逐圈算法]\label{zh:lem:klein-cycle-algorithm}
设 $\mathcal H$ 是固定有限函数集，每个非零函数都已给定一个非零常数倍的
$\Kstd$-基函数或 $\Ktw$-基函数表示。则 $\Holant(\mathcal H)\in\FP$。
\end{lemma}

\begin{proof}
若输入函网出现零函数点，其值直接为零。以下设所有函数点都非零。函数集固定，所以每个函数
的见证 pairing、每对内的矩阵及整体常数都可预处理。把输入函网
中的每个正元数函数点按该 pairing 拆成二元因子点，零元函数直接计入整体常数。
每个二元因子点有两个端口，原函网的内部边
仍连接这些端口，因此拆开后的每个连通分量都是一个圈。自环给出长度为一的圈，平行边可给
出长度为二的圈，都不例外。

给每条圈任选方向。经过一个二元因子时，若方向与其两个端口的记录次序一致就读矩阵 $M$，
反向就读 $M^{\mathsf T}$；普通内部边读作 $I$。依次对圈上的布尔变量求和，正是这些
$2\times2$ 矩阵乘积的迹。把所有圈的迹与各函数点的整体常数相乘，就得到原函网的值。
矩阵个数与输入大小成正比，所以算法是多项式时间。
\end{proof}

\subsection{上集结论}\label{zh:subsec:klein-exits}

\begin{theorem}[克莱因群上集的复杂性二分]\label{zh:thm:klein-upper-dichotomy}
设固定有限代数函数集 $\mathcal F$ 的非零二元普通可实现函数对应的射影群为克莱因群，
经定理~\ref{zh:thm:two-klein-normal-forms} 的正交变换后满足假设~\ref{zh:ass:klein-upper}。则：
\begin{enumerate}[label=\textup{(\arabic*)}]
 \item 标准型中，各函数可分别乘非零代数常数成为实值函数；所得实值函数集满足
       定理~\ref{zh:thm:real-holant} 的 condition \textup{(T)} 时，原问题属于 $\FP$，否则为
       $\sharpP$-hard；
 \item 麻花型中，每个非零普通可实现函数都是 $\Ktw$-基函数，原问题属于 $\FP$。
\end{enumerate}
\end{theorem}

\begin{proof}
两种规范形由定理~\ref{zh:thm:two-klein-normal-forms} 穷尽。标准型结论由
推论~\ref{zh:cor:standard-klein-real-holant} 得到；麻花型结论由
推论~\ref{zh:cor:twisted-internal-exit} 及引理~\ref{zh:lem:klein-cycle-algorithm} 得到。
\end{proof}

正文至此完成二二得四上集的克莱因情形；出现不可二二分解四元函数的下集暂不处理。

\paragraph*{证明工作说明}

粗略地划分，可认为前三版都由作者创作完成，第四版中麻花型情形的证明由 AI 创作完成，作者已经检查过完整证明。

细致点说，麻花型克莱因的情形由 AI 在作者的对话指引下创作完成。作者观察到，在标准型中，虚实的划分也是克莱因群 $\{I,Y\}$ 与 $\{X,Z\}$ 的划分，猜测其它两组划分——$\{I,X\}$ 与 $\{Y,Z\}$，以及 $\{I,Z\}$ 与 $\{X,Y\}$——应该也有相同的效果，并想知道麻花型的其它两组划分是否也有相同的效果。AI 证明了标准型中的结论，并且直接报告证明了结论更强的双活引理；双活引理在麻花型里也成立。这是作者在对话过程中的主要指引贡献，感觉作用已很轻微；前后提供材料、给出指令的对话也很长，貌似作用就更轻微了。

细致点说，截至目前为止，除了麻花型克莱因情形，只有标准型的四面体降元引理，是在作者忙于批改期末考试卷时，由付费 AI（参见第三版中的致谢）拿到命题之后自行证明完整，后经作者验证与重写。复盘时感觉，给出命题也很重要，而且该证明在作者的证明能力之内。其余部分只有免费 AI 参与检索、调整证明用词、简单矩阵计算、汉译英等较低创作性的工作，且效果一般。例如，AI 说不清用哪一个 Burnside 引理，最终还是靠搜索引擎搞定。

\clearpage
\appendix
\section{通往最后的路程}\label{zh:sec:appendix-road}
最近，含二元不等函数的二分定理\cite{zh:liu2026disequality}于 2026 年 8 月 31 日提交至
arXiv。我9月8日看到它，同时聊天听到的消息，
一个完全版的二分定理也已提交，正在等待网上公开显示，考虑几小时后，我打算快速整理一个冲刺版，
补齐含二元不等函数的二分定理到最终二分定理其余情形，但不包括已有完整独立稿件C2的情形，
需要完善的情形，除了奇二面体群之外，都在前三版中接近证明完毕，开通更强的付费AI帮我完善一些证明以及整理文字。

含二元不等函数的二分定理\cite{zh:liu2026disequality}这篇论文的证明很长，涵盖了许多情形，包括本文讨论的整个克莱因群情形，也包括其余的偶二面体群和三类正多面体旋转群。顺便一说，作者曾经试图统一所有含有克莱因群的大群的证明，想建立克莱因群四个矩阵的公共结论部分，然后再证明群里多出来的元素提供的能力都类似效果，分别达到各自的复杂性二分。这件没做成的事情，被这个二分定理做成了（只需辅佐分析一下这些大群里能有二元不等）。
以这个大定理为基础，原有九类群框架的最后阶段便有了一条较短的路线：作一个简短的群论
分析，找出哪些群能够在复正交全息变换后取得二元不等函数，再对这些情形调用该定理。
因此，从完成复杂性分类的角度看，正文可以保留克莱因上集的直接证明，下集则由附录中的
这条路线处理；本版不再展开下集的直接证明。

作完这一步分析，尚不能取得二元不等函数的情形只落在循环群与奇二面体群中，对应九类
中的 $C_1,C_2,C_k\ (k\ge3)$ 和 $D_{2n}\ (n\ge3\text{ 为奇数})$。
其中，奇数阶循环群全部留下；偶数阶循环群和奇二面体群只留下尚不能取得二元不等函数的
矩阵表示。恰好在今年春节后开学时，我的两名低年级研究生分别选择了 $C_2$ 和奇二面体群
作为探索与研究的方向。$C_2$ 的二分已经由 Yuxuan Zhou 在独立工作中完成
\cite{zh:zhou-draft}，我不是该文作者，本附录只引用其结果。
高阶循环群 $C_k$ 在第一稿 v1\cite{zh:xia2026v1}中已经证明；$C_1$ 的主体证明在
v2\cite{zh:xia2026v2}中给出，但当时还留下
一些有限情形没有验证完，一部分下集情形也未处理完整。本轮需要补充的内容中，大部分
恰好已经有我与学生此前完成的几乎纯人工（包括免费AI提供略强于搜索引擎的能力）工作作为基础，但经历了目前付费AI的完善与重写，
从正文到附录有些过程被完善不那么原汁原味了。

既然一个完全版的二分定理也已提交，正在等待网上公开显示，
借用一个我自己的比喻：同行的脚尖已经碰到了终点线，等论文公开显示，脚后跟也就跨过去了。
我想试一试，能否利用这中间的一小段时间，整理出一个 AI 辅佐的版本，也让自己的脚尖碰一碰
终点线，过一回跟着跨过终点线的瘾，为自己二十多年的追求作一个阶段性的了结。
这一版的冲刺以文献~\cite{zh:liu2026disequality}这个大定理为基石；本附录的任务，是把九类群
框架中可以直接借助它的部分摘出来，再补齐余下的证明。

本附录依次安排三个部分：
\begin{enumerate}[label=\textup{(\arabic*)}]
 \item 一个简短的分析，判定九类群中哪些情形能够取得二元不等函数并调用新定理；
 \item 循环群的证明，重点是在 AI 辅助下完善 $C_1$ 的遗留情形，整理既有的 $C_k$ 证明，
       并引用 $C_2$ 的独立结果；
 \item 奇二面体群的二分定理。本附录列出的这份证明，由 AI 在给定框架与研究问题后
       独立完成。
\end{enumerate}

\paragraph*{为什么还要保留原来的路线}
借助新定理完成这次冲刺以后，我仍希望沿原有框架，把不调用这一新定理的直接证明逐类
完成。这是通往最后的第二条路线，本附录暂不展开。

我猜想，文献~\cite{zh:liu2026disequality}的证明组织可能与此前的 $\#\mathrm{EO}$ 二分有相似之处，
即在研究过程中不断取得并加入新的辅助函数。
研究能够得到哪些辅助函数，以及得不到它们时原函数必须满足哪些性质，分别体现了我所说的
“辅佐性”和“自律性”。不断添加辅助函数，是 $\#\mathrm{EO}$ 框架性方法中我很喜欢的一部分；
其中我最喜欢的战略性方法，是实数化。这个框架的基本想法并不复杂，但把全部情形证明完整，
需要处理很多阶段，我更喜欢按阶段截断，每个阶段是个群的新情况方式，按群分情况，每个情况的条件更清晰。

用一个比喻来说，这样的证明仿佛要讲清楚一支队伍怎样从光杆司令起步，逐步拥有特种部队、
陆战队、海军，最后形成三军。每个阶段都有自己的人员和能力，也有相应的纪律与约束；
在证明中，就要分别说明每一阶段的辅佐性与自律性。我个人更喜欢直接采用九类群的分类，
一开始便依据完整的二元群，定位这一类已经具备的辅助能力，以及不能产生其他二元函数所
带来的约束，再在这一类内部展开分析。

陆战队有自己的人员配置和纪律，海军也有自己的人员配置和纪律。分别研究两者时，可以把
各自的辅佐性和自律性直接摆出来，而不必把“陆战队怎样发展成海军”作为证明主线，再区分
哪些情形能够发展、哪些不能发展：能够发展的情形增加了哪些人员，不能发展的情形又因
原有纪律而得到什么结论。我偏爱九类群框架，正是因为它让我更容易看清每一类的这些特点。

如第一版所讨论的，可用二元辅助函数的数量，以及它们在所讨论的矩阵空间或子空间中的
分布，会影响降元的力度。当辅助方向的数量超过相关空间的维数时，辅佐性通常已经足够强，
可以期待较有力的降元结论；这些群的证明方法也往往具有共同的抽象特点。
把这种共同特点解释给 AI，可以据此推进多面体群和较高阶二面体群的直接证明。
因此，借助文献~\cite{zh:liu2026disequality}冲刺以后，我仍愿意按原框架继续做一遍：
辅佐性较强的群有望采用相近的方法，只是逐项完成证明的工程量较大；辅助方向的数量少于
或恰好等于相关空间维数的群，类型较少，证明往往更有各自的特色，也更有“咬劲”。

下面只展开第一条路线。为便于从降元分析走到复杂性二分，先重列所需的模型与记号。
其中，乘积型类 $\mathcal P$ 用于一阶群的 EO 分类和奇二面体群最后的易算判据；
二元配对积类 $\mathcal P_{\mathcal B}$ 则用于描述降元过程中已经分解成二元因子的函数。

\paragraph*{本附录的问题、记号与归约}
一个 $m$ 元布尔函数是映射 $f:\{0,1\}^m\to\overline{\mathbb Q}$。
函数集固定且有限；函数值以代数数的精确表示给出。
本文把通常所说的张量网络称为函网。
一个函网由允许自环和重边的底图、各顶点上的函数以及各顶点边端的次序组成。
外部半边也有次序。把内部边变量求和而保留外部变量，得到网函
\[
 f_G(\mathbf x)=\sum_{\mathbf z\in\{0,1\}^{E(G)}}
       \prod_{v\in V(G)}f_v(\mathbf z,\mathbf x).
\]
每个内部变量在两个边端出现；自环在同一个顶点出现两次。
没有外部半边时，网函是一个常数，计算它的问题记为 $\Holant(\mathcal F)$。
空乘积为 $1$。零函数和零元函数按已知标量处理。
有限求和的结合律允许先结合任何子函网，因此一个固定构件可以替换其网函。

构件的网函称为普通可实现函数；分解得到的张量因子，按
引理~\ref{zh:lem:framework-available-closure} 一并作为可用函数。
函数切片用于读取指定输入上的函数值，插值则在使用处给出归约。
已知非零标量可以补偿：若某个构件实现 $cf$，在一个实例的 $r$ 处替换它后，
答案乘以 $c^r$，最后除去这个非零因子。

二元函数按其两个变量的次序写成矩阵。串接和交换两个外端口分别给出矩阵乘法和转置。
全文固定
\[
 I=\begin{pmatrix}1&0\\0&1\end{pmatrix},\qquad
 N=\begin{pmatrix}0&1\\1&0\end{pmatrix},\qquad
 Y=\begin{pmatrix}0&1\\-1&0\end{pmatrix}.
\]
$N$ 是二元不等函数；$Y$ 是反对称矩阵。
若使用四元群的传统记号，则另外两轴为
$X=iN$ 和 $Z=i\diag(1,-1)$。

\paragraph*{换基后的点函数与边}
对可逆矩阵 $P$，令 $\widehat f=(P^{-1})^{\otimes m}f$。
把每条旧边的两个相邻矩阵先结合，新边核为 $Q=P^{\mathsf T}P$。
因此整个问题精确变成边核为 $Q$、点函数为 $\widehat{\mathcal F}$ 的函网。
对二元函数 $A$，实际变换是
\begin{equation}\label{zh:eq:binary-congruence}
 \widehat A=P^{-1}AP^{-\mathsf T},\qquad
 \widehat A Q=P^{-1}AP.
\end{equation}
把新二元点与一条新边结合，就得到后一矩阵所表示的端口传递作用。
只有 $P^{\mathsf T}P=I$ 时，旧边核保持不变。
这里及以后 $\mathsf T$ 都是不带共轭的转置。

特别地，取
\begin{equation}\label{zh:eq:K-basis}
 K=\frac1{\sqrt2}\begin{pmatrix}1&1\\ i&-i\end{pmatrix},
 \qquad K^{\mathsf T}K=N.
\end{equation}
新问题有固定的 $N$ 边，二元点函数为 $K^{-1}AK^{-\mathsf T}$。
在这个模型中，以一个二元点 $A$ 短路两个端口时，有效收缩矩阵为 $NAN$，
而在一个端口上接入 $A$ 的传递矩阵为 $AN$。
下文也把这个原生 $N$ 边模型记为 $\Holant(N\mid\Gamma)$，
与正文的 $\Holant_N(\Gamma)$ 同义；把每条 $N$ 边看成一个二元顶点，就得到对应的二部函网。

\paragraph*{两个不同的乘积类}
函数类 $\mathcal P$ 是计数约束满足问题的乘积型类：一个函数可以写成
一元函数、相等约束和不等约束的乘积，变量允许在这些因子中重复。
它包括任意元数的广义相等函数。
另一方面，$\mathcal P_{\mathcal B}$ 是零函数以及在某份完美匹配上
把本群的满秩二元函数作张量积所得的函数，允许非零整体标量。
这些二元方向已按总框架加入当前函数集。
后者的每个变量只出现在一个二元因子中。
当二元矩阵均为满秩对角或反对角矩阵时，非零的 $2m$ 元
$\mathcal P_{\mathcal B}$ 函数恰有 $2^m$ 个支撑输入，且
$\mathcal P_{\mathcal B}\subseteq\mathcal P$。
四元相等函数 $=_4=\ket{0000}+\ket{1111}$ 属于 $\mathcal P$，
但不属于 $\mathcal P_{\mathcal B}$。

后文还使用三个已有易算类。仿射类 $\mathcal A$ 中的函数具有形式
$c\,\mathbf1_{L}(x)i^{Q(x)}$，其中 $L$ 是 $\mathbb F_2$ 上的仿射子空间，
$Q$ 是模四二次多项式，其交叉项系数为偶数。
取 $\alpha=e^{\pi i/4}$，$T=\diag(1,\alpha)$，定义
$\mathcal A_\alpha=\{T^{\otimes m}f:f\in\mathcal A,\ \arity(f)=m\}$。
局部仿射类 $\mathcal L$ 要求：对每个 $\sigma\in\supp(f)$，
函数 $\alpha^{\sum_j\sigma_jx_j}f(x)$ 属于 $\mathcal A$。
这些类均包括零函数。一个 $2d$ 元函数称为 EO 函数，是指其支撑只含重量为 $d$ 的输入。
一个 EO 函数具有 ARS 性质，是指
$f(x)=\overline{f(\bar x)}$，其中 $\bar x$ 是逐比特翻转。
这个等式只描述函数值的实结构；函网求和仍不含共轭。

\subsection{哪些群有二元不等函数}\label{zh:subsec:road-N}
\label{zh:sec:binary-group}
沿基础部分的群框架，$\mathcal B$ 包含全部可用二元函数的行列式一代表，
并对乘法与转置封闭；其有限性已由引理~\ref{zh:lem:framework-available-closure} 证明。
射影商群为
\[
 \mathcal G=\mathcal B/\{I,-I\}.
\]
本节只需分析：哪些群中能找到一个对称、迹为零的二元矩阵。
将它正交对角化，再作一次简单的正交换基，就得到二元不等函数 $N$。

\begin{lemma}\label{zh:lem:N-symmetric}
存在一个复正交全息变换，使当前问题普通实现 $N$ 的非零倍数，
当且仅当 $\mathcal B$ 中存在非零、对称、迹为零的矩阵。
\end{lemma}
\begin{proof}
正交合同也是相似，所以保持对称性和迹。
反过来，设 $A\in\mathcal B$ 对称且迹为零。
行列式为一给出特征值 $i,-i$。对应特征向量 $u,v$ 满足
$u^{\mathsf T}v=0$。两个特征方向都非迷向：若例如 $u^{\mathsf T}u=0$，
二维非退化双线性形式下 $u^\perp=\mathbb C u$，于是 $v$ 与 $u$ 共线，矛盾。
分别归一化后得到正交特征基，将 $A$ 化成 $i\diag(1,-1)$。
再使用 $2^{-1/2}\left(\begin{smallmatrix}1&1\\1&-1\end{smallmatrix}\right)$，
便得到 $iN$。所有标量均可补偿。
\end{proof}

\begin{lemma}\label{zh:lem:involution-route}
设有限群 $\mathcal G$ 的阶为偶数。若 $Y\notin\mathcal B$，则可以正交取得 $N$。
若 $Y\in\mathcal B$，则可以正交取得 $N$，当且仅当 $[Y]$ 属于
$\mathcal G$ 的某个 Klein 四元子群。
\end{lemma}
\begin{proof}
把群中既非单位元、阶也非二的元素按 $g,g^{-1}$ 配对。
由于群阶为偶数，非单位二阶元素的个数为奇数。
转置在这个奇数集合上作对合，所以固定某个陪集 $[A]$。
二阶非单位陪集的行列式一提升满足 $A^2=-I$、$\operatorname{tr}A=0$。
又因陪集固定，$A^{\mathsf T}=\pm A$。
反对称且行列式为一的二维矩阵只有 $\pm Y$。
所以 $Y\notin\mathcal B$ 时，必有对称迹零矩阵。

现在设 $Y\in\mathcal B$。任意对称迹零矩阵
$A=\left(\begin{smallmatrix}a&b\\b&-a\end{smallmatrix}\right)$ 都满足 $AY=-YA$。
因此 $[A],[Y]$ 为两个不同的可交换二阶元素，生成 Klein 四元群。
反过来，若 $[A]\ne[I],[Y]$ 是某个包含 $[Y]$ 的 Klein 子群中的元素，
则 $AY=\pm YA$。若二者交换，则 $A=uI+vY$；迹为零迫使 $u=0$，
从而 $[A]=[Y]$，矛盾。于是 $AY=-YA$，直接比较四个矩阵元即得
$A^{\mathsf T}=A$、$\operatorname{tr}A=0$。应用引理~\ref{zh:lem:N-symmetric}。
\end{proof}

在有限三维旋转群的分类中，上述判据有如下直接读法。
\begin{center}
\begin{tabular}{ll}
\toprule
射影群 & 能否正交取得 $N$\\
\midrule
奇数阶循环群（包括平凡群） & 不能\\
偶数阶循环群 & 当且仅当 $Y\notin\mathcal B$\\
奇二面体群 & 当且仅当 $Y\notin\mathcal B$\\
偶二面体群及 Klein 四元群 & 能\\
正四面体、正八面体、正二十面体旋转群 & 能\\
\bottomrule
\end{tabular}
\end{center}
这里的二面体群 $D_{2n}$ 的阶为 $2n$，奇偶指 $n$。
这张表不需要挑选一个任意的 Klein 子群并假设它被转置保持。
奇循环群没有二阶元素。循环群没有 Klein 子群；奇二面体群中两个不同二阶元素
的乘积有非平凡奇数阶，也没有 Klein 子群。
偶二面体群的中央半周旋转与任意另一半周旋转生成 Klein 子群，故每个二阶元素
都在这种子群中。
三种正多面体旋转群的每个半周旋转同样属于一组三条两两垂直轴上的半周旋转。
可以直接选坐标核验：正四面体取顶点 $(1,1,1),(1,-1,-1),(-1,1,-1),(-1,-1,1)$；
正八面体取三个坐标单位向量及其负向量，共六个顶点；正二十面体取
$(0,\pm1,\pm\varphi),(\pm1,\pm\varphi,0),(\pm\varphi,0,\pm1)$，
其中 $\varphi=(1+\sqrt5)/2$。
坐标轴上的三个半周旋转都保持相应顶点集。
正四面体的三个二阶旋转就是它们；正二十面体的半周轴因棱的传递性互相共轭。
正八面体还有经过相对棱中点的半周轴，代表轴 $(1,1,0)$ 与
$(1,-1,0)$、$(0,0,1)$ 组成所需的垂直三轴。
故每个二阶元素都位于一个 Klein 子群，再应用引理~\ref{zh:lem:involution-route}。

\begin{theorem}[含二元不等函数的已有二分\cite{zh:liu2026disequality}]
对任意有限代数复值布尔函数集 $\mathcal H$，
$\Holant(\mathcal H\cup\{N\})$ 有 $\FP/\sharpP$-hard 二分，
且易算判据可由精确代数输入判定。
\end{theorem}
因此一个实际可实现的对称迹零二元矩阵就足以结束相应分支。
该定理中的 $N$ 是可用点函数，不是一般换基后记在边上的核。

\subsection{循环群的完整证明}\label{zh:subsec:road-cyclic}
本节分为 $C_1$ 与 $C_k\ (k\ge3)$ 两个子节，每个子节依次讨论上集与下集；
$C_2$ 的独立结果先在这里说明。
沿前面的分解约定，上集的非零四元函数都是本群二元函数的配对积，下集则有真四元函数。
后面的降元证明引用第~\ref{zh:subsec:road-odd}~部分中配对积的共同结构引理；
这些引理在所引位置给出完整证明。

\paragraph*{二阶循环群的引用}\phantomsection\label{zh:sec:c2-citation}
二阶循环群的完整复杂性二分由周宇轩在独立工作中完成\cite{zh:zhou-draft}，
本文作者不是该工作的作者，此处只引用其结果。
我曾在中国科学院大学雁栖湖校区的讨论中，听取该稿的大部分内容，并跟随作者
周宇轩的讲解，检验了当时听到的部分。
对能取得二元不等函数的表示，第~\ref{zh:subsec:road-N}~部分已经给出出口；
剩余的反对称表示在普通相等边坐标下为
\[
 \mathcal B=\{\pm I,\pm Y\},\qquad
 Y=\begin{pmatrix}0&1\\-1&0\end{pmatrix},\qquad
 \mathcal G\cong C_2.
\]
本附录保留这一独立结果的引用，不复述其证明。

\addtocontents{toc}{\protect\setcounter{tocdepth}{4}}
\subsubsection{\texorpdfstring{$C_1$ 情形}{C1 情形}}\label{zh:subsec:cyclic-c1}

\paragraph{上集}\label{zh:subsec:c1-upper}\label{zh:sec:c1-strict}
在原生 $N$ 边坐标中，$C_1$ 的二元函数只有 $N$ 的非零倍数。
上集四元函数经分解引理取得两个二元因子，因此都是 $\lambda N\otimes N$，允许端口置换。
下面先把问题化为单函数，再整理 EO 部分的共同相位，最后用两份函数的连接迫使整体实数化。

\subparagraph*{四元配对积与证明目标}

取式~\eqref{zh:eq:K-basis} 的坐标，边核固定为 $N$，
当前有限函数集记为 $\Gamma$。假设每个非零函数均为偶元，并且
\begin{enumerate}[label=\textup{(C\arabic*)}]
\item\label{zh:item:c1-binary} 每个非零普通可实现二元函数都是 $N$ 的非零常数倍；
\item\label{zh:item:c1-quaternary} 每个非零普通可实现四元函数，在某份端口 pairing 下都是
      $\lambda N\otimes N$，其中 $\lambda\ne0$。
\end{enumerate}
由分解引理，第二项正是 $C_1$ 上集的四元条件。以下证明在这两个条件下完成二分。

对偶数 $m$，记
\[
 \mathrm{EO}^{m}=\{\alpha\in\{0,1\}^{m}:|\alpha|_1=m/2\},\qquad
 r(\alpha)=|\alpha|_1-|\alpha|_0.
\]
函数的 EO 限制指保留该层的值，并把层外置零。记 $[m]=\{1,\ldots,m\}$。

\begin{theorem}[$C_1$ 上集的二分]\label{zh:thm:c1-strict-dichotomy}
设 $\Gamma$ 是满足 \ref{zh:item:c1-binary}、\ref{zh:item:c1-quaternary} 的有限代数复值偶元函数集。
则 $\Holant(N\mid\Gamma)$ 在 $\FP$ 中，或者是 $\sharpP$-hard。
除直接零值及共同单侧的 EO 出口外，证明将原问题 Turing 等价地化为一个实值
普通 Holant 问题，或者在有限个实际 EO 构件中得到困难性出口。
\end{theorem}

\subparagraph*{固定标签列的零性共享与单函数化}

\begin{lemma}\label{zh:lem:c1-active-single}
固定一列顶点函数及各自端口次序，改变全部端口间的闭合完美匹配。
所有所得闭函网要么值全为零，要么值全非零。
后一情形中，一次只重组两条边，前后值之比属于 $\{1/2,1,2\}$。
因而要么每个非空实例的值为零，要么原问题 Turing 等价于
\[
 \Holant(N\mid\{H\}),
\]
其中 $H$ 是普通可实现的偶元函数，且每个单 $H$ 点的完全闭合均非零。
\end{lemma}
\begin{proof}
断开被重组的两条边，余下部分是一个四元网函 $q$。
若 $q=0$，两个闭合值同为零；否则按 \ref{zh:item:c1-quaternary}，
$q=\lambda N\otimes N$，允许端口置换。
任一完美匹配闭合的值都是 $2\lambda$ 或 $4\lambda$。
任意两个完美匹配可经有限次两边重组互相到达，故零性及非零比值沿路径传播。

称函数 $f\in\Gamma$ 活跃，若它的一个单点完全闭合非零；这等价于全部单点完全闭合非零。
若实例含非活跃函数，把其固定顶点列重组为各顶点分别闭合，得到值零，
所以原实例也为零。删去非活跃函数。
若没有剩余函数，非空实例恒为零；空网按约定取值 $1$。
否则取剩余全部函数的张量积为 $H$。
展开 $H$ 给出一个归约方向。反方向，用一份 $H$ 模拟一个原顶点，保留所需因子，
其余因子各用预先选定的非零闭合收起；所得已知非零标量可除去。
$H$ 至少有一个非零闭合，再用固定标签列的零性共享，即知全部闭合非零。
\end{proof}

\begin{corollary}\label{zh:cor:c1-active-powers}
对任意固定 $k\ge1$，每个单 $H^{\otimes k}$ 点完全闭合都非零。
由 $H^{\otimes k}$ 添加若干 $N$ 自环得到的函数 $F$ 也有此性质；
特别地，$F|_{\mathrm{EO}}\ne0$，而且 $F$ 的每个二端 $N$ 短路都非零。
\end{corollary}
\begin{proof}
将 $H^{\otimes k}$ 展开为 $k$ 个 $H$ 点。各点分别闭合给出一个非零值，
故全部闭合非零。$F$ 的每种完全闭合也是这列 $k$ 个 $H$ 点的一种完全闭合。
若其 EO 限制或某个二端短路为零，再把剩余端口闭合便得到零，矛盾。
\end{proof}

\subparagraph*{EO 翻转对称与共同相位}

\begin{lemma}\label{zh:lem:c1-eo-symmetry}
在 \ref{zh:item:c1-binary} 下，每个普通可实现偶元函数 $f$ 的 EO 限制都满足
$f(\alpha)=f(\bar\alpha)$。
\end{lemma}
\begin{proof}
对元数归纳。二元情形由 \ref{zh:item:c1-binary} 成立。
一次 $N$ 短路删去一个 $0$ 和一个 $1$，所以它与取 EO 限制交换。
由归纳假设，$f|_{\mathrm{EO}}$ 的每个 $N$ 短路都翻转对称。
令 $R(\alpha)=f|_{\mathrm{EO}}(\alpha)-f|_{\mathrm{EO}}(\bar\alpha)$，
则 $R$ 为 EO 函数，且全部 $N$ 短路为零。
若至少四元的 EO 函数 $R$ 非零，可在一个非零输入中选出三个位置形成 $100$，
记该局部串三个置换位置上的函数值为 $a,b,c$，其中 $a\ne0$。
三种短路在相应剩余输入上的值分别为 $a+b,a+c,b+c$，不可能全部为零。
这与 $R$ 的所有短路为零矛盾。因此 $R=0$，归纳完成。
\end{proof}

\begin{lemma}\label{zh:lem:c1-common-phase}
对引理~\ref{zh:lem:c1-active-single} 的 $H$，存在非零代数数 $\theta$，使
$H|_{\mathrm{EO}}/\theta$ 是非零有理值、翻转对称函数。
把 $H/\theta$ 重新记为 $H$ 后，对每个 $k\ge1$，
$H^{\otimes k}|_{\mathrm{EO}}$ 均为有理值且翻转对称。
\end{lemma}
\begin{proof}
对 $H$ 的端口完美匹配 $M$，沿 $M$ 加 $N$ 自环，记闭合值为 $z_M$。
选定 $M_0$ 并令 $\theta=z_{M_0}\ne0$。
引理~\ref{zh:lem:c1-active-single} 给出 $z_M=2^{t_M}\theta$，$t_M\in\mathbb Z$。

把每对互补 EO 输入的共同值作为未知量。各个完美匹配闭合给出有理系数线性方程，
右端为 $z_M$。该系统在翻转对称 EO 子空间上满秩。
事实上，若非零翻转对称 EO 函数 $R$ 属于齐次核，取支撑输入中的三个位置，
将相应局部串记作 $100$，其三个置换值记作 $a,b,c$，其中 $a\ne0$。
三个短路值 $a+b,a+c,b+c$ 不可能全为零。因此可保持非零地降元，
并始终保持 EO 支撑及翻转对称，最后得到非零二元 $\lambda N$。
闭合它得到 $2\lambda\ne0$，与原函数属于全部完美匹配方程的齐次核矛盾。
将线性系统除以 $\theta$，唯一解遂为有理数。

对 $H^{\otimes k}$ 重用以上论证，活动性由推论~\ref{zh:cor:c1-active-powers} 保证。
固定 $H$ 的一个非零 EO 输入，其 $k$ 次拼接值是非零有理数。
这使新共同相位也是有理数，故整个张量幂的 EO 层有理。
\end{proof}

\subparagraph*{偏差吸收与最小生成偏差}

若 $H$ 的全部支撑偏差非负，或全部非正，则每条 $N$ 边贡献一个 $0$ 和一个 $1$，
任一全局非零项的总偏差为零。因此全部严格偏差项均不贡献，原问题逐实例等于
$\Holant(N\mid H|_{\mathrm{EO}})$，由一般 EO 二分及其 FP 升级
\cite{zh:meng2025eo,zh:guan2026lp} 分类。以下只处理支撑中正、负偏差均存在的情形。

\begin{lemma}\label{zh:lem:c1-absorb-charge}
存在一个代数对角矩阵 $D=\diag(c^{1/2},c^{-1/2})$，$c\ne0$，
使 $\widetilde H=D^{\otimes m}H$ 全局翻转对称。
$DND=N$，故这是保持边核的全息变换；它不改变任何张量积幂的 EO 层。
\end{lemma}
\begin{proof}
先证支撑翻转封闭。若 $r(\alpha)=a>0$，取支撑输入 $\beta$ 满足
$r(\beta)=-b<0$。$b$ 份 $\alpha$ 与 $a$ 份 $\beta$ 拼成 EO 输入。
张量幂 EO 翻转对称给出
\[
 H(\alpha)^bH(\beta)^a=H(\bar\alpha)^bH(\bar\beta)^a.
\]
左边非零，故翻转值非零。负偏差同理，零偏差直接由 EO 对称性。

令 $u(\alpha)=H(\alpha)/H(\bar\alpha)$。
两个相同偏差的输入 $\alpha,\beta$ 拼接 $\alpha\bar\beta$，说明
$u(\alpha)=u(\beta)$，记为 $u_r$。
任意整数零和关系 $\sum_jn_jr_j=0$，把负次数换成翻转输入，
由 EO 翻转对称得 $\prod_ju_{r_j}^{n_j}=1$。
因此 $r\mapsto u_r$ 在所生成子群 $g\mathbb Z\le\mathbb Z$ 上定义乘法同态。
取一个代数根 $c^g=u_g$，有 $u_r=c^r$。
$D^{\otimes m}$ 在偏差 $r$ 上乘以 $c^{-r/2}$，故变换后的两个翻转值相等。
偏差为零时该乘数为一，边核等式可直接相乘验证。
\end{proof}

以后将 $\widetilde H$ 重新记为 $H$。条件 \ref{zh:item:c1-binary}、\ref{zh:item:c1-quaternary}
仍保持，因为 $DND=N$。$H$ 仍活跃，全局翻转对称，且每个张量幂的 EO 层有理。

\begin{lemma}[相位字符]\label{zh:lem:c1-phase-character}
设 $R=\{r(\alpha):\alpha\in\supp(H)\}$，
$g=\gcd\{|r|:r\in R,r\ne0\}$。
非零函数值在 $\mathbb C^*/\mathbb R^*$ 中的相位类只依赖偏差，
并给出一个同态
\[
 \chi:g\mathbb Z\longrightarrow\{\mathbb R^*,i\mathbb R^*\}.
\]
存在一个固定张量积幂 $H^{\otimes k}$，具有偏差恰为 $g$ 的非零输入。
\end{lemma}
\begin{proof}
$\alpha\bar\alpha$ 是 EO 输入，全局翻转对称给出 $H(\alpha)^2\in\mathbb R^*$。
故每个相位类的阶至多为二。若 $r(\alpha)=r(\beta)$，则
$\alpha\bar\beta$ 为 EO 输入，所以 $H(\alpha)H(\beta)\in\mathbb R^*$，
两个相位类相同。
任意整数零和关系，仍将负次数改为翻转输入；所得 EO 拼接值实且非零，
故相应相位积为单位类。这证明同态良好定义。
最后由 B\'ezout 整数关系写出 $g$；支撑翻转封闭，负次数可以用翻转输入的正次数代替。
相应输入拼接非零，得到所需固定张量幂。
\end{proof}

\begin{lemma}[保留最小生成偏差的整形]\label{zh:lem:c1-shape-generator}
由上一个引理的张量幂，可以普通实现一个 $g$ 元函数
\[
 F=Q+a(\ket{0^g}+\ket{1^g}),\qquad a\ne0,
\]
其中 $Q$ 是非零有理翻转对称 EO 函数，$a$ 的相位类为 $\chi(g)$，且 $g\ge6$。
$F$ 的每个单点完全闭合都非零。
\end{lemma}
\begin{proof}
保持一个偏差为 $g$ 的非零输入。若它还同时含 $0,1$，因 $g>0$，可选局部串 $011$。
记其三个置换值为 $a_0,b_0,c_0$，其中 $a_0\ne0$。
三个短路值 $a_0+b_0,a_0+c_0,b_0+c_0$ 不会全为零，故选择一个仍非零的短路。
它删去一个 $0$ 和一个 $1$，保持偏差并降低元数，有限步后输入成为全一串，剩余元数为 $g$。
收缩保持偏差，所以最终非零偏差仍是 $g$ 的整数倍；在 $g$ 个端口上只可能是 $\pm g$。
全局翻转对称给出两个相等且非零的端点值。
EO 层有理性及活动性由前述引理保持。
每个偏差 $g$ 的加项都处于同一实直线 $\chi(g)$，故非零和 $a$ 仍有该相位类。
若 $g=2$，出现非零二元对角项；若 $g=4$，出现四元非零端点。
分别违反 \ref{zh:item:c1-binary}、\ref{zh:item:c1-quaternary}，所以 $g\ge6$。
\end{proof}

\subparagraph*{把 EO 短路化成匹配积}

记 $\partial_{ij}f=\sum_{a\ne b}f(x_i=a,x_j=b,\ldots)$，剩余端口保持原次序。
对整形后的 $F$，每个 $J_{ij}=\partial_{ij}F=\partial_{ij}Q$ 都是非零、有理、
翻转对称 EO 函数，且单点完全闭合均非零。
ARS--EO 定理\cite[Theorem 3.1]{zh:cai2020arseo} 对有限集合
$\mathcal E=\{J_{ij}:i<j\}$ 给出如下结论：若
$\mathcal E\subseteq\mathcal A$ 或 $\mathcal E\subseteq\mathcal P$，则易算；否则困难。
这里 $\mathcal A$ 为仿射支撑上的四次单位根二次相位类：
函数形如 $\lambda\mathbf1_{Ax=b}i^{L(x)+2Q_2(x)}$，其中 $L$ 为整数线性式，
$Q_2$ 为整数二次式。$\mathcal P$ 是前文定义的乘积型类。
本处实值且翻转对称，因而满足该定理的 ARS 条件。
困难分支经固定构件展开返回原问题；以下在集合统一落入其中一个易类的分支工作。

\begin{lemma}[活动 EO 易类的清理]\label{zh:lem:c1-cleanup}
设 $J$ 是普通可实现的非零有理翻转对称 EO 函数，所有单点完全闭合均非零。
若 $J\in\mathcal P\cup\mathcal A$，则在 \ref{zh:item:c1-quaternary} 下，
$J$ 是非零常数乘一个 $N$ 的完美匹配乘积。
\end{lemma}
\begin{proof}
先设 $J\in\mathcal P$。按等式／不等式约束的连通块分解。
支撑翻转封闭排除固定变量块；每块支持两个互补赋值，权重分别为 $a,b\ne0$。
各块赋值独立而总支撑在 EO，因此每块自身平衡。
比较整体翻转并固定其余块，得到每块 $b/a=\pm1$。
若某块权比为 $-1$，把该块按一个支撑赋值中的相反位配对闭合，再将其余块内部闭合，
全值为零，矛盾。因此每块两个权重相等。
把其他块非零闭合，可以普通提取任一块。
若某块至少四元，用 $N$ 反复短路一位 $0$ 类端口和一位 $1$ 类端口，
得到 $D_4=\ket{0011}+\ket{1100}$ 的端口置换及非零倍数，违反
\ref{zh:item:c1-quaternary}。所以每块均为 $N$ 的倍数。

再设 $J\in\mathcal A$。先把仿射支撑参数化为
$x_i=\ell_i(z)+a_i$，其中 $\ell_i$ 是 $\mathbb F_2$ 上的线性式。
EO 条件等价于 $\sum_i(-1)^{\ell_i(z)+a_i}=0$ 对每个 $z$ 成立。
不同线性式的 Walsh 特征函数线性无关，所以对每个固定的 $\ell$，
出现为 $\ell(z)$ 和 $\ell(z)+1$ 的坐标个数相等。
将这些坐标逐对匹配，得到
\[
 J=\prod_{t=1}^d[x_t\ne y_t]A(x_1,\ldots,x_d),\qquad A\in\mathcal A.
\]
取 $A$ 的一个非零值 $\lambda$；有理性使全部非零值为 $\pm\lambda$。
两份 $J$ 逐坐标连接第一份 $y_t$ 与第二份 $x_t$，各使用一个 $N$ 边，
留下第一份 $x_t$ 与第二份 $y_t$。支持赋值唯一延拓，所得函数为
$\lambda^2\prod_t[x_t\ne y_t]\chi(x)$，其中 $\chi$ 是原仿射支撑的指示函数。

若 $\chi$ 为真仿射子空间，选原坐标中的自由变量，并选一条非自由变量关系
$x_j+\sum_{t\in T}x_t=c$。支撑翻转封闭，所以该行的变量总数为正偶数。
给其他变量的 $x_t,y_t$ 加 $N$ 自环，消去其余依赖变量和未出现的自由变量；
留下这一个方程，仅附带非零的 $2$ 的整数次幂标量。
若方程至少四变量，选其中两个自由变量 $p,q$，交叉短路
$x_p,y_q$ 与 $x_q,y_p$。这强制 $x_p=x_q$，消去它们并乘 $2$，
方程减少两个变量。继续到两个变量时，连同相反配对，得到 $D_4$ 的置换，矛盾。

所以 $\chi$ 为整个立方体，$A(x)=\lambda(-1)^{q(x)}$，$q$ 是次数至多二的
$\mathbb F_2$ 多项式。若 $q$ 有非零混合二次项，短路参与该项的 $x_j,y_j$，得到
\[
 \lambda(-1)^{q(0,z)}
 \bigl(1+(-1)^{q(1,z)-q(0,z)}\bigr).
\]
指数差为非恒定仿射线性式，所以所得函数非零且具有真仿射支撑；
重复刚才的平方及消元构件，又得到 $D_4$，矛盾。
于是 $q$ 只有线性项。若有非零线性系数，按全部原配对闭合 $J$ 会出现因子
$1-1=0$，与活动性矛盾。故 $q$ 常值，结论成立。
\end{proof}

\subparagraph*{全部短路为匹配乘积时的纯结构定理}

下面根据各个 $N$ 短路的匹配结构，拼出 EO 函数的整体结构。
把 $\partial_{ij}$ 也写作 $\partial^{ij}$。

本段始终只用二元不等函数短路。对不同指标 $i,j$，将统一短路符号简写为
\[
 \partial^{ij}Q:=\sum_{a\ne b}Q(x_i=a,x_j=b,\ldots).       
\]
这里只省略固定的边函数下标，剩余端口保持原次序。
若 $e=\{i,j\}$，也写 $\partial^e$。若 $M$ 是一组偶数个端口上的完美匹配，记
\[
 P_M=\prod_{\{u,v\}\in M}[x_u\ne x_v].                  
\]

先定义八元例外。令 $F'_8$ 在下列支撑上取值 $1$，在其余输入上取值 $0$：
\begin{equation}\label{zh:eq:c1-f8}
\begin{split}
\supp(F'_8)=\{x\in\{0,1\}^8:\ &\textstyle\sum_{i=1}^8x_i=4,\\
&x_1+x_2+x_3+x_4\equiv x_5+x_6+x_7+x_8\equiv0\pmod2,\\
&x_1+x_2+x_5+x_6\equiv x_1+x_3+x_5+x_7\equiv0\pmod2\}.
\end{split}
\end{equation}
它有十四个支撑串，是一个 0--1 翻转对称 EO 函数。下面所说的八元例外允许变量置换和
整体乘以非零常数。

\begin{theorem}[平凡群 EO 降元定理]\label{zh:thm:c1-eo-matching-reduction}
设 $Q$ 是至少六元的有理值、0--1 翻转对称 EO 函数。假设对每对不同指标 $i,j$，存在
非零常数 $\theta_{ij}$ 和剩余端口上的完美匹配 $M_{ij}$，使
\begin{equation}\label{zh:eq:c1-all-contractions}
 \partial^{ij}Q=\theta_{ij}P_{M_{ij}}.
\end{equation}
则以下两项恰有一项成立：
\begin{enumerate}[label=\textup{(\arabic*)}]
  \item $Q$ 是非零常数乘一个完美匹配的二元不等函数乘积；
  \item $Q$ 是八元函数，并且在变量置换后 $Q=\lambda F'_8$，其中 $\lambda\ne0$。
\end{enumerate}
\end{theorem}

证明分成“先处理统一的大元数，再直接算完两个有限底例”两段。下一条计算是两段共同使用
的局部规则。

\begin{lemma}[完美匹配乘积的一次收缩]\label{zh:lem:c1-matching-merge}
设 $M$ 是完美匹配。若 $\{a,b\}\in M$，则
\[
 \partial^{ab}P_M=2P_{M\setminus\{\{a,b\}\}}.
\]
若 $a,b$ 在 $M$ 中分别与 $a',b'$ 配对且 $a'\ne b$，则
\[
 \partial^{ab}P_M
 =P_{(M\setminus\{\{a,a'\},\{b,b'\}\})\cup\{\{a',b'\}\}}.
\]
\end{lemma}

\begin{proof}
第一种情况下，$(x_a,x_b)=(0,1),(1,0)$ 各贡献一次。第二种情况下，
$x_a\ne x_{a'}$、$x_a\ne x_b$、$x_b\ne x_{b'}$ 合起来恰好等价于
$x_{a'}\ne x_{b'}$，而每个满足它的 $(x_{a'},x_{b'})$ 只有一个 $(x_a,x_b)$ 与之对应。
其余完美匹配边不变。
另外，$P_M$ 的支撑唯一确定 $M$：同一匹配边的两个变量恒取相反值，
不同边上的变量可以独立赋值。因此非零匹配乘积相等时，可以同时比较匹配和系数。
\end{proof}

\begin{lemma}[所有二元不等短路为零]\label{zh:lem:c1-zero-merges}
若一个 $m\ge3$ 元函数 $R$ 满足 $\partial^{ij}R\equiv0$ 对每对 $i\ne j$ 都成立，则 $R$ 只可能
在 $0^m,1^m$ 上非零。
\end{lemma}

\begin{proof}
任取同时含有 $0,1$ 的输入，其中三个位置可重排为 $001$ 或 $110$。以 $001\delta$ 为例，
三个收缩方程给出
\[
 R(001\delta)+R(010\delta)=0,\quad
 R(001\delta)+R(100\delta)=0,\quad
 R(010\delta)+R(100\delta)=0.
\]
前两式相加再减去第三式，得到 $2R(001\delta)=0$。$110\delta$ 同理。
\end{proof}

\begin{lemma}[共同边从两次收缩传播]\label{zh:lem:c1-two-merge-propagation}
设 $Q$ 满足式~\eqref{zh:eq:c1-all-contractions}。设不同的指标对 $p,q$ 都与 $\{u,v\}$ 不交，
并且 $[x_u\ne x_v]$ 同时是 $\partial^pQ,\partial^qQ$ 的因子。则对任意与
$\{u,v\}\cup p\cup q$ 不交的指标对 $e$，$[x_u\ne x_v]$ 也是 $\partial^eQ$ 的因子。
\end{lemma}

\begin{proof}
反设 $u,v$ 在 $M_e$ 中不配成一对，而分别与 $u',v'$ 配对。比较
$\partial^p\partial^eQ=\partial^e\partial^pQ$。右端原有因子 $[x_u\ne x_v]$，且 $e$ 不触及它，所以收缩后仍保留。
左端从 $M_e$ 出发；由引理~\ref{zh:lem:c1-matching-merge}，原来没有边 $\{u,v\}$ 时，收缩
$p$ 后要新生这条边，唯一可能是 $p=\{u',v'\}$。同理比较 $q$ 得
$q=\{u',v'\}$，与 $p\ne q$ 矛盾。
\end{proof}

\begin{lemma}[固定的收缩三角形]\label{zh:lem:c1-obtain-triangle}
若 $Q$ 至少十元并满足式~\eqref{zh:eq:c1-all-contractions}，则存在两两不同的
$u,v,r,s,t$，使 $[x_u\ne x_v]$ 同时整除
\[
 \partial^{rs}Q,\qquad \partial^{st}Q,\qquad \partial^{rt}Q.                  \tag{C1.10}
\]
\end{lemma}

\begin{proof}
固定 $\partial^{12}Q$，并按其完美匹配重排其余端口，使
\[
 \partial^{12}Q=\theta_{12}[x_3\ne x_4][x_5\ne x_6]
 [x_7\ne x_8][x_9\ne x_{10}]\cdots .                    \tag{C1.11}
\]
比较 $\partial^{34}\partial^{12}Q$ 与 $\partial^{12}\partial^{34}Q$。前者的完美匹配是从式~(C1.11) 删除
$\{3,4\}$ 后留下的 $d-2$ 条边。后者从 $M_{34}$ 收缩 $1,2$ 时，至多把其中两条边换成
一条新边，故前述结果中至多一条边不是 $M_{34}$ 的原有边。因 $d\ge5$，有 $d-2\ge3$，
所以可以先固定一条边 $\{u,v\}$，使它同时属于 $M_{12},M_{34}$。

在 $[2d]\setminus\{1,2,3,4,u,v\}$ 中再取三个不同指标 $r,s,t$。依次用
引理~\ref{zh:lem:c1-two-merge-propagation}，固定 $p=\{1,2\}$、$q=\{3,4\}$，便得到
式~(C1.10)。这里先选 $u,v$，再选辅助三角形；后文不改变 $u,v$。
\end{proof}

\begin{lemma}[从三角形传播到全部收缩]\label{zh:lem:c1-triangle-to-all}
在引理~\ref{zh:lem:c1-obtain-triangle} 的条件下，固定其中的 $u,v,r,s,t$。则对每一对与
$\{u,v\}$ 不交的 $i,j$，都有
\[
 [x_u\ne x_v]\mid \partial^{ij}Q.                              \tag{C1.12}
\]
\end{lemma}

\begin{proof}
若 $\{i,j\}$ 与 $\{r,s,t\}$ 不交，取三角形中的任意两条不同边，直接用引理
\ref{zh:lem:c1-two-merge-propagation}。若相交，则
$\{u,v,r,s,t,i,j\}$ 至多有七个指标；因总元数至少十，可在其外另取三个不同指标
$r',s',t'$。第一种情形先说明这三个新收缩仍都含固定因子 $[x_u\ne x_v]$，从而得到一个
与 $\{i,j\}$ 不交的新三角形；再用第一种情形即得式~(C1.12)。改变的只是辅助三角形，
$u,v$ 始终固定。
\end{proof}

\begin{lemma}[至少十元时提取固定因子]\label{zh:lem:c1-large-factor}
若 $Q$ 至少十元并满足式~\eqref{zh:eq:c1-all-contractions}，则 $Q$ 是非零常数乘一个完美
匹配的二元不等函数乘积。
\end{lemma}

\begin{proof}
由前两条引理固定 $u,v$，使式~(C1.12) 对所有不触及 $u,v$ 的指标对成立。按
$(x_u,x_v)=00,01,10,11$ 把 $Q$ 写成四个 函数切片
$Q^{00}_{uv},Q^{01}_{uv},Q^{10}_{uv},Q^{11}_{uv}$。对每对 $i,j\notin\{u,v\}$，共同因子
逐个函数切片 给出
\[
 \partial^{ij}Q^{00}_{uv}=\partial^{ij}Q^{11}_{uv}=0,
 \qquad \partial^{ij}(Q^{01}_{uv}-Q^{10}_{uv})=0.                \tag{C1.13}
\]
由引理~\ref{zh:lem:c1-zero-merges}，这三个函数只可能在各自的全零、全一输入上非零。但 $Q$
是 $2d$ 元 EO 函数且 $d\ge5$：这些端点在四个函数切片 中的总重量分别属于
$\{0,2d-2\}$、$\{1,2d-1\}$、$\{2,2d\}$，都不等于 $d$。所以
\[
 Q^{00}_{uv}=Q^{11}_{uv}=0,\qquad
 Q^{01}_{uv}=Q^{10}_{uv}=:Q',
\]
即 $Q=[x_u\ne x_v]Q'$。最后使用尚未用过的收缩 $\partial^{uv}$：一方面
$\partial^{uv}Q=2Q'$，另一方面式~\eqref{zh:eq:c1-all-contractions} 给出
$\partial^{uv}Q=\theta_{uv}P_{M_{uv}}$。故
$Q'=(\theta_{uv}/2)P_{M_{uv}}$，结论成立。
\end{proof}

这段证明的量词次序是先得到一个固定见证
\[
 \exists\{u,v\}\ \forall\{i,j\}\cap\{u,v\}=\emptyset,
 \qquad [x_u\ne x_v]\mid \partial^{ij}Q,
\]
再从 $Q$ 本身提取这个因子。特别地，$u,v$ 不随 $i,j$ 或某个反证输入改变。

下面用数学计算完成六元、八元底例。继续使用式~\eqref{zh:eq:c1-all-contractions}
中的 $Q,\theta_{ij},M_{ij}$、前述 $P_M$ 和式~\eqref{zh:eq:c1-f8} 的 $F'_8$；
用 $ij$ 简记边 $\{i,j\}$。这部分是纯函数结构证明，不新增可实现性或分解前提。
六元直接写出线性分支；八元先确定全部匹配和系数，再恢复函数。
恢复的唯一性由引理~\ref{zh:lem:c1-zero-merges}给出：
若两个 EO 函数的全部短路相同，其差只可能在两个端点上非零，
而 EO 条件排除这两个端点，故二者相等。

\subparagraph*{六元：五个两参数族的逐项计算}

先通过整体非零缩放和变量置换固定
\[
 \partial^{12}Q=P_{\{34,56\}}.
\]
对三元子集 $A\subseteq[6]$，用 $1_A$ 表示其指示串，记
$q_A=Q(1_A)=q_{[6]\setminus A}$。
初始条件的全部解可写成
\[
\begin{array}{lll}
q_{134}=a,&q_{156}=-a,&q_{135}=b,\\
q_{146}=1-b,&q_{136}=c,&q_{145}=1-c,
\end{array}
\qquad y_r=q_{12r}\quad(3\le r\le6).
\]
所以只剩 $a,b,c,y_3,y_4,y_5,y_6$ 七个数。
四元匹配乘积在三对互补 EO 输入上的值，恰为
$(0,t,t),(t,0,t)$ 或 $(t,t,0)$，其中 $t\ne0$。
对以下五个短路逐项求和，得到
\begin{equation}\label{zh:eq:c1-finite-six-rows}
\begin{array}{c|ccc}
 e&\multicolumn{3}{c}{\partial^e Q\text{ 在三对互补 EO 输入上的值}}\\ \hline
 13&y_4-a&y_5+1-b&y_6+1-c\\
 14&y_3-a&y_5+c&y_6+b\\
 25&y_3+b&y_4+1-c&y_6-a\\
 23&y_4+a&y_5+b&y_6+c\\
 16&y_3+1-c&y_4+b&y_5+a
\end{array}
\end{equation}
例如，$25$ 行的第三项是 $q_{126}+q_{156}=y_6-a$；
$16$ 行的第一项是 $q_{123}+q_{236}=y_3+1-c$。
每一行必须恰有一个零项，其余两项相等且非零。

令 $h_e\in\{1,2,3\}$ 表示第 $e$ 行唯一的零坐标。
先用 $13,14$ 两行。若 $(h_{13},h_{14})$ 为 $(1,2)$ 或 $(1,3)$，
第一行的等值式迫使第二行三项全零；若为 $(2,1)$ 或 $(3,1)$，则反过来迫使
第一行三项全零，均不允许。因此只剩 $(1,1),(2,2),(2,3),(3,2),(3,3)$。
再看 $25$ 行，把它的三项记为 $A,B,C$。在 $(2,3)$ 分支，前两行给出
$2A+B+3C=0$；在 $(3,2)$ 分支给出 $A+2B+3C=0$。
而 $A,B,C$ 中一项为零、另两项同为非零 $t$，代入任一个式子都得到非零正整数乘
$t=0$，矛盾。在 $(2,2)$ 和 $(3,3)$ 分支，$25$ 行前两项相等，故只能第三项为零。
在 $(1,1)$ 分支，$25$ 行仍有三种选择。这就穷尽为以下五族；记 $w=y_5,z=y_6$：
\begingroup\small\setlength{\arraycolsep}{4pt}
\[
\begin{array}{c|c|ccc|c}
 &(h_{13},h_{14},h_{25})&a&b&c&y_3=y_4\\ \hline
\mathrm I&(1,1,1)&\frac{-w+2z-1}{3}&\frac{w-2z+1}{3}&\frac{-2w+z+1}{3}&a\\
\mathrm {II}&(1,1,2)&\frac{-w+2z-1}{3}&\frac{2w-z+2}{3}&\frac{-w+2z+2}{3}&a\\
\mathrm {III}&(1,1,3)&z&\frac{w-z+1}{2}&\frac{-w+z+1}{2}&z\\
\mathrm {IV}&(2,2,3)&z&w+1&-w&w+2z+1\\
\mathrm V&(3,3,3)&z&-z&z+1&w+2z+1
\end{array}
\]
\endgroup
表中每一行给出所列零项方程和等值方程的全部解，参数仍须满足各行的非零条件。
例如第一族先有 $y_3=y_4=a$、$b=c+w-z$；再由 $25$ 行得
$a+b=0$、$a+1-c=z-a$，逐项解出第一行。
现在先用 $16$ 行，再用 $23$ 行。记 $s=w+z$，代入得到
\begingroup\small\setlength{\arraycolsep}{3pt}
\[
\begin{array}{c|c|c|c}
 &16&23\text{（已代入第16行条件）}&z\\ \hline
\mathrm I&((s+1)/3,0,(2s-1)/3)&(2z-2,3-2z,2z-1)&1\\
\mathrm {II}&(0,(s+1)/3,(2s-1)/3)&(2z-2,4-2z,2z)&1\\
\mathrm {III}&((s+1)/2,(s+1)/2,s)&(2z,\frac12-2z,\frac12+2z)&0\\
\mathrm {IV}&(2s+2,2s+2,s)&(2z+1,1-2z,2z)&0\\
\mathrm V&(s+1,s+1,s)&(2z+1,-2z,2z+1)&0
\end{array}
\]
\endgroup
由于 $16$ 行必须恰有一项为零，第一、二族只能 $s=2$，后三族只能 $s=0$。
代入 $23$ 行后，逐一令它的三个坐标为零并检验另两项相等，恰得表中最后一列。
例如第一族的三个选择分别给 $z=1,3/2,1/2$；后二者的另两项分别为
$(1,2)$ 和 $(-1,2)$，故都排除。其余四行同样只有所列的一个值。
于是最后只剩五个点：
\begin{equation}\label{zh:eq:c1-finite-six-points}
\begin{array}{c|ccc|cccc|c}
&a&b&c&y_3&y_4&y_5&y_6&Q\\ \hline
1&0&0&0&0&0&1&1&P_{\{13,24,56\}}\\
2&0&1&1&0&0&1&1&P_{\{14,23,56\}}\\
3&0&\frac12&\frac12&0&0&0&0&\frac12P_{\{12,34,56\}}\\
4&0&1&0&1&1&0&0&P_{\{16,25,34\}}\\
5&0&0&1&1&1&0&0&P_{\{15,26,34\}}
\end{array}
\end{equation}
这五个函数均为匹配乘积，引理~\ref{zh:lem:c1-matching-merge}
同时验证它们未列在表中的全部短路。因此六元情形没有遗漏。

\subparagraph*{八元：从最大系数固定四对变量}

八元不必照六元那样列出全部 $35$ 个独立函数值。我们先确定全部匹配和系数，再用
前面的唯一性规则恢复函数。

对不相交的指标对 $e,f$，两次短路可以交换。按
引理~\ref{zh:lem:c1-matching-merge}，比较
$\partial^f\partial^e Q
=\partial^e\partial^f Q$ 得
$\theta_e/\theta_f\in\{1/2,1,2\}$。
如果 $e,f$ 相交，在其并集之外取一个指标对 $g$，分别与 $g$ 比较。
因此全部 $\theta_e$ 之比均为正实数。
从这 $28$ 个非零系数中取绝对值最大的一项，整体缩放并重编号，使
\[
 \theta_{12}=1,\qquad M_{12}=\{34,56,78\},\qquad
 0<\theta_e\le1\quad\text{对全部 }e.
\]
比较 $12$ 和 $34$ 的两次短路：
一侧的系数为 $2$，另一侧为 $\theta_{34}$ 或 $2\theta_{34}$。
最大性迫使 $\theta_{34}=1$，且 $12\in M_{34}$，从而
$M_{34}=\{12,56,78\}$。对 $56,78$ 同理。
记 $M_0=\{12,34,56,78\}$，便得到
\begin{equation}\label{zh:eq:c1-finite-four-edges}
 \partial^r Q=P_{M_0\setminus\{r\}}
 \qquad(r\in M_0).
\end{equation}
这里取最大值的集合始终是同一个 $Q$ 的 $28$ 次短路，没有改换函数或使用最小反例。

把这四对变量称为四个块 $A=(1,2),B=(3,4),C=(5,6),D=(7,8)$。
考虑跨两个块的指标对 $13$。分别与 $56,78$ 比较，两个双重短路在
式~\eqref{zh:eq:c1-finite-four-edges} 一侧的系数都是 $1$。
所以 $56,78$ 必须同时属于 $M_{13}$，或同时不属于它：
若只属于一条，两次比较将同时要求 $2\theta_{13}=1$ 和 $\theta_{13}=1$。
同时属于时，余边只能为 $24$。
同时不属于时也必须有 $24\in M_{13}$：否则，要通过短路 $56$ 新生出 $24$，
$2,4$ 的两个配偶必须为 $5,6$；要通过短路 $78$ 新生出同一条边，
这两个配偶又必须为 $7,8$，矛盾。
因此只有三种可能：
\begin{equation}\label{zh:eq:c1-finite-three-choices}
\begin{array}{c|c}
 M_{13}&\theta_{13}\\ \hline
 \{24,56,78\}&1/2\\
 \{24,57,68\}&1\\
 \{24,58,67\}&1
\end{array}
\end{equation}
对任一跨块指标对，同样的论证成立。第一种称为直配型：另外两块各自内部配对；
后两种称为交配型：另外两块之间取两种跨块匹配之一。

\subparagraph*{八元：三种块划分必须同时取同一型}

四块分成两组各两块，只有 $AB|CD,AC|BD,AD|BC$ 三种划分。
先固定 $AB|CD$，取 $e$ 跨 $A,B$，$f$ 跨 $C,D$。
若 $e$ 为直配型，则短路 $f$ 不删去 $M_e$ 中的一条边，
所以 $\partial^f\partial^e Q$ 的系数为 $1/2$。
若 $f$ 为交配型，另一侧系数只能为 $1$ 或 $2$，不可能相等。
故 $e$ 为直配型会迫使对侧的每个 $f$ 也为直配型；再反过来比较，
本侧的每个跨块指标对也为直配型。
所以在一份块划分内，八个跨块指标对要么全为直配型，要么全为交配型。

不同划分也不能混型。比较 $e=13$ 与 $f=25$。
在 $M_{13}$ 中，$2$ 总与 $4$ 配对，故 $25$ 不是匹配边；
在 $M_{25}$ 中，$1$ 总与 $6$ 配对，故 $13$ 也不是匹配边。
于是交换这两次短路直接给出 $\theta_{13}=\theta_{25}$，
从而 $AB|CD$ 与 $AC|BD$ 同型。改用 $e=13,f=27$，同理得到
$AB|CD$ 与 $AD|BC$ 同型。
所以三种划分要么全为直配型，要么全为交配型。

若全为直配型，式~\eqref{zh:eq:c1-finite-four-edges} 和
式~\eqref{zh:eq:c1-finite-three-choices} 已经给出全部 $28$ 次短路，
它们恰与 $\frac12P_{M_0}$ 的短路相同。由 EO 函数的唯一性规则，
\[
 Q=\tfrac12P_{\{12,34,56,78\}}.
\]

\subparagraph*{八元：交配型只剩两种方向}

以下所有跨块指标对都为交配型，故其系数均为 $1$。
在每个块内，按前述次序把两个端口标为 $0,1$。
两个块之间的一条边，用两个端口标号的异或 $p\in\{0,1\}$ 表示其方向；
每个方向恰有两条边，它们合起来是这两个块的一个完美匹配。

固定 $AB|CD$。对任意跨 $AB$ 的 $e$ 和跨 $CD$ 的 $f$，
交换短路并比较系数 $1$ 或 $2$，得到
\begin{equation}\label{zh:eq:c1-finite-incidence}
 f\in M_e\quad\Longleftrightarrow\quad e\in M_f.
\end{equation}
同一方向的两条 $AB$ 边 $e,e'$，在任何 $M_f$ 中必同时出现或同时不出现。
所以由式~\eqref{zh:eq:c1-finite-incidence}，$M_e,M_{e'}$ 在 $CD$ 上选择同一个方向。
这给出一个方向映射 $\phi:\{0,1\}\to\{0,1\}$。
它不可能为常值：否则取所选方向上的一条 $f$，由同一式子，
$M_f$ 必须同时含有全部四条 $AB$ 边，矛盾。
所以 $\phi$ 是双射，必形如 $p\mapsto p\oplus s_1$。
对另外两份块划分同理，分别得到 $s_2,s_3\in\{0,1\}$。
式~\eqref{zh:eq:c1-finite-incidence} 也确定了反向映射，不再有额外选择。

最后直接比较两组短路，确定三个比特。对 $13,25$，各自的匹配和双重短路匹配如下：
\[
\begin{array}{c|c|c}
 &\text{原匹配}&\text{再短路另一指标对后}\\ \hline
 s_1=0&M_{13}=\{24,57,68\}&\{47,68\}\\
 s_1=1&M_{13}=\{24,58,67\}&\{48,67\}\\ \hline
 s_2=0&M_{25}=\{16,38,47\}&\{47,68\}\\
 s_2=1&M_{25}=\{16,37,48\}&\{48,67\}
\end{array}
\]
因此 $s_1=s_2$。再比较 $13,27$：
\[
\begin{array}{c|c|c}
 &\text{原匹配}&\text{再短路另一指标对后}\\ \hline
 s_1=0&M_{13}=\{24,57,68\}&\{45,68\}\\
 s_1=1&M_{13}=\{24,58,67\}&\{46,58\}\\ \hline
 s_3=0&M_{27}=\{18,36,45\}&\{45,68\}\\
 s_3=1&M_{27}=\{18,35,46\}&\{46,58\}
\end{array}
\]
故 $s_1=s_3$。这样全部短路的匹配和系数只剩两份数据：
$s_1=s_2=s_3=0$，或 $s_1=s_2=s_3=1$。
由唯一性规则，每份数据至多对应一个 EO 函数。

\subparagraph*{正向验证八元例外并收尾}

还须证明上述两份数据确实来自函数，而不只是形式上相容的匹配。
把 $x_1,\ldots,x_8$ 依次标在 $\mathbb F_2^3$ 的
\[
 000,001,010,011,100,101,110,111
\]
八个点上。式~\eqref{zh:eq:c1-f8} 的四个校验向量为
\[
11110000,\quad00001111,\quad11001100,\quad10101010.
\]
它们张成常数函数和三个坐标函数的真值表空间。
这四个向量线性无关且两两正交，所以它们的共同零空间就是同一个四维空间。
非恒定仿射函数 $z\mapsto\alpha\cdot z+\beta$（$\alpha\ne0$）的重量均为 $4$；
两个恒定函数的重量为 $0,8$。所以 $F'_8$ 的十四个支撑串恰是
这十四个非恒定仿射函数的真值表。

固定不同点 $u,v$，令 $\delta=u+v$。在 $u,v$ 上取不同值等价于
$\alpha\cdot\delta=1$，共有四种 $\alpha$ 和两种 $\beta$。
其余六个点按 $z\leftrightarrow z+\delta$ 分成三对，每对的函数值相反。
上述八个仿射函数在这六点上的限制互不相同：
若两个限制相同，其差是一个在六点上为零的仿射函数；
而非零仿射函数的重量只能为 $4$ 或 $8$，不可能仅支撑在剩下两点。
三对不等条件恰有八个赋值，因而每个赋值恰好对应一个内部赋值。于是
\begin{equation}\label{zh:eq:c1-finite-f8-merges}
 \partial^{uv}F'_8
 =\prod_{\substack{\{z,z+\delta\}\subseteq
       \mathbb F_2^3\setminus\{u,v\}\\\text{每对只取一次}}}
 [x_z\ne x_{z+\delta}].
\end{equation}
这验证了全部 $28$ 次短路，而非只验证代表指标对。
其中 $M_{12}=\{34,56,78\}$，$M_{13}=\{24,57,68\}$，
故 $F'_8$ 正好对应三个方向比特均为 $0$ 的数据。
交换 $x_7,x_8$ 后，四条块内边的数据不变，而
$M_{13}$ 变为 $\{24,58,67\}$，所以对应三个比特均为 $1$ 的数据。
因此，在以最大短路系数归一化并按其匹配重编号后，八元只有
\begin{equation}\label{zh:eq:c1-finite-eight-points}
 \tfrac12P_{\{12,34,56,78\}},\qquad F'_8,\qquad (78)F'_8
\end{equation}
这三种可能。最后两种互为变量置换，正是定理中的一个例外类。

\begin{proof}[定理~\ref{zh:thm:c1-eo-matching-reduction} 的证明]
六元由式~\eqref{zh:eq:c1-finite-six-rows} 的穷尽分支计算得到
式~\eqref{zh:eq:c1-finite-six-points}。
八元先取最大短路系数固定四个块，由三种跨块选择、同型性和方向比特的相容关系，
得到式~\eqref{zh:eq:c1-finite-eight-points}。全部候选由匹配短路公式和
式~\eqref{zh:eq:c1-finite-f8-merges} 正向验证。
至少十元时使用引理~\ref{zh:lem:c1-large-factor}。
撤销有限底例中的非零缩放和变量置换，这些情形穷尽所有至少六元的偶数元数。
\end{proof}

\subparagraph*{两份整形函数迫使整体实数化}

\begin{lemma}\label{zh:lem:c1-double-F}
设 $F=Q+a(\ket{0^m}+\ket{1^m})$ 是引理~\ref{zh:lem:c1-shape-generator} 的函数。
若 $Q$ 满足定理~\ref{zh:thm:c1-eo-matching-reduction} 的全部前提，
则 $m=8$，$Q=\lambda F'_8$ 的端口置换，且 $a=\pm\lambda\in\mathbb R^*$。
\end{lemma}
\begin{proof}
先设 $m=2d$ 且 $Q=\lambda P_M$。
取两份 $F$，各保留 $M$ 中同一条边的两个端口，其余对应端口逐点用 $N$ 桥接。
$Q\times Q$ 给出 $d-1$ 个闭圈，贡献 $2^{d-1}\lambda^2N_{12}N_{34}$；
端点乘端点只有互补赋值可通过桥边，贡献 $a^2D_4$。
混合项因内部匹配要求相反、端点却全相同而为零。
所得四元的三个互补 EO 值为
\[
 (a^2,2^{d-1}\lambda^2,2^{d-1}\lambda^2).
\]
它们全非零，不能是任一非零 $N$ 匹配乘积，违反 \ref{zh:item:c1-quaternary}。

只剩 $Q=\lambda F'_8$。按 $\mathbb F_2^3$ 的八点标号，
两份 $F$ 各留同样两个点 $u,v$，其余六个同名点逐点用 $N$ 桥接。
两个仿射函数若在六点互补，就在全部八点互补。
指定 $u,v$ 上两个值，共有四个仿射函数；同值时其中恰有一个常值函数，
异值时四个均非常值。因此 $Q\times Q$ 输出
$4\lambda^2N_{13}N_{24}-\lambda^2D_4$。
混合项为零，因为非常值仿射函数不可能在六点同为常值。
加上端点乘端点，所得四元为
\[
 4\lambda^2N_{13}N_{24}+(a^2-\lambda^2)D_4,
\]
三个互补 EO 值为 $(a^2+3\lambda^2,0,4\lambda^2)$。
条件 \ref{zh:item:c1-quaternary} 迫使第一项等于第三项，故 $a^2=\lambda^2$。
由于 $Q$ 为非零有理函数，$\lambda$ 为非零有理数，所以 $a=\pm\lambda$ 为实数。
\end{proof}

\begin{proof}[定理~\ref{zh:thm:c1-strict-dichotomy} 的证明]
依引理~\ref{zh:lem:c1-active-single}，先处理恒零分支，或者 Turing 等价地单函数化。
共同单侧分支送入一般 EO 二分。余下分支经共同相位及偏差吸收，得到活跃、
翻转对称的 $H$，全部张量幂的 EO 层有理。
引理~\ref{zh:lem:c1-phase-character} 取全部支撑偏差生成子群的正生成元 $g$，
引理~\ref{zh:lem:c1-shape-generator} 得到 $F$。
对全部 $\partial_{ij}F$ 同时使用 ARS--EO 二分，若困难便通过实际构件返回原问题；
否则引理~\ref{zh:lem:c1-cleanup} 使全部非零短路成为匹配乘积。
于是纯结构定理及引理~\ref{zh:lem:c1-double-F} 给出实端点值 $a$。
但其相位类为 $\chi(g)$，因此 $\chi(g)=\mathbb R^*$。
相位同态遂处处为单位类，整个 $H$ 都实值。

回到普通相等边坐标，使用式~\eqref{zh:eq:K-basis}。
对输出 $y$，有
\[
 (K^{\otimes m}H)(y)
 =2^{-m/2}i^{|y|}\sum_x(-1)^{y\cdot x}H(x).
\]
翻转对称使右侧实数和在 $|y|$ 奇数时为零；偶数时 $i^{|y|}$ 是实数。
所以变换后的整个单函数实值。
Real-Holant 二分\cite{zh:shao2020real} 分类该问题，先前 Turing 等价把结论返回原 $\Gamma$。
\end{proof}

\paragraph{下集}\label{zh:subsec:c1-lower}\label{zh:sec:c1-quaternary}
取一个真四元函数。先用二元封闭条件确定它的形式，再按两个端点是否非零分情况。
双端点用已有四元二分；单端点用插值得到钉子；没有端点时按 EO 层的中心数处理。

\subparagraph*{真四元函数的二分}

本节在原生 $N$ 边模型中讨论 $C_1$ 的下集。
设 $\Gamma$ 是有限代数复值偶元函数集，按前面的框架约定，
所有可用非零二元函数恰为 $\lambda N$，$\lambda\ne0$。
记 $\delta_x$ 为仅在输入 $x$ 取值 $1$ 的函数。

\begin{proposition}[$C_1$ 下集的二分]\label{zh:prop:c1-quaternary-interface}
在上述二元封闭条件下，若有一个真四元函数，则 $\Holant(N\mid\Gamma)$
在 $\FP$ 中，或者是 $\sharpP$-hard。
\end{proposition}

\subparagraph*{由二元闭包强制的四元形式}

\begin{lemma}\label{zh:lem:c1-general-q4-shape}
每个普通四元构件都可写成
\begin{equation}\label{zh:eq:c1-general-q4-shape}
 q=s\delta_{0000}+t\delta_{1111}
   +a(\delta_{0011}+\delta_{1100})
   +b(\delta_{0101}+\delta_{1010})
   +c(\delta_{0110}+\delta_{1001}).
\end{equation}
\end{lemma}
\begin{proof}
把任意两个端口 $i,j$ 用一条 $N$ 边短路，输出二元矩阵的两个对角值必须为零。
设 $u_i=q(e_i)$ 是四个重量一的系数，便有 $u_i+u_j=0$ 对所有 $i\ne j$ 成立。
取任意三个指标可得其系数全为零，继而第四个也为零。
重量三的四个系数用同一论证处理。

设三个重量二互补对的差分别为
\[
 d_a=q_{0011}-q_{1100},\quad
 d_b=q_{0101}-q_{1010},\quad
 d_c=q_{0110}-q_{1001}.
\]
短路 $12$、$13$、$23$ 后，二元矩阵的两个反对角值相等，分别给出
$d_b=d_c$、$d_a=d_c$、$d_a=-d_b$。
因此三个差均为零，得到所述形式。$s,t$ 在这些短路中不出现，因而没有受到约束。
\end{proof}

\subparagraph*{两个非零端点的出口}

\begin{lemma}\label{zh:lem:c1-general-two-endpoints}
若式~\eqref{zh:eq:c1-general-q4-shape} 中 $st\ne0$，则整个函数集已有可判定二分。
\end{lemma}
\begin{proof}
若 $q$ 不可作非平凡张量积分解，直接使用\cite{zh:liu2026disequality}的
Quaternary $K$-Holant dichotomy。其条件正是偶支撑、两个非零端点和
张量积不可分解，结论给出 $\Holant(N\mid\Gamma)$ 的二分。

若 $q$ 可分解，两个非零端点迫使每个因子在其全零、全一输入上均非零。
$1+3$ 分解会同时产生奇重量非零输入，故不可能。
在某个 $2+2$ 分解下，两个因子都必须是满秩对角二元函数：
任一因子的反对角非零项与另一因子的全零项相乘都会违反偶支撑。
由分解引理，这两个因子可用，与二元函数只有 $N$ 的非零倍数矛盾。
因此可分解的情形不在本分支中。
\end{proof}

\subparagraph*{任意单端点四元函数的统一插值}

\begin{lemma}[单端点 Jordan 插值]\label{zh:lem:c1-general-jordan}
若式~\eqref{zh:eq:c1-general-q4-shape} 恰有一个非零端点，则任意多次使用该端点
四元钉子的实例均 Turing 归约到 $\Holant(N\mid\Gamma)$。
因此整个函数集已有可判定二分。
\end{lemma}
\begin{proof}
先构造一列四元链，使端点项比其余项多出一个链长因子；再用插值把端点项单独取出。
先设 $s\ne0,t=0$，目标钉子为 $E=\delta_{0000}$。
若 $a=b=c=0$，它已是实际构件的非零倍数。
否则置换端口，使 $a\ne0$。按 $12\mid34$ 的自然顺序展平，写
\[
 Q=\begin{pmatrix}
 s&0&0&a\\0&b&c&0\\0&c&b&0\\a&0&0&0
 \end{pmatrix},\qquad J=N\otimes N.
\]
取 $m\ge1$ 份 $q$，将每份的两个列端口分别用 $N$ 边连接下一份的两个行端口。
所得实际四元构件 $q_m$ 的矩阵为
\begin{equation}\label{zh:eq:c1-general-chain}
 Q_m=(QJ)^{m-1}Q=(QJ)^mJ.
\end{equation}
$QJ$ 在 $00,11$ 子空间上的矩阵是
$\left(\begin{smallmatrix}a&s\\0&a\end{smallmatrix}\right)$；
在 $01,10$ 子空间上是
$\left(\begin{smallmatrix}c&b\\b&c\end{smallmatrix}\right)$，后者总是半单的。
记
\[
 A=\delta_{0101}+\delta_{1010},\quad
 B=\delta_{0110}+\delta_{1001},\quad
 C=\delta_{0011}+\delta_{1100}.
\]
直接取幂得到精确公式
\begin{equation}\label{zh:eq:c1-general-jordan-expansion}
 q_m=m\frac{s}{a}a^mE+a^mC
      +(c+b)^m\frac{A+B}{2}
      +(c-b)^m\frac{B-A}{2}.
\end{equation}
底数为零的项对所有 $m\ge1$ 都为零，故可直接删除。

考虑任意包含 $M\ge1$ 个 $E$ 顶点的目标函网。用结合律，把这些点与其余部分分开：
前者为 $E^{\otimes M}$，后者结合成一个固定的上下文。
将前者同时替换为 $q_m^{\otimes M}$，所得值记为 $Z_m$。
随 $m$ 改变的部分给出方程系数，固定上下文给出待恢复的未知量。
令 $\Lambda$ 是 $a,c+b,c-b$ 中不同非零数的集合；其大小至多为三。
将相同的底数聚合后，多线性展开给出
\begin{equation}\label{zh:eq:c1-general-exp-poly}
 Z_m=\sum_{\mu\in\mathcal L_M}P_\mu(m)\mu^m,
 \qquad \deg P_\mu\le M,
\end{equation}
其中 $\mathcal L_M$ 是从 $\Lambda$ 选取 $M$ 项相乘得到的不同非零数集合。
它的大小至多 $\binom{M+2}{2}$。
公式~\eqref{zh:eq:c1-general-jordan-expansion} 中唯有端点项含有因子 $m$，所以
\begin{equation}\label{zh:eq:c1-general-leading-coeff}
 [m^M]P_{a^M}(m)=\left(\frac{s}{a}\right)^M Z_{\mathrm{target}}.
\end{equation}
不同谱底数之间即使存在单位根关系或者乘积重合，也不改变此式：
次数 $M$ 要求所有 $M$ 个位置都选取端点项。

下面明确验证有限次插值的可逆性。
令 $L=(M+1)|\mathcal L_M|=O(M^3)$，用 $m=1,\ldots,L$ 的样本，
对未知系数 $[m^j]P_\mu$ 求解矩阵
\[
 V_{m,(\mu,j)}=m^j\mu^m,\qquad 0\le j\le M.
\]
若存在行依赖 $(v_1,\ldots,v_L)$，令 $p(z)=\sum_{m=1}^L v_mz^m$。
依赖关系表示 $(z\,d/dz)^jp(\mu)=0$，对所有 $\mu\in\mathcal L_M$ 及
$0\le j\le M$ 成立。由于 $\mu\ne0$，这些关系与
$p^{(j)}(\mu)=0$ 等价：两套导数之间的转换是对角非零的三角转换。
所以 $p$ 被次数为 $L$ 的多项式
$\prod_{\mu\in\mathcal L_M}(z-\mu)^{M+1}$ 整除。
又 $p(0)=0$，而该乘积的常数项非零，故 $p=0$。
因此 $V$ 可逆。恢复式~\eqref{zh:eq:c1-general-leading-coeff} 的系数并除以
已知非零数 $(s/a)^M$ 即得目标值。
每个样本使用 $mM=O(M^4)$ 份固定构件，总查询数为 $O(M^3)$。
所有数都属于由固定代数系数生成的固定数域；精确比较、求幂和线性求解的位复杂度
均为输入大小的多项式。若预先缩放了 $q$，同时补偿相应的已知幂次即可。

当 $s=0,t\ne0$ 时，相同计算的端点子块为下三角 Jordan 块，提取
$\delta_{1111}$；这是直接的矩阵计算，不需要额外提供翻位函数。

得到四元钉子后，用分解定理~\ref{zh:thm:decomposition}：或者问题已在 $\FP$ 中，
或者取得非零一元因子 $\delta_0$ 或 $\delta_1$，再由奇元二分定理~\ref{zh:thm:odd-arity-fp} 完成。
\end{proof}

\subparagraph*{纯 EO 四元函数的中心数分派}

\begin{lemma}\label{zh:lem:c1-general-three-centers}
若 $s=t=0$ 且 $abc\ne0$，则 $\Holant(N\mid\Gamma)$ 是 $\sharpP$-hard。
\end{lemma}
\begin{proof}
Cai--Fu--Xia 的六顶点二分\cite[Theorem~3.1]{zh:cai2018sixvertex}
适用于一般图上的 $\Holant(N\mid\{q\})$。
其例外为 $q\in\mathcal A$、$q\in\mathcal P$，或三个互补对中每对至少有一项为零。
这里 $q$ 的支撑大小恰为六；$\mathcal A$ 的非零支撑是仿射空间，
$\mathcal P$ 的非零支撑是互补两点块的独立乘积，两类的支撑大小都是 $2$ 的幂。
所以前两项不可能，第三项也不成立。因此单 $q$ 问题已困难；
用实际构件代回把困难性传给整个 $\Gamma$。
\end{proof}

\begin{lemma}[两个中心的链式出口]\label{zh:lem:c1-general-two-centers}
设 $s=t=0$，恰有两个中心非零。置换端口后写成
\[
 q=aA+bB,\qquad
 A=\delta_{0101}+\delta_{1010},\quad
 B=\delta_{0110}+\delta_{1001},\qquad ab\ne0.
\]
若 $a=b$，则 $q=aN_{12}N_{34}$；若 $a=-b$，则 $q=aY_{12}Y_{34}$。
在其余情形，任意多次使用 $A$ 的实例均 Turing 归约到原问题，
从而整个函数集已有可判定二分。
\end{lemma}
\begin{proof}
按 $13\mid24$ 展平，有
\[
 Q=\diag(a,b,b,a)J,\qquad J=N\otimes N.
\]
长度为 $m$ 的两线 $N$ 链给出实际构件
\begin{equation}\label{zh:eq:c1-general-two-center-powers}
 q_m=a^mA+b^mB.
\end{equation}
设 $r=a/b$ 不是单位根。把目标函网中的 $M$ 个 $A$ 点结合为 $A^{\otimes M}$，
其余部分作为固定上下文。全部换成 $q_m$ 后，将输出除以 $b^{mM}$，便得到一个次数至多为 $M$ 的
多项式在 $r^m$ 上的值。取 $m=1,\ldots,M+1$，这些点互不相同；
普通 Vandermonde 插值的最高次系数就是目标实例值。
这给出 $O(M)$ 次查询，每次 $O(M^2)$ 份固定构件。
由此得到 $\Gamma\cup\{A\}$ 的任意多次 Turing 访问。
引理~\ref{zh:lem:cyc-complement-carrier} 分类包含互补两点四元 $A$ 的整个函数集，
归约复合后返回原问题。

现设 $r$ 是单位根且 $r\ne\pm1$。
改用 $12\mid34$ 展平，非零中间块是
$\left(\begin{smallmatrix}a&b\\b&a\end{smallmatrix}\right)$。
两份 $q$ 用两条 $N$ 边连接，所得中间块为
\begin{equation}\label{zh:eq:c1-general-root-conversion}
 \begin{pmatrix}a&b\\b&a\end{pmatrix}
 \begin{pmatrix}0&1\\1&0\end{pmatrix}
 \begin{pmatrix}a&b\\b&a\end{pmatrix}
 =\begin{pmatrix}2ab&a^2+b^2\\a^2+b^2&2ab\end{pmatrix}.
\end{equation}
若 $r=\pm i$，这已经是互补两点四元的非零倍数。
否则新的两个系数均非零，其比值之一为
\[
 \frac{a^2+b^2}{2ab}=\frac{r+r^{-1}}2.
\]
写 $r=e^{i\theta}$，该数为非零实数 $\cos\theta$，严格介于 $-1$ 与 $1$ 之间，
因而不是单位根。对这个实际构件应用上一段即可。
最后 $a=\pm b$ 的两个张量等式逐项核验成立。
\end{proof}

\begin{proof}[命题~\ref{zh:prop:c1-quaternary-interface} 的证明]
取一个真四元函数，用引理~\ref{zh:lem:c1-general-q4-shape} 写出它的形式。
两个非零端点由引理~\ref{zh:lem:c1-general-two-endpoints} 处理，
恰一个非零端点由引理~\ref{zh:lem:c1-general-jordan} 处理。
其余为纯 EO。
三个非零中心由引理~\ref{zh:lem:c1-general-three-centers} 给出困难性；
恰一个非零中心是互补两点四元，由引理~\ref{zh:lem:cyc-complement-carrier} 处理。
两个非零中心由引理~\ref{zh:lem:c1-general-two-centers} 处理：
$a=\pm b$ 时均可二二分解，与真四元条件矛盾；其余情形给出二分。
没有非零中心则函数为零，已被排除。
\end{proof}

\subsubsection{\texorpdfstring{$C_k$ 情形（$k\ge3$）}{Ck 情形（k≥3）}}\label{zh:sec:cyclic}
本节采用总框架的可用函数约定：构件及其非零张量因子都计入可用函数。
二元封闭条件是，全部非零可用二元函数均满秩，其行列式归一化代表的射影群为
$\Ggroup\cong C_k$，$k\ge3$。有限个二元方向已一次加入当前函数集。
证明中选用的其他函数也只有有限多个，按总框架同时加入后进行构件计算。
记 $\mu_k=\{z\in\mathbb C:z^k=1\}$ 为 $k$ 次单位根的集合。
有非迷向特征向量时，正交三角化和转置封闭把整个循环群化为对角矩阵。
若 $k$ 为偶数，其二阶陪集具有对称迹零代表，故已进入二元不等函数的已有二分。
以下分别在等边对角奇 $k$ 模型与迷向特征基给出的原生 $N$ 边模型中工作，
先讨论上集的降元，再讨论下集的四元出口，最后合并两集的结论。

称一个非零函数为\emph{不可分解}，若在任意端口置换及非零标量调整后，
都不能写成两个正元数函数的张量积。由分解定理，未解决情形中各个非零因子都可用。
因此，上集的四元函数既然能够二二分解，就必为本群二元函数的配对积。

\begin{theorem}[原生不等边模型的四元二分\cite{zh:liu2026disequality}]
\label{zh:thm:cyc-external-q4}
设有限代数函数集 $\Gamma$ 的边核为 $N$，且普通实现一个张量不可分解四元函数 $q$，满足
\[
 q_x=0\quad(|x|\text{ 为奇数}),\qquad q_{0000}q_{1111}\ne0.
\]
则整个 $\Gamma$ 问题有可判定的 $\FP/\sharpP$-hard 二分。
其易算条件是整个函数集满足文献的七项谓词 $\operatorname{Tract}(K\Gamma)$。
\end{theorem}
这里使用\cite{zh:liu2026disequality}的 Quaternary $K$-Holant dichotomy，
其证明不假设 $I$ 是可用顶点。矩阵 $K$ 与式~\eqref{zh:eq:K-basis} 完全一致。
若 $q$ 是已有构件，可把它和使用的有限个辅助加入函数集；问题普通构件等价。

\subparagraph*{迷向循环必须保留原生不等边}

若循环生成元 $R$ 的两条特征方向迷向，取 $K_0^{\mathsf T}K_0=N$ 的特征基，
并令 $D=K_0^{-1}RK_0$。
相似表示为 $D^j$，但实际二元点函数为
\begin{equation}\label{zh:eq:cyc-native-binary}
 A_j=K_0^{-1}R^jK_0^{-\mathsf T}=D^jN,
 \qquad NA_jN=ND^j,\qquad A_jN=D^j.
\end{equation}
三者分别是二元函数、两端口短路的有效矩阵和单端口传递矩阵。
因此这里的端口作用是对角加权，短路比较的是 $01,10$ 两个切片。
迷向分支均在此原生 $N$ 边模型中工作，允许 $k\ge3$ 为奇数或偶数。

对四元函数，两个不同的反对角相位探针及反对角二元输出，
迫使剩余相等输入上的两个系数分别为零。
每个四位奇重量串都能选出两个不同的比特且剩余两位相等，
因此所有四元构件均偶支撑。
写其一个自然展平为
\begin{equation}\label{zh:eq:cyc-native-q4}
 Q=\begin{pmatrix}
 s&0&0&a\\0&c&e&0\\0&f&d&0\\b&0&0&t
 \end{pmatrix}.
\end{equation}
三个有效反对角探针对中心块应用引理~\ref{zh:lem:oddq-root-ratio}，得到：
每个中心互补对同时为零或同时非零；非零时两项之比在 $\mu_k$ 中。
若两个中心对活跃，对应四项全非零，根比例引理的最后一支说明其乘积相等，
且四个值之间的比例都在 $\mu_k$ 中。
因此三个中心对全活跃时
\begin{equation}\label{zh:eq:cyc-native-center-products}
 ab=cd=ef=P\ne0.
\end{equation}
这些比例关系只涉及中心值，$s,t$ 留待四元分类时讨论。

\paragraph{上集}\label{zh:subsec:ck-upper}
上集的四元函数都是本群二元函数的配对积。
下面沿第一版的思路，对高元非配对积函数降元；相等边模型比较 $00,11$ 切片，
原生不等边模型比较 $01,10$ 切片，分别证明所需的配对结构。

\subparagraph*{相等边模型的降元}

先处理等边对角循环。证明分三步：先取得全零输入上的非零值，
再用三个对角辅助恢复切片的配对结构，最后合并三个视角，把整个函数拼回配对积。

\begin{lemma}[对角循环的降元]\label{zh:lem:cyc-diagonal-arity}
在本节等边对角循环的二元封闭条件下，$k\ge3$。
若有偶元可实现函数不属于本群的二元配对积类，
则有一个四元可实现函数不属于该类。
\end{lemma}
\begin{proof}
反设无此四元函数，取最小非配对积函数 $F$，元数 $2m\ge6$。
最小性取在当前固定函数集的普通构件中；所有元数更小的偶元可实现函数
均为零或本群满秩对角二元函数的配对积。
先说明三个非零对角配对积若满足一条系数均非零的线性关系，必须使用同一匹配。
否则，引理~\ref{zh:lem:odda-three-products} 的支撑计数迫使第三个支撑
恰为前两个的对称差；但全零输入属于每个纯对角配对支撑，
不属于任何两者的对称差，矛盾。

从 $F$ 的一个非零输入中选三个位置，其中至少两个位置 $p,q$ 同值。
故 $F_{pq}^{00},F_{pq}^{11}$ 至少一个非零。
它们不能恰有一个非零：否则任一对角短路得到该非零切片的倍数，
所以该切片是低元对角配对积。沿它的匹配逐对选择非零闭合的探针，
在原 $F$ 上作这些短路，会留下恰有一个非零对角值的二元函数；
完整二元闭包又保证其非对角项为零，违反满秩条件。

记这两个非零切片为 $A,B$。三个探针给出 $A+t_jB$。
若 $A,B$ 相关，它们都与某个非零低元对角配对积成比例，
因而 $A(0)\ne0$，即 $F(0^{2m})\ne0$。
若二者无关，三个短路均非零且两两不成比例，刚才的论证给出同一匹配。
由引理~\ref{zh:lem:odda-one-changing-factor}，它们除至多一个配对外有共同因子，
解回切片得到 $A=E_0Q$、$B=E_1Q$，其中 $Q$ 是实际满秩对角配对积，
$E_0,E_1$ 均非零。非零闭合 $Q$，在原 $F$ 上得到一个实际四元函数。
它的 $00$ 切片非零，所以不是零函数；低元条件迫使它为满秩对角配对积。
其全零值因此非零，故 $E_0(00)\ne0$，再次得到 $F(0^{2m})\ne0$。

此后，对任意两个端口重复上述三个探针及共同因子论证。
引理~\ref{zh:lem:odda-common-slice-factor}、\ref{zh:lem:odda-two-slice-forms}
的证明仅使用对角探针、共同匹配以及实际四元闭包，故在此全部适用。
每一对 $00,11$ 切片因此为成比例的完整配对，或同一对上的互补钉位配对。
再用引理~\ref{zh:lem:odda-three-views} 和~\ref{zh:lem:odda-two-rhinos}，
选出 $x,y,z,w$ 的三个视角及至少两份互补钉位配对。
这些证明也只使用对角闭合，不使用反对角端口作用。
由布尔性，$x=z$、$y=z$、$x=y$ 至少一个成立；相应切片条件遂迫使
$x=y$ 且 $z=w$。提出共同外因子并将其非零闭合，得到实际四元剩余因子。
该因子由低元条件仍为对角配对积，所以 $F$ 本身是配对积，矛盾。
这个收尾过程与定理~\ref{zh:thm:odd-arity-reduction} 相同；本段第一步取得全零非零值后，
其余步骤只需对角辅助。
\end{proof}

\subparagraph*{原生不等边模型的降元}\phantomsection\label{zh:sec:native-cyclic-arity}

原生不等边模型的短路比较 $01,10$ 两个切片。下面也分三步证明：
先把函数分成 EO 部分与两个端点；再利用短路的共同配对，把 EO 部分的逐点幂化成
无权匹配；最后排除八元例外，恢复原来的二元权重。
所得最小反例只剩匹配部分与端点，留待下一段处理。

本段在原生 $N$ 边模型中工作，二元封闭条件给出全部非零可用二元函数，忽略非零标量后为：
\[
 A_j=D^jN,\qquad
 D=\operatorname{diag}(\lambda,\lambda^{-1}),\qquad
 \lambda=e^{\pi i/k},\qquad k\ge3,
\]
其中 $j$ 遍历整数。这里不限制 $k$ 的奇偶。记 $\mathcal P_{\mathcal B}$ 为零函数与
实际可实现的满秩反对角二元函数的配对张量积，允许非零整体标量。
每个非零二元因子的两个反对角值之比属于 $\mu_k$。
连接一个 $A_j$ 的两个端口时，实际有效核为 $NA_jN$；随着 $j$ 变化，忽略非零标量，
这些核恰可写成
\begin{equation}\label{zh:eq:nca-probes}
 K_t=\begin{pmatrix}0&1\\t&0\end{pmatrix},\qquad t\in\mu_k.
\end{equation}
因此本节的实际短路形如 $F_{pq}^{01}+tF_{pq}^{10}$。

以下考虑非零普通可实现函数 $F$，元数为 $m=2d\ge6$，并假设所有普通可实现的更低偶元函数
均属于 $\mathcal P_{\mathcal B}$。最终将取最小非配对积函数以满足此假设。
下面的逐点幂、切片及 EO 限制均用来分析函数值，构件仍按指定的二元辅助连接。

\subparagraph*{中间层与配对一致性}

\begin{lemma}[只余 EO 层与两个端点]\label{zh:lem:nca-eo-endpoints}
存在一个 EO 函数 $Q$ 及常数 $a,b$，使
\[
 F=Q+a\delta_{0^m}+b\delta_{1^m}.
\]
这里允许 $Q=0$。
\end{lemma}

\begin{proof}
任取同时含 $0,1$ 且重量不是 $d$ 的输入。在其中选一对取值相反的端口 $p,q$。
删去这两个比特后，剩余输入仍不在相应 EO 层上。每个实际短路
$F_{pq}^{01}+tF_{pq}^{10}$ 是低元反对角配对积或零，因而为 EO 函数。
在该剩余输入处，用两个不同参数 $t$ 得到两个零值方程，迫使两个切片值各自为零。
所以所有既非 EO 又非全零、全一的函数值都为零。
\end{proof}

\begin{lemma}[循环探针的共同配对]\label{zh:lem:nca-common-matching}
固定任意两个端口 $p,q$。所有非零实际短路
$F_{pq}^{01}+tF_{pq}^{10}$，$t\in\mu_k$，具有相同支撑与匹配。
\end{lemma}

\begin{proof}
写 $A=F_{pq}^{01}$、$B=F_{pq}^{10}$。若两者相关，所有非零结果成比例。
以下设两者无关，则 $k$ 个短路均非零且射影方向互异。

若 $k\ge4$，引理~\ref{zh:lem:odda-three-products} 说明：一旦同一直线上的三个
配对积出现不同支撑，该直线上恰好只有三个配对积射影点，其余非零点在局部四变量上
有六点支撑。因此不可能同时存在至少四个不同的配对积射影点，结论成立。

只剩 $k=3$。假设三个结果 $U,V,W$ 出现四变量例外，写其非退化关系为
$\alpha U+\beta V+\gamma W=0$。每个实际反对角二元因子的两个值的三次幂相同，
所以 $U,V,W$ 在各自支撑上的逐点三次幂分别为非零常数 $u,v,w$。
例外中，三个两两支撑交集均非空，而三重交集为空。
分别在只属于 $U,V$、只属于 $U,W$、只属于 $V,W$ 的位置上代入关系并取三次幂，得
\[
 \alpha^3u=-\beta^3v,\qquad
 \alpha^3u=-\gamma^3w,\qquad
 \beta^3v=-\gamma^3w.
\]
这些等式与三者非零矛盾。因此三次根情形的例外也不可能出现。
\end{proof}

\begin{lemma}[两个反向切片的局部形式]\label{zh:lem:nca-slice-forms}
对任意端口 $p,q$，或者 $F_{pq}^{01}=F_{pq}^{10}=0$，或者可选剩余变量中的一对
$z,w$ 并写成
\begin{equation}\label{zh:eq:nca-slice-factors}
 F_{pq}^{01}=E_0(z,w)Q_{\mathrm{out}},\qquad
 F_{pq}^{10}=E_1(z,w)Q_{\mathrm{out}},
\end{equation}
其中 $Q_{\mathrm{out}}$ 是实际满秩反对角二元因子的配对积，并且恰有两种非零形式：
\begin{enumerate}
 \item $E_0=\eta_0 E$、$E_1=\eta_1 E$，其中 $E$ 为满秩反对角函数，
 其两个非零值之比属于 $\mu_k$，且 $\eta_1/\eta_0\in\mu_k$；
 \item $E_0,E_1$ 分别只在 $01,10$ 非零，或分别只在 $10,01$ 非零，
 其两个非零值之比属于 $\mu_k$。
\end{enumerate}
两种形式中的所有显示非零系数均非零。
\end{lemma}

\begin{proof}
若两个切片相关且不都为零，取一个非零短路的实际配对分解并提出除一对外的因子。
若切片无关，三个短路的共同配对由引理~\ref{zh:lem:nca-common-matching} 保证，
再用引理~\ref{zh:lem:odda-one-changing-factor}，至多一个配对上的因子变化。
两种情形都得到式~\eqref{zh:eq:nca-slice-factors}，其中 $E_0,E_1$ 暂为可能退化的
反对角函数，不都为零。

对 $Q_{\mathrm{out}}$ 的每个因子选一个有效核 $K_t$，使其闭合值非零。
闭合值是非零一次式，至多一个 $t$ 使其为零。把这些实际辅助连到原函数 $F$ 上，
得到非零四元函数 $H(p,q,z,w)$。其 $01,10$ 切片是同一非零常数乘 $E_0,E_1$。
低元假设使 $H$ 为实际反对角配对积。

若 $H$ 的匹配为 $pq|zw$，两切片为同一满秩 $zw$ 因子的倍数；
其比例由 $pq$ 因子的两个反对角值给出，属于 $\mu_k$。
若匹配交叉，则 $p,q$ 取相反值时，$z,w$ 被分别固定为两个互补的反向输入。
对应函数值之比是两个二元因子的比值之积，也属于 $\mu_k$。
这证明两种形式，并同时排除了恰有一个切片为零的可能。
\end{proof}

\subparagraph*{逐点幂的全部不等短路}

对任意函数 $f$，记 $f^{[k]}(x)=f(x)^k$，表示逐点幂而非矩阵幂。
以下先设引理~\ref{zh:lem:nca-eo-endpoints} 中的 $Q\ne0$，并令
\[
 S=Q^{[k]}.
\]
它是非零 EO 函数。记 $\partial^{pq}$ 为无权二元不等短路，与
第~\ref{zh:sec:c1-strict}~节一致。

\begin{lemma}[逐点幂的短路为匹配]\label{zh:lem:nca-power-caps}
对每个 $p\ne q$，存在非零常数 $\theta_{pq}$ 及剩余变量上的完美匹配 $M_{pq}$，使
\begin{equation}\label{zh:eq:nca-all-power-caps}
 \partial^{pq}S=\theta_{pq}P_{M_{pq}},\qquad
 P_M=\prod_{\{u,v\}\in M}[x_u\ne x_v].
\end{equation}
\end{lemma}

\begin{proof}
端点项在 $01,10$ 切片上为零，因此这些切片同时也是 $Q$ 的相应切片。于是
\[
 \partial^{pq}S=(F_{pq}^{01})^{[k]}+(F_{pq}^{10})^{[k]}.
\]
这不是将短路和整体取幂交换，而是两个切片逐点幂之和。
若两切片皆零，短路为零。否则用引理~\ref{zh:lem:nca-slice-forms}。
$Q_{\mathrm{out}}^{[k]}$ 是非零常数乘无权不等匹配。
在第一种形式中，$E^{[k]}$ 是非零常数乘 $N$，而
$\eta_0^k=\eta_1^k\ne0$；相加的系数为 $2\eta_0^k$，非零。
在第二种形式中，两个非零钉位值的 $k$ 次幂相同且非零，相加也是非零常数乘 $N$。
所以每个短路为零或非零常数乘无权不等匹配。

现在排除零短路。若 $\partial^eS=0$，而与 $e$ 不交的一对端口 $f$ 满足
$\partial^fS\ne0$，则交换短路次序得到
\[
 0=\partial^f\partial^eS=\partial^e\partial^fS.
\]
右端却是非零不等匹配积的一次不等短路，仍非零：沿原匹配边短路产生因子 $2$，
跨两条匹配边短路则合并这两条边且产生因子 $1$。
因此一个零短路迫使所有与它不交的短路都为零。

所有二元端口对组成的“不交”图连通。事实上，两对若不交就相邻；若相交且不同，
它们的并只有三个端口，因 $m\ge6$，可在并外选一对，得到长度二的路径。
故一个零短路会迫使全部短路为零。
非零 EO 函数不可能全部短路为零：在一个非零输入中选三个位置为 $100$，
把其三个置换位置的值记为 $a,b,c$，其中 $a\ne0$；相应短路值为
$a+b,a+c,b+c$，不能全为零。于是所有短路均非零，得到
式~\eqref{zh:eq:nca-all-power-caps}。
\end{proof}

\begin{lemma}[有理化与翻转对称]\label{zh:lem:nca-rational-power}
存在非零常数 $\theta$，使 $R=S/\theta$ 是非零有理值、0--1 翻转对称 EO 函数，
且其每个不等短路仍为非零常数乘无权不等匹配。
\end{lemma}

\begin{proof}
对不交端口对 $e,f$ 交换两次短路。匹配积的一次短路只产生系数 $1$ 或 $2$，所以
式~\eqref{zh:eq:nca-all-power-caps} 给出
\[
 \theta_e c_e P_{M'}=\theta_f c_f P_{M''},\qquad c_e,c_f\in\{1,2\}.
\]
两边非零，支撑相等，因而匹配相等且系数相等。于是
$\theta_e/\theta_f\in\{1/2,1,2\}$。
沿上述不交图取路径，选定一个 $e_0$ 并令 $\theta=\theta_{e_0}$，便得
$\theta_e/\theta\in2^{\mathbb Z}$ 对所有 $e$ 成立。
因此 $R=S/\theta$ 的全部短路都是有理值。

把所有 EO 输入上的值作为未知数，所有不等短路组成一个有理系数线性系统。
这个联合线性映射在 EO 子空间上单射：其核中的非零 EO 函数会与上面的三个局部和
论证矛盾。系统有唯一解且右端有理，所以高斯消元说明 $R$ 本身有理。

每个无权不等匹配积均翻转对称，而不等短路与逐位翻转交换。
因此 $R(x)-R(\bar x)$ 的全部短路为零；它仍是 EO 函数，单射性遂给出
$R(x)=R(\bar x)$。
\end{proof}

\subparagraph*{排除八元例外并恢复原函数的权重}

由引理~\ref{zh:lem:nca-rational-power}，可以对 $R$ 使用
定理~\ref{zh:thm:c1-eo-matching-reduction}。结论是 $S$ 为非零常数乘一个无权不等匹配，
或者在变量置换后为八元例外 $F'_8$ 的非零倍数。

\begin{lemma}[八元例外不可能出现]\label{zh:lem:nca-exclude-f8}
这里 $S$ 不可能属于八元例外。因此存在非零常数 $c$ 和完美匹配 $M$，使
$S=cP_M$。
\end{lemma}

\begin{proof}
式~\eqref{zh:eq:c1-f8} 的十四个支撑输入可以视为：在 $\mathbb F_2^3$ 的八个点上，
所有非恒定仿射函数 $a+b\cdot v$ 的真值表。为验证这一表示，四个线性奇偶校验方程独立，
其解空间维数为四；全部十六个仿射真值表都满足这些方程，因此恰好给出该空间。
去掉重量为零和八的两个常数函数，便留下十四个重量为四的输入。

固定任意两个不同坐标 $v,w$ 为 $0,1$，在四个仿射系数上施加两个独立方程，
恰留下四个真值表。任意第三坐标 $u$ 的求值线性式不在前两个求值式的线性张成中：
它的常数项系数为一，排除了零及前两式的和，而 $u\ne v,w$ 排除了前两式本身。
所以每个剩余坐标在这四张表上都能取零和一。

若 $S$ 是 $F'_8$ 的非零倍数，逐点幂不改变支撑，所以 $Q$ 也具有该支撑。
其任一 $01$ 切片有四个支撑输入，且六个剩余端口均无固定值。
另一方面，引理~\ref{zh:lem:nca-slice-forms} 把该切片写成 $E_0Q_{\mathrm{out}}$。
此时 $Q_{\mathrm{out}}$ 有两个满秩二元因子，支撑大小为四。故 $E_0$ 的支撑只能有
一个输入，从而固定了两个剩余端口，与刚才的性质矛盾。
\end{proof}

\begin{lemma}[恢复实际加权匹配]\label{zh:lem:nca-recover-weighted}
若 $Q\ne0$，则 $Q\in\mathcal P_{\mathcal B}$，其中的二元因子确实普通可实现。
\end{lemma}

\begin{proof}
由 $S=cP_M$ 和逐点幂的支撑保持性，$Q$ 的支撑恰为 $P_M$ 的完整支撑。
对原匹配 $M$ 中的任意一对 $p,q$，两个切片 $F_{pq}^{01}=Q_{pq}^{01}$ 和
$F_{pq}^{10}=Q_{pq}^{10}$ 都具有剩余匹配的完整支撑，大小为 $2^{d-1}$。
在引理~\ref{zh:lem:nca-slice-forms} 中，外部因子支撑大小为 $2^{d-2}$，
所以局部因子必须满秩，不能属于互补钉位支。于是这两个切片相互成比例。

以每条匹配边的两个方向为一个逻辑比特。刚得到的比例说明，改变任一逻辑比特时，
函数值之比与其余逻辑比特无关。因完整匹配支撑上的值都非零，可逐比特改变一个固定输入，
得到
\[
 Q=\gamma\prod_{e\in M}B_e,\qquad \gamma\ne0,
\]
其中各 $B_e$ 为满秩反对角函数。局部形式已经给出其比值属于 $\mu_k$。

还可直接验证因子的实际可实现性。在原函数 $F$ 上，沿 $M$ 的其余匹配边选取
闭合值非零的有效核 $K_t$。每个被闭合的反对角核都杀掉全零、全一端点项，
因此留下 $B_e$ 的非零倍数，标量是其他因子的非零闭合值之积。
这确为普通二元构件；完整二元闭包保证其属于实际二元集合。
因此 $Q\in\mathcal P_{\mathcal B}$。
\end{proof}

\begin{theorem}[原生循环最小反例的形状]\label{zh:thm:native-cyclic-shape}
在本节二元封闭条件下，假设每个非零普通四元构件均属于
$\mathcal P_{\mathcal B}$。若存在普通可实现的偶元非配对积函数，取其中元数最小者 $F$。
则其元数为偶数 $m\ge6$，并且
\[
 F=Q+a\delta_{0^m}+b\delta_{1^m},
\]
其中 $Q=0$ 或 $Q$ 是实际加权反对角二元配对积，且 $a,b$ 至少一个非零。
此结论对所有 $k\ge3$ 成立。
\end{theorem}

\begin{proof}
最小性使所有更低偶元普通可实现函数属于 $\mathcal P_{\mathcal B}$，而零元、二元和四元
已排除，所以 $m\ge6$。引理~\ref{zh:lem:nca-eo-endpoints} 给出显示的分层。
若 $Q=0$，无需取逐点幂或调用 EO 结构定理；此时 $F$ 已是纯端点函数。
若 $Q\ne0$，依次应用上述配对一致性、逐点幂短路、有理化、八元排除及权重恢复引理，
得到 $Q\in\mathcal P_{\mathcal B}$。
若再有 $a=b=0$，则 $F=Q\in\mathcal P_{\mathcal B}$，与选择矛盾。
\end{proof}

\subparagraph*{端点部分的两份构件与多线插值}\phantomsection\label{zh:subsec:native-cyclic-tail}

继续使用原生 $N$ 边模型，二元封闭条件给出的函数为 $D^jN$。
端口上连接实际二元函数 $D^jN$ 的传递矩阵是 $D^j$，故对任意
$r\in\mu_k$，$\diag(1,r)$ 都是可用的单端口传递，允许已知非零标量。
以下两条引理处理定理~\ref{zh:thm:native-cyclic-shape} 得到的端点部分。
两条引理都适用于任意 $k\ge3$。

\begin{lemma}[两端点的四元支撑矛盾]\label{zh:lem:native-cyclic-two-endpoints}
设所有非零普通四元构件均属于反对角二元配对积类
$\mathcal P_{\Bgroup}$。若一个 $2d\ge6$ 元普通构件满足
\begin{equation}\label{zh:eq:native-cyclic-tail-shape}
 F=Q+a\delta_{0^{2d}}+b\delta_{1^{2d}},
\end{equation}
其中 $Q=0$，或者在某个匹配上
\[
 Q=\prod_{i=1}^{d}E_i(x_i,y_i),\qquad
 E_i=\begin{pmatrix}0&u_i\\v_i&0\end{pmatrix},\qquad u_iv_i\ne0,
\]
则 $ab=0$。这里排除双端点的结论不需要各 $E_i$ 的根单位比例条件。
\end{lemma}
\begin{proof}
先设 $Q\ne0$。取两份整构件 $F$，各留下第一个匹配对的两个端口；
对 $i\ge2$，把两份的 $x_i$ 用一条 $N$ 边连接，$y_i$ 也对应连接。
四个外端口依次为 $x_1,y_1,x'_1,y'_1$。
$Q$ 与 $Q$ 配对时，第 $i$ 个内部匹配对贡献
\[
 \sum_{x,y}E_i(x,y)E_i(1-x,1-y)=2u_iv_i.
\]
记 $C=\prod_{i=2}^{d}2u_iv_i\ne0$，这部分输出为 $CE_1\otimes E_1$。
两个相反端点的配对贡献 $ab(\delta_{0011}+\delta_{1100})$；
相同端点不能跨 $N$ 边连接。一个端点与一个 $Q$ 的混合项也为零：
至少一个完整内部匹配对被迫为 $00$ 或 $11$，其反对角因子为零。
因此实际四元输出恰为
\begin{equation}\label{zh:eq:native-cyclic-tail-six-support}
 H=CE_1\otimes E_1+ab(\delta_{0011}+\delta_{1100}).
\end{equation}
第一项恰有 $0101,0110,1001,1010$ 四个非零支撑点，
与第二项的 $0011,1100$ 不相交。
若 $ab\ne0$，$H$ 有六个支撑点，而任意两个满秩反对角因子的匹配积恰有四个，矛盾。

$Q=0$ 时，任意选取两对外端口作同样连接，输出为
$ab(\delta_{0011}+\delta_{1100})$。
$ab\ne0$ 时这有两个支撑点，同样不是实际二元配对积。
因此两种情形均有 $ab=0$。
\end{proof}

\begin{lemma}[任意偶元单端点的插值出口]\label{zh:lem:native-cyclic-endpoint-interpolation}
设有限代数偶元函数集 $\Gamma$ 的边核为 $N$，完整普通二元闭包的非零方向恰为 $D^jN$。
若 $\Gamma$ 普通实现式~\eqref{zh:eq:native-cyclic-tail-shape} 中恰有一个非零端点的函数，
其中 $Q=0$，或者 $Q$ 是满秩反对角匹配积且各因子的非零项之比属于 $\mu_k$，
则整个 $\Gamma$ 问题已有可判定的 $\FP/\sharpP$-hard 二分。
本引理不假设所有普通四元构件均为配对积。
\end{lemma}
\begin{proof}
先设 $Q\ne0$、$a\ne0$、$b=0$。
在每对的第一个端口作用实际可用的传递
$T_i=\diag(1,u_i/v_i)$，则 $T_iE_i=u_iN$。
作用于整份 $F$ 后除去已知非零标量 $c=\prod_i u_i$，得到
\[
 F_* =\prod_{i=1}^{d}N(x_i,y_i)+s\delta_{0^{2d}},\qquad s=a/c\ne0,
\]
其中普通实现允许再带一个已知非零标量。
若最初仅知 $Q$ 的所有非零值具有相同的 $k$ 次幂，
固定其他匹配对而反转第 $i$ 对即可推出 $(u_i/v_i)^k=1$，故同样允许这些传递。

按每个匹配对一端为行、另一端为列展平。记
$J=N^{\otimes d}$、$z=0^d$、$o=1^d$，矩阵表达为
\[
 F_*=J+sE_{z,z},\qquad F_*J=I+sE_{z,o},\qquad E_{z,o}^{\,2}=0.
\]
串联 $\ell$ 份整构件，每个相邻连接用 $d$ 条原生 $N$ 边，把列端口与下一份的行端口对应连接。
实际输出为
\begin{equation}\label{zh:eq:native-cyclic-multiline-chain}
 (F_*J)^{\ell-1}F_*=(I+\ell sE_{z,o})J=J+\ell sE_{z,z}.
\end{equation}
构件链中的顶点都是整份 $F_*$，相邻两份之间沿原生 $N$ 边连接。
若实际标准构件是 $\kappa F_*$，长度 $\ell$ 的链还带已知非零因子 $\kappa^\ell$。

给定含 $M$ 处目标 $\delta_{0^{2d}}$ 的实例，先把这些目标点与其余固定函网分开。
把每一处都换成长度 $\ell$ 的链，对这 $M$ 个位置作多线性展开。
询问答案除去 $\kappa^{\ell M}$ 后，是 $\ell$ 的至多 $M$ 次多项式，
其 $\ell^M$ 系数恰为 $s^M$ 乘目标实例的答案。
取 $\ell=1,\ldots,M+1$，用 $M+1$ 次询问插值最高次系数并除以 $s^M$。
每次替换使用 $O(M^2)$ 份固定构件。因此
\begin{equation}\label{zh:eq:native-cyclic-endpoint-access}
 \Holant(N\mid\Gamma,\delta_{0^{2d}})
 \equiv_T\Holant(N\mid\Gamma).
\end{equation}
反向归约因原函数集保留而直接成立。

若非零端点为 $1^{2d}$，同样归一后得到 $J+sE_{o,o}$，
此时 $F_*J=I+sE_{o,z}$ 且 $E_{o,z}^{\,2}=0$。
同样串联与插值便取得 $\delta_{1^{2d}}$。
若 $Q=0$，所需单端点张量已经普通可用，无需插值。

至此可以加入一个非零一元函数的张量幂而保持 Turing 等价。
用式~\eqref{zh:eq:K-basis} 的固定 $K$ 变换把整个问题写回相等边模型，
新增函数仍是非零一元函数的张量幂。
由分解定理~\ref{zh:thm:decomposition}，或者问题属于 $\FP$，或者可以加入这个一元因子。
后一情形由奇元二分定理~\ref{zh:thm:odd-arity-fp} 完成。
\end{proof}

\paragraph{下集}\label{zh:subsec:ck-lower}
下集已经有真四元函数。分别写出两种边模型中的四元形式，利用可用二元函数调整系数，
再构造相等函数、钉子或满足已有四元二分条件的函数。

\subparagraph*{等边奇阶对角群的实际四端口归一化}

以下先令 $k\ge3$ 为奇数。忽略已知非零标量，可用对角辅助为
$D_r=\diag(1,r)$，$r\in\mu_k$。
四元函数 $q$ 的所有二元短路均对角，两个独立对角探针因而迫使奇重量输入为零。
记
\[
 s=q_{0000},\quad t=q_{1111},\qquad
 a_I=q_{\mathbf1_I},\quad b_I=q_{\mathbf1_{I^c}},\qquad I=12,13,14.
\]
对每个互补对，在两种方向上使用全部 $D_r$，应用
引理~\ref{zh:lem:oddq-root-ratio}。它给出：$s,t$ 不能恰有一个非零；
每个中心对同时为零或同时非零；非零时 $b_I/a_I\in\mu_k$。
若 $st\ne0$，还有 $t/s\in\mu_k$，且每个非零中心对满足 $a_Ib_I=st$，
其中全部非零数与 $s$ 的比例都在 $\mu_k$ 中。

\begin{lemma}[奇阶四端口平衡]\label{zh:lem:cyc-balance}
只用实际可用的 $D_r$ 在四个端口上的作用，可以使 $q$ 的每个互补对数值相等。
若端点非零，还可使两端点相等，并且每个非零中心对的公共值都等于端点值。
\end{lemma}
\begin{proof}
四个端口分别使用 $D_{r_i}$，记 $R=\prod_i r_i$、$u_I=\prod_{i\in I}r_i$。
若端点非零，令 $R=s/t$；否则令 $R=1$。每个活跃中心对取
\[
 u_I^2=(b_I/a_I)R.
\]
奇数阶群 $\mu_k$ 中平方是双射，故可以选择这些 $u_I$；非活跃对任取。
任何指定的 $(R,u_{12},u_{13},u_{14})\in\mu_k^4$ 都有实际四端口解：
\[
 r_1^2=u_{12}u_{13}u_{14}/R,\qquad
 r_2=u_{12}/r_1,\quad r_3=u_{13}/r_1,\quad r_4=u_{14}/r_1.
\]
于是新的中心值为 $a_Iu_I=b_IR/u_I$，端点为 $s,tR$。
若 $st\ne0$，每个中心公共值 $c$ 满足 $c^2=s^2$，且 $c/s\in\mu_k$。
由于 $-1\notin\mu_k$，只能有 $c=s$。所有作用都是实际二元连接。
\end{proof}

设平衡后的函数为 $h$，端点公共值为 $s$，三个中心公共值为 $a,b,c$，零对记零。
对整个函数集同步作 $K^{-1}$，边核从 $I$ 变为 $N$。
实际新函数 $h'=(K^{-1})^{\otimes4}h$ 满足
\begin{equation}\label{zh:eq:cyc-K-fourier}
 h'_y=\frac14\sum_x h_x(-i)^{|x|}(-1)^{x\cdot y},\qquad
 h'_{0000}=h'_{1111}=\frac{s-a-b-c}{2}.
\end{equation}
因为 $h_x=h_{\bar x}$ 且 $h$ 偶支撑，互补项在奇重量 $y$ 上相消，故 $h'$ 仍偶支撑。
逐端口可逆变换保持张量不可分解性。

\begin{proposition}[等边奇阶循环的四元格]\label{zh:prop:cyc-equality-q4}
在本节的二元封闭条件下，设等边对角群为奇阶 $C_k$，$k\ge3$。
若有一个四元函数不是本群二元函数的配对积，则整个问题有复杂性二分。
\end{proposition}
\begin{proof}
先设 $st\ne0$，有 $r$ 个活跃中心对。引理~\ref{zh:lem:cyc-balance} 后，
式~\eqref{zh:eq:cyc-K-fourier} 的新端点为 $s(1-r)/2$。
$r=0,2,3$ 时非零；这些格中，若有正元数分解，
两个端点及偶支撑迫使它为两个满秩对角二元因子的乘积，支撑只能有四点，
因此这些格都是张量不可分解的。适用定理~\ref{zh:thm:cyc-external-q4}。
具体地，偶支撑和每个变量同时取两值排除一元因子；任意 $2+2$ 分解的
两个因子必须各有固定奇偶，否则它们独立取值会产生奇支撑输入。
两个端点又迫使它们都为满秩对角函数。
$r=1$ 时 $a_Ib_I=st$ 已给出两个可用对角二元因子的乘积，不是当前非配对积格。

再设 $s=t=0$。一个活跃中心对给出互补两点支撑，平衡后新端点为 $-a/2\ne0$，
适用同一定理。两个活跃中心对的公共值记为 $a,b$。
若 $a^2\ne b^2$，则 $a+b\ne0$，相应的唯一候选 $2+2$ 分解行列式非零，
所以仍是张量不可分解的四元入口。
若 $a^2=b^2$，适当端口次序下为非零倍的 $N\otimes N$ 或 $Y\otimes Y$。
分解后 $N$ 或 $Y$ 也是可用二元函数，而本群的二元函数全为对角矩阵，
与二元封闭条件矛盾。

最后排除三个中心对均活跃的格。平衡后 $a,b,c$ 均非零。
固定一个 $2+2$ 展平及两份构件的对应端口方向，写为
\[
 Q_a=\begin{pmatrix}
 0&0&0&a\\0&b&c&0\\0&c&b&0\\a&0&0&0
 \end{pmatrix}.
\]
两份的前两端口分别由普通相等边对应连接，保留后两端口，实际输出为
\[
 Q_a^{\mathsf T}Q_a=
 \begin{pmatrix}
 a^2&0&0&0\\0&b^2+c^2&2bc&0\\
 0&2bc&b^2+c^2&0\\0&0&0&a^2
 \end{pmatrix}.
\]
两个端点及中心对 $(2bc,2bc)$ 非零，根单位比例引理迫使 $a^4=4b^2c^2$。
另外两种配对同理给出 $b^4=4a^2c^2$、$c^4=4a^2b^2$。
相乘并除去非零 $(abc)^4$ 得 $1=64$，矛盾。
因此三个中心对全非零的情形也与二元封闭条件矛盾。
\end{proof}

\subparagraph*{任意互补两点四元函数的载体出口}

\begin{lemma}[互补两点载体]\label{zh:lem:cyc-complement-carrier}
设原生 $N$ 边有限代数函数集 $\Gamma$ 普通实现
\[
 g=c_0\delta_x+c_1\delta_{\bar x},\qquad
 x\in\{0,1\}^4,\quad c_0c_1\ne0.
\]
则整个函数集已有可判定的 $\FP/\sharpP$-hard 二分。
这里 $x$ 可以重量为二，不要求 $g$ 有非零端点。
\end{lemma}
\begin{proof}
使用\cite{zh:liu2026disequality}的 Universal one-pair reduction 和
Single-occurrence replication，引入仅作 Turing 载体的有序函数
$\Delta=\delta_0\otimes\delta_1$。
第一条引理对任意有限函数集保证至多一个 $\Delta$ 的实例可归约到原问题；
先以这一份 $\Delta$ 开始构造。

把 $\Delta$ 的两个端口经 $N$ 边分别接到两份 $g$ 的同一编号端口。
两份 $g$ 在该端口被迫取相反值，故选定互补的两项。
余下六个外端口恰有三个 $0$、三个 $1$，系数为 $c_0c_1$。
固定重排后得到
\[
 c_0c_1\,\Delta(P)\otimes\Delta^{\otimes2}(H).
\]
$P$ 保留为有序两端口载体。只把这个保留块送入下一个同型上下文，
不识别其他外端口；嵌套 $\ell$ 次便以唯一的初始 $\Delta$ 实现
\[
 (c_0c_1)^\ell\Delta(P)\otimes
 \bigotimes_{j=1}^{\ell}\Delta^{\otimes2}(H_j),
\]
大小 $O(\ell)$。$\Delta$ 两端用一条 $N$ 边闭合的值为 $1$，
故满足单次复制引理的非零闭实例条件。
因此任意多次 $\Delta^{\otimes2}$ 的使用均 Turing 归约到 $\Gamma$。
具体地，先给目标实例并上一份由 $\Delta$ 的 $N$ 自环组成的闭实例 $J$，
其值为 $1$。从这份唯一载体出发嵌套复制，填入全部目标位置；
调用单次 $\Delta$ 的归约后，除去 $(c_0c_1)^\ell Z(J)$ 即得原答案。

至此可以加入 $\Delta^{\otimes2}$ 而保持 Turing 等价。
用固定 $K$ 变换把整个问题写回相等边模型，新增函数仍是非零一元函数的张量积。
由分解定理~\ref{zh:thm:decomposition}，或者问题属于 $\FP$，或者可加入这些一元因子。
后一情形由奇元二分定理~\ref{zh:thm:odd-arity-fp} 完成。
\end{proof}

\subparagraph*{单端点四元函数的两线链插值}

\begin{lemma}[单端点插值出口]\label{zh:lem:cyc-native-one-endpoint}
设原生 $N$ 边函数集的完整普通二元闭包，其非零方向恰为式~\eqref{zh:eq:cyc-native-binary} 所列，
若四元构件恰有一个非零端点且中心最多有两个活跃互补对，
则整个函数集已有可判定的二分。
\end{lemma}
\begin{proof}
以非零端点 $0000$ 为例，$1111$ 同理。
没有中心项时已是非零倍的 $\delta_{0000}$。
一个活跃中心对的两值之比在 $\mu_k$，可用已有对角端口传递归一化并置换端口，得到
\[
 q=s\delta_{0000}+\delta_{0011}+\delta_{1100},\qquad s\ne0.
\]
按 $12\mid34$ 展平，在 $\{00,11\}$ 子块中
\[
 Q=\begin{pmatrix}s&1\\1&0\end{pmatrix},\qquad
 J=(N\otimes N)|_{\{00,11\}}=\begin{pmatrix}0&1\\1&0\end{pmatrix}.
\]
把 $m$ 份构件串成链，每次以两条 $N$ 边连接前一份的列端口和后一份的行端口，
输出为 $(QJ)^{m-1}Q$。直接计算给出
\[
 (QJ)^{m-1}Q=\begin{pmatrix}ms&1\\1&0\end{pmatrix}.
\]

两个中心对活跃时，四个中心值的秩一关系和 $\mu_k$ 比例，
给出两个可用反对角二元因子的乘积。
利用实际对角传递在端口上归一化，得到数值表达
\[
 q=s\delta_{0000}+N_{12}N_{34}.
\]
下面直接把整份 $q$ 串成链。
改按 $13\mid24$ 展平，令 $J=N\otimes N$，有
\[
 Q=J+sE_{00,00},\qquad QJ=I+sE_{00,11},\qquad E_{00,11}^2=0.
\]
同样的两线链实际给出 $(QJ)^{m-1}Q=J+msE_{00,00}$。
故两种情形均以 $O(m)$ 个顶点实现 $q_m=ms\delta_{0000}+q_0$，其中 $q_0$ 固定。
若标准构件实际实现 $c q$，则长度 $m$ 的链实现 $c^m q_m$。
下述含 $M$ 个替换的询问须先除去已知非零标量 $c^{mM}$。

考虑含 $M$ 个目标 $\delta_{0000}$ 的实例，先把这些目标点与其余固定函网分开。
把全部目标替换为 $q_m$，对这 $M$ 个位置作多线性展开。
答案是 $m$ 的至多 $M$ 次多项式，最高次系数恰为
$s^M$ 乘目标实例答案。
取 $m=1,\ldots,M+1$，作 $M+1$ 次询问，每次替换大小 $O(M^2)$，
插值恢复最高次系数并除以 $s^M$。
由此得到任意多次 $\delta_{0000}$ 的 Turing 访问。
将整个问题变回相等边模型后，新增函数仍是非零一元函数的张量幂。
由分解定理~\ref{zh:thm:decomposition}，或者问题属于 $\FP$，或者可加入这个一元因子，
再由奇元二分定理~\ref{zh:thm:odd-arity-fp} 完成。
\end{proof}

\begin{lemma}[零端点乘积的三块矛盾]\label{zh:lem:cyc-native-three-centers}
在本节原生 $N$ 边模型的二元封闭条件下，$st=0$ 的四元函数不可能有三个活跃中心互补对。
\end{lemma}
\begin{proof}
采用式~\eqref{zh:eq:cyc-native-q4} 的端口方向。
两份构件的行端口由两条 $N$ 边对应相连，保留列端口，实际输出
$Q^{\mathsf T}(N\otimes N)Q$。
由式~\eqref{zh:eq:cyc-native-center-products}，其三个中心互补对的乘积是
\[
 (st+P)^2,\qquad (2cf)(2de)=4P^2,\qquad (cd+ef)^2=4P^2.
\]
$st=0$ 时三对都活跃，再次应用中心比例刚性，应有 $P^2=4P^2$，矛盾。
这里两份函数之间的连接矩阵为 $N\otimes N$。
\end{proof}

\begin{proposition}[原生不等边循环的四元出口]\label{zh:prop:cyc-native-q4}
在本节原生 $N$ 边模型的二元封闭条件下，$k\ge3$。
若有一个四元函数不是本群二元函数的配对积，则整个问题有可判定的复杂性二分。
\end{proposition}
\begin{proof}
四元函数已为偶支撑。若两个端点均非零且张量不可分解，直接用
定理~\ref{zh:thm:cyc-external-q4}。
若可分解，两个非零端点排除一元因子，偶支撑迫使它分成两个满秩对角二元因子。
分解后这两个因子也可用，与本群二元函数全为反对角矩阵矛盾。

恰有一个非零端点时，三中心对已由引理~\ref{zh:lem:cyc-native-three-centers} 排除，
其余由引理~\ref{zh:lem:cyc-native-one-endpoint} 完成。
两端点均为零时：零中心是零函数；一个活跃中心对由
引理~\ref{zh:lem:cyc-complement-carrier} 完成；两个活跃中心对按中心秩一关系已是
两个实际可用反对角因子的乘积，与前提矛盾；三个活跃中心对再次由三块矛盾排除。
\end{proof}

\subparagraph*{只有 EO 四元函数的例子}
\begin{example}\label{zh:ex:cyclic-no-endpoint}
在原生 $N$ 边模型中，考虑
$\Gamma=\{=_6\}\cup\{D^jN:0\le j<k\}$，其中
$D=\diag(e^{\pi i/k},e^{-\pi i/k})$。
压缩二元顶点后，相等顶点之间形成带非零权重的不等边。
每个连通分量若不满足二染色条件，贡献为零；否则只有两个全局取值，
可直接计算，所以这个问题在 $\FP$ 中。
完整射影二元群为 $C_k$。

另一方面，一个非零连通四端口构件的二部相等顶点数若为 $a,b$，
两侧的外端口数为 $u,v$，则度数求和给出
$6(a-b)=u-v$，而 $u+v=4$。
所以 $a=b$、$u=v=2$，四元支撑只能在 EO 层。
只有二元顶点的分量也不改变这一结论。
因此，即使原函数有两个非零端点，也未必能构造出双端点非零的四元函数。
两份 $=_6$ 用四条 $N$ 边相连，确实得到互补两点的 EO 四元函数，
这类四元函数由引理~\ref{zh:lem:cyc-complement-carrier} 处理。
\end{example}

\subparagraph*{合并上下集的结论}

上集降元与下集四元分类合在一起，得到两种边模型的完整二分。

\begin{corollary}\label{zh:cor:cyc-equality-dichotomy}
在本节等边对角奇阶 $C_k$ 的二元封闭条件下，$k\ge3$，
所有有限代数复值偶元函数集均有 $\FP/\sharpP$-hard 二分。
\end{corollary}
\begin{proof}
若全部输入函数是实际二元配对积，展开这些固定分解后，
每个连通分量是一个二元矩阵圈；按方向乘矩阵再取迹，便可多项式时间求值。
自环和重边包含在同一算法中。
否则由引理~\ref{zh:lem:cyc-diagonal-arity} 得到四元非配对积构件，
再用命题~\ref{zh:prop:cyc-equality-q4} 完成。
\end{proof}

\begin{corollary}[原生不等边高阶循环分支的完整二分]\label{zh:cor:native-cyclic-dichotomy}
设 $k\ge3$，有限代数复值偶元函数集 $\Gamma$ 在原生 $N$ 边模型中满足本节的二元封闭条件。
则整个 $\Gamma$ 问题有可判定的 $\FP/\sharpP$-hard 二分。
\end{corollary}
\begin{proof}
如果所有输入函数均属于实际二元配对积类，展开这些固定分解。
每个连通分量是一个二元函数的圈，其值是交替排列实际二元矩阵与边核 $N$ 后所得乘积的迹。
因此整个问题多项式时间可算；零函数、零元标量、自环和重边均按同一规则处理。

否则，某个原输入已是非配对积函数。若存在普通可实现的四元非配对积函数，
直接应用命题~\ref{zh:prop:cyc-native-q4}。
余下所有普通四元构件都是零或本群二元配对积，
在当前固定函数集的普通非配对积构件中取元数最小的 $F$。
它确实存在，因为上述输入函数属于这个集合；
二元封闭条件与四元假设又保证其偶元数至少为六。
定理~\ref{zh:thm:native-cyclic-shape} 给出
$F=Q+a\delta_{0^{2d}}+b\delta_{1^{2d}}$，其中 $Q=0$，
或者 $Q$ 是具有 $\mu_k$ 因子比例的实际反对角匹配积。
若 $a=b=0$，则 $F$ 为零或属于实际二元配对积类，违反选取。
引理~\ref{zh:lem:native-cyclic-two-endpoints} 排除 $ab\ne0$，
因此恰有一个端点非零。
无论 $Q$ 是否为零，引理~\ref{zh:lem:native-cyclic-endpoint-interpolation}
均给出原问题的二分。

最后说明如何判定。二元配对积的成员关系可用有限匹配枚举和代数数运算判定。
对已经加入本群二元方向的固定函数集，若输入函数不全在配对积类中，
就枚举普通构件，寻找以下两种函数之一：
四元非配对积构件，或具有式~\eqref{zh:eq:native-cyclic-tail-shape} 形式、
至少一个端点非零且 $Q$ 满足上述条件的构件。
前面的存在性证明保证这项枚举终止。
第一类证书使用四元二分；第二类若有双端点，
引理~\ref{zh:lem:native-cyclic-two-endpoints} 的两份构件连接直接构造四元非配对积，
而单端点直接使用插值引理。各证书条件通过有限支撑、匹配和根单位检查即可判定，
随后外部二分的完整函数集谓词也是可判定的。
因此这一判定过程会终止，并给出原问题的复杂性二分。
\end{proof}

\addtocontents{toc}{\protect\setcounter{tocdepth}{2}}

\subsection{奇二面体群}\label{zh:subsec:road-odd}
奇二面体群的二分定理另可参见 Jiayi Zheng 的草稿\cite{zh:zheng-draft}。本附录保留一份AI证明的过程。

对尚不能取得二元不等函数的奇二面体群，先沿原有的
“正交三角化、再由转置封闭迫使对角化”步骤取得规范形，
然后使用对角与反对角辅助函数降至四元，最后给出整个函数集的复杂性二分。

\subsubsection{无二元不等函数时的规范形}
\begin{lemma}\label{zh:lem:odd-dihedral-normal}
若 $\mathcal G$ 是阶为 $2n$ 的二面体群，$n\ge3$ 为奇数，
且引理~\ref{zh:lem:N-symmetric} 的条件不成立，则存在复正交基，使
\begin{equation}\label{zh:eq:odd-normal}
 \mathcal B=\langle D,Y\rangle,\qquad
 D=\diag(\lambda,\lambda^{-1}),\qquad \lambda=e^{\pi i/n}.
\end{equation}
\end{lemma}
\begin{proof}
写 $\mathcal G=\langle r,s:r^n=s^2=1,\ srs=r^{-1}\rangle$。
非单位二阶元素恰为 $s,sr,\ldots,sr^{n-1}$，共 $n$ 个。
转置在这些陪集上是一个对合，故至少固定一个。
其提升 $S$ 的特征值为 $i,-i$，且 $S^{\mathsf T}=S$ 或 $S^{\mathsf T}=-S$。
前一情形会给出 $N$，故只能有 $S=\pm Y$，从而 $Y\in\mathcal B$。

循环子群 $C=\langle r\rangle$ 恰由所有奇数阶元素组成，所以转置保持 $C$。
取提升 $R$ 使 $R^n=-I$；$C$ 的完全逆像为 $\langle R\rangle$。
若 $R$ 的两个特征方向都迷向，则它们是
$\mathbb C(1,i)^{\mathsf T}$ 和 $\mathbb C(1,-i)^{\mathsf T}$，
也是 $Y$ 的两个特征方向，所以 $RY=YR$。
但在奇二面体群中，$[Y]$ 是循环子群外的二阶元素，
它把 $r$ 共轭为 $r^{-1}\ne r$，矛盾。
因此 $R$ 有非迷向特征向量 $u$，可补成复正交基。
在此基下 $R$ 为上三角，$\langle R\rangle$ 的每个元素也为上三角。
这个循环子群仍对转置封闭，故 $R^{\mathsf T}$ 也为上三角；于是 $R$ 对角。
重选循环生成元，把它写成式~\eqref{zh:eq:odd-normal} 的 $D$。
二维正交矩阵 $M$ 满足 $M^{\mathsf T}YM=(\det M)Y$，所以 $Y$ 仍在群中。
$[D]$ 与 $[Y]$ 已生成整个商群，二者及 $-I=D^n$ 因而生成 $\mathcal B$。
\end{proof}

这里先对循环子群的完全逆像作正交三角化，再由这个子群的转置封闭性得到对角化。

\subsubsection{奇二面体群的直接降元}
\label{zh:sec:odd-arity}

本节采用相等边模型。设全部可实现的非零二元函数均为满秩对角或反对角矩阵，
其行列式归一化矩阵集合记为 $\mathcal B$，并保留正负两种代表；相应商群为
$\mathcal G=\mathcal B/\{I,-I\}$。假设可实现三个互不成比例的对角辅助函数，
忽略可补偿的非零整体标量，写成
\begin{equation}\label{zh:eq:odda-three-probes}
 L_j=\operatorname{diag}(1,t_j),\qquad
 t_j\ne0,\qquad t_i\ne t_j\quad(i\ne j),\qquad j=1,2,3.
\end{equation}
另假设可实现一个可逆反对角端口作用及其逆作用。这些操作保持实际二元函数集合。
奇二面体标准形
\[
 \mathcal B=\langle D,Y\rangle,\qquad
 D=\operatorname{diag}(\lambda,\lambda^{-1}),\qquad
 \lambda=e^{\pi i/n},\quad n\ge3\text{ 为奇数},\qquad
 Y=\begin{pmatrix}0&1\\-1&0\end{pmatrix}
\]
满足这些条件：三个对角方向可取自 $I,D,D^2$，反对角作用可取 $Y$，其逆为 $-Y$。

记 $\mathcal P_{\mathcal B}$ 为零函数以及如下配对积：在一个完美匹配上，把实际
可实现的满秩二元函数作张量积，再乘非零整体标量。即其非零 $2d$ 元成员具有形式
\[
 C(x_1,\ldots,x_{2d})=
 c\prod_{\{u,v\}\in M}B_{uv}(x_u,x_v),\qquad c\ne0.
\]
每个二元因子支持一对相等输入或一对不等输入，故
\begin{equation}\label{zh:eq:odda-product-support-size}
 |\operatorname{supp}(C)|=2^d.
\end{equation}
零元常数按空配对积处理。

证明的步骤与高阶循环群相同。先由三个短路的线性关系得到共同配对，
再恢复 $00,11$ 两个切片的结构，最后合并三个视角，得到两个犀牛角和一个羚羊角，
把整个函数拼回配对积。这里增加的反对角辅助用于翻位及排除四变量例外。
切片和线性组合用来分析函数值，构件按指定的二元辅助连接。

\paragraph{配对积的三项关系}

\begin{lemma}[共同配对或四变量例外]\label{zh:lem:odda-three-products}
设 $U,V,W$ 是由满秩对角或反对角二元函数组成的非零 $2d$ 元配对积，满足
\[
 aU+bV+cW=0,\qquad abc\ne0,
\]
并且两两不成比例。则存在以下两种情形之一：
\begin{enumerate}
 \item 三者具有相同支撑，因而具有相同匹配及每条边的相等或不等约束；
 \item 三个匹配只在四个变量上不同，在这四个变量上采用三种不同匹配，外部变量
 上可提出同一个非零配对积因子 $Q$。对四个变量作适当翻位后，三个局部支撑的并
 恰为全部六个重量为二的输入；三者张成的二维空间中，除这三个射影点外，每个非零
 局部函数都在这六个位置全部非零。
\end{enumerate}
\end{lemma}

\begin{proof}
满秩单项二元配对积的支撑决定匹配和边约束：对任意一个变量，只有它的配对变量
与其恒等或恒反，其余配对中的自由比特独立变化。若 $U,V$ 支撑相同，线性关系使
$\operatorname{supp}(W)$ 包含于这一支撑；三者支撑大小相同，故必然相等。

若 $U,V$ 支撑不同，非退化关系给出
\[
 \operatorname{supp}(U)\mathbin{\triangle}\operatorname{supp}(V)
 \subseteq\operatorname{supp}(W).
\]
由式~\eqref{zh:eq:odda-product-support-size}，
\begin{equation}\label{zh:eq:odda-half-intersection}
 |\operatorname{supp}(U)\cap\operatorname{supp}(V)|\ge2^{d-1}.
\end{equation}
两个匹配的并分解成交替圈，公共边视为含两条重边的圈。若圈上的相等、不等约束不相容，
交集为空；若相容，每个圈只留下一个自由比特。一个含 $2r$ 个顶点的圈，在单个匹配中
留下 $r$ 个自由比特，在交集中只留下一个。因此不同匹配的交集至多为 $2^{d-1}$，
等号仅在恰有一个四变量交替圈、其余边及约束完全相同时成立。
同一匹配的不同边约束给出不相交支撑，也不能满足式~\eqref{zh:eq:odda-half-intersection}。

所以交集大小恰为 $2^{d-1}$，而 $W$ 的支撑恰为上述对称差。四变量以外的匹配不变，
四变量内部则必须采用第三种匹配。局部交集有两个输入；在任意一个这样的输入上，
$W$ 为零而 $U,V$ 非零，代入关系即知 $U,V$ 的外部乘积因子成比例。在局部对称差的
任意一个输入上再代入关系，得到 $W$ 的外部因子也与它们成比例。于是可提出共同因子
$Q$。

为说明局部支撑的标准形式，把前两个局部匹配记为 $12|34$ 与 $13|24$。
选它们共同支撑中的一个输入，逐位翻成 $0000$，就把前两个匹配的约束都变成相等；
第三个匹配 $14|23$ 的约束因此都是不等。再翻转第 2、3 个变量，三个局部支撑便成为
\[
 \operatorname{supp}(Y_{12}Y_{34}),\quad
 \operatorname{supp}(Y_{13}Y_{24}),\quad
 \operatorname{supp}(Y_{14}Y_{23}),\qquad
 Y=\begin{pmatrix}0&1\\-1&0\end{pmatrix}.
\]
这里仅用这些函数表示支撑，不要求局部非零值与 $Y$ 的乘积相同。三个支撑的并是
四位串中全部六个重量为二的输入。

最后看局部二维线性空间。两个独有支撑部分在非退化线性组合中不会消失；它们的两个
交集位置则必须在同一个参数下同时相消，否则不会存在第三个配对积。因此除前两个
函数自身及这个唯一相消方向之外，其余非零组合都在六个位置全部非零。
\end{proof}

四变量例外本身确实存在，例如
\[
 Y_{12}Y_{34}-Y_{13}Y_{24}+Y_{14}Y_{23}=0.
\]
下一步利用不同端口的短路排除这个四变量例外。

\paragraph{换端口排除四变量例外}

以下几个引理共同考虑一个非零可实现的 $2m$ 元函数 $F$，其中 $m\ge3$，并假设
所有可实现的更低偶数元函数均属于 $\mathcal P_{\mathcal B}$。这项假设将在最终的
最小反例证明中满足。

\begin{lemma}[三个短路的配对一致]\label{zh:lem:odda-cap-consistency}
固定 $F$ 的任意两个端口 $p,q$。用式~\eqref{zh:eq:odda-three-probes} 中的三个对角
辅助短路后，所有非零结果具有相同支撑，因而采用相同匹配及边约束。
\end{lemma}

\begin{proof}
记四个数学切片为
\[
 A=F_{pq}^{00},\quad B=F_{pq}^{11},\quad
 C=F_{pq}^{01},\quad E=F_{pq}^{10}.
\]
三个实际短路结果为 $H_j=A+t_jB$。若 $A,B$ 相关，所有非零结果成比例，结论成立。
若二者无关，三个结果都非零、两两不成比例，且满足三系数均非零的线性关系。
低元假设使每个 $H_j$ 都是配对积，因此可以使用
引理~\ref{zh:lem:odda-three-products}。

假设出现其四变量例外，局部变量记为 $x_1,\ldots,x_4$，其余变量记为 $z$。
提出公共外因子 $Q(z)$，并令
\[
 \mathcal S=\operatorname{supp}(Q),\qquad s=|\mathcal S|=2^{m-3}.
\]
用已有的可逆反对角作用实施引理中的局部翻位。这些作用只是给输入翻位并乘非零权重，
保持支撑大小、配对积归属及低元假设。仍沿用原来的字母记号。
三个特殊支撑的并此时为
\[
 E_2\times\mathcal S,\qquad
 E_2=\{x\in\{0,1\}^4:|x|=2\}.
\]
$A,B$ 在该二维空间中分别对应参数 $0,\infty$，而三个配对积参数 $t_j$ 非零且有限。
因此引理~\ref{zh:lem:odda-three-products} 的末项说明：$A,B$ 恰在全部 $6s$ 个位置
非零，其余位置为零。

固定 $i\in\{1,2,3,4\}$，改在端口 $p,x_i$ 上使用同样三个对角辅助，所得函数为
$K_j^{(i)}$。其剩余端口为 $q,x_{\ne i},z$，并且
\begin{align*}
 K_j^{(i)}(q=0)&=A(x_i=0)+t_jE(x_i=1),\\
 K_j^{(i)}(q=1)&=C(x_i=0)+t_jB(x_i=1).
\end{align*}
考虑位置集合
\begin{equation}\label{zh:eq:odda-anchor-set}
 q+|x_{\ne i}|=2,\qquad z\in\mathcal S.
\end{equation}
它有 $6s$ 个位置。在每个位置，上述关于 $t_j$ 的一次式至少有一个已知非零系数：
$q=0$ 时来自 $A$，$q=1$ 时来自 $B$。每个一次式在三个不同参数下至多一次为零，
所以
\begin{equation}\label{zh:eq:odda-support-lower-bound}
 \sum_{j=1}^{3}|\operatorname{supp}(K_j^{(i)})|\ge12s.
\end{equation}
另一方面，各结果为严格低元的可实现函数，故为零或配对积。其元数是 $2m-2$，
非零时支撑大小为 $2^{m-1}=4s$，所以式~\eqref{zh:eq:odda-support-lower-bound}
两边必须取等号。于是式~\eqref{zh:eq:odda-anchor-set} 以外的所有位置上，三个结果
都为零；一次式在三个不同参数上为零，两个系数也都为零。

遍历四个 $i$，由此得到：$C$ 的非零输入若有 $x_i=0$，则 $|x|=1$；
$E$ 的非零输入若有 $x_i=1$，则 $|x|=3$。因此它们在局部四变量上的可能支撑满足
\[
 \operatorname{supp}_x(C)\subseteq E_1\cup\{1111\},\qquad
 \operatorname{supp}_x(E)\subseteq E_3\cup\{0000\},
\]
其中 $E_r=\{x:|x|=r\}$，$\operatorname{supp}_x$ 表示把非零支撑投影到这四个变量。
这些包含关系没有对外部变量作额外假设，故也包括可能出现在 $\mathcal S$ 外的输入。

再交换 $p,q$，对端口 $q,x_i$ 作同样的三个短路。完全相同的计数给出相反限制
\[
 \operatorname{supp}_x(C)\subseteq E_3\cup\{0000\},\qquad
 \operatorname{supp}_x(E)\subseteq E_1\cup\{1111\}.
\]
两组允许集合的交均为空，故 $C=E=0$。但此时每个 $K_j^{(i)}$ 都在
式~\eqref{zh:eq:odda-anchor-set} 的全部 $6s$ 个位置非零，超过配对积允许的 $4s$，
矛盾。四变量例外不可能出现。
\end{proof}

\paragraph{两个切片的共同因子}

在需要的端口接入可逆反对角辅助，可把 $F$ 的某个非零输入翻到全零。
这一操作及其逆都可实现，故保持 $\mathcal P_{\mathcal B}$ 归属及低元假设。
下面不妨已安排
\begin{equation}\label{zh:eq:odda-all-zero-nonzero}
 F(0^{2m})\ne0.
\end{equation}

\begin{lemma}[至多一个因子改变]\label{zh:lem:odda-one-changing-factor}
设三个非零函数在同一个变量分组上都为张量积，三者两两不成比例，且满足非退化
三项线性关系。则任取其中两个，其不成比例的因子至多出现在一个变量组上。
\end{lemma}

\begin{proof}
若至少两个变量组上的因子不成比例，取其中一个组与其余变量作矩阵展平。
前两函数分别为 $uv^{\mathsf T}$、$xy^{\mathsf T}$，其中 $u,x$ 无关，$v,y$ 也无关。
两个系数非零的和因此秩为二，不可能是第三个函数的秩一展平。
\end{proof}

\begin{lemma}[切片的共同外因子]\label{zh:lem:odda-common-slice-factor}
对任意端口 $p,q$，可在剩余变量中选一对 $z,w$，写成
\begin{equation}\label{zh:eq:odda-slice-factorization}
 F_{pq}^{00}=E_0(z,w)Q,\qquad
 F_{pq}^{11}=E_1(z,w)Q.
\end{equation}
这里 $Q$ 是实际可实现的满秩对角二元因子的配对积，$E_0,E_1$ 是暂时允许退化的
对角矩阵，且 $E_0(00)\ne0$。
\end{lemma}

\begin{proof}
三个短路在剩余变量全零处的值为
\[
 F_{pq}^{00}(0)+t_jF_{pq}^{11}(0).
\]
第一项非零，故至少两个结果的支撑含全零。满秩单项二元配对积的支撑若含全零，
它的所有二元因子必为对角函数。

若两个切片相关，所有非零短路与上述纯对角配对积成比例，结论立即成立。
若两个切片无关，三个短路均非零，且由引理~\ref{zh:lem:odda-cap-consistency}
具有同一支撑和匹配，因此三者都是纯对角配对积。再用
引理~\ref{zh:lem:odda-one-changing-factor}，至多一个配对上的因子不成比例。
把其余公共因子提出为 $Q$，并从两个短路的线性方程恢复切片，得到
式~\eqref{zh:eq:odda-slice-factorization}。可取 $Q$ 的因子为某个实际短路结果的
分解因子，故它们确实来自实际二元函数集合。最后由
式~\eqref{zh:eq:odda-all-zero-nonzero} 得 $E_0(00)\ne0$。
\end{proof}

\begin{lemma}[完整配对与互补钉位配对]\label{zh:lem:odda-two-slice-forms}
在式~\eqref{zh:eq:odda-slice-factorization} 中，恰有以下两种情形之一：
\begin{enumerate}
 \item $E_0,E_1$ 均满秩且相互成比例；两个切片具有同一完整对角配对支撑。
 \item $E_0$ 只在 $00$ 非零，$E_1$ 只在 $11$ 非零；两个切片在同一对 $z,w$
 上分别钉为 $00$、$11$，其余配对因子均满秩。
\end{enumerate}
特别地，两个切片都非零。
\end{lemma}

\begin{proof}
对 $Q$ 的每个满秩对角因子，选一个 $L_j$ 使其双线性闭合值非零。
这种闭合值是参数 $t_j$ 的非零一次式，故至多一个探针使它为零，总能选择。
在原函数 $F$ 的相应外部配对上连接这些辅助，得到实际四元函数 $H(p,q,z,w)$。
其两个切片为同一个非零标量乘 $E_0,E_1$，并且 $H(0000)\ne0$。

四元严格低于 $2m$，故 $H\in\mathcal P_{\mathcal B}$。又因其支撑含全零，
它是纯对角配对积。若匹配为 $pq|zw$，两个切片都是同一个满秩 $zw$ 因子的非零倍数。
若匹配为 $pz|qw$ 或 $pw|qz$，钉 $p=q=0$ 时只留下 $z=w=0$，钉 $p=q=1$ 时
只留下 $z=w=1$，相应值均非零。这恰给出两种情形。
\end{proof}

为沿用切片结构的直观名称，称第一种完整配对切片为\emph{羚羊}；称第二种在一对变量上
固定为 $00$ 或 $11$ 的切片为\emph{0 犀牛}或\emph{1 犀牛}，把这一对变量称为犀牛角。
以下用羚羊和犀牛指代这两种切片结构。
第一种情形中，两个切片实际成比例于某个非零短路结果，所以其全部配对因子也可取为
已有满秩对角函数；第二种情形中，犀牛角以外的公共因子由
引理~\ref{zh:lem:odda-common-slice-factor} 保证来自已有二元函数。

\paragraph{三个视角的配对拼接}

\begin{lemma}[四变量内部配对]\label{zh:lem:odda-three-views}
存在四个不同的变量 $x,y,z,w$，使三个切片
\[
 F_{xy}^{00},\qquad F_{xz}^{00},\qquad F_{yz}^{00}
\]
分别把 $z,y,x$ 与 $w$ 配对。若其中某个切片是犀牛，其犀牛角就是这份内部配对。
三个切片在其余变量上的匹配一致，且外部乘积因子相互成比例。
\end{lemma}

\begin{proof}
取任意 $F_{xy}^{00}$。若它是羚羊，任选剩余变量 $z$，令 $w$ 为其伙伴；
若它是犀牛，令 $z,w$ 为犀牛角。无论哪种情形，在共同切片 $F_{xyz}^{000}$ 的
剩余变量中，恰好只有 $w$ 被迫取零，其他变量仍按满秩对角因子自由配对。

现在从 $F_{xz}^{00}$ 的视角看这个切片，即再把 $y$ 钉为零。若该视角是羚羊，
唯一进一步固定的变量必是 $y$ 的伙伴，所以伙伴为 $w$。若该视角是犀牛而 $y$
不在犀牛角上，钉 $y$ 会使犀牛角的两个变量及 $y$ 的伙伴共三个剩余变量被迫取零，
矛盾。因此 $y$ 必在犀牛角上，另一端就是 $w$。对 $F_{yz}^{00}$ 同理，$x$ 必与
$w$ 配对，若有犀牛角也只能是 $x,w$。

再将各视角在 $x,y,z,w$ 中的剩余两个变量钉为零，得到同一个非零函数
$F_{xyzw}^{0000}$。它在外部变量上是满秩对角配对积，支撑决定匹配，所以三个外部匹配
一致；函数值本身相同又说明三份外部乘积因子成比例。
\end{proof}

\begin{lemma}[至少两个犀牛视角]\label{zh:lem:odda-two-rhinos}
引理~\ref{zh:lem:odda-three-views} 的三个切片中，至少两个是 0 犀牛。
\end{lemma}

\begin{proof}
若至少两个是羚羊，重命名 $x,y,z$ 后，不妨为 $F_{xy}^{00}$ 与 $F_{xz}^{00}$。
它们的外部因子成比例，所以可选择同一组对角辅助使该公共因子非零闭合。
在原函数 $F$ 上连接这些辅助，得到实际四元函数 $H(x,y,z,w)$，并有
\[
 H(0000)\ne0,\qquad H(0011)\ne0,\qquad H(0101)\ne0.
\]
低元假设与 $H(0000)\ne0$ 迫使 $H$ 为纯对角配对积。但是三种四变量匹配中，
$xy|zw$ 不允许 $0101$，$xz|yw$ 不允许 $0011$，而 $xw|yz$ 对这两点都不允许。
矛盾。
\end{proof}

\begin{theorem}[降到四元非配对积]\label{zh:thm:odd-arity-reduction}
在本节的二元刚性和辅助函数条件下，若可实现某个偶数元函数不属于
$\mathcal P_{\mathcal B}$，则可实现一个四元函数也不属于
$\mathcal P_{\mathcal B}$。
\end{theorem}

\begin{proof}
反设没有这样的四元函数。取所有可实现的非 $\mathcal P_{\mathcal B}$ 函数中元数最小者
$F$。零元、二元和四元都已排除，故 $F$ 的元数为 $2m\ge6$，所有可实现的更低偶数元
函数均属于 $\mathcal P_{\mathcal B}$。使用可逆反对角端口作用，使
$F(0^{2m})\ne0$；撤销这些作用不改变 $\mathcal P_{\mathcal B}$ 归属。

应用上述引理得到四个变量 $x,y,z,w$。重命名 $x,y,z$，使
$F_{xz}^{00}$ 与 $F_{yz}^{00}$ 是 0 犀牛，其角分别为 $y,w$ 与 $x,w$。
由引理~\ref{zh:lem:odda-two-slice-forms} 的互补 $11$ 切片，在 $F$ 的每个非零输入上：
\begin{itemize}
 \item 若 $x=z$，则 $x=y=z=w$；
 \item 若 $y=z$，则 $x=y=z=w$；
 \item 若 $x=y$，由第三个视角的 $00$ 或 $11$ 切片，必有 $z=w$。
\end{itemize}
三个比特 $x,y,z$ 必有两个相等，故整个支撑满足
\begin{equation}\label{zh:eq:odda-final-support}
 x=y,\qquad z=w.
\end{equation}

再看 $F_{xy}^{00}$ 与 $F_{xy}^{11}$。若它们是羚羊，两者相互成比例，且 $z,w$ 配对；
若它们是互补犀牛，其角就是 $z,w$。两种情形都可将实际二元函数构成的共同外因子
提出，写成
\[
 F_{xy}^{00}=R_0(z,w)Q,\qquad
 F_{xy}^{11}=R_1(z,w)Q.
\]
式~\eqref{zh:eq:odda-final-support} 已排除 $F_{xy}^{01}$ 与 $F_{xy}^{10}$，因此这两个
切片的公共因子也是整个函数的因子：
\[
 F=J(x,y,z,w)Q.
\]
在 $Q$ 的每个配对上选择一个闭合值非零的实际对角辅助，便普通实现 $J$ 的非零标量倍。
$J$ 为严格低元的四元函数，所以由低元假设属于 $\mathcal P_{\mathcal B}$。
$Q$ 本来就是实际二元配对积，故 $F\in\mathcal P_{\mathcal B}$，与选择矛盾。
\end{proof}

\subsubsection{奇二面体群的四元基例与复杂性出口}
\label{zh:sec:oddq}

本节始终采用相等边模型和普通转置。设全部可实现的非零二元函数经行列式归一化、
并保留两种正负代表后组成
\[
 \mathcal B=\langle D,Y\rangle,\qquad
 D=\operatorname{diag}(\lambda,\lambda^{-1}),\qquad
 \lambda=e^{\pi i/n},\qquad n\ge3\text{ 为奇数},\qquad
 Y=\begin{pmatrix}0&1\\-1&0\end{pmatrix}.
\]
其射影商群为 $\mathcal G=\mathcal B/\{I,-I\}$。本节处理的正是没有
$N=\left(\begin{smallmatrix}0&1\\1&0\end{smallmatrix}\right)$ 的奇二面体标准形。
所有实际构件允许带已知的非零整体标量，使用时按构件出现次数补偿。

记 $\mu_n$ 为全部 $n$ 次单位根。两类实际可用的二元探针，忽略非零标量，可写成
\begin{equation}\label{zh:eq:oddq-probes}
 L_\rho=\begin{pmatrix}1&0\\0&\rho\end{pmatrix},\qquad
 K_\rho=\begin{pmatrix}0&1\\-\rho&0\end{pmatrix},
 \qquad \rho\in\mu_n.
\end{equation}
因此，非零对角二元函数的两个非零值之比属于 $\mu_n$；非零反对角二元函数的
下、上非零值之比属于 $-\mu_n$。后一情形的两个值的 $n$ 次幂互为负数。

记 $\mathcal P_{\mathcal B}$ 为零函数以及在任意完美匹配上由实际可实现的满秩二元
函数作张量积、再乘非零标量得到的函数。另以 $\mathcal P$ 表示通常的 CSP 乘积型类，
其定义允许一元函数、二元相等和二元不等的因子共享变量。两类不同；例如四元相等
$\mathrm{EQ}_4$ 属于 $\mathcal P$，却不属于 $\mathcal P_{\mathcal B}$。

\paragraph{短路探针与四元函数的总奇偶}

\begin{lemma}[固定总奇偶]\label{zh:lem:oddq-parity}
在上述二元刚性条件下，每个非零可实现四元函数的支撑具有固定的总奇偶。
在奇支撑函数的任意一个端口连接 $Y$，可将其变为偶支撑函数；这一可逆的实际端口操作
保持函数是否属于 $\mathcal P_{\mathcal B}$。
\end{lemma}

\begin{proof}
令 $V_0,V_1$ 分别为所有对角、反对角二元矩阵构成的二维空间，故
$V_0\cap V_1=\{0\}$。固定四元函数的一份 $2+2$ 配对，对其中一对端口短路给出
线性映射 $T:\mathbb C^{2\times2}\to\mathbb C^{2\times2}$。

首先证明 $T(V_0)$ 包含于某个 $V_i$。秩零时显然。秩一时，式
\eqref{zh:eq:oddq-probes} 中的 $L_\rho$ 张成 $V_0$，所以至少一个探针有非零像；
这个允许的二元像属于某个 $V_i$，并且张成整个像直线。秩二时，$T|_{V_0}$ 单射，
至少三个不同的探针射影方向给出三个不同的非零像方向。若二维像空间既不包含于
$V_0$ 也不包含于 $V_1$，它与两者的交各至多为一维，至多能包含两个允许的射影方向，
矛盾。这里秩一情形已包含部分短路结果为零的情况；秩二情形没有非零探针落在核中。

对 $K_\rho$ 作同样论证，得 $T(V_1)$ 也包含于某个 $V_i$。再从另一对端口短路，
得到相反方向的同一结论。于是把四元展平的行、列各按奇偶分为两组后，每组行和每组列
都至多有一个非零的奇偶块。这个 $2\times2$ 块表的非零块构成部分匹配：两个非零块
要么都在对角位置，要么都在反对角位置；只有一个非零块时也有固定总奇偶。

端口上连接 $Y$ 只翻转该输入并乘非零权重，其逆操作由 $Y^{-1}=-Y$ 实现。
它改变总奇偶，并将任何二元因子变为已有矩阵的乘积，必要时带普通转置。
因此它双向保持 $\mathcal P_{\mathcal B}$ 归属。
\end{proof}

\paragraph{根单位比例与六个奇偶块}

\begin{lemma}[双端口根单位比例]\label{zh:lem:oddq-root-ratio}
设 $n\ge3$，且
\[
 Q=\begin{pmatrix}u&v\\w&z\end{pmatrix}.
\]
若对每个 $\rho\in\mu_n$ 都有
\begin{equation}\label{zh:eq:oddq-root-two-directions}
 (u+\rho w)^n=(v+\rho z)^n,\qquad
 (u+\rho v)^n=(w+\rho z)^n,
\end{equation}
则恰有下列四种情形之一：
\begin{enumerate}
 \item 四项全零；
 \item $v=w=0$，$u,z\ne0$，且 $u^n=z^n$；
 \item $u=z=0$，$v,w\ne0$，且 $v^n=w^n$；
 \item 四项全非零，$uz=vw$，且四项的 $n$ 次幂全部相等。
\end{enumerate}
反过来，这四种情形均满足式~\eqref{zh:eq:oddq-root-two-directions}。
\end{lemma}

\begin{proof}
在全部 $n$ 次单位根上展开并比较离散 Fourier 系数。常数项的两条等式为
\[
 u^n+w^n=v^n+z^n,\qquad
 u^n+v^n=w^n+z^n,
\]
所以 $u^n=z^n$、$v^n=w^n$。其余系数在约去非零的二项式系数后给出
\[
 u^{n-r}w^r=v^{n-r}z^r,\qquad
 u^{n-r}v^r=w^{n-r}z^r\quad(1\le r<n).
\]
零模式立即只有前三种或者四项全非零。全非零时，对第一族等式的 $r=1,2$ 两项
相除，得到 $uz=vw$，代回 $r=1$ 得 $u^n=v^n$。充分性逐项代入即可。
\end{proof}

由引理~\ref{zh:lem:oddq-parity}，以下不妨已把四元函数 $F$ 变为偶支撑。
采用自然展平：行指标为 $x_1x_2$，列指标为 $x_3x_4$，均按 $00,01,10,11$ 排列。
写作
\begin{equation}\label{zh:eq:oddq-eight-parameters}
 N_{12|34}(F)=
 \begin{pmatrix}
 s&0&0&e\\
 0&a&c&0\\
 0&d&b&0\\
 f&0&0&t
 \end{pmatrix}.
\end{equation}
即
\[
 s=F_{0000},\quad t=F_{1111},\quad
 a=F_{0101},\ b=F_{1010},\quad
 c=F_{0110},\ d=F_{1001},\quad
 e=F_{0011},\ f=F_{1100}.
\]
三个配对的偶块与奇块分别为
\begin{equation}\label{zh:eq:oddq-six-blocks}
\begin{array}{c|c|c}
 \text{配对}&\text{偶块}&\text{奇块}\\ \hline
 12|34&\begin{pmatrix}s&e\\f&t\end{pmatrix}
       &\begin{pmatrix}a&c\\d&b\end{pmatrix}\\[4pt]
 13|24&\begin{pmatrix}s&a\\b&t\end{pmatrix}
       &\begin{pmatrix}e&c\\d&f\end{pmatrix}\\[4pt]
 14|23&\begin{pmatrix}s&c\\d&t\end{pmatrix}
       &\begin{pmatrix}e&a\\b&f\end{pmatrix}
\end{array}
\end{equation}
偶块的行、列指标为 $00,11$，奇块的行、列指标为 $01,10$。

\begin{lemma}[奇偶块的比例约束]\label{zh:lem:oddq-block-rigidity}
式~\eqref{zh:eq:oddq-six-blocks} 中每个偶块满足引理~\ref{zh:lem:oddq-root-ratio}。
每个奇块 $Q=\left(\begin{smallmatrix}u&v\\w&z\end{smallmatrix}\right)$ 的变形
\[
 Q^\sharp=\begin{pmatrix}u&-v\\-w&z\end{pmatrix}
\]
也满足该引理。尤其，奇块四项全非零时必有
\begin{equation}\label{zh:eq:oddq-odd-block-sign}
 uz=vw,\qquad u^n=z^n=-v^n=-w^n.
\end{equation}
\end{lemma}

\begin{proof}
对偶块，从两端分别用 $L_\rho$ 短路，所得二元函数对角或为零，直接给出
式~\eqref{zh:eq:oddq-root-two-directions}。

对奇块，从两端分别用 $K_\rho$ 短路，所得二元函数反对角或为零，故
\[
 (u-\rho w)^n=-(v-\rho z)^n,\qquad
 (u-\rho v)^n=-(w-\rho z)^n.
\]
因为 $n$ 为奇数，这正是 $Q^\sharp$ 满足
式~\eqref{zh:eq:oddq-root-two-directions}。零短路仍满足这些等式。
\end{proof}

\paragraph{四元函数的穷尽分类}

\begin{theorem}[四元相等或实际二元乘积]\label{zh:thm:oddq-quaternary}
在本节的二元刚性条件下，任意非零可实现四元函数 $F$ 恰属于以下两种互斥情形之一：
\begin{enumerate}
 \item $F\in\mathcal P_{\mathcal B}$；
 \item $F$ 的支撑为两个互补输入，并且在其端口连接已有的 $Y$ 和 $D^j$ 后，可普通
 实现非零标量倍的四元相等函数 $\mathrm{EQ}_4$。
\end{enumerate}
\end{theorem}

\begin{proof}
先由引理~\ref{zh:lem:oddq-parity} 作可逆端口操作，得到式
\eqref{zh:eq:oddq-eight-parameters}。三个偶块的比例约束说明：$s,t$ 同时为零或同时
非零，三对 $(a,b),(c,d),(e,f)$ 也各自同时为零或同时非零。

若 $s,t$ 均非零，每个非零数对 $(p,q)$ 都满足
\[
 pq=st,\qquad p^n=q^n=s^n=t^n.
\]
若两个数对均非零，对应奇块却由式~\eqref{zh:eq:oddq-odd-block-sign} 要求两对的非零
$n$ 次幂互为负，矛盾。所以至多一个数对非零。

三对全零时，$F=s\lvert0000\rangle+t\lvert1111\rangle$，且 $s^n=t^n$。
恰一个数对非零时，置换端口后不妨为 $(e,f)$。此时 $st=ef$，并且
\[
 F(x_1,x_2,x_3,x_4)=A(x_1,x_2)B(x_3,x_4),\qquad
 A=\operatorname{diag}(1,f/s),\quad B=\operatorname{diag}(s,e).
\]
两个对角比都在 $\mu_n$ 中，所以两个因子确实是已有 $D^j$ 的非零标量倍。
这证明实际的 $\mathcal P_{\mathcal B}$ 归属，而不只是代数可分解。

再设 $s=t=0$。每个非零数对内部的 $n$ 次幂相等。三对全零会给出 $F=0$，已经排除。
恰一对非零时，置换端口后不妨为 $(e,f)$。在端口 3、4 各连接 $Y$，就把支撑
$\{0011,1100\}$ 变成 $\{0000,1111\}$，两非零值仍有相同的 $n$ 次幂。

恰两对非零时，置换端口后不妨为 $(a,b),(c,d)$。由奇块约束，
\[
 ab=cd,\qquad a^n=b^n=-c^n=-d^n.
\]
因此
\[
 F(x_1,x_2,x_3,x_4)=A(x_1,x_2)B(x_3,x_4),\qquad
 A=\begin{pmatrix}0&1\\d/a&0\end{pmatrix},\quad
 B=\begin{pmatrix}0&a\\c&0\end{pmatrix}.
\]
两个反对角比 $d/a,c/a$ 的 $n$ 次幂均为 $-1$，因 $n$ 为奇数，它们都属于
$-\mu_n$。所以两个因子分别是已有 $D^jY$ 的非零标量倍。

三对全非零时，记三对的共同 $n$ 次幂为 $A_0,C_0,E_0$，它们均非零。
三个奇块要求
\[
 A_0=-C_0,\qquad E_0=-C_0,\qquad E_0=-A_0.
\]
前两式给出 $A_0=E_0$，与第三式矛盾。这就穷尽了所有零模式。

剩下的非乘积情形经实际端口操作已成为
\[
 H=s\lvert0000\rangle+t\lvert1111\rangle,
 \qquad s,t\ne0,\qquad(s/t)^n=1.
\]
选 $j\in\{0,\ldots,n-1\}$ 使 $\lambda^{-2j}=s/t$，并在任一端口连接 $D^j$。
两个非零值分别变成 $s\lambda^j$ 和 $t\lambda^{-j}=s\lambda^j$，故普通实现
$s\lambda^j\mathrm{EQ}_4$。反向恢复开始时的可逆端口操作即可解释原函数。
两个情形互斥，因为其支撑大小分别为四和二。
\end{proof}

\begin{corollary}[直接降元的四元出口]\label{zh:cor:oddq-eq4-exit}
结合定理~\ref{zh:thm:odd-arity-reduction}，若存在可实现的偶元函数不在
$\mathcal P_{\mathcal B}$，则普通实现非零标量倍的 $\mathrm{EQ}_4$。
\end{corollary}

\begin{proof}
定理~\ref{zh:thm:odd-arity-reduction} 给出四元函数 $F\notin\mathcal P_{\mathcal B}$，再用
定理~\ref{zh:thm:oddq-quaternary}，便实现非零倍的 $\mathrm{EQ}_4$。
\end{proof}

\paragraph{整个函数集合的复杂性二分}

以下使用复值布尔 $\#\mathrm{CSP}^{2}$ 的集合级二分定理
\cite[Theorem~16]{zh:lin2021cspd}：固定有限代数复值函数集 $\mathcal H$ 的
$\#\mathrm{CSP}^{2}$ 问题，在整个集合包含于
$\mathcal P,\mathcal A,\mathcal A^\alpha,\mathcal L$ 中某一个时属于
$\mathrm{FP}$，否则为 $\#\mathrm P$-hard。这里 $\alpha^2=i$，
$M_\alpha=\operatorname{diag}(1,\alpha)$，
$\mathcal A^\alpha=\{M_\alpha^{\otimes k}g:g\in\mathcal A\}$。
局部仿射类 $\mathcal L$ 要求：对每个支撑点 $\sigma$，在第 $j$ 个端口施加
$\operatorname{diag}(1,\alpha^{\sigma_j})$ 后的函数都属于 $\mathcal A$。

\begin{lemma}[奇阶对角工具排除仿射易类]\label{zh:lem:oddq-exclude-affine}
本节的 $D$ 不属于 $\mathcal A\cup\mathcal A^\alpha\cup\mathcal L$。
\end{lemma}

\begin{proof}
它的两个非零值之比 $\rho=\lambda^{-2}$ 具有奇数阶 $n\ge3$，所以
$\rho\notin\mu_4$。仿射函数任意两个非零值之比在 $\mu_4$ 中，因此
$D\notin\mathcal A$。

若 $D\in\mathcal A^\alpha$，则
$(M_\alpha^{-1}\otimes M_\alpha^{-1})D$ 属于 $\mathcal A$，其两个非零值之比为
$\alpha^{-2}\rho=-i\rho$，仍不在 $\mu_4$，矛盾。
若 $D\in\mathcal L$，取支撑点 $00$，局部仿射定义直接要求未经改变的 $D$
属于 $\mathcal A$，也矛盾。
\end{proof}

\begin{theorem}[无不等函数的奇二面体标准形二分]\label{zh:thm:oddq-dichotomy}
设 $\mathcal F$ 是固定有限的代数复值布尔偶元函数集，并满足本节的二元群前提。
则在该标准坐标中，若 $\mathcal F\subseteq\mathcal P$，有
$\operatorname{Holant}(\mathcal F)\in\mathrm{FP}$；否则
$\operatorname{Holant}(\mathcal F)$ 为 $\#\mathrm P$-hard。
\end{theorem}

\begin{proof}
若 $\mathcal F\subseteq\mathcal P$，把乘积型约束展开为一元权重及相等、不等关系。
每个一致的连通分量只需合计两个可能赋值的权重；不一致的分量贡献零。因此有多项式
时间算法。

反之，$\mathcal F\not\subseteq\mathcal P$，而
$\mathcal P_{\mathcal B}\subseteq\mathcal P$，所以存在输入函数不在
$\mathcal P_{\mathcal B}$。推论~\ref{zh:cor:oddq-eq4-exit} 普通实现非零标量倍的
$\mathrm{EQ}_4$。

令 $\mathcal H=\mathcal F\cup\{D,Y\}$。新增函数已有固定实现构件，故
$\operatorname{Holant}(\mathcal H)\equiv_T\operatorname{Holant}(\mathcal F)$。
普通边提供 $\mathrm{EQ}_2$；把 $\mathrm{EQ}_{2k}$ 与 $\mathrm{EQ}_4$ 各取一个
端口相接就得到 $\mathrm{EQ}_{2k+2}$。因此每个出现偶数次的 CSP 变量都能用对应的
偶元相等构件连接，给出
\[
 \#\mathrm{CSP}^{2}(\mathcal H)
 \equiv_T\operatorname{Holant}(\mathcal H).
\]
相反方向直接把每条 Holant 边看作出现两次的 CSP 变量。所有构件中的已知非零标量
按出现次数补偿。

集合 $\mathcal H$ 不包含于 $\mathcal P$，而其成员 $D$ 由
引理~\ref{zh:lem:oddq-exclude-affine} 同时排除另外三个易类。集合级
$\#\mathrm{CSP}^{2}$ 二分定理遂给出困难性。
\end{proof}

若此前曾对原问题作统一复正交变换，应当对变换后处于本节标准形的整个函数集检验
$\mathcal F\subseteq\mathcal P$。

\paragraph*{第一条路线的接续}\label{zh:sec:scope}
以上三部分把可取得二元不等函数的群交给文献~\cite{zh:liu2026disequality}，
并展开一阶群、高阶循环群和奇二面体群的证明；
二阶循环群的结果独立引用文献~\cite{zh:zhou-draft}。
这些出口均针对整个函数集，取得一个四元函数只是归约中的中间步骤。
以后仍将按原有九类群框架继续第二条路线，逐类整理不调用含二元不等函数二分定理的证明。

\end{document}